\documentclass[12pt]{article}
\usepackage{fullpage,url,amsmath,amssymb,array,multirow,graphicx,microtype}
\newtheorem{theorem}{Theorem}

\newtheorem{lemma}[theorem]{Lemma}

\newtheorem{corollary}[theorem]{Corollary}
\newtheorem{claim}[theorem]{Claim}
\newtheorem{fact}[theorem]{Fact}

\newtheorem{remk}[theorem]{Remark}
\newtheorem{exmp}[theorem]{Example}

\def\FullBox{\hbox{\vrule width 8pt height 8pt depth 0pt}}

\def\qed{\ifmmode\qquad\FullBox\else{
\nobreak\hfil
\penalty50\hskip1em\null\nobreak\hfil\FullBox
\parfillskip=0pt\finalhyphendemerits=0\endgraf}\fi}

\def\qedsketch{\ifmmode\Box\else{\unskip\nobreak\hfil
\penalty50\hskip1em\null\nobreak\hfil$\Box$
\parfillskip=0pt\finalhyphendemerits=0\endgraf}\fi}

\newenvironment{proof}{\begin{trivlist} \item {\bf Proof:~~}}
  {\qed\end{trivlist}}

\newcommand{\etal}{{\it et~al.\ }}
\newcommand{\ie} {{\it i.e.,\ }}
\newcommand{\eg} {{\it e.g.,\ }}
\newcommand{\cf}{{\it cf.,\ }}

\newcommand{\eqdef}{\mathbin{\stackrel{\rm def}{=}}}
\newcommand{\R}{{\mathbb R}}
\newcommand{\N}{{\mathbb{N}}}
\newcommand{\Z}{{\mathbb Z}}
\newcommand{\F}{{\mathbb F}}
\newcommand{\poly}{{\mathrm{poly}}}
\newcommand{\polylog}{{\mathrm{polylog}}}
\newcommand{\loglog}{{\mathop{\mathrm{loglog}}}}
\newcommand{\zo}{\{0,1\}}
\newcommand{\suchthat}{{\;\; : \;\;}}
\newcommand{\pr}[2][]{\Pr_{#1}\left[#2\right]}
\newcommand{\deffont}{\em}
\newcommand{\getsr}{\mathbin{\stackrel{\mbox{\tiny R}}{\gets}}}
\newcommand{\Exp}{\mathop{\mathrm E}\displaylimits}
\newcommand{\Var}{\mathop{\mathrm Var}\displaylimits}
\newcommand{\xor}{\oplus}
\newcommand{\GF}{\mathrm{GF}}

\def\textprob#1{\textmd{\textsc{#1}}}
\newcommand{\mathprob}[1]{\mbox{\textmd{\textsc{#1}}}}
\newcommand{\SAT}{\mathprob{SAT}}
\newcommand{\QuadRes}{\textprob{Quadratic Residuosity}}
\newcommand{\CktApprox}{\mathprob{Circuit-Approx}}
\newcommand{\CktRelApprox}{\mathprob{Circuit-RelApprox}}
\newcommand{\DNFRelApprox}{\mathprob{DNF-RelApprox}}
\newcommand{\GraphNoniso}{\textprob{Graph Nonisomorphism}}
\newcommand{\GNI}{\mathprob{GNI}}
\newcommand{\GraphIso}{\textprob{Graph Isomorphism}}
\newcommand{\GI}{\mathprob{GI}}
\newcommand{\MinCut}{\textprob{Min-Cut}}
\newcommand{\MaxCut}{\textprob{Max-Cut}}
\newcommand{\IdentityTest}{\textprob{Identity Testing}}
\newcommand{\GraphConn}{\textprob{Graph Connectivity}}

\newcommand{\yes}{{\sc yes}}
\newcommand{\no}{{\sc no}}

\newcommand{\class}[1]{\mathbf{#1}}
\newcommand{\coclass}[1]{\mathbf{co\mbox{-}#1}} 
\newcommand{\BPP}{\class{BPP}}
\newcommand{\NP}{\class{NP}}
\newcommand{\RP}{\class{RP}}
\newcommand{\coRP}{\coclass{RP}}
\newcommand{\ZPP}{\class{ZPP}}
\newcommand{\RNC}{\class{RNC}}
\newcommand{\RL}{\class{RL}}
\renewcommand{\L}{\class{L}}
\newcommand{\coRL}{\coclass{RL}}
\newcommand{\IP}{\class{IP}}
\newcommand{\AM}{\class{AM}}
\newcommand{\MA}{\class{MA}}
\renewcommand{\P}{\class{P}}
\newcommand\prBPP{\class{prBPP}}
\newcommand\prRP{\class{prRP}}
\newcommand\prP{\class{prP}}
\newcommand{\Ppoly}{\class{P/poly}}
\newcommand{\DTIME}{\class{DTIME}}
\newcommand{\ETIME}{\class{E}}
\newcommand{\BPTIME}{\class{BPTIME}}
\newcommand{\EXP}{\class{EXP}}
\newcommand{\SUBEXP}{\class{SUBEXP}}
\newcommand{\qP}{\class{\tilde{P}}}
\newcommand{\PH}{\class{PH}}
\newcommand{\NC}{\class{NC}}
\newcommand{\PSPACE}{\class{PSPACE}}
\newcommand{\quasiP}{\class{\tilde{P}}}

\newcommand{\negl}{{\mathrm{neg}}}
\newcommand{\Diam}{\mathrm{Diam}}
\newcommand{\Cut}{\mathrm{Cut}}
\newcommand{\pf}{\mathit{pf}}
\newcommand{\Col}{\mathrm{Col}}
\newcommand{\Supp}{\mathrm{Supp}}

\newcommand{\accept}{\mathtt{accept}}
\newcommand{\reject}{\mathtt{reject}}
\newcommand{\fail}{\mathtt{fail}}
\newcommand{\halt}{\mathtt{halt}}

\newcommand{\HFam}{\mathcal{H}}
\newcommand{\FFam}{\mathcal{F}}
\newcommand{\Dom}{\mathcal{D}}
\newcommand{\Rng}{\mathcal{R}}

\newcommand{\Hall}{\mathrm{H}}
\newcommand{\Hmin}{\mathrm{H}_{\infty}}
\newcommand{\HRen}{\mathrm{H}_2}
\newcommand{\HSha}{\mathrm{H}_{\mathit{Sh}}}
\newcommand{\Ext}{\mathrm{Ext}}
\newcommand{\Con}{\mathrm{Con}}
\newcommand{\Samp}{\mathrm{Smp}}
\newcommand{\Enc}{\mathrm{Enc}}
\newcommand{\Code}{\mathcal{C}}

\newcommand{\zigzag}{\mathbin{\raisebox{.2ex}{
      \hspace{-.4em}$\bigcirc$\hspace{-.75em}{\rm z}\hspace{.15em}}}}

\newcommand{\eps}{\varepsilon}
\newcommand{\ci} {\stackrel{\rm{c}}{\equiv}}
\newcommand{\Time}{\mathrm{time}}
\newcommand{\inner}[2]{\left[#1,#2\right]}

\newcommand{\GL}{\operatorname{GL}}
\newcommand{\maj}{\operatornamewithlimits{maj}}
\DeclareMathOperator{\Tr}{Tr}

\newcommand{\field}[1]{\mathbb{#1}}

\newcommand{\E}{\mathbf{E}}
\newcommand{\equivalent}{\ensuremath{\Leftrightarrow}}

\newcommand{\eat}[1]{}
\newcommand{\myhrule}{\rule[.5pt]{\hsize}{.5pt}}

\newcommand{\dom}{\ensuremath{\mathcal{D}}}

\newcounter{ccc}
\newcommand{\bcc}{\setcounter{ccc}{1}\theccc}
\newcommand{\icc}{\addtocounter{ccc}{1}\theccc}
\newcommand{\bfromcc}[1]{\setcounter{ccc}{#1} \theccc}

\newcounter{rc}
\newcommand{\brc}{\setcounter{rc}{0}}
\newcommand{\irc}{\addtocounter{rc}{1} \therc.}
\newcommand{\stream}{\ensuremath{\mathcal{S}}}
\newcommand{\sign}{\text{sgn}}

\newcommand\var{\mbox{\rm Var}}
\newcommand{\widebar}[1]{\ensuremath{\overline{#1}}}
\providecommand{\abs}[1]{\lvert#1\rvert}
\providecommand{\norm}[1]{\left\lVert{#1}\right\rVert}
\newcommand{\normb}[1]{\bigl\lVert {#1} \bigr\rVert}
\newcommand{\normB}[1]{\Bigl\lVert {#1} \Bigr\rVert}
\newcommand{\normbb}[1]{\biggl \lVert {#1} \biggr\rVert}
\providecommand{\real}{\text{Re }}

\providecommand{\card}[1]{\bigl\lvert#1\bigr\rvert}
\newcommand{\cardB}[1]{\Bigl \lvert {#1} \Bigr\rvert}
\newcommand{\cardbb}[1]{\biggl \lvert {#1} \biggr \rvert}
\providecommand{\prob}[1]{\ensuremath{\text{Pr}\left[ #1\right]}}
\newcommand{\probb}[1]{\ensuremath{\text{Pr}\bigl[ #1\bigr]}}
\newcommand{\probB}[1]{\ensuremath{\text{Pr}\Bigl[ #1\Bigr]}}

\newcommand{\expect}[1]{\ensuremath{\E\left[#1 \right]}}
\newcommand{\expects}[1]{\ensuremath{\E[ {#1} ]}}
\newcommand{\expectb}[1]{\ensuremath{\E\bigl[ {#1} \bigr]}}
\newcommand{\expectB}[1]{\ensuremath{\E\Bigl[ {#1} \Bigr]}}
\newcommand{\expectbb}[1]{\ensuremath{\E\biggl[ {#1} \biggr]}}
\newcommand{\expectBB}[1]{\ensuremath{\E\Biggl[ {#1} \Biggr]}}

\renewcommand{\exp}[1]{\ensuremath{\textrm{exp}\left\{ {#1} \right\}}}
\newcommand{\eqclass}[1]{\ensuremath{\langle {#1} \rangle}}

\newcommand{\probsub}[2]{\ensuremath{\text{Pr}_{#1}\left[ #2 \right]}}
\newcommand{\probsubb}[2]{\ensuremath{\text{Pr}_{#1}\bigl[ #2 \bigr]}}
\newcommand{\probsubB}[2]{\ensuremath{\text{Pr}_{#1}\Bigl[ #2 \Bigr]}}
\newcommand{\probsubbb}[2]{\ensuremath{\text{Pr}_{#1}\biggl[ #2 \biggr]}}
\providecommand{\variance}[1]{\textrm{Var}\left[ {#1} \right]}
\newcommand{\varianceb}[1]{\ensuremath{\textrm{Var}\bigl[ {#1} \bigr]}}
\newcommand{\varianceB}[1]{ \ensuremath{\textrm{Var}\Bigl[ {#1} \Bigr]}}
\newcommand{\variancebb}[1]{ \ensuremath{\textrm{ Var}\biggl[ {#1} \biggr]}}

\providecommand{\covar}[1]{\text{Cov}\left[ {#1} \right]}
\newcommand{\covarb}[1]{\ensuremath{\textrm{Cov}\bigl[ {#1} \bigr]}}
\newcommand{\covarB}[1]{ \ensuremath{\textrm{Cov}\Bigl[ {#1} \Bigr]}}
\newcommand{\covarbb}[1]{ \ensuremath{\textrm{Cov}\biggl[ {#1} \biggr]}}

\providecommand{\expectsub}[2]{\ensuremath{\mathbf{E}_{{#1}}\left[ {#2} \right]}}
\providecommand{\varsub}[2]{\textrm{Var}_{{#1}}\left[ {#2} \right]}
\providecommand{\varsubb}[2]{\textrm{Var}_{{#1}}\bigl[ {#2} \bigr]}
\providecommand{\varsubB}[2]{\textrm{Var}_{{#1}}\Bigl[ {#2} \Bigr]}
\providecommand{\varsubbb}[2]{\textrm{Var}_{{#1}}\biggl {#2} \biggr]}

\providecommand{\expectsubb}[2]{\mathrm{E}_{{#1}}\bigl[ {#2} \bigr]}
\providecommand{\expectsubbb}[2]{\mathrm{E}_{{#1}}\biggl[ {#2} \biggr]}
\providecommand{\expectsubB}[2]{\mathrm{E}_{{#1}}\Bigl[ {#2} \Bigr]}
\newcommand{\tpest}{\textsf{TPest}}
\newcommand{\est}{\textsf{AvgEst}}
\newcommand{\D}{\ensuremath{\mathcal{D}}}
\newcommand{\countsketch}{\textsf{CountSketch}}

\newcommand{\Hss}{\textsf{Hss}}

\newcommand{\Card}[1]{\biggl\lvert {#1} \biggr \rvert}
\newcommand{\nocollision}{\textsc{nocollision}}

\newcommand{\ip}[2]{\ensuremath{\langle {#1}, {#2} \rangle}}

\newcommand{\nocoll}{\textsc{nocollision}}

\newcommand{\powlow}[2]{\ensuremath{ {#1}^{\underline{{#2}}}}}

\newcommand{\lmargin}{\text{lmargin}}
\newcommand{\rmargin}{\text{rmargin}}
\newcommand{\midreg}{\text{mid}}
\newcommand{\goodest}{\textsc{goodest}}
\newcommand{\accuest}{\textsc{accuest}}
\newcommand{\bin}{\text{Binom}}
\newcommand{\smallhh}{\textsc{smallhh}}

\newcommand{\ftwores}[1]{\ensuremath{F_2^{\text{res}}\left({#1}\right)}}

\newcommand{\1}{\textbf{1}}

\newcommand{\level}{\mathsf{level}}

\newcommand{\G}{\ensuremath{\mathcal{G}}}
\renewcommand{\H}{\ensuremath{\mathcal{H}}}

\newcommand{\bvtheta}{\ensuremath{\bar{\vartheta}}}

\newcommand{\smallH}{\textsc{small-h}}
\newcommand{\lastlevel}{\textsc{goodlastlevel}}
\newcommand{\smallU}{\textsc{small-u}}
\newcommand{\smallu}{\textsc{small-u}}
\newcommand{\goodtopk}{\textsc{goodtopk}}
\newcommand{\gooditems}{\textsc{gooditems}}
\newcommand{\discover}{\textsc{discover}}
\newcommand{\ftworesb}[1]{\ensuremath{F_2^{\text{res}}({#1})}}

\newcommand{\fastams}{\textsf{Fast-AMS}}

\newcommand{\goodsign}{\ensuremath{\G_{\text{Sgn}}}}
\newcommand{\goodH}{\ensuremath{\G_H}}

\newcommand{\goodh}{\text{GoodH}}
\newcommand{\rank}{\text{rank}}

\newcommand{\tp}{\textsc{tp}}
\newcommand{\avtp}{\textsc{avgtp}}
\newcommand{\ghss}{\textsc{ghss}}
\newcommand{\shelf}{\textsc{shelf}}
\newcommand{\hh}{\textsf{HH}}
\newcommand{\topk}{\text{\sc Topk}}
\newcommand{\hattopk}{\ensuremath{\overline{\textsc{Topk}}}}
\newcommand{\epsbar}{\ensuremath{\bar{\epsilon}}}
\newcommand{\conj}[1]{\overline{#1}}
\newcommand{\sgn}{\textrm{sgn}}
\newcommand{\goodftwo}{\textsc{goodf}\ensuremath{_2}}
\newcommand{\smallres}{\textsc{smallres}}
\newcommand{\median}{\text{median}}

\newcommand{\consistw}{\textsf{~consist. with~}}

\title{High Probability Frequency Moment Sketches}
\author {Sumit Ganguly \\ Indian Institute of Technology,\\ Kanpur, India \\ \texttt{sganguly@cse.iitk.ac.in}
\and David P. Woodruff \\ School of Computing,\\ Carnegie Mellon University, \\ Pittsburg, USA \\ \texttt{dwoodruf@cs.cmu.edu}}

\begin{document}

\maketitle

\begin{abstract}
We consider the problem of sketching the $p$-th frequency moment of a vector, $p>2$, with multiplicative error at most $1\pm \epsilon$ and \emph{with high confidence} $1-\delta$. Despite the long sequence of work on this problem, tight bounds on this quantity are only known for constant $\delta$. While one can obtain an upper bound with error probability $\delta$ by repeating a sketching algorithm with constant error probability $O(\log(1/\delta))$ times in parallel, and taking the median of the outputs, we show this is a suboptimal algorithm! Namely, we show {\it optimal} upper and lower bounds of $\Theta(n^{1-2/p} \log(1/\delta) + n^{1-2/p} \log^{2/p} (1/\delta) \log n)$ on the sketching dimension, for any constant approximation. Our result should be contrasted with results for estimating frequency moments for $1 \leq p \leq 2$, for which we show the optimal algorithm for general $\delta$ is obtained by repeating the optimal algorithm for constant error probability $O(\log(1/\delta))$ times and taking the median output. We also obtain a matching lower bound for this problem, up to constant factors.
\end{abstract}

%
\section{Introduction} \label{sec:intro}
The frequency moments problem is a very well-studied and foundational problem in the data stream literature. In the data stream model, an algorithm may use only  sub-linear
 memory  and a single pass over the data to summarize a  data stream that
appears as  a sequence of  incremental updates.  A  data stream may be
viewed as a sequence of $m$  records of the form $( (i_1, v_1), (i_2, v_2), \ldots, (i_m, v_m))$, where, $i_j
\in
 [n] = \{1,2, \ldots, n\}$ and $v_j   \in  \R$.
The record $(i_j,v_j)$ changes the $i_j$th coordinate $x_{i_j}$ of an underlying
$n$-dimensional vector $x$ to  $ x_{i_j} + v_{j}$. Equivalently, for $i \in [n]$,
$x_i = \sum_{j: i_j = i} v_j$. Note that $v_j$ may be positive or negative, which corresponds to the
so-called turnstile model in data streams. Also, the $i$-th coordinate of $x$ is sometimes referred to
as the {\it frequency} of item $i$, though note that it can be negative in the turnstile model.
The $p$-th
moment of $x$ is defined to be $F_p = \sum_{i \in [n]}\abs{x_i}^p$, for a real number
$p\ge 0$, which for $p \geq 1$ corresponds to the $p$-th power of the $\ell_p$-norm $\|x\|_p^p$ of $x$.

The $F_p$ estimation problem with approximation parameter $\epsilon$ and failure probability $\delta$ is:
design an algorithm that makes one pass over the input stream and  returns  $\hat{F}_p$ such that
$\probb{\abs{\hat{F}_p - F_p} \le \epsilon F_p} \ge 1- \delta$. Such an algorithm is also referred to as an $(\epsilon, \delta)$-approximation of $F_p$. This is a  problem that is among the ones that has received the most attention in the data stream literature, and we only give a partial list of work on this problem \cite{AMS99,Andoni12,ako11,anpw:icalp13,B-YJKS:stoc02,bgks06,bo10,cks03,ck04,g04,g04b,g11,indy:focs00,indyk2005optimal,knpw11,knw:pods10b,knw:soda10,liwood:random13,mw10,ww15}.

We study the class of algorithms based on linear sketches, which store only a sketch $S \cdot x$ of the input vector $x$ and a (possibly randomized) matrix $A$. This model is well-studied, both for the problem of estimating norms and frequency moments
\cite{anpw:icalp13,hw13,liwood:random13,pw12}, and for other problems such as estimating matrix
norms \cite{lnw14a}, and matching size
\cite{AKLY16,k15}. The efficiency is measured in terms of the {\it sketching dimension} which is the maximum
number of rows of a matrix $S$ used by the algorithm. Since the algorithm is randomized, it may
choose different $S$ based on its randomness, so the maximum is taken over its randomness. Linear sketches are particularly useful for data streams since given an update
$(i_j, v_j)$, one can update $Sx$ as $S(x+v_j e_{i_j}) = Sx + Sv_j e_{i_j}$,
where $e_{i_j}$ is the standard unit vector in the $i_j$-th direction. They are
also used in distributed environments, since given $S \cdot x$ and618 $S \cdot y$,
one can add these to obtain $S \cdot (x+y)$, the sketch of $x+y$.

When $0 < p \leq 2$, one can achieve a sketching dimension of $O(\epsilon^{-2} \log(1/\delta))$ independent of $n$ \cite{AMS99,knpw11,knw:soda10}, while for $p = 0$ the sketching
dimension is $O(\epsilon^{-2} (\log(1/\epsilon) + \log \log n) \log(1/\delta))$ \cite{knw:pods10b}. For $p = 2$ there is a sketching lower bound of $\Omega(\epsilon^{-2} \log(1/\delta))$ \cite{kmn11}, which implies an optimal algorithm for general $\delta$ is to run an optimal algorithm with error probability $1/3$ and take the median of $O(\log(1/\delta))$ independent repetitions. As a side result, we show in the full version a lower bound
of $\Omega(\epsilon^{-2} \log(1/\delta))$ for any $1 \leq p < 2$, which shows this
strategy of amplifying the success probability by $O(\log 1/\delta)$ independent
repetitions is also optimal for any $1 \leq p < 2$.

Perhaps surprisingly, for $p > 2$, the sketching dimension needs to be polynomial in $n$, as first shown in \cite{pw12}, with the best known lower bounds being $\Omega(n^{1-2/p} \log n)$ \cite{anpw:icalp13} for constant $\epsilon$ and $\delta$, and $\Omega(n^{1-2/p} \epsilon^{-2})$ for constant $\delta$ \cite{liwood:random13}. Regarding upper
bounds, we present the long list of bounds in Table \ref{tab:results}. The best known
upper bound is $O(n^{1-2/p}\epsilon^{-2} \log(1/\delta) + n^{1-2/p} \epsilon^{-4/p} \log n \log(1/\delta))$ \cite{g11}. This is tight only when $\epsilon$ and $\delta$ are constant, in which case it matches \cite{anpw:icalp13}, or when $\delta$ is constant and $\epsilon < 1/\poly(\log n)$, since it matches \cite{liwood:random13}.
%
%
\begin{table*}[!ht]
\centering
\scalebox{1}{
\begin{tabular}{|c|c|}
\hline
$F_p$ Algorithm & Sketching Dimension\\
\hline
\cite{indyk2005optimal}  & $O(n^{1-2/p}\epsilon^{-O(1)}\log^{O(1)} n \log(1/\delta))$ \\
\cite{bgks06} & $O(n^{1-2/p}\epsilon^{-2-4/p}\log n \log(M) \log(1/\delta))$ \\
\cite{mw10} & $O(n^{1-2/p} \epsilon^{-O(1)}\log^{O(1)} n \log(1/\delta))$ \\
\cite{ako11} & $O(n^{1-2/p} \epsilon^{-2-6/p}\log n \log(1/\delta))$ \\
\cite{bo10} & $O(n^{1-2/p} \epsilon^{-2-4/p}\log n \cdot g(p,n) \log(1/\delta))$ \\
\cite{Andoni12} & $O(n^{1-2/p} \log n \epsilon^{-O(1)} \log(1/\delta))$\\
\cite{g11}, {\bf Best upper bound} & $O(n^{1-2/p}\epsilon^{-2} \log(1/\delta) + n^{1-2/p} \epsilon^{-4/p} \log n \log(1/\delta))$\\
\hline
\end{tabular}
}
\caption{
Here, $g(p,n) = \min_{c \textrm{ constant }} g_c(n)$, where $g_1(n) = \log n$, $g_c(n) = \log(g_{c-1}(n))/(1-2/p)$.
We start the upper bound timeline with \cite{indyk2005optimal}, since that
is the first work which achieved an exponent of $1-2/p$ for $n$. For earlier work which achieved worse exponents for $n$,
see \cite{AMS99,ck04,g04,g04b}.
}\label{tab:results}
\end{table*}
%
%
\subsection{Our Contributions} In this work, we show {\it optimal} upper and lower bounds of $\Theta(n^{1-2/p} \log(1/\delta) + n^{1-2/p} \log^{2/p} (1/\delta)\allowbreak \log n)$ on the sketching dimension for $F_p$-estimation, for any $p > 2$, and for any constant $\epsilon$. Our upper bound shows, perhaps surprisingly, that the optimal bound is {\it not} to run $O(\log (1/\delta))$ independent repetitions of a constant success probability algorithm and report the median of the outputs. Indeed, such an algorithm would give a worse $O(n^{1-2/p} \log(1/\delta) \log n)$ sketching dimension.

For general $\epsilon$, our upper bound is $O(n^{1-2/p} \epsilon^{-2} \log (1/\delta) + n^{1-2/p} \epsilon^{-4/p} \log ^{2/p}(1/\delta) \log n)$ and our lower bound is
 $\Omega(n^{1-2/p} \epsilon^{-2} \log (1/\delta) + n^{1-2/p} \epsilon^{-2/p} \log^{2/p} (1/\delta) \log n)$, which differ by at most an $\epsilon^{-2/p}$ factor. Our results thus come close to resolving the complexity for general $\epsilon$ as well.

Our results should be contrasted to $1 \leq p \leq 2$,
for which the optimal sketching dimension for such $p$ is $\Theta(\epsilon^{-2} \log(1/\delta))$, and so for these $p$ it is optimal to run $O(\log(1/\delta))$ independent repetitions of a constant probability algorithm. Here we strengthen the $\Omega(\epsilon^{-2} \log(1/\delta))$ bound for $p = 2$ of \cite{kmn11} by showing the same bound for $1 \leq p \leq 2$.

\subsubsection{Overview of Upper Bound}
In order  to obtain a confidence of $1-\delta$, we use the $d = \lceil\log (1/\delta)\rceil$th moment of an  estimate $\hat{F}_p$ of $F_p$. Since we are unable to use the $d$th moment of the Taylor polynomial estimator of \cite{g:arxiv15}, we employ a different estimator $X_i$ for estimating individual coordinates $\abs{x_i}$ and use it as $X_i^p$ to estimate $\abs{x_i}^p$. This estimator is based on (a) using random $q$th roots of unity for sketches instead of standard Rademacher variables, and (b) taking the \emph{average} of the estimates from those tables where the item does not collide with the set of  top-$k$ estimated heavy hitters. 

{\bf The Shelf Structure.}
The algorithm uses two structures, namely, a \ghss-like structure from \cite{g:arxiv15} and a new shelf structure , which is our main algorithmic novelty (both formally defined later). The shelf structure is necessary when the failure probability is $\delta =n^{-\omega(1)}$; otherwise, for $\delta =n^{-\Theta(1)}$, somewhat surprisingly the \ghss~structure of \cite{g:arxiv15} alone suffices with parameter $C = n^{1-2/p}( \epsilon^{-2} \log (1/\delta)/\log (n) + \epsilon^{-4/p} \log^{2/p}(1/\delta))$ and number of measurements $O(C \log n)$, which requires a some-what  intricate $d$-th moment analysis of the \ghss~ structure.


\begin{figure}[htbp]
\includegraphics[width=6.5in]{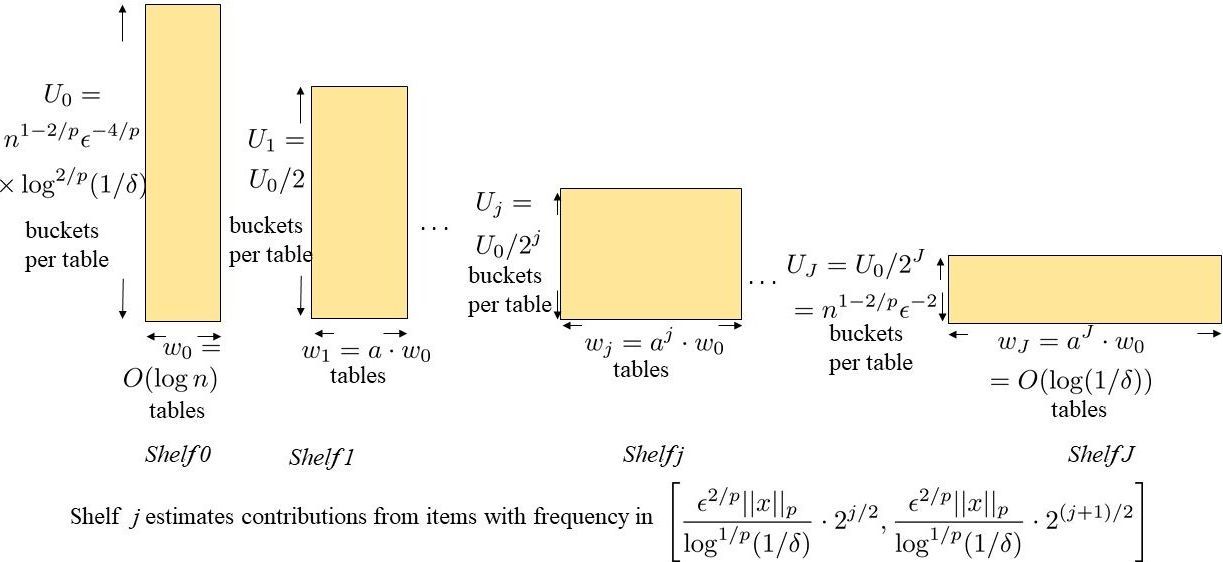}
\caption{Shelf structure and level sets for each shelf index $j$ whose contribution to $F_p$ is estimated accurately.}
\label{fig:shelf}
\end{figure}

The shelf structure is partitioned into shelves, indexed from $j=0, \ldots, J$, for a value $J$ which is specified below. Each shelf consists of a pair of  \countsketch~like structures, $\hh_j$ and $\est_j$. The number of buckets in the tables of the $j$th shelf is $H_j$ and the number of tables in the $j$th shelf of the $\hh_j$ structure is $w_j$ and of the $\est_j$ structure is $2w_j$.
We set $H_J = \Theta( n^{1-2/p} \epsilon^{-2})$ and $w_J = \Theta(\log (1/\delta))$, while $H_0 = \Theta(n^{1-2/p} \epsilon^{-4/p} \log^{2/p}(1/\delta))$ and $w_0 = s= \Theta(\log n)$.
In particular, the shelf numbered zero coincides with the \ghss~level zero.
The input vector $x$ is provided as input to all the shelves' structures. The levels of the \ghss~structure and the shelves of the shelf structure can also be viewed as a single structure starting from shelf numbered $J, J-1, \ldots, 0$, and level numbers $1,2 \ldots, L$.
Here we consider the interesting case when $H_J = o(H_0)$, otherwise, (i.e., when $H_J = \Omega(H_0)$) there are just two shelves and $J=1$.
The table height $H_j = H_0b^j$ decays geometrically with parameter $0<b<1$ and the table width $w_j = w_0a^j $ increases geometrically with parameter $a>1$. Note that the parameters $a$ and $b$ determine $J$. By requiring that $\abs{1-ab} = \Omega(1)$, we ensure that the  total number of measurements of the shelf structure is  $ \sum_{j=0}^J H_jw_j = O(H_0w_0+ H_Jw_J)$, no matter which value of $J$ we choose. For the shelf structure, frequency-wise thresholds are defined as  $U_j = O(\hat{F}_2/H_j)^{1/2}$, for $j = 0,1, \ldots, J$. The shelf frequency group corresponding to shelf $j$ is $S_j = [U_j, U_{j+1})$, where, $U_{J+1} = \infty$ and $U_0 = T_0$. We sometimes conflate $S_j$ with the set of items whose frequency belongs to $S_j$.  The frequency group $G_0$ is defined as $ [T_0, U_1] $ and coincides with $S_0$. 

So why a shelf structure? Suppose for simplicity that $\epsilon$ is a constant. Consider a vector $x$ which has a constant number of ``large'' coordinates of value $\Theta(n^{1/p})$, and $\Theta(n)$ remaining ``small'' coordinates of absolute value $O(1)$. Then we need to find all the large coordinates to accurately estimate $F_p$ up to a small constant factor. This is well-known to be possible with $\Theta(n^{1-2/p})$ buckets in the $J$-th shelf, since with probability $1-\delta$, each of the large coordinates will not collide with any other large coordinate in more than a small constant fraction of tables. Note that in each table, in each bucket containing a large coordinate, the ``noise'' in the bucket from small coordinates will be $C n^{1/p}$ for an arbitrarily small constant $C > 0$ with constant probability, and so this will happen in most buckets containing a large coordinate in most tables with probability $1-\delta$.

However, now consider a vector $x$ which has $\Theta(\log(1/\delta))$ ``large-ish'' coordinates of value $\Theta(n^{1/p}/\log^{1/p}(1/\delta))$, and $\Theta(n)$ remaining ``small'' coordinates of absolute value $O(1)$, as before. Then we again need to find most of the ``large-ish'' coordinates to accurately estimate $F_p$ up to a constant factor. We also {\it cannot} subsample and try to estimate how many large-ish coordinates there are from a subsample. Indeed, since there are only $O(\log(1/\delta))$ total large-ish coordinates, sub-sampling would not accurately estimate this total with probability at least $1-\delta$. However, to find these ``large-ish'' coordinates, we need to increase the number of buckets from $\Theta(n^{1-2/p})$ to $\Theta(n^{1-2/p} \cdot \log^{2/p}(1/\delta))$ just so that in a bucket containing one of these coordinates, with constant probability the noise will not be too large. But if we then want this to happen for a $1-\delta$ fraction of tables, we still need $\Theta(\log(1/\delta))$ tables, which gives overall $\Theta(n^{1-2/p} \cdot \log^{1+2/p}(1/\delta))$ measurements, which is above our desired total of $O(n^{1-2/p} (\log(1/\delta) + \log(n) \log^{2/p}(1/\delta)))$ measurements.

So what went wrong? The key idea in our analysis is to relax the requirement of trying to recover all the larg-ish coordinates with probability $1-\delta$. Suppose instead of $\Theta(\log(1/\delta))$ tables we just use $\Theta(\log n)$ tables. Then with probability $1-1/n$, there may be two large-ish coordinates which collide and cancel with each other in every single table, and we have no way of recovering them. However, we are able to show that with probability $1-\delta$, only $O(\log(1/\delta)/\log n)$ large-ish coordinates will fall into this category, and neglecting this roughly $(1-1/\log n)$ fraction of the large-ish coordinates will not affect our estimate of $F_p$ by more than a constant factor. And indeed, our $0$-th shelf has exactly $\Theta(n^{1-2/p} \cdot \log^{2/p}(1/\delta))$ buckets and $\Theta(\log n)$ tables, so is exactly suited for finding these large-ish coordinates. In general, we can show that one of our shelves will be able to handle every vector with coordinates of magnitude between the large and large-ish coordinates. Again, by choosing the shelf structure carefully, the total number of measurements is dominated by that in the zero-th plus the $J$-th shelf, giving us $O(n^{1-2/p} (\log(1/\delta) + \log(n) \log^{2/p}(1/\delta)))$ total measurements, and explaining where the $\log^{2/p}(1/\delta)$ in the upper bound comes from.

{\bf The Non-Large-ish Coordinates.}
Our shelves are designed to estimate the contribution to $F_p$ from all coordinates of absolute value at least $\Theta(n^{1/p}/\log^{1/p}(1/\delta))$. For coordinates of smaller value, we can now afford to sub-sample and apply the same $0$-th shelf structure to estimate their contribution to $F_p$. We apply
the \ghss~structure, which is analogous to the structure presented in \cite{g:arxiv15} and has $L+1$ levels corresponding to $l=0, \ldots, L$, and consists of a pair of \countsketch~like structures $\hh_l$ and $\est_l$ at each level. The sub-sampling technique and the associated frequency-wise thresholds and frequency groups are defined analogously (with new parameters) to \cite{g:arxiv15}.

A notable difference with \cite{g:arxiv15} is that
the \est~structures in the \ghss~and shelf structures use
 complex $q$th roots of unity and return the average of table estimates instead of the median of table estimates used by \countsketch, which are novelties in this context, though have been used for other data stream problems \cite{knpw:stoc11}.
We have that $\expect{X_i^p} = \abs{x_i}^p (1 \pm n^{-\Omega(1)})$ for
our estimator $X_i$ of $\abs{x_i}$, and thus $X_i^p$ provides a nearly unbiased
estimator of $\abs{x_i}^p$.
Additionally, we use  averaging in the definition of $X_i$ instead of the median to allow for a tractable, though intricate calculation of the $d$-th moment of the sum of the $p$-th powers of $X_i$. 

\subsubsection{Overview of Lower Bounds} We give an overview for the case of constant $\epsilon$. In both cases we start by applying Yao's minimax principle for which we fix $S$ and then design a pair of distributions $\alpha$ and $\beta$ which must be distinguished by an $(\epsilon, \delta)$-approximation algorithm for $F_p$. We can also assume the rows of $S$ are orthonormal, since a change of basis to the row space of $S$ can always be applied in post-processing.

{\bf Our $\Omega(n^{1-2/p} \epsilon^{-2/p} (\log^{2/p} 1/\delta) \log n)$ bound.} This is our technically more involved lower bound. We first upper bound the variation distance using the $\chi^2$-divergence as in \cite{anpw:icalp13} 
 and work only with the latter. We let $\alpha = N(0, I_n)$ be an $n$-dimensional isotropic Gaussian distribution, while $\beta$ is a distribution formed by sampling an $x \sim N(0,I_n)$, together with a random subset $T \subset [n]$ of size $O(\log(1/\delta))$, and outputting $z = x + \sum_{i \in T} (Cn^{1/p}/t^{1/p}) e_i$, where $e_i$ is the $i$-th standard unit vector and $C > 0$ is a constant. For $y \sim \alpha$ and $z \sim \beta$, one can show that with probability $1-O(\delta)$, one has that $\|z\|_p^p$ is a constant factor larger than $\|y\|_p^p$, since $\|y\|_p^p$ and $\|x\|_p^p$ are concentrated at $\Theta(n)$, while $\sum_{i \in T} C^p n/t = \Theta(n)$.

A common technique in upper bounds, including our own, is the notion of subsampling, whereby a random fraction of roughly $1/2^i$ of the $n$ coordinates are sampled, for each value of $i \in O(\log n)$, and information is then gathered for each $i$ and combined into an overall estimate of $F_p$. We choose our hard distributions so that {\it subsampling does not help.} Indeed, if one subsamples half of the coordinates of $z \sim \beta$, with probability $\Omega(\delta)$ all of the coordinates in $T$ will be removed, at which point $z$ is indistinguishable from $y \sim \alpha$. Therefore, our pair of distributions suggests itself as being hard for $(\Theta(1), \delta)$-approximate $F_p$ algorithms.

What drives our analysis is conditioning our distributions on an event $\mathcal{G}$ which only happens with probability $\Omega(\delta)$. Note that for any algorithm which can distinguish samples from $\alpha$ from those from $\beta$ with probability at least $1-\delta$, it must still have probability $9/10$, say, of distinguishing the distributions given an event $\mathcal{G}$ which occurs for samples drawn from $\beta$. The event $\mathcal{G}$ corresponds to every $i \in T$ having the property that the corresponding column $S_i$ of our sketching matrix $S$ has squared length at most $2r/n$, where $r$ is the number of rows of $S$. By a Markov bound, half of the columns of $S$ have this property, and since $T$ has size $O(\log 1/\delta)$, with probability $\Omega(\delta)$, event $\mathcal{G}$ occurs.

We analyze the $\chi^2$-divergence of the distributions $\alpha$ and $\beta$ conditioned on $\mathcal{G}$. One technique helpful for this is an equality given by Fact \ref{mixture}, which states that for $p$ a distribution on $\R^n$, that $\chi^2(N(0, I_n) \ast p, N(0, I_n)) = {\bf E}[e^{\langle X, X' \rangle}] - 1,$ where $X$ and $X'$ are independently drawn from $p$. This equality was used in \cite{anpw:icalp13,lnw14a,woodruff:book} among other places. In our case, the inner product of $X$ and $X'$ corresponds to an inner product $P$ of two independent random sums of $t$ columns of $S$, restricted to only those columns with squared length at most $2r/n$. Let the $t$ columns forming $X$ be denoted by $T$ and the $t$ columns forming $X'$ be denoted by $U$.

Critical to our analysis is bounding ${\bf E}[P^j]$ for large powers of $j$, see Lemma \ref{lem:key}. One can think of indexing the rows of $S^TS$ by $T$ and the columns of $S^TS$ by $U$, where $S^TS$ is an $n \times n$ matrix. Let $M$ denote the resulting submatrix.
The inner product of interest is then $e_T^T M e_U$, where $e_T = \sum_{i \in T} e_i$ and $e_{U} = \sum_{i \in U} e_i$.

Our bound in Lemma \ref{lem:key} is very sensitive to minor changes. Indeed, if instead of showing ${\bf E}[P^j] \leq \left (\frac{t^2}{r^{1/2}} \right ) \cdot \left (\frac{16r}{n} \right )^j$, we had shown  ${\bf E}[P^j] \leq \left (\frac{t^2}{r^{1/2}} \right ) \cdot \left (\frac{16rt}{n} \right )^j$ or  ${\bf E}[P^j] \leq \left (\frac{t^2}{r^{1/2}} \right ) \cdot \left (\frac{16r \log n}{n} \right )^j$, our resulting bound for the $\chi^2$-divergence would be larger than $1$. For instance, a natural approach is to instead consider $e_T = \sum_{i \in T} \sigma_i e_i$ and $e_{U} = \sum_{i \in U} \sigma_i e_i$ where the $\sigma_i$ are independent random signs (i.e., $\Pr[\sigma_i = 1] = \Pr[\sigma_i = -1] = 1/2$), which would correspond to redefining the distribution $\beta$ above to sample $z = x + \sum_{i \in T} (Cn^{1/p}/t^{1/p}) \sigma_i e_i$. Without further conditioning the $\sigma_i$ variables, the $\chi^2$-divergence can be as large as $n^{\Theta(\log(1/\delta))}$. This is because with probability roughly $2^{-2t}$, over the choice of the $\sigma_i$, one has $\sum_{i \in T} \sigma_i e_i$ and $\sum_{i \in U} \sigma_i e_i$ both being very well aligned with the top singular vector of $M$ (if say, $S$ were a random matrix with orthonormal rows), at which point our desired inner product is too large. Instead, by setting all $\sigma_i = 1$, that is, by considering $e_T = \sum_{i \in T} e_i$ and $e_{U} = \sum_{i \in U} e_i$ as we do, we rule out this possibility.

We prove Lemma \ref{lem:key} by expanding ${\bf E}[P^j]$ into a sum of products,
each having the form $\prod_{w=1}^j |\langle S_{a_w}, S_{b_w} \rangle|$
where the $S_{a_w}, S_{b_w}$ are columns of $S$. One thing that matters in
such products is the multiplicities of duplicate columns that appear in a
product. We split
the summation by what we call $y$-{\it patterns}.
We can think of a $y$-pattern as a partition of $\{1, 2, \ldots, j\}$ into $y$ non-empty pieces.
We can also define a $z$-pattern as a partition of $\{1, 2, \ldots, j\}$ into $z$ non-empty pieces.
We analyze the expectation
for a particular pair $P, Q$, where $P$
is a $y$-pattern and $Q$ is a $z$-pattern for some $y, z \in \{1, 2, \ldots, j\}$, that is, we only sum over
pairs of $j$-tuples $a_1, \ldots, a_j$ and $b_1, \ldots, b_j$ for which for each
non-empty piece $\{d_1, \ldots, d_{\ell}\}$ in $P$, where $d_i \in \{1, 2, \ldots, j\}$ for all $i$
and $\ell \leq j$,
we have
$a_{d_1} = a_{d_2} = \cdots = a_{d_{\ell}}$. Similarly
for each $\{e_1, \ldots, e_{m}\}$ in $Q$, where $e_i \in \{1, 2, \ldots, j\}$ for all $i$ and $m \leq j$,
we have $b_{e_1} = b_{e_2} = \cdots = b_{e_m}$. We also require
if $d, d' \in \{1, 2, \ldots, j\}$ are in different pieces of $P$,
then $a_{d} \neq a_{d'}$. Similarly, if $e, e' \in \{1, 2, \ldots, j\}$ are in different pieces of $Q$, then
$b_e \neq b_{e'}$. Thus, each pair of $j$-tuples is valid for exactly one pair $P, Q$ of patterns.

The valid pairs of $j$-tuples for $P$ and $Q$
define a bipartite multi-graph as follows.
In the left partition we create a node for each non-empty piece of $P$, and in the right
partition we create a node for each non-empty piece of $Q$. We include an edge from a node
$a$ in the left to a node $b$ in the right if $i \in a$ and $i \in b$ for some $i \in \{1, 2, \ldots, j\}$.
If there is more than one such $i$, we include an edge with multiplicity corresponding to the number of such $i$.
This bipartite graph only depends on $P$ and $Q$. We consider a maximum matching in this multi-graph,
and we upper bound the contribution of valid pairs for $P$ and $Q$ based on that matching. By
summing over all pairs $P, Q$, we obtain our bound on ${\bf E}[P^j]$.

{\bf Our $\Omega(n^{1-2/p} \epsilon^{-2} \log(1/\delta))$ bound.} This bound uses the same distributions $\alpha$ and $\beta$ as in \cite{liwood:random13}, where an $\Omega(n^{1-2/p}\epsilon^{-2})$ bound was shown, but we strengthen it to hold for general $\delta$. To do so, we use an exact characterization of the variation distance between multi-variate Gaussians with shifted mean by relating it to the univariate case (given in the full version), and a strong concentration of bounded Lipshitz functions with respect to the Euclidean norm (given in the full version). These enable us to show with probability $1-O(\delta)$, vectors sampled from $\alpha$ and $\beta$ have $\l_p$-norm differing by a $1+\epsilon$ factor. By the definition of $\alpha$ and $\beta$, we can then reduce the problem to distinguishing an isotropic Gaussian from an isotropic Gaussian plus a small multiple of a fixed column of $S$, which typically has small norm since $S$ has orthonormal rows. We then apply a bound as derived above (see full version).

{\bf Our $\Omega(\epsilon^{-2} \log(1/\delta))$ bound for $1 \leq p < 2$.} This lower bound uses similar techniques to our
lower bound of $\Omega(n^{1-2/p} \epsilon^{-2} \log(1/\delta))$, but considers distinguishing an isotropic Gaussian
$N(0, I_n)$ from an $N(0, (1+\epsilon) I_n)$ random variable. Here we set $n = \Theta(\epsilon^{-2} \log(1/\delta))$,
and show the $p$-norms of samples from the two distributions differ by a $(1+\epsilon)$-factor with probability $1-\delta$.
Using that $S$ has orthonormal rows, the images of the two distributions under our sketching matrix $S$ correspond to
$N(0, I_r)$ and $N(0, (1+\epsilon) I_r)$, where $r$ is the number of rows of $S$. The result then follows by using the product structure of Hellinger distance.

\section{Our Lower Bounds} \label{sec:lb:intro}
We first present an overview of the lower bounds in a little more detail.  We defer both our $\Omega(n^{1-2/p} \epsilon^{-2} \log(1/\delta))$ lower bound for $p > 2$ and our $\Omega(\epsilon^{-2} \log(1/\delta))$
lower bound for $1 \leq p < 2$ entirely to the Appendix. Here we focus on our lower bound of $\Omega(n^{1-2/p} \epsilon^{-2/p} (\log^{2/p}(1/\delta)) \log n)$ for $p > 2$.
See also Section \ref{sec:intro} for an overview of all of our lower bounds.

We assume $\delta$-{\bf Bound4}, which is that $\log(1/\delta) \leq (n^{1-2/p} \eps^{-2/p} (\log^{2/p} 1/\delta) \log n)^{1/4} n^{-c'}$,
for a sufficiently small constant $c' > 0$. Since $p > 2$ is an absolute constant, independent of $n$, this just states that $\delta \geq 2^{-n^{c''}}$ for a sufficiently small constant $c'' > 0$. There are
other bounds - $\delta$-{\bf Bound1}, $\delta$-{\bf Bound2}, and $\delta$-{\bf Bound3} - see the full version in the Appendix, but these are not assumptions but rather implied by relations between the various parameters (e.g., otherwise the $\Omega(n^{1-2/p} \epsilon^{-2} \log(1/\delta))$ lower bound is stronger).
%
%

Let $p$ and $q$ be probability density functions of continuous distributions. The $\chi^2$-divergence from $p$ to $q$ is
$\chi^2(p, q) = \int_x \left (\frac{p(x)}{q(x)} - 1 \right )^2 q(x) dx.$

\begin{fact}(\cite{Tsybakov}, p.90) \label{chitvdSHORT}
For any two distributions $p$ and $q$, we have $D_{TV}(p, q) \leq \sqrt{\chi^2(p,q)}$.
  \end{fact}
We need a fact about the distance between a Gaussian location mixture to a Gaussian distribution.
\begin{fact}(p.97 of \cite{is03})\label{mixtureSHORT}
  Let $p$ be a distribution on $\mathbb{R}^n$. Then
  $\chi^2(N(0,I_n) \ast p, N(0,I_n)) = {\bf E}[e^{\langle X, X' \rangle}] - 1,$
  where $X$ and $X'$ are independently drawn from $p$.
\end{fact}
%
Let $T$ be a sample of $t \eqdef \log_3(1/\sqrt{\delta})$ coordinates $i \in [n]$ without replacement.

\noindent
{\bf Case 1:}
Suppose $y \sim N(0,I_n)$, and let $\alpha'$ be the distribution of $y$.

\noindent
{\bf Case 2:}
Let $z = x + \sum_{i \in T} \frac{C' \epsilon^{1/p} E_{n-t}}{t^{1/p}} e_i,$ where $x \sim N(0, I_n)$
and $E_{n-t} = {\bf E}_{x \sim N(0, I_{n-t})}[\|x\|_p]$.
Note that $x$ and $T$ are independent. Also, $C' > 0$ is a sufficiently large constant.
Let $\beta'$ be the distribution of $z$.

In the full version (in the Appendix) we show that for the sketching algorithm to be correct,
$D_{TV}(\bar{\alpha'}, \bar{\beta'}) \geq 1-2\delta$, where
$\bar{\alpha'}$ is the distribution of $S \cdot y$ for $y \sim \alpha'$ and
$\bar{\beta'}$ is the distribution of $S \cdot z$ for $z \sim \beta'$.

Fix an $r \times n$ matrix $S$ with orthonormal rows. Important to our proof will be the existence of a
subset $W$ of $n/2$ of the columns for which $\|S_i\|^2 \leq 2r/n$ for all $i \in W$.
To see that $W$ exists,
consider a uniformly random column $S_i$ for $i \in [n]$. Then ${\bf E}[\|S_i\|^2] = r/n$ and so by Markov's inequality, at least a $1/2$-fraction
of columns $S_i$ satisfy $\|S_i\|^2 \leq 2r/n$. We fix $W$ to be an arbitrary subset of $n/2$ of these columns.

Suppose we sample $t$ columns of $S$ without replacement, indexed by $T \subset [n]$.
Let $\mathcal{G}$ be the event that the set $T$ of sampled columns belongs to the set $W$.

\begin{lemma}\label{lemma:condition}
  $\Pr[\mathcal{G}] \geq \sqrt{\delta}$.
\end{lemma}
Let $\alpha_G = \bar{\alpha'} \mid \mathcal{G}$ and
$\beta_G = \bar{\beta'} \mid \mathcal{G}$. By the triangle inequality,
$1- 2\delta \leq D_{TV}(\bar{\alpha'}, \bar{\beta'}) \leq \Pr[\mathcal{G}] D_{TV}(\alpha_g, \beta_G) + 1-\Pr[\mathcal{G}] \leq \frac{\sqrt{\delta}}{2} D_{TV}(\alpha_G, \beta_G) + 1-\frac{\sqrt{\delta}}{2},$
which implies that $1 - 4\sqrt{\delta} \leq D_{TV}(\alpha_G, \beta_G)$. We can assume $\delta$ is less than a sufficiently small positive
constant, and so it suffices to show for sketching dimension $r = o(n^{1-2/p} \eps^{-2/p} (\log^{2/p} 1/\delta) \log n)$, that
$D_{TV}(\alpha_G, \beta_G) \leq 1/2$. By Fact \ref{chitvdSHORT}, it suffices to show $\chi^2(\alpha_G, \beta_G) \leq 1/4$.

Since $S$ has orthonormal
rows, $\bar{\alpha'}$ is distributed as $N(0, I_r)$. Note that, by definition of $\alpha$, we in fact have $\bar{\alpha'} = \alpha_G$ since
conditioning on $\mathcal{G}$ does not affect this distribution. On the other hand, $\beta_G$ is a Gaussian location mixture, that is, it has
the form $N(0, I_r) \ast p$, where $p$ is the distribution of a random variable
chosen by sampling a set $T$ subject to event $\mathcal{G}$ occurring
and outputting $\sum_{i \in T} \frac{C' \epsilon^{1/p} E_{n-t} S_i}{t^{1/p}}$.
We can thus apply Fact \ref{mixtureSHORT} and
it suffices to show for $r = o(n^{1-2/p} \eps^{-2/p} (\log^{2/p} 1/\delta) \log n)$ that
$\expect{e^{\frac{(C')^2 \epsilon^{2/p} E_{n-t}^2}{t^{2/p}} \langle \sum_{i \in T} S_i, \sum_{j \in U} S_j \rangle}} - 1 \leq \frac{1}{4},$
where the expectation is over independent samples $T$ and $U$ conditioned on $\mathcal{G}$. Note that
under this conditioning $T$ and $U$ are uniformly random subsets of $W$.

To bound the $\chi^2$-divergence, we define variables $x_{T,U}$, where
  $x_{T, U} = \frac{(C')^2 \epsilon^{2/p} E_{n-t}^2}{t^{2/p}} \langle \sum_{i \in T} S_i, \allowbreak \sum_{j \in U} S_j \rangle$.
Consider the following, where the expectation is over
independent samples $T$ and $U$ conditioned on $\mathcal{G}$:
\begin{align*}
&\mathbf{E}\biggl[\mathrm{exp}\biggl\{\frac{(C')^2 \epsilon^{2/p} E_{n-t}^2}{t^{2/p}} \langle \sum_{i \in T} S_i, \sum_{j \in U} S_j \rangle\biggr\}\biggr] = \mathbf{E} \bigl[e^{x_{T,U}}\bigr] = \sum_{0 \leq j < \infty} \mathbf{E} \biggl [\frac{x_{T,U}^j}{j!} \biggr ] \\
& = 1 +
\sum_{j \geq 1} \frac{(C')^{2j} \eps^{2j/p} E_{n-t}^{2j}}{t^{2j/p} j!} \mathbf{E} \left [\langle \sum_{i \in T} S_i, \sum_{j \in U} S_j \rangle^j \right ]\\
& = 1 + \sum_{j \geq 1} \frac{O(1)^{2j} \eps^{2j/p} n^{2j/p}}{t^{2j/p} j!} \mathbf{E} \left [\langle \sum_{i \in T} S_i, \sum_{j \in U} S_j \rangle^j \right ].
\end{align*}

The final equality uses that $E_{n-t} = \Theta(n^{1/p})$ and here
$O(1)^{2j}$ denotes an absolute constant raised to the $2j$-th power.
We can think of $T$ as indexing a subset of rows of $S^T S$ and $U$ indexing a subset of columns.
Let $M$ denote the resulting $t \times t$ submatrix of $S^T S$.
Then $\langle \sum_{i \in T} S_i, \sum_{j \in U} S_j \rangle = \sum_{i, j \in [t]} M_{i,j} \leq \sum_{i, j \in [t]} |M_{i,j}| \eqdef P$, and we seek
to understand the value of ${\bf E}[P^j]$ for integers $j \geq 1$.

The following lemma is the key to the argument; its proof is described in Section \ref{sec:intro}. The
proof is based on defining $y$-patterns and looking at matchings in an associated bipartite
multi-graph.

\begin{lemma}\label{lem:keySHORT}
  For integers $j \geq 1$, ${\bf E}[P^j] \leq \left (\frac{t^2}{r^{1/2}} \right ) \cdot \left (\frac{16r}{n} \right )^j.$
\end{lemma}
Given the previous lemma, by $\delta$-{\bf Bound4}, we have $\frac{t^2}{r^{1/2}} = \frac{1}{n^{\Omega(1)}}$, and therefore Lemma \ref{lem:keySHORT} establishes that
${\bf E}[P^j] \leq \frac{1}{n^{\Omega(1)}} \cdot \left (\frac{16r}{n} \right )^j.$
We thus have,
\begin{align*}
\mathbf{E}\biggl[\mathrm{exp}\biggl\{\frac{(C')^2 \epsilon^{2/p} E_{n-t}^2}{t^{2/p}} \langle \sum_{i \in T} S_i, \sum_{j \in U} S_j \rangle \biggr\} \biggr]
& =   \mathbf{E} [e^{x_{T,U}}] = 1 + \frac{1}{n^{\Omega(1)}} \cdot \sum_{j \geq 1} \frac{O(1)^{2j} \epsilon^{2j/p} n^{2j/p}}{j! t^{2j/p}} \cdot \left (\frac{r}{n} \right )^j\\
& = 1 + \frac{1}{n^{\Omega(1)}} \cdot \sum_{j \geq 1} \frac{(c \log n)^j}{j!} \\ &  \leq 1 + \frac{1}{n^{\Omega(1)}} \cdot e^{c (\log n)} \\ & \leq 1 + \frac{1}{4},
\end{align*}
since $c > 0$ is an arbitrarily small constant independent of the constant in the $n^{\Omega(1)}$. The proof is complete.

For $1 \leq p < 2$, we now show that the sketching
dimension is $\Omega(\epsilon^{-2} \log(1/\delta))$, which as discussed
in Section \ref{sec:intro}, matches known upper bounds up to a constant
factor.
\begin{theorem} \label{thm:smallp}
  The sketching dimension for $(\epsilon,\delta)$-approximating $F_p$
  for $1 \leq p < 2$ is $\Omega(\epsilon^{-2} \log(1/\delta))$.
\end{theorem}

\section{Algorithm } \label{sec:algo}
 \label{sec:ghss}
As outlined earlier, the algorithm uses two level-based  structures, namely, \ghss, which is similar to the \ghss~structure presented in \cite{g:arxiv15}, and the shelf structure. The shelf structure is needed only when $ \delta = n^{-\omega(1)}$, otherwise, the \ghss~structure suffices. The \ghss~  has $L+1$ levels, corresponding to $l=0,1, \ldots, L$, and  the shelf structure  has $J$ shelves numbered $0, 1, \ldots, J$. In particular, shelf 0 is identical to \ghss~level 0.

\subsection{Estimating $F_p$}

\emph{\ghss~structure.} Corresponding to each \ghss~level $l \in \{0,1, \ldots, L-1\}$, a pair of \countsketch~like structures named $\hh_l = \hh(C_l,s)$  and $\est_l = \est(C_l,2s)$  are kept. Both structures  $\hh(C_l,s)$ and $\est(C_l,2s)$ are very similar to \countsketch~ structures and have $s$ and $2s$ independent repetitions respectively, with  $16C_l$ buckets in each repetition (table).
Here,  $s = \Theta (\log n)$.   Recall that  $C = n^{1-2/p}(\epsilon^{-2}\log(1/\delta)/\log(n) + \epsilon^{-4/p}\log^{2/p}(1/\delta))$, $C =C_0 = \Theta(p^2 n^{1-2/p} \epsilon^{-4/p} \log^{2/p}(1/\delta))$ and $C_l = C_0 \alpha^l$, for $l=0,1,2, \ldots, L-1$, where,   $\alpha = 1- (1-2/p)\nu$ and $\nu$ is a constant (e.g., 0.01). The number of levels is $L = \lceil \log_{2\alpha} (n/C) \rceil $.
The final level $L$ of the \ghss~structure  uses an $\ell_2/\ell_1$ deterministic sparse-recovery algorithm  \cite{crt:ieeetit06a,donoho:tit06}. We will show that the number of items that are subsampled into level $L$ is $O(C_L)$ with probability $1-O(\delta)$ and therefore by the theorems proved in  \cite{crt:ieeetit06a,donoho:tit06}, by using $O(C_L \log (n/C_L))$ measurements, all  these item frequencies are recovered deterministically. 
\eat{ for the recovery of  non-zero frequencies of all  those items $i$ that hash to level $L$. We  show that the number of items that map to the sub-sampled stream at level $L$ is $O(C_L)$ (i.e., the frequency vector of $\stream_L$ is $O(C_L)$-sparse) with probability $1- \delta^{\Omega(1)}$.}
Following \cite{g:arxiv15}, the \ghss~structure subsamples the stream hierarchically using independent random hash functions $g_1, \ldots, g_L : [n] \rightarrow \{0,1\}$. All items are mapped to level 0; an item is mapped to each of levels $1$ through $l$ iff $g_1(i) = \ldots = g_l(i)=1$, where, the  $g_l$'s are $O(\log (1/\delta) + \log n)$-wise independent. 

\eat{\emph{\ghss~sampling.} For the \ghss~structure, the  input stream $\stream$ is sub-sampled hierarchically to produce random sub-streams $\stream_0, \stream_1, \ldots, \stream_L$, corresponding to each of the levels $0, \ldots, L$.
 The stream
$\stream_0$ is the  entire input stream. For each $l=1, \ldots, L$, $\stream_{l}$ is obtained by sampling each   item $i$
appearing in $\stream_{l-1}$ with probability $1/2$. If $i$ is sampled, then all its records $(i,v)$ are
included in $\stream_1$, otherwise none of its records are included. The sampling uses independently chosen  random hash functions
  $g_1, g_2, \ldots, g_{L}$  each mapping $[n] \rightarrow \{0,1\}$.  $i$ is included in
  $\stream_l$ iff $g_1(i) =1,  g_2(i)=1,  \ldots, g_{l}(i)=1$. The $g_l$'s are chosen
  from a $ O(\log (1/\delta) + \log n)$-wise independent hash family.}

\vspace*{0.1in}
\emph{\hh~and \est~structures.} The \hh$(C_l,s)$ is  a \countsketch~structure \cite{ccf:icalp02}. The $\est_l = $ \est$(C_l,2s)$ structure is similar, except that   instead of Rademacher sketches, it uses  random $q$th roots of unity sketches, where, $q \ge \Theta(\log (1/\delta)+\log n)$. At level $l$ and for table indexed $r \in [2s]$, the corresponding hash function is $h_{lr}: [n] \rightarrow [16C_l]$, and the sketch for  bucket index $b $ is given by  $T_{lr}[b] = \sum_{h_{lr}(i)=b} x_i \omega_{lr}(i)$, where, $\{\omega_{lr}(i)\}_{i \in [n]}$ is a random family of $q$th roots of unity that is $O(\log (1/\delta)+ \log n)$-wise independent. The hash functions across the tables and distinct levels, and the seeds of the family of the  random roots of unity used by the $\est_l$ structures are independent. 

\vspace*{0.1in}
\emph{Shelf structure.} The shelves,  indexed from $j=0, \ldots, J$, each also consist of an analogous pair of structures, namely,  \hh$(H_j, w_j)$ and \est$(H_j, 2w_j)$, each of which are \countsketch-like structures. The number of independent  repetitions in the $\hh(H_j, w_j)$ and $\est(H_j, 2w_j)$  structures are $w_j$ and $2w_j$ respectively.  The number of buckets per hash table is $O(H_j)$ in either of the structures. The $\hh(H_j, w_j)$ is exactly a \countsketch~ structure.
Analogous to the $\est_l$ structures of the \ghss~levels, the  $\est$ structures of the shelves also use sketches using $q$th roots of unity, instead of Rademacher sketches. In particular,  $H_0 = C_0$ and $w_0 = s$, ensuring that shelf 0 coincides with level 0 of \ghss. 
Further,   $H_J =  \Theta(n^{1-2/p} \epsilon^{-2})$ and $w_J = O(\log (1/\delta))$. We therefore have  two  cases, namely, (1) $H_J = \Omega(H_0)$, or, (2) $H_J = o(H_0)$. We consider each of the two cases next.

In case (1), $H_J  = \Omega(H_0) \ge cH_0$ for some constant $c$. The total number of sketches used  by the $J$th shelf is $(cH_0)O(\log (1/\delta))$. Up to constant factor, therefore, the $J$th shelf has higher width (i.e., higher number $O(\log (1/\delta))$)  of independent repetitions compared to shelf 0 (which has $O(\log n)$ independent repetitions) and has same or higher height (i.e., number of buckets in a repetition) namely $O(H_0)$. In this case, one can set $J=1$, and have only two shelves. This considerably simplifies the analysis.

The other case, namely, when,  $H_J = o( H_0)$ is  more
 interesting. Here, we let $H_j = H_0 b^j$, for a geometric decay  parameter $b < 1$ and $b = \Omega(1)$. The latter constraint $b = \Omega(1)$ is a technical constraint whose need becomes clear from the analysis. The table widths increase geometrically as $w_j = w_0 a^j$, for  a parameter $a> 1$. The
 total measurements used by the shelf structure is $\sum_{j=0}^J H_j w_j = H_0w_0 \sum_{j=0}^J (ab)^j = O(\max(H_0w_0, H_Jw_J))$, provided, \eat{ $ab$ is not too close to 1, that is,} $\abs{1-ab} = \Omega(1)$, or, equivalently,  $\abs{\ln (ab)} = \Omega(1)$. The entire stream $\stream$ is provided as input to each of the shelves $j=0,1, \ldots, J$, that is, there is no  sampling.

\vspace*{0.1in}
\emph{Frequency groups, thresholds, estimates  and samples}. Let $B = \Theta(C)$ and  $\epsbar = (B/C)^{1/2} = \Theta (1/p)$. Let $\hat{F}_2$ be an estimate for $F_2 = \norm{x}_2^2$ satisfying $F_2 \le \hat{F}_2  \le  (1+ O(1/p))F_2$  with probability $1-O(\delta)$. (Throughout the paper, it suffices to let $O(\delta)$ denote $\delta/100$. In general, it is $c\delta$ for any constant $c$ that can be embedded into the constants of the structures used by the algorithm).  Define frequency thresholds for \ghss~levels as follows: $T_0 = ( \hat{F}_2/B )^{1/2}$,  $T_l= \left( 2\alpha
 \right)^{-l/2} T_0$ and $Q_l = T_l(1 -\epsbar)$, for $l \in [L-1]$.
Let $Q_L, T_L = 0^+$ (i.e., $a \ge T_L$ iff $a > 0$).
For shelf $j=0, \ldots, J$, let  $E_j = \epsbar^2 H_j$. For shelf $j$, define  the frequency threshold
$U_j = (\hat{F}_2/E_j)^{1/2}$ and let  $U_{J+1} = \infty$. For \ghss~level indices $l =0, \ldots, L-1$, let $\hat{x}_{il}$ denote  the estimate for $x_i$ obtained
using $\hh_l$, and (overloading notation), for shelf indices,  $j =0, \ldots, J$,  let $\hat{x}_{ij}$ denote the estimate for $x_i$ obtained from the $\hh$ structure of shelf $j$. For $l= L$,  $\hat{x}_{iL}$ denotes the estimate returned from the $\ell_2/\ell_1$ sparse recovery structure at level $L$.

\vspace*{0.1in}
\emph{Discovering Items.} We say that $i$ is \emph{discovered} at shelf $j \in [J]$,  provided, $(1-\epsbar)U_j \le \abs{\hat{x}_{ij}} \le (1+\epsbar) U_{j+1}$ and $j \in [J]$ is the \emph{highest} numbered shelf with this property. We say that $i$ is discovered at \ghss~level $l \in \{0, \ldots, L\}$,  if $i$ is not discovered at any shelf indexed $j \in [J]$, and  $l$ is the \emph{smallest} level  such that  $T_l(1-\epsbar) < \hat{x}_{il} \le T_{l-1}(1+\epsbar)  $. 
If $i$ is discovered at shelf $j$,  then, $i$ is included in the shelf sample $\bar{S}_j$.
If $i$  is discovered at  level $l \in [0, 1, \ldots, L]$ and
$\abs{\hat{x}_{il}} \ge T_l$, then, $i$ is included in the level sample $\bar{G}_l$.
If $i$ is discovered at level $l$ and   $ T_l(1-\epsbar)  < \abs{\hat{x}_{il}} < T_l$   then, $i$ is placed in $\bar{G}_{l+1}$ iff the random toss of an  unbiased coin $K_i$ lands heads; and upon tails, $i$is not placed in any sample group. The \ghss~level sampling scheme is similar to \cite{g:arxiv15}.

\vspace*{0.1in}
\emph{The averaged estimator and \nocollision.} For each item $i$ included in a group sample $\bar{G}_l$ or shelf sample $\bar{S}_j$, an estimate $X_i$ for $\abs{x_i}$   is obtained  using the corresponding \est~structure of that level or  shelf, provided  the event \nocollision$(i)$ succeeds. \eat{ This estimate $X_i$ is  conditioned on an event \nocollision$(i)$, that is, if $\nocollision(i)$ fails to hold, then $X_i$ is not defined and $i$ is removed from the corresponding group or shelf sample.
 Suppose $i$ is sampled into $\bar{G}_l$ and let  $\hattopk(C_l)$ denote the top-$C_l$ items with respect to the estimated frequency $\abs{\hat{x}}_{i'l}$ (resp. \hattopk$(E_j)$ if $i$ is discovered in shelf $j$).}If $i$ is sampled into $\bar{G}_l$, then  \nocollision$(i)$ holds if  there is a set $R_l(i) \subset [2s]$ of table indices of the \est$_l$~structure such that for each $r \in R_l(i)$, $i$ does not collide under the hash function $h_{lr}$ with any of the items that are the  top-$C_l$ absolute  estimated frequencies using $\hh_l$. An analogous definition holds if $i$ is included in the $j$th shelf sample.  Assuming \nocollision$(i)$ holds,  the estimate $X_i$ is defined as the average of the estimates obtained from the tables whose indices are in the set $R_l(i)$ ( resp. $R_j(i)$ if $i$ was discovered in shelf $j$), that is,
$$X_i = (1/\abs{R(i)}) \sum_{r \in R(i)} T_r[h_r(i)] \cdot \conj{\omega_r(i)} \cdot \sgn(\hat{x_i}) \enspace . $$
Further, we check whether $(1-\epsbar)T_l \le X_i \le (1+\epsbar) T_{l-1}$ (resp. $(1-\epsbar) U_j \le X_i \le (1+\epsbar) U_{j+1}$, if $i$ is in shelf $j$ sample), otherwise, $i$ is dropped from the sample.

\vspace*{0.1in}
\emph{Estimating $F_p$.} The estimate for the $p$th frequency moment, $\hat{F}_p$, is the sum of  the  contribution from the shelf samples $\bar{S}_j, j \in [J]$, and  the contribution from the sample groups $\bar{G}_l$, $l = 0, \ldots, L$. For an item $i \in \bar{G}_l$, let $l_d(i)$ be the level at which an item $i$ is discovered. Define 
\begin{align*}
\hat{F}_p^{\shelf} & = \sum_{j=1}^J \sum \left\{ X_i^p \mid i \in \bar{S}_j, (1-\epsbar)U_j \le \abs{X_j} \le (1+\epsbar) U_{j+1}  \right\}, \text{  and } \\
\hat{F}_p^{\ghss} &= \sum_{l=0}^{L}
 2^L \sum \left\{X_i^p \mid i \in \bar{G}_l, l_d(i) < L,
 (1-\epsbar) T_{l_d} \le X_i < (1+\epsbar) T_{l_d-1}\right\} + 2^L \sum_{l_d(i) = L} \abs{\hat{x}_i}^p.
\end{align*}

The final estimate is $\hat{F}_p = \hat{F}_p^{\shelf} + \hat{F}_p^{\ghss}$.



\subsection{Analysis}
\label{sec:ghssreview:prop}
\emph{Notation.}
Let $\ftwores{k}$ be the sum of the squares of all coordinates except the top-$k$ absolute coordinates. That is, suppose the items are placed in decreasing order as per their absolute value of $\abs{x_i}$, that is, let $t$ is an ordering (permutation) of the coordinates $[n]$ such that   $\abs{x_{t_1}} \ge \abs{x_{t,2}} \ge \ldots \abs{x_{t_n}}$, with ties broken arbitrarily. Then, for $0 \le k \le n$, define  $\ftwores{k} = \sum_{j=k+1}^n \abs{x_{t_j}}^2 \enspace . $

For a \ghss~level $l \in [L]$,  $\ftwores{l,k}$ is the random $k$-residual second moment of the frequency vector  in the sampled  substream $\stream_l$. Define the following events.

\begin{align*}
(1) & ~\goodftwo \equiv   F_2 \le  \hat{F}_2  \le (1+0.001/(2p))F_2,\\
(2) &~\smallres_l  \equiv   \ftwores{2C_l,l} \le  1.5 \ftwores{ \lceil
(2\alpha)^l C \rceil }/2^{l-1}, ~~l =0, 1, \ldots, L,\\
(3) &~\smallres \equiv \forall l \in \{0, 1, \ldots, L\}~ \smallres_l,\\
(4) &~\lastlevel~ \equiv  (\hat{f}_{iL} = f_{i}) \text{ and } \forall i \not\in \stream_L, (\hat{f}_{iL}  = 0) \enspace .
\end{align*}
We  condition the analysis on the ``good event''  $ \G \equiv \goodftwo \wedge \smallres \wedge \lastlevel$, that we show holds with probability $ 1- \min(O(\delta), n^{-\Omega(1)})$.
\begin{lemma} \label{lem:GoodEvents:extn}
\G~holds with probability $1- \min(O(\delta), n^{-\Omega(1)})$.
\end{lemma}

The range of item frequencies is subdivided into frequency groups , so that each item belongs to exactly one shelf frequency group or to exactly one \ghss~frequency group.
The frequency group corresponding to the shelf $j$ is $[U_j, U_{j+1})$, for $j=1, \ldots, J$, where, $U_{J+1} = \infty$ and $U_0 =T_0$.  
The frequency group corresponding to level $l$ of \ghss~is $[T_l, T_{l-1})$, where, $T_L = 0$ and $T_{-1} = U_1$. Let $S_j$ (resp. $G_l)$ denote the set of items whose frequency belongs to the frequency group corresponding to shelf $j$ (resp. group $l$).
A few  other events  are used  in the analysis. If $i \in G_l$, then, $\prob{\nocollision(i)} \ge 1- \exp{-\Theta(\log n)}$ as shown in \cite{g:arxiv15} (Lemma 30). If $i \in S_j$, $\prob{\nocollision(i)} \ge 1- \exp{-\Theta(w_j)}$.
We condition some parts of the analysis  on the following additional events. 
\[
\begin{array}{ll}
(\bfromcc{5}) & ~\goodest(i) \equiv \forall l \in [0, \ldots, L], i \in \stream_l \Rightarrow \abs{\hat{x}_{il} - x_i} \le \bigl(
\ftwores{2C_l,l}/{C_l} \bigr)^{1/2} \\
(\icc) & ~\accuest(i) \equiv \forall l \in [0, \ldots, L], i \in \stream_l \Rightarrow
\abs{\hat{x}_{il} - x_i} \le \bigl(
\ftwores{  (2\alpha)^l C  }/(2(2\alpha)^l C)\bigr)^{1/2} \enspace .
\end{array} \]
As shown in \cite{g:arxiv15}, (a) $\goodest(i)$ and $\accuest(i)$ each hold with probability $1- n^{-\Omega(1)}$, and, (b)  $\goodest(i) \wedge \smallres$ imply the event \accuest$(i)$.    For an item $i$ that is discovered at some shelf $j$, the event $\goodest(i)$ is the same as the event $\accuest(i)$ and  is defined as $\abs{\hat{x}_{ij}- x_i} \le \bigl( \ftwores{U_j}/U_j\bigr)^{1/2}$ and holds with probability $1- \exp{-\Theta(w_j)}$.\eat{ Let $\accuest_1 = \{i : \accuest(i) \text{ holds}\}$. Similarly, define the set $\nocollision_1$ and $\smallhh_1$.  Define $\gooditems_1 = \nocollision_1 \cap \accuest_1 \cap \smallhh_1$ .}

\noindent
Lemma~\ref{lem:approxdwise} extends the approximate 2-wise independence property of the sampling scheme of \cite{g:arxiv15} to an approximate $d$-wise independence property.
\begin{lemma} \label{lem:approxdwise}
Let $I = \{i_1,  \ldots, i_{d} \}\subset [n]$ and  $1 \le h \le d$. Let $\accuest(\{i_1, \ldots, i_h\}) \equiv$ $\bigwedge_{k=1}^h \accuest(i_k)$. Then, assuming $d$-wise independence of the hash functions,\\
\begin{multline*}
\sum_{\substack{l_j =0,1, \ldots, L, \\ \forall j =1,2, \ldots, h}}
2^{l_1+ l_2 + \ldots + l_h} \prob{  \bigwedge_{j=1}^h i_j \in \bar{G}_{l_j} \bigl\vert \bigwedge_{j=h+1}^d i_j \in \stream_{l_j},\G, \accuest(\{i_1, \ldots, i_h\})} \in
\\ \prod_{j=1}^h \left(1 \pm 2^{\level(i_j) + 1} n^{-c}\right).
\end{multline*}
\end{lemma}


Lemma~\ref{lem:csk1} bounds $\abs{X_i - \expect{X_i}}$ using the $2d$th moment method. 
It uses Lemma~\ref{lem:2dmom1} as a key component.
\begin{lemma} \label{lem:2dmom1} Let $Z = \sum_{j =1}^t  a_j \omega(j)\chi(j)$, where, $\{\omega(j)\}_{j=1}^t$ is a family of random and  $2d$-wise independent family  roots of the equation $x^q = 1$, $q > 2d$  and integral. Let  $\{\chi(j)\}$ be a $2d$-wise independent family of  indicator variables such that  $ \prob{\chi(j)=1} = 1/C$ and is independent of the $\omega_j$'s. If $  \norm{a}_2^2 \ge  4d \norm{a}_{\infty}^2 C$, then, $\expect{(Z\bar{Z})^d} \le (2)\left( \frac{d \norm{a}_2^2}{C} \right)^d  \enspace . $
\end{lemma}

\begin{lemma} \label{lem:csk1} Suppose $d \le O(\log n)$ and  even and let $s \ge  300 \log (n)$. Then we have that
$\textrm{Pr}\bigl\{\abs{X_i - \abs{x_i}} > \left( \frac{d \ftwores{2C}}{(s/9)C }\right)^{1/2}\mid \nocollision,\goodest\bigr\} < 2^{-2d+1}\enspace .$
\end{lemma}

The use of $q$th roots of unity for the sketches used in the \est~structures allows us nearly unbiased estimators for $X_i^p$. This was first observed in \cite{knpw:stoc10}.
\begin{lemma}[\cite{knpw:stoc10}] \label{lem:expecttheta}
$\card{\expect{X_i^p}-  \abs{x_i}^p\mid \G,\goodest, \nocollision} \le \abs{x_i}^p n^{-\Omega(1)}$.
\end{lemma}

\noindent
For $i \in [n]$, let $x_{li} $ be an indicator variable that is 1 iff $i \in \stream_l$.
Let  $X_i$ denote $\abs{\hat{x}_{i}}$ when $l_d(i) = L$ and otherwise,  let its meaning be unchanged. Let $z_{il}$ be an indicator variable
that is 1 if $i \in \bar{G}_l $  and is  0 otherwise.
Define $$
\hat{F}_p =\sum_{i\in [n]} Y_i, ~~~~~ \text{ where, } 
Y_i = \sum_{l'=0}^{L}2^{l'} z_{il'} X_i^p \enspace . $$
 Let $\H = \G \cap \nocollision \cap \goodest$ and $G' \subset [n]$ be the set of items $G' = \lmargin(G_0)\cup_{l=1}^L G_l$.



 \begin{lemma} \label{lem:dmf2}
Let $B \ge O( n^{1-2/p} \epsilon^{-4/p} \log^{2/p} (1/\delta)))$.
For integral $0 \le d_1, d_2 \le \lceil \log (1/\delta)\rceil $, we have,
$\mathbf{E}\bigl[ \left( \sum_{i \in G'} (Y_i-\expect{Y_i \mid \H})\right)^{d_1} \left( \sum_{i \in G'} (\conj{Y_i} - \expect{\conj{Y_i} \mid \H})\right)^{d_2}\big\vert \H \bigr] \le \left( \frac{ \epsilon F_p}{20}\right)^{d_1+d_2} \enspace . $
\end{lemma}

\subsection{Analysis for the case $\delta \ge n^{-O(1)}$}
For the case $\delta = n^{-O(1)}$, the shelf structure is not needed. Redefine the group $G_0$  to correspond to  the frequency range $[T_0, \infty]$.
The lemmas in this section assume that the family $\{\omega_{lr}(i)\}_{i \in [n]}$ are $O(\log (1/\delta) + \log (n))$-wise independent, and independent across $l,r$ and all hash functions are also $O(\log (1/\delta) + \log (n))$-wise independent.
\begin{lemma} \label{lem:midcentral:conj}
Let $1 \le e,g \le \lceil\log (1/\delta)\rceil, l \in \midreg(G_0)$ and  $\abs{x_l} \ge \bigl( \frac{\ftwores{C}}{C} \bigr)^{1/2}$.
Then,\\ $\expect{ \left( Y_l - \expect{Y_l \mid \H}\right)^e \left(\conj{Y_l} - \expect{ \conj{Y_l} \mid H}\right)^g \mid \H}$ is real and is at most $\left( \frac{a\abs{x_l}^{2p-2}  \ftwores{C}}{\rho C} \right)^{(e+g)/2}$ for some constant $a$. Further, $ \card{\expect{ \left( Y_l - \expect{Y_l \mid \H} \right)^e} \mid \H} \le  \abs{x_l}^{pe} n^{-\Omega(e)} \enspace .$
\end{lemma}
The calculation of the $d$th central moment for the  contribution to $\hat{F}_p$  from the items in $\midreg(G_0)$ requires an upper bound on the following combinatorial sums.

\begin{align} \label{eq:comb:basic}
Q(S_1, S_2) & = \sum_{q=1}^{\min(S_1, S_2)} \sum_{\substack{e_1 + \ldots + e_q = S_1\\ e_j's \ge 1}}
\sum_{\substack{g_1 + \ldots + g_q = S_2 \\ g_j's \ge 1}} \binom{S_1}{e_1, \ldots, e_q} \binom{S_2}{g_1, \ldots, g_q}  \notag  \\
&= \sum_{\{i_1, \ldots, i_q\}} \prod_{r=1}^q \abs{x_{i_r}}^{(p-1)(e_r+g_r)} \prod_{r=1}^q (e_r+g_r)^{(e_r+g_r)/2}, \text{ and } \\
R(S) &  = \sum_{q=1}^{\lfloor  S/2 \rfloor} \sum_{h_1 + \ldots + h_q = S, h_j's \ge 2} \binom{S}{h_1, \ldots, h_q} \sum_{\{i_1, \ldots, i_q\}}\prod_{r\in [q]} \abs{x_{i_r}}^{(p-1)h_r} \prod_{r \in [q]} h_r^{h_r/2}. \label{eq:comb:Rdef}
\end{align}

\begin{lemma} \label{lem:comb:final}
$Q(S_1, S_2) \le R(S_1+S_2) \le \left(16e(S_1+S_2) F_{2p-2} \right)^{(S_1+S_2)/2}$.
\end{lemma}

\begin{lemma} \label{lem:dcentmom:midG0:2}
Let  $C \ge O(n^{1-2/p}/\log (n)) \epsilon^{-2}\log (1/\delta))$. Then, for $0 \le d_1, d_2 \le \log(1/\delta)$, the following expectation is real and is bounded above as follows.
\begin{gather*}
\expect{ \biggl( \sum_{i \in \midreg(G_0)} (Y_i - \expect{Y_i \mid \H}) \biggr)^{d_1} \biggl( \sum_{i \in \midreg(G_0)} (\conj{Y_i} - \expect{ \conj{Y_i}\mid H}) \biggr)^{d_2}\mid \H} \le \left( \frac{\epsilon F_p}{10} \right)^{d_1+d_2}.
\end{gather*}
\end{lemma}

\begin{lemma}  \label{lem:dcentmom} Let $C \ge Kn^{1-2/p}\epsilon^{-2}\log(1/\delta)/\log (n) + L n^{1-2/p} \epsilon^{-4/p} \log^{2/p}(1/\delta)$, where, $K,L$ are constants.\eat{ the constant from Lemma~\ref{lem:dcentmom:midG0:2} and $L$ is the constant from Lemma~\ref{lem:dmf2}.  Let $B$ be such that  $C/B \ge (5p)^2$.}
Then, for $d = \lceil \log (1/\delta) \rceil$, the following expectation is real and is bounded as follows.
\begin{align*}
\expect{ \bigl( \sum_{i \in [n]} (Y_i - \expect{Y_i \mid \H}) \bigr)^d \bigl( \sum_{i \in [n]} ( \conj{Y_i} - \expect{\conj{Y_i} \mid \H})\bigr)^d  \mid \H} \le \bigl( \frac{ \epsilon F_p}{5} \bigr)^{2d}.
\end{align*}
 It follows that  $\prob{ \bigl\lvert\hat{F}_p - F_p \bigr\rvert \ge (\epsilon/2) F_p} \le \delta$.
\end{lemma}
Since, $\H$ holds with probability $ 1- 2^{-\Theta(s)} = 1- 1/n^{-c}$, for any constant $c$ by choosing $s = \Theta(\log n)$ appropriately, we have the following theorem.

\begin{theorem}\label{thm:ub} For each $0 < \epsilon < 1$ and $7/8 \ge \delta \ge n^{-c}$, for any constant $c$, there is a sketching  algorithm that $(\epsilon, \delta)$-approximates $F_p$ with sketching dimension
$O\bigl(n^{1-2/p}\bigl(\epsilon^{-2} \log (1/\delta) + \epsilon^{-4/p} \log^{2/p}(1/\delta) \allowbreak \log n\bigr)\bigr)$ and update time (per stream update)  $O((\log n) \log (1/\delta)))$.
\end{theorem}

\subsection{Analysis for the case $\delta = n^{-\omega(1)}$}
We now extend the analysis for failure probability $\delta$ smaller than $n^{-\Theta(1)}$ and up to $\delta = 2^{-n^{\Omega(1)}}$.
For the \ghss~structure, $\nocollision$ and $\goodest$ may  hold only with probability $1-n^{-\Theta(1)}$. We first  show that the number of items that fail  to satisfy \nocollision~or  \goodest~ is at most $O(\log(1/\delta)/\log n)$ with probability $1 - O(\delta)$. The following lemmas assume the parameter sizes for $B,C, C_l, H_J$ and $H_j$ as described earlier.
\begin{lemma} \label{lem:Gfails}
With probability $1- O(\delta)$, the number of elements for which $\goodest$ or $\nocollision$ fails is at most $O(\log (1/\delta))/(\log n)$.
\end{lemma}
Thus, it is possible that legitimate items are not discovered, or are dropped due to collisions, or mistakenly classified and their contribution added to samples. Let $\text{Error}^{\ghss}$ denote the total contribution of such items to $\hat{F}_p^{\ghss}$  and let $\text{Error}^{\shelf}$ denote the error arising in the estimate of $\hat{F}_p^{\shelf}$ due to analogous errors. As described earlier, we mainly emphasize  the more interesting and complicated case when $H_J = o(H_0)$ (otherwise, $J=1$). \eat{We let   $H_j = H_0 b^j$, $j \in [0,\ldots, J]$ for a geometric decay parameter $b < 1$ and $b = \Omega (1)$ and $w_j = w_0a^j$, for a geometric rising ratio $a > 1$. It is shown in the full version that $a,b$ can be chosen so that $b= \Omega(1)$ and  $\abs{\ln(ab)} = \Omega(1)$.} 

\begin{lemma} \label{lem:error1}
$\text{Error}^{\ghss} \le O (\epsilon^2 F_p/\log n) $ and $\text{Error}^{\shelf} \le O(\max(\epsilon^2 F_p/(\log n) , O(\epsilon^p F_p))) $, each   with probability $1 - \delta/n^{\Omega(1)}$.
\end{lemma}

We first prove a refinement of Lemma~\ref{lem:midcentral:conj}.
\begin{lemma}\label{lem:refine:shelf:conj} [Refinement of Lemma~\ref{lem:midcentral:conj}.]
Let $1 \le e,g \le \lceil\log (1/\delta)\rceil$, $l \in S_j$ and  $\log (1/\delta) = \omega (\log n)$. Assume that $\accuest(l)$ holds and  $H_j \ge \Omega(p^2 E_J)$ and $\abs{x_l} \ge  (F_2/E_j)^{1/2}$.   Then the following expectation is real and is bounded above as follows. 
$$ \expect{ \left( \left(1 + \frac{Z_l}{\abs{x_l}} \right)^p-1\right)^e \left( \left(1 + \frac{\conj{Z_l}}{\abs{x_l}} \right)^p-1\right)^g\mid \H} \le  c^h \abs{x_l}^{-h} \biggl( \frac{F_2}{H_j} \biggr)^{h/2} \bigl( \min\bigl( h/w_j, 1 \bigr) \bigr)^{h/2} $$
where, $h = e+g$ and $c $ is an absolute constant.
Therefore
\begin{align*}
\expect{ \left( Y_l - \expect{Y_l}\right)^e \left(\conj{Y_l} - \expect{ \conj{Y_l}}\right)^g}
\le c^h \abs{x_l}^{(p-1)h} \left( \frac{F_2}{H_j} \right)^{h/2} \left( \min\left( h/w_j, 1 \right) \right)^{h/2} \enspace .
\end{align*}
\end{lemma}

Lemma~\ref{lem:dmom:shelf} considers the $2d$th central moment of the contribution to $\hat{F}_p^{\shelf}$ from all but the outermost  shelf, and  from the set of  outermost shelf items denoted $S_J$, separately. Let $S' = S_1 \cup \ldots \cup S_{J-1}$.
\begin{lemma} \label{lem:dmom:shelf}
Let $0 \le d_1, d_2 \le \lceil \log (1/\delta) \rceil$ and integral and $c_1,c_2$ be constants. Then,
\begin{align*}
\expect{ \bigl( \sum_{i \in  S'} (Y_i - \expect{Y_i \mid \H}) \bigr)^{d_1} \bigl( \sum_{i \in S'} (\conj{Y_i} - \conj{\expect{Y_i \mid\H}}) \bigr)^{d_2} \mid\H} \le \left( c_1\epsilon F_p \right)^{d_1 + d_2}.
\end{align*}
\begin{align*}
\expect{ \bigl( \sum_{i \in S_J} (Y_i - \expect{Y_i \mid\H}) \bigr)^{d_1} \bigl( \sum_{i \in S_J} (\conj{Y_i} - \conj{\expect{Y_i \mid \H}}) \bigr)^{d_2} \mid \H} \le \left( c_2\epsilon F_p \right)^{d_1 + d_2}.
\end{align*}
\end{lemma}
Combining Lemmas~\ref{lem:dmf2}, ~\ref{lem:dcentmom:midG0:2} and ~\ref{lem:dmom:shelf} with  Lemma~\ref{lem:error1}, we obtain the following.
\begin{lemma} \label{lem:dcentmomsmalldelta} 
 $\exists$ constant $c$ s.t. for  $1 \le d \le \lceil \log (1/\delta) \rceil$, the following holds. $$\expect{ \bigl(  \sum_{i \in  [n]} (Y_i - \expect{Y_i \mid \H}) \bigr)^d \bigl( \sum_{i \in  [n]} ( \conj{Y_i} - \expect{\conj{Y_i} \mid \H})\bigr)^d \mid \H} \le ( c \epsilon F_p )^{2d} \enspace . $$
Hence, $\prob{ \card{ \hat{F}_p - F_p} \le \epsilon F_p)} < \delta \enspace.$
\end{lemma}

\begin{theorem} \label{thm:ubfinal}  For each $0 < \epsilon < 1$ and $7/8 \ge\delta \ge  2^{-n^{\Omega(1)}}$, there is a sketching  algorithm that $(\epsilon, \delta)$-approximates $F_p$ with sketching dimension
$O\bigl(n^{1-2/p}\bigl(\epsilon^{-2} \log (1/\delta) + \epsilon^{-4/p} \log^{2/p}(1/\delta)\log n\bigr)\bigr)$ and update time (per stream update)  $O((\log n)\log (1/\delta))$.
\end{theorem}

\bibliographystyle{alpha}
\bibliography{streams}




\appendix

\section{Proofs for our Algorithm} \label{app:algoproofs}

We first  present some details of  the \ghss~structure that were originally presented in  summary form in Section~\ref{sec:algo}.

\emph{\ghss~sampling.} For the \ghss~structure, the  input stream $\stream$ is sub-sampled hierarchically to produce random sub-streams $\stream_0, \stream_1, \ldots, \stream_L$, corresponding to each of the levels $0, \ldots, L$.
 The stream
$\stream_0$ is the  entire input stream. For each $l=1, \ldots, L$, $\stream_{l}$ is obtained by sampling each   item $i$
appearing in $\stream_{l-1}$ with probability $1/2$. If $i$ is sampled, then all its records $(i,v)$ are
included in $\stream_1$, otherwise none of its records are included. The sampling uses independently chosen  random hash functions
  $g_1, g_2, \ldots, g_{L}$  each mapping $[n] \rightarrow \{0,1\}$.  $i$ is included in
  $\stream_l$ iff $g_1(i) =1,  g_2(i)=1,  \ldots, g_{l}(i)=1$. The $g_l$'s are chosen
  from a $ O(\log (1/\delta) + \log n)$-wise independent hash family.

\eat{\emph{\hh~and \est~structures.} The \hh$(C_l,s)$ is precisely a \countsketch~structure \cite{ccf:icalp02}. The \est$(C_l,2s)$ structure  is its modification that,   instead of using Rademacher sketches, applies  sketches using random $q$th roots of unity, where, $q = O(\log (1/\delta))$. Thus, at level $l$ and for table indexed $j$ with hash function $h_{lj}: [n] \rightarrow [16C_l]$,  and bucket index $b \in [16C_l]$, the sketch is $T_{lj}[b] = \sum_{h_j(i)=b} x_i \omega_{lj}(i)$, where, $\{\omega_{lj}(i)\}_{i \in [n]}$ is a random family of $q$th roots of unity that is $O(\log (1/\delta))$-wise independent. It is assumed that the hash functions across the tables and distinct levels, and the seeds of the family of the  random $q$th roots of unity, are independent. The estimate for $x_i$ obtained from the $j$th table of  level $l$ structure is defined as $\hat{X}_{lj} = T_{lj}[h_{lj}(i)] \conj{\omega_{lj}(i)}$. The error,  $Z_{lj} = \hat{X}_{lj} - x_i$ satisfies $ \expect{Z_{lj}^t} = 0$, for $t < q$.  We will define an estimate $X_i$ for $x_i$ as the averages of $\hat{X}_{lj}$ for those $j$'s that satisfy the \nocollision~property defined later. It can then be shown that, (a) $\expects{X_i^p} = \abs{x_i}^p (1 \pm n^{-\Omega(1)})$ as observed in \cite{knpw:stoc10}, and (b) the $d = \log (1/\delta)$th moment scaled sum of the $X_i^p$ terms is amenable to calculation.
}

Let $C = C_0 = \Theta(p^2 n^{1-2/p} \epsilon^{4/p} \log^{2/p}(1/\delta))$ be the height of \countsketch~and \est~structures at level 0 of the \ghss~structure. The height of the \countsketch~and \est~structures at level $l$ of the \ghss~structure is defined as
$$C_l = \alpha^l C_0, ~~~l = 0,1, \ldots, L \enspace . $$
Let $\epsbar = 1/(54p)$ be a constant. Let $B$ be another parameter closely related to $C$ as follows.
$$ B = \epsbar^2 C, \text{ and  let } B_l = \epsbar^2 C_l, \text{ for } l=0,1, \ldots, L\enspace. $$
The level-wise frequency thresholds are defined as follows.
\begin{align*}
T_0 &= \left( \frac{\hat{F}_2}{B} \right)^{1/2},
T_l = \left( \frac{1}{(2\alpha)} \right)^{1/2} T_0  = \left( \frac{\hat{F}_2}{2^l B_l} \right)^{1/2} l=0,1, \ldots, L-1 \enspace .
\end{align*}
$T_L$ is defined as $0^+$, that is, $a> T_L$ iff $a > 0$.
Another threshold $Q_l$ is used for defining when an item is discovered at level $l$, and is defined as follows.
\begin{align*}
Q_l = T_l (1- \epsbar) ~~l=0,1, \ldots, L-1 \enspace .
\end{align*}
Similarly, $Q_L$ is defined as $0^+$.

The \ghss~level groups are sets of items identified with  frequency ranges and essentially follows the scheme of \cite{g:arxiv15}.
 The group $G_l$ consists of all items in the  frequency range $[T_l, T_{l-1})$, for $l =0,1, \ldots, L-1$. That is,
 However, the group $G_0$ consists of the frequency range $[T_0, U_1)$, that is, $T_{-1}$ is identified with $U_1$ of the shelf structure. The group $G_L$ is identified with the frequency range $(T_0, T_{L-1})$.
 The ratio $T_{l-1}/T_{l} = (2\alpha)^{1/2}$, for  $l=0, 1 , \ldots, L-1$.  More precisely, we have the following group definitions.
\begin{align*}
G_l & = \{i: T_l \le \abs{x_i} < T_{l-1} \}, ~~l=0,1, \ldots, L-1 \\
G_L & = \{i: 0 < \abs{x_i} < T_{L-1}\}
\end{align*}
where, for the definition of the group $G_0$ it is assumed that $T_{-1} = U_1$ (and $T_0 = U_0$).

For analysis purposes, the  \ghss~ level groups $G_0, G_1, \ldots, G_L$ are  partitioned into subsets  $\lmargin(G_l), \midreg(G_l)$ and $\rmargin(G_l)$ \cite{g:arxiv15} as follows.
 \begin{align*}
 \lmargin(G_l) &= \{i: T_l \le \abs{x_i}  <  T_l(1+  \epsbar), \\
 \rmargin(G_l) &= \{ i :  T_{l-1}(1 - 2\epsbar) \le \abs{x_i} <  T_{l-1} \}, \text{ and }  \\
 \midreg(G_l) &= \{ i: T_l+ T_{l}\epsbar \le \abs{x_i}  < T_{l-1} - 2T_{l-1}\epsbar\} \enspace .
 \end{align*}
For $G_0$, there is no $\rmargin(G_0)$ defined, and analogously for $G_L$, there is no $\lmargin(G_L)$ defined.  Instead, \midreg$(G_0)$ and $\midreg(G_L)$ are extended as follows.
\begin{align*}
\midreg(G_0) &= \{i: \abs{x_i} \ge T_0(1+\epsbar)\}, \\
\midreg(G_L) &=\{i: 0 < \abs{x_i} < T_{L-1}(1-2\epsbar)\}
\end{align*}

Throughout the analysis, we will obtain bounds on expressions involving probability of events conditioned on $\G$. It is often much easier to prove the same expressions without conditioning on $\G$, and then deriving upper and lower bounds on the probability conditioned by $\G$. The following lemma from \cite{g:arxiv15} (Fact 21) is useful for this purpose.
\begin{lemma}[Fact 21 in \cite{g:arxiv15}.] \label{lem:fact}
Let $E$ and $G$ be a pair of events. Then,
$$ \left\lvert \prob{E \mid F} - \prob{E} \right\rvert \le \frac{ 1- \prob{F}}{\prob{F}}
$$
\end{lemma}

\subsection{Proofs}

The following lemma is a slight modification of Lemma 33 of \cite{g:arxiv15} by reducing the conditions on which the probability event depends. It may be noted that the event $\G$ of \cite{g:arxiv15} has been considerably trimmed to define the event denoted by $\G$ in this work.
\begin{lemma}[ Modified from \cite{g:arxiv15}] \label{lem:margin}
Let $i \in G_l$.
\begin{enumerate}
\item If $i \in \midreg(G_l)$, then,  $$\card{2^l\prob{ i \in \bar{G}_l \mid \G, \goodest(i)} - 1} \le 2^l \min(O(\delta), n^{-\Omega(1)}) \enspace. $$
Further, conditional on $\G \wedge \accuest(i)$,  (i)  $i \in \bar{G}_l $ iff $i \in \stream_l$,  and, (ii) $i$ may not  belong to any $\bar{G}_{l'}$, for $l' \ne l$.
\item If $i \in \lmargin(G_l)$, then
\begin{gather*}
 \card{2^{l+1} \prob{i \in \bar{G}_{l+1} \mid \G, \goodest(i)} + 2^l\prob{i \in \bar{G}_l \mid \G, \goodest(i)} -
 1}\\  \le 2^l \min(O(\delta), n^{-\Omega(1)}) \enspace. \end{gather*}
Further, conditional on $\G \wedge \goodest(i)(i)$,  $i$ may  belong to either $ \bar{G}_l$ or $\bar{G}_{l+1}$,
but not to any other sampled group.
\item If $i \in \rmargin(G_l)$, then
 \begin{gather*}
 \card{ 2^l \prob{ i \in \bar{G}_l \mid \G, \goodest(i)} + 2^{l-1}\prob{ i \in \bar{G}_{l-1} \mid \G, \goodest(i)} -
 1}\\  \le 2^l \min(O(\delta), n^{-\Omega(1)}) \enspace .
\end{gather*}
 Further, conditional on $\G \wedge \accuest(i)$,   $i$ can belong to  either $ \bar{G}_{l-1}$ or $
 \bar{G}_{l}$ and not to any other sampled group.
 \end{enumerate}
\end{lemma}
The proof of Lemma~\ref{lem:margin} is similar to the  proof of  Lemma 33 in \cite{g:arxiv15}. We provide an outline here to emphasize the slight generality of the current version. The main difference in the statement is that the event $\G$ here is the conjunction of significantly fewer events than the event $\G$ defined in \cite{g:arxiv15}.
\begin{proof}

Let $i \in \midreg(G_l)$. Suppose $i \in \stream_l$. We are given that $\G$ and $\goodest(i)$ hold. It is easy to see that \smallres $\wedge$ \goodest$(i)$ imply  $\accuest(i)$, which therefore also holds. We therefore have,
$$\abs{\hat{x}_{il} - x_i} \le \frac{ \ftwores{\lceil (2\alpha)^l C}}{4 (2 \alpha)^l C}  \le \epsbar \left( \frac{\hat{F_2}}{ (4) 2^l  B_l} \right)^{1/2} \le \epsbar T_l$$
where the first inequality follows from \accuest$(i)$ and the second uses $\goodftwo$ followed by a relaxation using the definitions of the thresholds. Thus
$$ \abs{\hat{x}_{il}} \ge \abs{x_i} - \epsbar T_l \ge (1+\epsbar)T_l - \epsbar T_l = T_l \enspace. $$
Since, $ T_l > Q_l$, $\abs{\hat{x}_{il}} > Q_l$ and by the  definition of discovery of item, $i$ qualifies to be  discovered at level $l$.
Further, a direct calculation shows that if $i \in \stream_r$ for any $r < l$ then, $i$ cannot be discovered at any level $r < l$.  Indeed, by \accuest$(i)$ we have,
$
\abs{\hat{x}_{ir} - x_i} \le \epsbar T_r
$
and therefore,
\begin{align*}
\abs{\hat{x}_{ir}} &  \le \abs{x_i} + \epsbar T_r < (1-2\epsbar)T_{l-1} + \epsbar T_r   \le (1-\epsbar)T_{r} \le Q_r
\end{align*}
that is, $i$ is not discovered at any level $r < l$. A similar calculation shows that $i$ is not discovered at any shelf $j \in [J]$.

Thus, under the presumptions of $\G \wedge \goodest(i)$,  if $i \in \stream_l$, then  $i$ is discovered at level $l$ and is not discovered at any lower level, that is, $l$ is the lowest level at which $i$ is discovered. Secondly, $\abs{\hat{x}_{il}} \ge T_l$ implying that $i$ is included into the group sample at level $l$, that is, $i \in \bar{G}_l$.

Therefore, for any level $l \ge 1$,
\begin{align*}
\prob{i \in \bar{G}_l \mid \G, \goodest(i),  i \in \stream_l} = 1
\end{align*}
By law of total probability,
\begin{align*}
 &\prob{i \in \bar{G}_l \mid \G, \goodest(i)} \\
  &= \prob{ i \in \bar{G}_l \mid \G, \goodest(i), i \in \stream_l} \prob{ i \in \stream_l \mid \G, \goodest(i)}\\
  & ~~+ \prob{ i \in \bar{G}_l \mid \G, \goodest(i), i \not\in \stream_l} \prob{i \not\in \stream_l \mid \G, \goodest(i)} \\
  & = 1 \cdot \prob{ i \in \stream_l \mid \G, \goodest(i)} + 0 \cdot \prob{i \not\in \stream_l \mid \G, \goodest(i), i \not\in \stream_l}\\
  & = \prob{i \in \stream_l \mid \G, \goodest(i)} \\
  & \in \prob{i \in \stream_l} \pm \left( \frac{ 1- \prob{\G, \goodest(i)}}{\prob{\G, \goodest(i)}} \right), ~~~~ \text{ by Lemma~\ref{lem:fact}} \\
  & = 2^{-l}  \pm \min(n^{-\Omega(1)},O(\delta)) \enspace .
\end{align*}
The last step is obtained  as follows. $\accuest(i)$ is implied by $\smallres \wedge \goodest(i)$. Since $s = \Theta(\log n)$, $\G$ holds with probability $1-\min(O(\delta), \exp{-\Omega(s)} = n^{-\Omega(1)})$ and \goodest$(i)$ holds with probability $1-n^{-\Omega(1)}$. Multiplying by $2^l$ on both sides and substracting 1, we get
\begin{align*}
\left\lvert 2^l \prob{i \in \bar{G}_l \mid \G, \goodest(i)}- 1\right\rvert \le 2^l \min(n^{-\Omega(1)},O(\delta)) \enspace .
\end{align*}

Items 2 and 3 in the statement of the lemma can be proved in a similar manner by following the steps  in the proof of Lemma 33 of \cite{g:arxiv15} and simplifying them in the above manner.
\end{proof}

\eat{
Lemma~\ref{lem:hsscond} essentially repeats the results of Lemma~\ref{lem:margin},
conditional upon the event that another item maps to some sampled group. This property is
useful in variance calculations later.

\begin{lemma}[\cite{g:arxiv15}] \label{lem:hsscond}  Let $i,j \in [n]$,  $i \ne j$ and $j \in  G_r$. Then,
$
  \sum_{r'=0}^L 2^{r'} \probb{j \in \bar{G}_{r'} \mid i \in \stream_l, \G} \allowbreak  =1 \pm O(2^{r}n^{-c}) \enspace .$ In particular, the following hold.
\begin{enumerate}
\item If $j \in \midreg(G_r)$, then $
  2^r \probb{ j \in \bar{G}_r \mid i \in \stream_l, \G}\allowbreak  = 1 \pm  2^r n^{-c}
$ and  for any $r \ne r'$, $ \probb{j \in \bar{G}_{r'}\mid i \in \stream_l, \G} = 0$.

\item If $j \in \lmargin(G_r)$, then,
$
  2^{r+1} \prob{j \in \bar{G}_{r+1} \mid i \in \stream_l, \G} + 2^r\prob{j \in
 \bar{G}_r \mid i \in \stream_l,\G}\allowbreak
  = 1 \pm   2^{r+1} n^{-c} \enspace .
$
Further, for any $ r' \not\in \{r,r+1\},  \prob{j \in \bar{G}_{r'}\mid i \in \stream_l, \G}  = 0$.
\item If $j \in \rmargin(G_r)$, then  $
 2^{r} \prob{ j \in \bar{G}_r \mid i \in \stream_{l}, \G} + 2^{r-1}\prob{ j \in
 \bar{G}_{r-1} \mid i \in \stream_{l},\G} \allowbreak
  = 1 \pm  2^{r+1} n^{-c} $.
 Further, for any $ r' \not\in \{r-1,r\},  \prob{j \in \bar{G}_{r'}\mid i \in \stream_l, \G} = 0$.
 \end{enumerate}
\end{lemma}

}
Suppose it is given that $i$ belongs to \ghss~group $G_l$, for some level $l$. Then, let $\level(i)$ denote this value of $l$. From the definition of threholds $T_{l'}$, for $l' \in [0, \ldots, L]$ this equals
Let $\level(i)$ denote the true  ``level'' of $i$, that is,
 \begin{align*}\level(i) =
 \begin{cases} 0 & \text{ if } x_i^2 \ge \hat{F}_2/B  \\
 \left\lfloor 2\log_{2\alpha} \left(  \hat{F}_2/( x_i^2B)\right)\right\rfloor & \text{ otherwise.}
 \end{cases}
 \end{align*}
Let $i$ be an item that belongs to \ghss~group $G_l$, where, $l = \level(i)$. Assuming  $\G$ and $\goodest(i)$, by Lemma~\ref{lem:margin}, $i$ may either be correctly classified into $\bar{G}_{\level(i)}$. However, if $i \in \lmargin(G_l)$, then due to estimation errors of $\hat{x}_{il}$, $i$ may be classified to belong to either $\bar{G}_{\level(i)}$ or to  $\bar{G}_{\level(i)+1}$. Finally, if $i \in \rmargin(G_l)$, then due to estimation errors again $i$ may be classified to belong to either $\bar{G}_{\level(i)}$ or to $\bar{G}_{\level(i)-1}$. In each case,  assuming $\G$ and $\goodest(i)$, there is zero probability that $i$ would be classified into a third group.

In certain equations, it is sometimes needed to sum or iterate over the possible groups each item $i$  can be sampled into.  Under the conditions $\G$ and $\goodest(i)$, it then  suffices to iterate over only the  three groups $\level(i)-1, \level(i), \level(i)+1$.
Given $i \in [n]$ and $l \in \{0,1, \ldots, L\}$, we use the notation $l ~\consistw i$ to denote that  $l \in \{\level(i)-1, \level(i), \level(i)+1\}$. \eat{If  $i \in \lmargin(G_{\level(i)})$, then as proved by Lemma~\ref{lem:margin}, it is possible that $i$ is classified into the sampled  group at either  level $\level(i)$ or $\level(i)+1$ but to no other level. If $i \in \rmargin(G_{\level(i)})$, then $i$ may be classified into the  sampled group at level $\level(i)$ or $\level(i)-1$, but to no other level. Finally, if $i \in \midreg(G_{\level(i)})$, then $i$ is classified only into $\bar{G}_{\level(i)}$. Thus, if $i$ is sampled into a group at level $l$, it follows that $l ~\consistw i$.}

\begin{lemma} \label{lem:basecase}
Let $i_1,  \ldots, i_d \in [n]$. Then,
\begin{align*}
&\sum_{l_1 \consistw~ i_1} 2^{l_1} \prob{ i_1 \in \bar{G}_{l_1} \vert \bigwedge_{j=h+1}^d i_j \in \stream_{l_j}, \G, \goodest(i_1)}\\ & \hspace*{1.0in} = 1 \pm 2^{\level(i_1)} \min (n^{-\Omega(1)}, O(\delta)) \enspace .
\end{align*}
\end{lemma}

\begin{proof}

\emph{Case 1:} $i \in \midreg(G_l)$. Then, conditional on $\G$ and $\goodest(i)$, as shown in the proof of Lemma~\ref{lem:margin}, we have,  (i) $i \in \bar{G}_l$ iff $i \in \stream_l$, and (ii) $i \not\in \bar{G}_r$, for any $r \ne l$.
Therefore, in this case,
\begin{multline}\label{eq:lem:induction:mid1}
 \sum_{l_1 \text{ consist. with } i_1} 2^{l_1} \prob{ i_1 \in \bar{G}_{l_1} \left\vert \bigwedge_{j=2}^d i_j \in \stream_{l_j} \right., \G, \goodest(i)}\\  = 2^l \prob{i \in \stream_l  \left\vert \bigwedge_{j=2}^d i_j \in \stream_{l_j} \right., \G, \goodest(i)}
\end{multline}
Now, we have, $ \prob{i \in \stream_l  \left\vert \bigwedge_{j=2}^d i_j \in \stream_{l_j} \right.} = 2^{-l}$ assuming that (i) the hash functions $g_l$ are each drawn from a $d$-wise independent family, for each $l \in [L]$,  and, (ii) the $g_l$'s are independent across $l$. Therefore, from Lemma~\ref{lem:fact}, we have,
$$\left\lvert \prob{i \in \stream_l  \left\vert \bigwedge_{j=2}^d i_j \in \stream_{l_j} \right., \G, \goodest(i)} -2^{-l} \right\rvert \le \min(n^{-\Omega(1)}, O(\delta))$$
Multiplying above equation by $2^l$ and substituting in Eqn.~\eqref{eq:lem:induction:mid1}, we obtain
\begin{align*}
& 2^l \prob{i \in \stream_l  \left\vert \bigwedge_{j=2}^d i_j \in \stream_{l_j} \right., \G, \goodest(i)}  =1 \pm 2^l \min(n^{-\Omega(1)}, O(\delta)) \enspace .
\end{align*}

\emph{ Case 2:} $i_1 \in \lmargin(G_l)$.  Then, conditional on $\G$ and \goodest$(i)$, the following statements follow from Lemma~\ref{lem:margin}.
\begin{enumerate}
\item $i_1$ cannot be discovered at level smaller than $l$ (with prob. 1).
\item If $i_1 \in \stream_l$, then it is discovered at level $l$ (with probability 1).
\item $i_1 \in \bar{G}_l$ iff $i_1 \in \stream_l$ and $\abs{\hat{x}_{i_1,l}} \ge T_l$.
\item $i_1 \in \bar{G}_{l+1}$ iff $i_1 \in \stream_1$ and $\abs{\hat{x}_{i_1,1}} < T_l$ and a random coin $K_i$ turns heads.
\end{enumerate}
Let $E_{2,d}$ denote the event
$$E_{2,d} = \bigwedge_{j=2}^d i_j \in \stream_{l_j} \enspace . $$
Therefore,
\begin{align}
& \sum_{l_1 \text{ consist. with } i_1} 2^{l_1} \prob{ i_1 \in \bar{G}_{l_1} \vert E_{2,d}, \G, \goodest(i_1)} \notag \\
&= 2^l \prob{i_1 \in \bar{G}_{l} \vert E_{2,d}, \G, \goodest(i_1)}
+ 2^{l+1} \prob{ i_1 \in \bar{G}_{l+1} \vert E_{2,d}, \G, \goodest(i_1)} \notag \\
& = 2^l \prob{ \abs{\hat{x}_{il}} \ge T_l, i_1 \in \stream_l \vert E_{2,d}, \G, \goodest(i_1)}\\
& \hspace*{0.2in}+ 2^{l+1} \prob{ \abs{\hat{x}_{il}} < T_l, i_1 \in \stream_l, K_i=1 \vert E_{2,d}, \G, \goodest(i_1)} \notag \\
& = 2^l\prob{ \abs{\hat{x}_{il}} \ge T_l  \vert  i_1 \in \stream_l, E_{2,d}, \G, \goodest(i_1)} \prob{ i_1 \in \stream_l \vert E_{2,d}, \G, \goodest(i_1)} \notag \\
&\hspace*{0.2in} + 2^{l} \prob{ \abs{\hat{x}_{il}} < T_l \vert  i_1 \in \stream_l, E_{2,d}, \G} \prob{ i_1 \in \stream_l, \vert E_{2,d}, \G, \goodest(i_1)} \notag \\
& = 2^l \prob{ i_1 \in \stream_l \vert E_{2,d}, \G, \goodest(i_1)} \label{eq:lem:induction:base:lmarg}
\end{align}
Now, by $d$-wise independence of the hash functions $g_1, \ldots, g_l$, we have,
$\prob{ i_1 \in \stream_l, \vert E_{2,d}} = 2^{-l}$. Therefore,
\begin{align*}
\prob{ i_1 \in \stream_l \vert E_{2,d}, \G, \goodest(i_1)} = 2^{-l} \pm \min( n^{-\Omega(1)}, O(\delta))
\end{align*}
Substituting in Eqn.~\eqref{eq:lem:induction:base:lmarg}, we have,
\begin{align*}
2^l \prob{ i_1 \in \stream_l \vert E_{2,d}, \G, \goodest(i_1)}
 = 1 \pm 2^l\min( n^{-\Omega(1)}, O(\delta)) \enspace .
\end{align*}

\emph{Case 3:} $i_1 \in \rmargin(G_l)$. As before, let $E_{2,d}$ denote the event $ \bigwedge_{j=2}^d i_j \in \stream_{l_j}$.  By Lemma~\ref{lem:margin} and assuming $\G$ and \goodest$(i)$,  we have the following observations.
\begin{enumerate}
\item It is possible for $i_1 \in \rmargin(G_l)$ to be discovered at level $l-1$. This happens if $\abs{\hat{x}_{i_1, l-1}} > Q_{l-1} = T_{l-1}(1-\epsbar)$.
\item It is not possible (i.e., is a zero probability event) that $i_1$ is discovered at levels lower than $l-1$.
\item It is also possible for $i_1$ to be classified into the sample at level $l-1$, that is, $i_1 \in \bar{G}_{l-1}$. This happens iff $i \in \stream_{l-1}$ and  $\abs{\hat{x}_{i_1, l-1}} \ge T_{l-1}$.
\item It is possible for $i_1$ to be classified  into the sample at level $l$, that is, $i_1 \in \bar{G}_l$. This can happen in one of the two mutually exclusive ways.
    \begin{enumerate}
    \item $i \in \stream_{l-1}$ and $Q_{l-1} < \abs{\hat{x}_{i_1, l-1}} < T_{l-1}$ and $K_i =1$.
    \item $i \in \stream_{l} $ and $\abs{\hat{x}_{i_1, l-1}} \le Q_{l-1}$.
    \end{enumerate}
\end{enumerate}

The \emph{LHS} can be written as follows. For $i_1 \in \rmargin(G_l)$, the levels consistent with $i_1$ are $l-1$ and $l$. Thus, we have,
\begin{align} \label{eq:lem:induction:basermarg}
& \sum_{l_1 \text{ consist. with } i_1} 2^{l_1} \prob{ i_1 \in \bar{G}_{l_1} \vert E_{2,d}, \G, \goodest(i_1)} \notag \\
&= 2^{l-1} \prob{i_1 \in \bar{G}_{l-1}  \left\vert E_{2,d} \right., \G, \goodest(i_1)} + 2^l \prob{i_1 \in \bar{G}_l  \left\vert E_{2,d}\right., \G, \goodest(i_1)} \notag \\
& = 2^{l-1} \prob{\abs{\hat{x}_{i_1,l-1}} \ge T_{l-1}, i_1 \in \stream_{l-1}\left\vert E_{2,d}\right., \G, \goodest(i_1)}\notag \\
 & ~+   2^l \prob{ Q_{l-1} \le \abs{\hat{x}_{i_1, l-1} }< T_{l-1}, K_i = 1,  i_1 \in \stream_{l-1} \mid E_{2,d}, \G, \goodest(i_1)} \notag \\
& ~~~~~  +  2^l \prob{ \abs{\hat{x}_{i_1, l-1}} < Q_{l-1}, g_l(i_1)=1, i_1 \in \stream_{l-1} \mid E_{2,d}, \G, \goodest(i_1)}
\end{align}
The final  expression is the sum of three terms. We consider these terms individually and then combine them. The first term can be written as
\begin{align} \label{eq:lem:induction:base:rmarg:1a}
& \text{ Term 1} \notag \\
& =2^{l-1} \prob{\abs{\hat{x}_{i_1,l-1}} \ge T_{l-1}, i_1 \in \stream_{l-1}\vert E_{2,d}, \G, \goodest(i_1)} \notag \\
& = 2^{l-1} \prob{\abs{\hat{x}_{i_1,l-1}} \ge T_{l-1} \vert i_1 \in \stream_{l-1}, E_{2,d}, \G, \goodest(i_1)} \notag \\
& \hspace*{1.0in} \cdot \prob{i_1 \in \stream_{l-1} \mid E_{2,d}, \G, \goodest(i_1)}
\end{align}
Now, $\prob{i_1 \in \stream_{l-1} \vert E_{2,d}, \G, \goodest(i_1)} = 2^{-(l-1)} \pm \min(n^{-\Omega(1)}, O(\delta))$. Substituting in Eqn.~\eqref{eq:lem:induction:base:rmarg:1a}, we obtain,
\begin{align}\label{eq:lem:induction:base:rmarg:1b}
\text{Term 1} =  &\prob{\abs{\hat{x}_{i_1,l-1}} \ge T_{l-1} \vert i_1 \in \stream_{l-1}, E_{2,d}, \G, \goodest(i_1)} \notag \\ &\hspace*{1.0in}\left(1 \pm 2^{l-1}\min(n^{-\Omega(1)}, O(\delta))\right) \enspace .
\end{align}
 The second term is
\begin{align}\label{eq:lem:induction:base:rmarg:2a}
& \text{Term 2} \notag \\
&= 2^l \prob{ Q_{l-1} < \abs{\hat{x}_{i_1, l-1} }< T_{l-1}, K_i = 1,  i_1 \in \stream_{l-1} \mid E_{2,d}, \G, \goodest(i_1)} \notag \\
& = 2^{l-1} \prob{ Q_{l-1} <\abs{\hat{x}_{i_1, l-1} }< T_{l-1}, i_1 \in \stream_{l-1} \mid E_{2,d}, \G, \goodest(i_1)}  \enspace .
\end{align} since, $K_i = 1$ happens with probability $1/2$ and is independent of all other random terms occurring  in the  expression. By definition of conditional probability,
Eqn.~\eqref{eq:lem:induction:base:rmarg:2a} equals
\begin{align}\label{eq:lem:induction:base:rmarg:2b}
&\text{Term 2} \notag \\
& = 2^{l-1} \prob{ Q_{l-1} <\abs{\hat{x}_{i_1, l-1} }< T_{l-1} \mid i_1 \in \stream_{l-1}, E_{2,d}, \G, \goodest(i_1)} \notag \\
& \hspace*{1.0in} \cdot  \prob{ i_1 \in \stream_{l-1} \mid E_{2,d}, \G, \goodest(i_1)}\notag \\
& = 2^{l-1} \prob{ Q_{l-1} <\abs{\hat{x}_{i_1, l-1} }< T_{l-1} \mid i_1 \in \stream_{l-1}, E_{2,d}, \G, \goodest(i_1)}\notag  \\
& \hspace*{1.0in} \cdot \left( 2^{-(l-1)} \pm \min(n^{-\Omega(1)}, O(\delta))\right) \notag \\
& = \prob{ Q_{l-1} <\abs{\hat{x}_{i_1, l-1} }< T_{l-1} \mid i_1 \in \stream_{l-1}, E_{2,d}, \G, \goodest(i_1)} \notag \\ & \hspace*{1.0in} \cdot  \left(1 \pm 2^{(l-1)} \min(n^{-\Omega(1)}, O(\delta))\right)   \enspace .
\end{align}
Adding the simplified terms for the first and second expression from Eqns. ~\eqref{eq:lem:induction:base:rmarg:1b} and ~\eqref{eq:lem:induction:base:rmarg:2b}, we obtain,
\begin{align} \label{eq:lem:induction:base:rmarg:12c}
\text{ Term 1} + \text{ Term 2}\notag  &= \prob{ \abs{\hat{x}_{i_1, l-1} }> Q_{l-1} \mid i_1 \in \stream_{l-1}, E_{2,d}, \G, \goodest(i_1)}\\ &  \left(1 \pm 2^{(l-1)} \min(n^{-\Omega(1)}, O(\delta))\right)\enspace .  \enspace .
\end{align}

We now consider Term 3. By definition of conditional probability, we have,
\begin{align} \label{eq:lem:induction:base:rmarg:3aa}
& \text{Term 3} \notag \\
& =  2^l \prob{ \abs{\hat{x}_{i_1, l-1}} < Q_{l-1}, g_l(i_1)=1, i_1 \in \stream_{l-1} \mid E_{2,d}, \G, \goodest(i_1)} \notag \\
& = 2^l\prob{ \abs{\hat{x}_{i_1, l-1}} < Q_{l-1}, g_l(i_1) = 1 \mid  i_1 \in \stream_{l-1}, E_{2,d}, \G \goodest(i_1)} \notag \\
& \hspace*{1.0in} \cdot  \prob{  i_1 \in \stream_{l-1} \mid E_{2,d}, \G, \goodest(i_1)} \enspace .
\end{align}
As argued earlier, $ \prob{  i_1 \in \stream_{l-1} \mid E_{2,d}, \G, \goodest(i_1)} = 2^{-(l-1)}  \pm \min(n^{-\Omega(1)}, O(\delta))$. Simplifying by substituting in Eqn.~\eqref{eq:lem:induction:base:rmarg:3aa}, we obtain,
\begin{align} \label{eq:lem:induction:base:rmarg:3a}
& \text{Term3} \notag \\
& = 2\prob{ \abs{\hat{x}_{i_1, l-1}} < Q_{l-1}, g_l(i_1) = 1 \mid  i_1 \in \stream_{l-1}, E_{2,d}, \G,  \goodest(i_1)}\notag \\ & \hspace*{1.0in} \cdot  \left( 1 \pm 2^{l-1} \min(n^{-\Omega(1)}, O(\delta))\right) \enspace .
\end{align}
Consider a related probability
$
\prob{\abs{\hat{x}_{i_1, l-1}} < Q_{l-1}, g_l(i_1) = 1 \mid  i_1 \in \stream_{l-1}, E_{2,d}, \goodest(i_1) }$. Since, the event $g_l(i_1)=1$ is independent  of (i) the event $i_1 \in \stream_{l-1}$ (by independence of $g_l$ and $g_r$'s for all $r \ne l$),  (ii) the event $E_{2,d}$ by independence and $d$-wise independence of $g_l$, and (iii) $\goodest(i_1)$ which considers the inferences obtained until level $l-1$, we have,
\begin{align} \label{eq:lem:induction:base:rmarg:3b}
&\prob{\abs{\hat{x}_{i_1, l-1}} < Q_{l-1}, g_l(i_1) = 1 \mid  i_1 \in \stream_{l-1}, E_{2,d}, \goodest(i_1) } \notag \\
& = \prob{\abs{\hat{x}_{i_1, l-1}} < Q_{l-1} \mid  i_1 \in \stream_{l-1}, E_{2,d}, \goodest(i_1) } \notag \\
& \hspace*{1.0in} \cdot  \prob{ g_l(i_1) = 1 \mid  i_1 \in \stream_{l-1}, E_{2,d}, \goodest(i_1) } \enspace .
\end{align}
The  second product probability term  simplifies to  $(1/2) \pm \min(n^{-\Omega(1)}, O(\delta))$. Substituting in Eqn.~\eqref{eq:lem:induction:base:rmarg:3b}, we have,
\begin{align}\label{eq:lem:induction:base:rmarg:3c}
&\prob{\abs{\hat{x}_{i_1, l-1}} < Q_{l-1}, g_l(i_1) = 1 \mid  i_1 \in \stream_{l-1}, E_{2,d}, \goodest(i_1)} \notag\\
& = \prob{\abs{\hat{x}_{i_1, l-1}} < Q_{l-1} \mid  i_1 \in \stream_{l-1}, E_{2,d}, \goodest(i_1) } \left(1/2 \pm \min(n^{-\Omega(1)}, O(\delta)) \right) \enspace .
\end{align}
Therefore,
\begin{align}\label{eq:lem:induction:base:rmarg:3d}
&\prob{\abs{\hat{x}_{i_1, l-1}} < Q_{l-1}, g_l(i_1) = 1 \mid  i_1 \in \stream_{l-1}, E_{2,d}, \goodest(i_1), \G} \notag\\
& = \prob{\abs{\hat{x}_{i_1, l-1}} < Q_{l-1}, g_l(i_1) = 1 \mid  i_1 \in \stream_{l-1}, E_{2,d}, \goodest(i_1)} \pm \min(n^{-\Omega(1)}, O(\delta)) \notag \\
& = \prob{\abs{\hat{x}_{i_1, l-1}} < Q_{l-1} \mid  i_1 \in \stream_{l-1}, E_{2,d}, \goodest(i_1) } \left(1/2 \pm \min(n^{-\Omega(1)}, O(\delta)) \right) \notag \\
& \hspace*{1.0in} \pm \min(n^{-\Omega(1)}, O(\delta)) \enspace .
\end{align}
Substituting in Eqn.~\eqref{eq:lem:induction:base:rmarg:3a}, we have,
\begin{align}\label{eq:lem:induction:base:rmarg:3e}
& \text{Term3} \notag \\
& = 2\prob{ \abs{\hat{x}_{i_1, l-1}} < Q_{l-1}, g_l(i_1) = 1 \mid  i_1 \in \stream_{l-1}, E_{2,d}, \G,  \goodest(i_1)} \notag \\
& \hspace*{1.0in} \cdot  \left( 1 \pm 2^{l-1} \min(n^{-\Omega(1)}, O(\delta))\right) \notag \\
& = 2\prob{\abs{\hat{x}_{i_1, l-1}} < Q_{l-1} \mid  i_1 \in \stream_{l-1}, E_{2,d}, \goodest(i_1) }\left(1/2 \pm \min(n^{-\Omega(1)}, O(\delta)) \right) \notag \\
& \hspace*{1.0in} \pm \min(n^{-\Omega(1)}, O(\delta)) \notag \\
& = \prob{\abs{\hat{x}_{i_1, l-1}} < Q_{l-1} \mid  i_1 \in \stream_{l-1}, E_{2,d}, \goodest(i_1) }\left(1 \pm \min(n^{-\Omega(1)}, O(\delta)) \right) \notag \\
& \hspace*{1.0in} \pm \min(n^{-\Omega(1)}, O(\delta)) \enspace .
\end{align}
Adding Eqns.~\eqref{eq:lem:induction:base:rmarg:12c} and ~\eqref{eq:lem:induction:base:rmarg:3e}, we obtain
\begin{align*}
&\text{ Term 1 } + \text{ Term 2} + \text{ Term 3}  \\
& = \left( 1 \pm 2^{l-1} \min(n^{-\Omega(1)}, O(\delta)) \right) \pm \min(n^{-\Omega(1)}, O(\delta))\\
& = 1 \pm 2^l \min(n^{-\Omega(1)}, O(\delta)) \enspace .
\end{align*}
\end{proof}

\begin{corollary} \label{cor:basecase1}
Let $i_1,  \ldots, i_d \in [n]$ and $S \subset \{i_1, \ldots, i_d\}$ containing $i_1$. Then,
\begin{gather*}
\sum_{l_1 \consistw~ i_1} 2^{l_1} \prob{ i_1 \in \bar{G}_{l_1} \left\vert \bigwedge_{j=h+1}^d i_j \in \stream_{l_j}, \G, \bigwedge_{j \in S}\goodest(j)\right.} \\ = 1 \pm 2^{\level(i_1)} \abs{S}\min (n^{-\Omega(1)}, O(\delta)) \enspace .
\end{gather*}
\end{corollary}

\begin{proof} The proof proceeds identically as the proof  for Lemma~\ref{lem:basecase}.

\end{proof}

\begin{lemma} \label{lem:approxdwise:restate} [Re-statement of Lemma~\ref{lem:approxdwise}.]
Let $i_1,  \ldots, i_{d} \in [n]$ and distinct  $1 \le h \le d$. Then,
\begin{gather*}
\sum_{\substack{l_j \consistw i_j\\ \forall j \in [h]}}
2^{l_1+ l_2 + \ldots + l_h} \prob{\left.  \bigwedge_{j=1}^h i_j \in \bar{G}_{l_j} \right\vert \bigwedge_{j=h+1}^d i_j \in \stream_{l_j}, \G, \bigwedge_{j=1}^d \goodest(i_j)} \\ \in \prod_{j=1}^h \left(1 \pm 2^{\level(i_j) + 1}\cdot \left(\min(n^{-\Omega(1)},\delta) + d n^{-\Omega(1)}\right)\right)
\end{gather*}
\end{lemma}

\eat{\begin{lemma}
Let $I = \{i_1,  \ldots, i_{d} \}\subset [n]$ and  $1 \le h \le d$. Let $\accuest(\{i_1, \ldots, i_h\})$ denote the event $\bigwedge_{k=1}^h \accuest(i)$. Then, assuming $d$-wise independence of the hash functions $g_1, \ldots, g_L$, the following holds.
\begin{align*}
\sum_{\substack{l_j =0,1, \ldots, L, \\ \text{ for each } j =1,2, \ldots, h}}
2^{l_1+ l_2 + \ldots + l_h} \prob{\left.  \bigwedge_{j=1}^h i_j \in \bar{G}_{l_j} \right\vert \bigwedge_{j=h+1}^d i_j \in \stream_{l_j}, \G, \accuest(\{i_1, \ldots, i_h\})} \in \prod_{j=1}^h \left((1 \pm 2^{\level(i_j) + 1})(O(\delta) + d n^{-\Omega(1)})\right)
\end{align*}
\end{lemma}
}

\begin{proof}
The proof proceeds by induction on $h$.

The base case occurs when $h=1$ and has been proved in Lemma~\ref{lem:basecase} and Corollary~\ref{cor:basecase1}.

\emph{Induction Case.} Let $E_{h+1,d} $ denote the event $\bigwedge_{j=h+1}^d i_j \in \stream_{l_j}$, $G_h$ denote the event $ \bigwedge_{j=1}^h i_j \in \bar{G}_{l_j}$ and $K_d = \bigwedge_{j=1}^d \goodest(i_j)$. Also, for simplicity, let $\delta'$ denote $\min(n^{-\Omega(1)}, O(\delta))$.  In this notation, the statement of the lemma can be written as
\begin{align} \label{eq:indn:1}
& \sum_{\substack{l_j \text{ consistent with } i_j\\ \forall j \in [h]}}
2^{l_1 + \ldots + l_{h+1}} \prob{ G_h \mid E_{h+1,d},  \G, K_d}
 \in \prod_{j=1}^h (1 \pm 2^{\level(i_j) + 1})\left( \delta' + dn^{-\Omega(1)}\right) \enspace .
\end{align}
We assume that for a fixed $h \in [d-1]$,  the induction hypothesis holds, that is,
\begin{align*}
& \sum_{\substack{l_j \text{ consistent with } i_j\\ \forall j \in [h]}}
2^{l_1 + \ldots + l_h} \prob{G_h\vert E_{h+1,d}, \G, K_d} \in\prod_{j=1}^{h+1} (1 \pm 2^{\level(i_j) + 1})\left( \delta' + dn^{-\Omega(1)}\right)  \enspace .
\end{align*}
We now have to show  that the following equation holds.
\begin{align*}
 \sum_{\substack{l_j \text{ consistent with } i_j\\ \forall j \in [h+1]}}
2^{l_1 + \ldots + l_{h+1}} \prob{G_{h+1}\vert E_{h+2,d}, \G, K_{d}} \in \prod_{j=1}^{h+1} (1 \pm 2^{\level(i_j) + 1})\left( \delta' + dn^{-\Omega(1)}\right) \enspace .
\end{align*}

\emph{Induction Case 1:}  Suppose there exists a $j \in [h+1]$ such that $i_j \in \midreg(G_l)$. Without loss of generality, let $j=h+1$ and  $i_{h+1} \in \midreg(G_l)$ (this can be done by rearranging the indices $j_1, \ldots, j_{h+1}$, without affecting the statement). 
Then,
\begin{align} \label{eq:lem:ind:mid1}
& \sum_{\substack{l_j \text{ consistent with } i_j\notag \\ \forall j \in [h+1]}}
2^{l_1 + \ldots l_{h+1}} \prob{G_{h+1}\vert E_{h+2,d}, \G, K_d} \notag \\
& = \sum_{\substack{l_j \text{ consistent with } i_j\notag \\ \forall j \in \{1,\ldots, h\}}}
2^{l_1 + \ldots l_{h} + l} \prob{ G_h \wedge (i_{h+1} \in \stream_l) \vert E_{h+2,d}, \G, K_d} \notag \\
& = \sum_{\substack{l_j \text{ consistent with } i_j \\ \forall j \in \{1,\ldots, h\}}}
2^{l+l_1 + \ldots l_{h}}\prob{ G_h\vert i_{h+1} \in \stream_l \wedge E_{h+2,d}, \G, K_d} \prob{ i_1 \in \stream_l\mid E_{h+2,d},\G, K_d}
\end{align}
The term
\begin{align*}
&\prob{ G_h\vert i_{h+1} \in \stream_l \wedge E_{h+2,d}, \G, K_d}\\
 &= \prob{G_h \vert E_{h+1,d}, \G, K_d}\prod_{j=1}^h (1 \pm 2^{\level(i_j) + 1}\left( O(\delta) + dn^{-\Omega(1)} \right)
\end{align*}
 by the induction hypothesis.
The second term in Eqn.~\eqref{eq:lem:ind:mid1} is evaluated as follows. By $d$-wise  independence, $\prob{ i_1 \in \stream_l\mid E_{h+2,d}} = 2^{-l}$. Also, $\G$ holds with probability $1-\delta'$ and $K_d$ holds with probability $1-dn^{-\Omega(1)}$. Therefore, $\G \wedge K_d$ holds with probability $1-\delta' - dn^{-\Omega(1)}$.  Hence, by Lemma~\ref{lem:fact},
$$\prob{ i_1 \in \stream_l\mid E_{h+2,d},\G, K_d}  \in 2^{-l} \pm O(\delta' + dn^{-\Omega(1)}) \enspace .$$
Therefore, Eqn.~\eqref{eq:lem:ind:mid1} can be bounded as
\begin{align*}
\prod_{j=1}^{h+1} (1 \pm 2^{\level(i_j) + 1}\left( O(\delta) + dn^{-\Omega(1)} \right)
\end{align*}
there by proving this case.

\emph{Induction Case 2:} Suppose there exists an index $j \in [h+1]$ such that $i_j \in \lmargin(G_l)$. By reordering the indices, we can assume without loss of generality that $i_{h+1} \in \lmargin(G_l)$. Conditional on $\G$ and \goodest$(i)$, as shown in Lemma~\ref{lem:margin}, the following holds.
\begin{enumerate}
\item  $i_{h+1} \in \bar{G}_l$ iff $i_{h+1} \in \stream_l$ and $\abs{\hat{x}_{h+1,l}} \ge T_l$.
\item $i_{h+1} \in \bar{G}_{l+1}$ iff $i_{h+1} \in \stream_l$ and $Q_l < \abs{\hat{x}_{h+1,l}} < T_l$.
\item $i_{h+1} \in \bar{G}_r$, for any $r \in [L] \setminus \{l, l+1\}$.
\end{enumerate}
Then,
\begin{align} \label{eq:lem:ind:lmarg:1}
& \sum_{\substack{l_j \text{ consistent with } i_j\\ \forall j \in [h+1]}}
2^{l_1 + \ldots l_{h+1}} \prob{G_{h+1} \vert E_{h+2,d}, \G, K_d} \notag  \\
& = \sum_{\substack{l_j \text{ consistent with } i_j\\ \forall j \in [h]}}
2^{l_1 + \ldots + l_h + l} \prob{G_h \wedge (i_{h+1} \in \stream_l \wedge \abs{\hat{x}_{i_{h+1},l}} \ge T_l) \vert E_{h+2,d}, \G, K_d} \notag \\
& +\sum_{\substack{l_j \text{ consistent with } i_j\\ \forall j \in [h]}}
2^{l_1 + \ldots + l_h + l+1} \notag \\
& \hspace*{1.0in} \cdot \prob{G_h \wedge (i_{h+1} \in \stream_l \wedge   (\abs{\hat{x}_{i_{h+1},l}} < T_l, K_{i_{h+1}}=1 ) \vert E_{h+2,d}, \G, K_d}
\end{align}
Since the coin toss $K_{i_{h+1}} = 1$ occurs with probability 1/2 and is independent of all the other random bits, the second  probability expression in the \emph{RHS} of Eqn.~\eqref{eq:lem:ind:lmarg:1}, we have,
\begin{align*}
&\prob{G_h \wedge (i_{h+1} \in \stream_l \wedge \abs{\hat{x}_{i_{h+1},l}} < T_l, K_{i_{h+1}} =1) \vert E_{h+2,d}, \G, K_d}\\
& = (1/2) \prob{G_h \wedge (i_{h+1} \in \stream_l \wedge \abs{\hat{x}_{i_{h+1},l}} < T_l) \vert E_{h+2,d}, \G, K_d} \enspace .
\end{align*}
Thus, the \emph{RHS} of Eqn.~\eqref{eq:lem:ind:lmarg:1} is simplified as follows.
\begin{align} \label{eq:lem:ind:lmarg:2}
& \sum_{\substack{l_j \text{ consistent with } i_j\\ \forall j \in [h]}} \left(
2^{l_1 + \ldots + l_h + l} \prob{G_h \wedge (i_{h+1} \in \stream_l \wedge \abs{\hat{x}_{i_{h+1},l}} \ge T_l) \vert E_{h+2,d}, \G, K_d} \right . \notag \\
& \left. ~~+
2^{l_1 + \ldots + l_h + l+1} (1/2)\prob{G_h \wedge (i_{h+1} \in \stream_l \wedge   \abs{\hat{x}_{i_{h+1},l}} < T_l ) \vert E_{h+2,d}, \G, K_d} \right) \notag \\
& = \sum_{\substack{l_j \text{ consistent with } i_j\\ \forall j \in [h]}}
2^{l_1 + \ldots + l_h + l} \prob{G_h \wedge i_{h+1} \in \stream_l\vert E_{h+2,d}, \G, K_d}  \notag \\
& =  \sum_{\substack{l_j \text{ consistent with } i_j\\ \forall j \in [h]}}
2^{l_1 + \ldots + l_h + l} \prob{G_h \vert i_{h+1} \in \stream_l\wedge  E_{h+2,d}, \G, K_d} \notag \\
& \hspace*{1.0in}  \cdot \prob{i_{h+1} \in \stream_l \vert E_{h+2,d}, \G, K_d}
\end{align}
where the first step follows from  probability axioms of union of exclusive events, and the last step follows from the definition of conditional probability.
Now, as discussed earlier in this proof, $\prob{i_{h+1} \in \stream_l \vert E_{h+2,d}} = 2^{-l}$ and therefore, $\prob{i_{h+1} \in \stream_l \vert E_{h+2,d}, \G, K_d} = 2^{-l} \pm \min(O(\delta), n^{-\Omega(1)})+ dn^{-\Omega(1)}$. Also, the event $i_{h+1} \in \stream_l \wedge E_{h+2, d} = E_{h+1, d}$. Hence, the probability expression in the \emph{RHS} of Eqn.~\eqref{eq:lem:ind:lmarg:2} is

\begin{align*} 
& = \sum_{\substack{l_j \text{ consistent with } i_j\\ \forall j \in [h]}}
2^{l_1 + \ldots + l_h}\prob{G_h \vert E_{h+1,d}, \G, K_d} (1 \pm 2^l \min(O(\delta), n^{-\Omega(1)})+ dn^{-\Omega(1)}) \notag\\
& = (1 \pm 2^l \min(O(\delta), n^{-\Omega(1)})+ dn^{-\Omega(1)}) \prod_{j=1}^h  \left(1 \pm 2^{l_j+1} \min(O(\delta), n^{-\Omega(1)}) \right)
\end{align*}
by the induction hypothesis, thereby proving this induction case.

\emph{Induction Case 3:} For this case, it suffices to assume that for each $j \in [h+1]$, $i_j \in \rmargin(G_{l_j})$. Assume that we re-order the items so that $ \abs{x_{i_{h+1}}} \ge \abs{x_{i_{h}}} \ge \ldots \ge \abs{x_{i_1}}$. Let $\level(i_{h+1}) = l$. Therefore,
the expression for the \emph{LHS} is

\begin{align} \label{eq:lem:ind:rmarg:1}
\sum_{\substack{ l_j \consistw i_j \\ \forall j \in [h+1]}}
2^{l_1 + l_2 + \ldots + l_{h+1}} \prob{ G_{h+1} \vert E_{h+2,d}, \G, K_d} \enspace .
\end{align}

We are given that $i_{h+1} \in \rmargin(G_l)$. Conditional on $\G$ and \goodest~, we have the following.
\begin{enumerate}
\item $i_{h+1}$ may be discovered at level $l-1$ (only if $i_{h+1} \in \stream_{l-1})$ or possibly at level $l$ (only if $i_{h+1} \in \stream_l)$, but at no lower level than $l-1$.
\item $i_{h+1} \in \bar{G}_{l-1}$ if $i_{h+1} \in \stream_{l-1}$ and $\abs{\hat{x}_{i_{h+1}, l-1}} \ge T_{l-1}$.
\item $i_{h+1} \in \bar{G}_l$ in one of two mutually exclusive ways.
\begin{enumerate}
\item $i \in \stream_{l-1}$ and $Q_{l-1} = T_{l-1} (1- \epsbar)) < \abs{\hat{x}_{i_{h+1}, l-1}}< T_l$ and the coin toss $K_i =1 $.
\item $i \in \stream_{l-1}$ and $g_l(i_{h+1}) = 1$ and  $\abs{\hat{x}_{i_{h+1}, l-1}} < Q_{l-1}$.
\end{enumerate}
\end{enumerate}
The  probability term in Eqn.~\eqref{eq:lem:ind:rmarg:1} can be written as follows.
\begin{align} \label{eq:lem:ind:rmarg:2}
& \sum_{l_{h+1} \consistw i_{h+1}} 2^{l_{h+1}}\prob{ G_{h+1} \vert E_{h+2,d}, \G, K_d} \notag \\
& = 2^{l-1} \prob{G_h \wedge i_{h+1} \in \bar{G}_{l-1}  \vert E_{h+2,d}, \G, K_d} \notag \\
& ~~+ 2^l \prob{ G_h \wedge i_{h+1} \in \bar{G}_l \vert E_{h+2, d}, \G, K_d} \notag\\
& = 2^{l-1} \prob{G_h, i_{h+1} \in \stream_{l-1}, \abs{\hat{x}_{h+1,l-1}} \ge T_{l-1} \vert E_{h+2,d}, \G, K_d}\notag \\
& ~~+ 2^l  \prob{G_h, i_{h+1} \in \stream_{l-1},  Q_{l-1} < \abs{\hat{x}_{h+1, l-1}} < T_{l-1}, K_i = 1 \mid E_{h+2, d}, \G, K_d } \notag \\
&~~+ 2^l \prob{ G_h, i_{h+1} \in \stream_{l-1}, g_l(i_{h+1}) = 1,  Q_{l-1} > \abs{\hat{x}_{i_{h+1}, l-1}} \mid E_{h+2,d}, \G, K_d}
\end{align}

We simplify the above equation to consider the following expression, which we will relate back to the \emph{RHS} of Eqn.~\eqref{eq:lem:ind:rmarg:2}.
\begin{align}\label{eq:lem:ind:rmarg:3}
& = 2^{l-1} \prob{G_h, i_{h+1} \in \stream_{l-1}, \abs{\hat{x}_{h+1,l-1}} \ge T_{l-1} \vert E_{h+2,d}, } \notag \\
& ~~+ 2^l  \prob{G_h, i_{h+1} \in \stream_{l-1},  Q_{l-1} < \abs{\hat{x}_{h+1, l-1}} < T_{l-1}, K_i = 1 \mid E_{h+2, d}} \notag \\
&~~+ 2^l \prob{ G_h, i_{h+1} \in \stream_{l-1}, g_l(i_{h+1}) = 1,  Q_{l-1} > \abs{\hat{x}_{i_{h+1}, l-1}} \mid E_{h+2,d}} \notag\\
& = 2^{l-1}\prob{G_h, i_{h+1} \in \stream_{l-1}, \abs{\hat{x}_{h+1,l-1}} \ge T_{l-1} \vert E_{h+2,d}, } \notag \\
&~~+ 2^{l-1} \prob{G_h, i_{h+1} \in \stream_{l-1},  Q_{l-1} < \abs{\hat{x}_{h+1, l-1}} < T_{l-1} \mid E_{h+2, d}} \notag \\
&~~+ 2^{l-1} \prob{ G_h, i_{h+1} \in \stream_{l-1},  Q_{l-1} > \abs{\hat{x}_{i_{h+1}, l-1}} \mid E_{h+2,d}}
\end{align}

The last step follows since $g_l(i_{h+1})$ is independent of $g_r$ for $r \ne l$ and comes from a $d$-wise independent family.
So Eqn.~\eqref{eq:lem:ind:rmarg:3} becomes
\begin{align}\label{eq:lem:ind:rmarg:4}
= 2^{l-1} \prob{G_h, i_{h+1} \in \stream_{l-1} \mid E_{h+2,d}}
\end{align}
Each of the summands in the expression of the \emph{RHS} of Eqn.~\eqref{eq:lem:ind:rmarg:2} can be written as follows.

\begin{align}
&2^{l-1} \prob{G_h, i_{h+1} \in \stream_{l-1}, \abs{\hat{x}_{h+1,l-1}} \ge T_{l-1} \vert E_{h+2,d}, \G, K_d} \notag\\
& = 2^{l-1} \prob{G_h, i_{h+1} \in \stream_{l-1}, \abs{\hat{x}_{h+1,l-1}} \ge T_{l-1} \vert E_{h+2,d}} \pm 2^{l-1} \left( \delta' + dn^{-\Omega(1)} \right)  \enspace .  \label{eq:lem:ind:rmarg:5a} \\
& 2^l  \prob{G_h, i_{h+1} \in \stream_{l-1},  Q_{l-1} < \abs{\hat{x}_{h+1, l-1}} < T_{l-1}, K_i = 1 \mid E_{h+2, d}, \G, K_d}  \notag \\
&  = 2^l \prob{G_h, i_{h+1} \in \stream_{l-1},  Q_{l-1} < \abs{\hat{x}_{h+1, l-1}} < T_{l-1}, K_i = 1 \mid E_{h+2, d}} \pm 2^{l-1} (\delta' + dn^{-\Omega(1)} )
\label{eq:lem:ind:rmarg:5b} \\
&2^{l-1} \prob{ G_h, i_{h+1} \in \stream_{l-1}, g_l(i_{h+1}) = 1,  Q_{l-1} > \abs{\hat{x}_{i_{h+1}, l-1}} \mid E_{h+2,d}, \G, K_d} \notag \\
& = 2^{l-1}\prob{ G_h, i_{h+1} \in \stream_{l-1},   Q_{l-1} > \abs{\hat{x}_{i_{h+1}, l-1}} \mid E_{h+2,d}} \pm 2^{l-1}\left( \delta' + dn^{-\Omega(1)} \right)
\label{eq:lem:ind:rmarg:5c}
\end{align}

Adding Eqns. ~\eqref{eq:lem:ind:rmarg:5a} through ~\eqref{eq:lem:ind:rmarg:5c}, we have,
\begin{align}  \label{eq:lem:ind:rmarg:6}
&\sum_{l_{h+1} \consistw i_{h+1}} 2^{l_{h+1}}\prob{ G_{h+1} \vert E_{h+2,d}, \G, K_d} \notag \\
& = 2^{l-1} \prob{G_h, i_{h+1} \in \stream_{l-1} \mid E_{h+2,d}} \pm 3 \cdot 2^{l-1} ( \delta' + dn^{-\Omega(1)} )
\end{align}

Further, consider the probability term in the \emph{RHS} of Eqn. ~\eqref{eq:lem:ind:rmarg:6}. We have,
\begin{align*}
&\prob{G_h, i_{h+1} \in \stream_{l-1} \mid E_{h+2,d}} \\
& = \prob{G_h \mid E_{h+2,d}, i_{h+1} \in \stream_{l-1}} \cdot \prob{i_{h+1} \in \stream_{l-1} \mid E_{h+2,d}} \\
& = 2^{-(l-1)}\prob{G_h\mid E_{h+1,d}}
\end{align*}

Substituting in Eqn.~\eqref{eq:lem:ind:rmarg:6}, we have,
\begin{align*}
&\sum_{l_{h+1} \consistw i_{h+1}} 2^{l_{h+1}}\prob{ G_{h+1} \vert E_{h+2,d}, \G, K_d} \notag \\
& =2^{l-1}\prob{G_h\mid E_{h+1,d}}  \pm 3 \cdot 2^{l-1} \left( \delta' + dn^{-\Omega(1)} \right) \notag \\
& = \prob{G_h \mid E_{h+1,d}, \G, K_d} \pm 4\cdot 2^{l-1} \left( \delta' + dn^{-\Omega(1)} \right)
\end{align*}
Therefore,
\begin{align*}
&\sum_{\substack { l_j \consistw i_j \\ \forall j \in [h+1]}} 2^{l_1 + \ldots + l_{h+1}}
\prob{G_{h+1} \mid E_{h+2,d}, \G, K_d} \notag \\
& = \prod_{j=1}^h (1 \pm 2^{\level(i_j)+1}( \delta' + dn^{-\Omega(1)})\pm 4\cdot 2^{l-1} \left( \delta' + dn^{-\Omega(1)} \right)
\end{align*}
which proves the property.

\end{proof}

\subsection{$p$th power estimator for $\abs{x_i}^p$ } \label{sec:binest}

\begin{lemma}\label{lem:expectZZbar} Let $W = \sum_{j =1}^n a_j \omega(j)$ where, $\{\omega(j)\}_{j=1}^n$ is a 2d-wise independent family of randomly chosen roots of $x^q=1$ for any integer $q > d$, where, $d \in \Z^+$.  Then, for  any $d \ge 1$,
$\expect{ (W\conj{W})^d} \le d! \norm{a}_2^{2d}  $, where
$a$ is the $n$-dimensional vector $(a_1, a_2, \ldots, a_n)$.
\end{lemma}

\begin{proof}

\begin{align*}
 \expect{ (W\conj{W})^d} & =  \sum_{e_1 + \ldots + e_n = d} \sum_{g_1 + \ldots + g_n = d}  \binom{d}{e_1, \ldots, e_n} \binom{d}{g_1, \ldots, g_n}  \prod_{j=1}^n a_j^{e_j + g_j} \expect{ \omega_j^{e_j} \conj{\omega_j}^{g_j}}  \\
& = \sum_{e_1 + \ldots  + e_n = d}  \binom{d}{e_1, \ldots, e_n}^2  a_j^{2e_j} \notag  \\
& \le d! \sum_{e_1 +\ldots + e_n = d} \binom{d}{e_1, \ldots, e_n} a_j^{2e_j} \notag \\
& = d! \norm{a}_2^{2d} \notag
\end{align*}
The first step views $(W\conj{W})^d$ as $W^d \conj{W}^d$ and takes the product of the  multinomial expansion for  $W^d $ and $\conj{W}^d$;  and subsequently uses linearity of expectation. The second step follows since $\expect{\omega_j^{e_j} \conj{\omega_j}^{g_j}} = 0$ unless $e_j = g_j$, in which case, it is 1.

\end{proof}

\begin{lemma}[Re-statement of Lemma~\ref{lem:2dmom1}.]
Let $Z = \sum_{j =1}^t  a_j \omega(j)\chi(j)$, where, $\{\omega(j)\}_{j=1}^t$ is a family of random and  $2d$-wise independent family  roots of the equation $x^q = 1$, $q > 2d$  and integral. Let  $\{\chi(j)\}$ be a $2d$-wise independent family of  indicator variables such that  $ \prob{\chi(j)=1} = 1/C$ and is independent of the $\omega_j$'s. If $  \norm{a}_2^2 \ge  4d \norm{a}_{\infty}^2 C$, then, $$ \expect{(Z\bar{Z})^d} \le (2)\left( \frac{d \norm{a}_2^2}{C} \right)^d   $$
where, $a$ is the $t$-dimensional vector $(a_1, a_2, \ldots, a_t)$.
\end{lemma}

\begin{proof} [Proof of Lemma~\ref{lem:2dmom1}.]
Assume $K = \norm{a}_{\infty}$.
\begin{multline*}
\expect{(Z\bar{Z})^d} = \sum_{\substack{e_1 + \ldots +e_t=d \\ e_j\text{'s } \ge 0}} \sum_{\substack{g_1 + \ldots + g_t = d \\ g_j\text{'s } \ge 0}} \binom{d}{e_1, \ldots, e_t} \binom{d}{g_1, \ldots, g_t}  \\
\cdot \prod_{u=1}^t a_{u}^{e_u+g_u} \expect{\omega^{e_u}(j) \bar{\omega}^{g_u}(j)} \expect{ \chi^{e_u+g_u}(j) } ~~
\end{multline*}
\text{(by $2d$-wise independence)}.
The expectation is 0 unless $e_u = g_u$ for each $u \in [t]$. Therefore,
\begin{align}
&\expect{(Z\bar{Z})^d} \\
& = \sum_{r=1}^d \sum_{e_1 + \ldots + e_r=d, e_j \text{'s} \ge 1} \binom{d}{e_1, \ldots, e_r}^2 \sum_{1\le i_1 < \ldots < i_r \le t}  \prod_{u=1}^r \frac{ a_{i_u}^{2e_u}}{C} \notag \\
& \le  K^{2d} \sum_{r=1}^d \frac{1}{K^{2r}} \sum_{e_1 + \ldots + e_r=d} \binom{d}{e_1, \ldots, e_r}^2 \sum_{1 \le i_1 < \ldots < i_r \le t}  \prod_{u=1}^r\frac{ a_{i_u}^2}{C} \notag \\ & \hspace*{3.0in} \text{(since, $a_{i_u}^{2e_u} \le K^{2e_u-2} a_{i_u}^2)$}\notag \\
& \le K^{2d}  \sum_{r=1}^d \frac{d!}{K^{2r}} \sum_{e_1 + \ldots + e_r=d} \binom{d}{e_1, \ldots, e_r}  \frac{1}{r!} \left(\frac{ \norm{a}_2^2}{C} \right)^r \label{eq:2dmt0}  \\
& \le K^{2d} d! \sum_{r=1}^d \frac{1}{r!} \left(\frac{ \norm{a}_2^2}{K^2C} \right)^r \frac{d^r}{(d-r+1)} \cdot r^{d-r} \label{eq:2dmt1}
\end{align}
Eqn.~\eqref{eq:2dmt0} follows by observing that (i)  $ \sum_{1 \le i_1 < \ldots < i_r \le t}  \prod_{u=1}^r\frac{ a_{i_u}^2}{C} \le \frac{1}{r!} \left( \frac{ \norm{a}_2^2}{C} \right)^r$, which can be seen by expanding  $ \frac{1}{C^r} \left( \sum_{j=1}^t a_j^2\right)^r$, and,  (ii) $\binom{d}{e_1, \ldots e_r}^2 \le d! \binom{d}{e_1, \ldots, e_r}$.  \\
Eqn.~\eqref{eq:2dmt1} is obtained as follows. Let $e'_j = e_j-1$. Then,
\begin{align*}
\sum_{e_1 + \ldots + e_r =d, e_j\text{'s} \ge 1} & \binom{d}{e_1, \ldots, e_r} \\
 & = \left(\frac{d(d-1)\ldots  (d-r+1)}{e_1 e_2 \ldots e_r}\right) \sum_{e'_1 + \ldots + e'_r=d-r, e'_j \text{'s}\ge 0} \binom{d-r}{e'_1, e'_2, \ldots, e'_r}
\end{align*}
Now, $\frac{d(d-1)\ldots  (d-r)}{e_1 e_2 \ldots e_r}\le \left(\frac{d^r}{d-r+1}\right)$, where the product $e_1 e_2 \ldots e_r$ is minimized,  subject to $e_1, \ldots, e_r\ge 1$ and $e_1 + \ldots + e_r = d$,  by letting $e_1, \ldots, e_{r-1}$ to be 1 and $e_r = d-r+1$. Also, $$ \sum_{e'_1 + \ldots + e'_r=d-r, e'_j \text{'s}\ge 0} \binom{d-r}{e'_1, e'_2, \ldots, e'_r} =
(\underbrace{1 + 1 + \ldots + 1}_{\text{ $r$ times }})^{d-r} = r^{d-r}  $$
Combining, this gives Eqn.~\eqref{eq:2dmt1}.

Consider the sum in Eqn.~\eqref{eq:2dmt1}. The ratio of the $r+1$ th term to the $r$th term, for $r =1,2, \ldots, d-1$, is
$$
\frac{d\norm{a}_2^2}{(r+1)K^2 C} \cdot \left( \frac{d-r+1}{d-r}\right) \cdot \frac{ (r+1)^{d-r-1}}{r^{d-r}}  \ge \frac{ d \norm{a}_2^2}{2(r+1)^2 K^2 C} \ge \frac{4d^2}{2(r+1)^2} \ge 2 \enspace .
$$
Hence, the sum in Eqn.~\eqref{eq:2dmt1} is a geometrically increasing sequence with common ratio at least 2. The sum is therefore bounded above by 2 $\times$ the final term, that is,
\begin{align*}
 \expect{(Z\bar{Z})^d} &  \le 2\left( \frac{ d\norm{a}_2^2}{C} \right)^d \enspace .
 \end{align*}
\end{proof}

For $r \in [2s]$, let  $\chi_{rij}$ be an indicator variable that is 1 if $h_r(i) = h_r(j)$ and is 0 otherwise. For $j \in [n]$, let  $\hat{x}_j$ be an estimate for $x_j$ obtained using independent means (e.g., \countsketch).

For $i \in [n]$ and $r\in [2s]$, let $X_{ri}$ denote the estimate for $x_i$ returned from table $r \in R(i)$. That is,
  \begin{align} \label{eq:Xri}
  X_{ri}= T[h_r(i)] \cdot \conj{\omega_r(i)} \cdot  \sgn(\hat{x}_i)  = \abs{x_i} + \sum_{j \ne i} x_j \omega_r(j) \conj{\omega_r(i)} \chi_{rij}\sgn(\hat{x}_i)
  \end{align}
assuming $\sgn(\hat{x}_i) = \sgn(x_i)$.
  We denote this as  $X_{ri} = \abs{x_i} + Z_{ri}$, where, for $r \in R(i)$,
  \begin{align} \label{eq:Zri} Z_{ri} = \sum_{j \ne i}  x_j \omega_r(j) \conj{\omega_r(i)} \chi_{rij}\sgn(\hat{x}_i) \enspace .
\end{align}Here, $Z_{ri}$ is the error in the estimation of $\abs{x_i}$ from table $T_r$, $r \in R(i)$. Therefore,
\begin{align} \label{eq:Xi}
  X_i = \frac{1}{\abs{R(i)}} \sum_{r \in R(i)} X_{ri} = \abs{x_i} + \frac{1}{\abs{R(i)}} \sum_{r \in R(i)} Z_{ri} = \abs{x_i} + Z_i  \enspace .
  \end{align}
where,
$Z_i$ is  the average of the $Z_{ri}$'s, that is,
\begin{align} \label{eq:Zi} Z_i = \frac{1}{\abs{R(i)}} \sum_{r \in R(i)} Z_{ri}
\end{align}

Let
\begin{align} \label{eq:Zprimeri}Z'_{ri} = \sum_{j \ne i} x_j \omega_r(j) \chi_{rij}
\end{align} which gives
\begin{align} \label{eq:Zri2}
Z_{ri} = Z'_{ri} \conj{\omega_r(i)} \sgn(\hat{x}_i) \enspace .
\end{align}
Define
\begin{align} \label{eq:Zprimei} Z'_i = \sum_{r \in R(i)}Z'_{ri} \enspace . \end{align}

Let \goodest~define the event $ \forall i \in [n],  \abs{\hat{x}_i - x_i} \le \left( \frac{ \ftwores{2C}}{2C} \right)^{1/2}$.

\begin{lemma} \label{lem:csk1:repeat} [Restatement of Lemma~\ref{lem:csk1}.]Suppose $d = O(\log n)$, $s \ge  300 \log (n)$ and $s \ge 20d$. Let \nocollision$(i)$ hold. Let $\rho = \abs{R(i)}$ be the  number of tables in \est$_l$ structure where $i$ does not collide with the \hattopk$(C_l)$ items and $l$ is the level of discovery of $i$. Let $\H_i$ be the event $\G \wedge \nocollision(i) \wedge \goodest(i)$.
We have that,
\begin{gather*}
\expect{ ((X_i - \abs{x_i})^d(\conj{X_i} - \abs{x_i})^d\mid \H_i} \le 2\left( \frac{d  \ftwores{2C}}{(\rho/9) C} \right)^d, \text{ and } \\
\prob{\abs{X_i - \abs{x_i}} > 2\left( \frac{ d\ftwores{2C}}{(s/9)C} \right)^{1/2}\mid \H_i} < 2^{-2d+1}\enspace .
\end{gather*}

\end{lemma}

\begin{proof} [Proof of Lemma~\ref{lem:csk1:repeat}.]
\eat{ Let $R(i)$ be the set of table indices for which $i$ does not collide with  any of the $\hattopk(2C)$ elements, that is, the top-$2C$ elements by $\abs{\hat{x}_i}$.
}

From \cite{g:isaac12} (Lemma 1b), we have conditional on \goodest(i), that, $\ftwores{\hattopk(2C)} \le 9\ftwores{2C}$. Consider the vector $x_{[n]\setminus \hattopk(2C)}$. Then, by Lemma 1(a) of \cite{g:isaac12}, we have
\begin{align} \label{eq:fres}\norm{x_{[n]\setminus\hattopk(2C)}}_{\infty} \le (1+ \sqrt{2})\left(\frac{\ftwores{2C}}{2C} \right)^{1/2} \enspace .
\end{align}
From Eqn.~\eqref{eq:Xri}, we have  $\expect{X_{ri} \mid \H_i} = \abs{x_i}$,  and hence, from Eqn.~\eqref{eq:Xi},  that  $\expect{X_i} = \abs{x_i}$.
Let $\rho = \abs{R(i)}$. Therefore,  $X_i - \expect{X_i} = Z_i$, and
\begin{align*}
X_i  - \expect{X_i \mid \H_i} & = \rho^{-1}\sum_{r \in R(i)} Z_{ri} & \text{ by Eqn.~\eqref{eq:Xi}}\\
 &= \rho^{-1} \sum_{r \in R(i)} Z'_{ri} \conj{\omega_{r(i)}}\sgn(\hat{x}_i) & \text{ by Eqn.~\eqref{eq:Zprimeri}.}
\end{align*}

Therefore,
\begin{align} \label{eq:csktemp1}
&\expect{Z_i^d \conj{Z_i}^d \mid \H_i}\notag \\
& = \expect{(X_i  - \expect{X_i \mid \H_i})(\conj{X_i - \expect{X_i \mid \H_i}})^d \mid \H_i}  \notag \\
 & = \rho^{-2d} \expect{  \left( \sum_{r \in R(i)} Z'_{ri} \conj{\omega_r(i)}\right)^d \left( \sum_{r\in R(i)} \conj{Z'_{ri}} \omega_r(i) \right)^d \mid \H_i } \notag \\
&= \rho^{-2d}\mathbf{E}\left[\left (\sum_{r \in R(i) } \sum_{j \ne i} x_j \omega_r(j) \conj{\omega_r(i) } \chi_{rji}\right)^d  \left(\sum_{r \in R(i)}\sum_{j \ne i} x_j \conj{\omega_r(j) } \omega_r(i) \chi_{rji}\right)^d \biggl\vert \H_i \right]\notag \\
& = \rho^{-2d}\sum_{\substack{\sum_{r\in R(i)} \sum_{j \ne i} e_{rj} = d\notag \\
e_{rj}\text{'s} \ge 0}} \binom{ d}{e_{11}, \ldots, e_{\rho n}}^2 \prod_{\substack{r\in R(i)\notag \\ j \in [n]\setminus \hattopk(C)}} x_j^{2e_{rj}} \prob{\bigwedge_{(r,j): e_{rj} \ge 1} \chi_{rij}=1 \mid \H_i} \notag \\
& = \rho^{-2d}\expect{ (Z'_i \overline{Z'_i})^d \mid\H_i} \enspace .
\end{align}
from Eqn.~\eqref{eq:Zprimei} defining $Z'_i$.
In this calculation, we have used that   $\expect{ (\omega_u(j))^{e_u} (\conj{\omega_u(j)})^{g_u} }= 0$ unless $e_u = g_u$ (this holds irrespective of the conditioning on $\H_i$, since the family $\omega$ is independent of the other random variables that define $\H_i$.

We can now apply Lemma~\ref{lem:2dmom1} to obtain a bound on $\expect{(Z'_i \conj{Z'_i})^d}$, conditional on \goodest(i)~and \nocollision(i). Let $a$ be the vector with $\rho = \abs{R(i)}$ copies of the vector $x_{[n] \setminus (\hattopk(2C)\cup\{i\})}$. Thus, $a$ is a $\rho n$-dimensional vector, with zeros in all coordinates corresponding to $i$ and elements of $\hattopk(2C)$. Hence,
 $$\norm{a}_2^2 = \rho \norm{x_{[n] \setminus \hattopk(2C)}}_2^2 \ge  \rho \ftwores{2C} \enspace .$$
 Further, as shown in \cite{g:arxiv15}, for any set of $S \subset [n] \setminus ( \hattopk(C) \cup \{i\})$,
 $$\prob{\wedge_{j \in S} \chi_{rij}=1 \mid \nocollision(i)} = (16C)^{-\abs{S}}(1 \pm n^{-\Omega(1)}) \enspace .$$
Thus, the same holds conditional on $\H_i$.

Now
\begin{align*}
 4d \norm{a}_{\infty}^2  (16C) \le \frac{ 4d(1+\sqrt{2})^2 \ftwores{2C} (16C)}{(2C)} \le  (200)d \ftwores{2C} \le \rho \ftwores{2C} \le  \norm{a}_2^2
\end{align*}
since, $\rho \ge s \ge  300\log (n)$, the premise of Lemma~\ref{lem:2dmom1} holds.

It follows from  Lemma~\ref{lem:2dmom1} that
\begin{align} \label{eq:csktemp2} 
\expect{(Z'_i\conj{Z'_i}))^d \mid \H_i}  \le 2\left( \frac{d\rho  (9)\ftwores{2C}}{C} \right)^d \enspace .
\end{align}
and, hence from Eqn.~\eqref{eq:csktemp1}, that
\begin{align}\expect{ (Z_i \conj{Z_i})^d} \le 2\left( \frac{d  \ftwores{2C}}{(\rho/9) C} \right)^d \label{eq:cskdmom}
\end{align}

Therefore,
\begin{multline*}
\prob{ \card{X_i - \expect {X_i}} > T\mid \H_i} =\prob{ \abs{Z_i}> T \mid \H_i}= \prob{  (Z_i \conj{Z_i})^d >   T^{2d}\mid  \H_i}\\ \le \frac{ \expect{(Z_i \bar{Z_i})^d \mid  \H_i}}{T^{2d}}
 \le 2\left( \frac{ d  \ftwores{2C}}{(\rho/9) C T^2}\right)^d \enspace .
 \end{multline*}
Letting $\rho \ge s$ and since, $\expect{X_i\mid \goodest(i),\nocollision(i)} = \abs{x_i}$, we have, $$\prob{\abs{X_i - \abs{x_i} \mid \H_i} > 2\left( \frac{ d\ftwores{2C}}{(s/9) C} \right)^{1/2}} \le 2^{-2d+1}  $$

\end{proof}

The following lemma uses a standard property of  $q$th roots of unity, that $\expect{\omega^v}=0$, where,   $\omega$ is a randomly chosen  root of  $x^q = 1$ and $0\le v \le q-1$.

\begin{lemma} \label{lem:expZipowv} Assume that the family $\{\omega(j)\}$ are $k+1$-wise independent and  randomly chosen roots of $x^q =1$, where, $q > k+1$ is any integer. Then,
for $ v = 1,2 \ldots, k$, $\expect{Z_i^v} = 0$.
\end{lemma}
\begin{proof} [Proof of Lemma~\ref{lem:expZipowv}]
For $v =1,2, \ldots, k$ and $l \in R(i)$, we have,

\begin{align*}
&\expect{Z_{li}^v} \\ &= \sgn^v(\hat{x}_i)\expect{\conj{\omega_{li}}^v} \expect{ \left(\sum_{j\ne i} x_j \omega_{lj} \chi_{lij}\right)^v} \\
 & = \sgn^v(\hat{x}_i) \cdot  0 \cdot \left( \sum_{r=1}^v \sum_{e_1 + \ldots + e_r = v, e_j's \ge 1} \binom{v}{e_1, \ldots, e_r} \sum_{1\le j_1 < \ldots < j_r \le n, j_u's \ne i} \frac{1}{C^r}\prod_{u=1}^r x_{j_u}^{e_u}\expect{\omega_{li}^{e_u}} \right) \\
 & = 0 \cdot 0 = 0
\end{align*}
since, each of the terms $\expect{\omega_{li}^{e_u}} = 0$.

Therefore,  for $v=1,2, \ldots, k$,
$
\expect{Z_i^v} = \expect{\left(\sum_{l=1}^s Z_{li}\right)^v}  = 0
$,
since, the $Z_{li}$'s are independent across the $l$'s.

\end{proof}

\begin{lemma}[Restatement of Lemma \ref{lem:expecttheta}]
Let $\{\omega_r(j)\}_{j \in [n]}$ be a family of $d+1$-wise independent variables, where, $d = O(\log n)$, $s \ge 300 \log (n) $ and $s \ge 72 d$. Further suppose that $\abs{x_i} > 8\left( \frac{\ftwores{2C}}{2C} \right)^{1/2}$.
Then, with  probability $1-n^{-\Omega(1)}$,
$$\card{\expect{X_i^p}-  \abs{x_i}^p\mid \H_i} \le \abs{x_i}^p n^{-\Omega(1)} \enspace .
$$
\end{lemma}

\begin{proof}
We have, $X_i = \abs{x_i} + Z_i$, where, $Z_i = \frac{1}{\rho} \sum_{r \in R(i)} Z_{ri}$ and \\ $Z_{ri} =  \sum_{j \ne i} x_j \omega_r(j) \conj{\omega_r(i)} \chi_{rij} \sgn(x_i)$.
By Lemma~\ref{lem:csk1},
$$ \prob{ \abs{Z_i} > 2 \left( \frac{ d \ftwores{2C}}{ (s/9) C} \right)^{1/2} \mid \H_i} \le 2^{-2d+1} \enspace .  $$
Since $s \ge 72d$, $(s/9)/d  \ge 8$. Also, since, $d = \Theta(\log n)$, $2^{-2d+1} = n^{-\Omega(1)}$. Therefore,
$$\prob{ \abs{Z_i} >  \left( \frac{  \ftwores{2C}}{ 2 C} \right)^{1/2} \mid \H_i} \le n^{-\Omega(1)} \enspace . $$
Therefore,
\begin{align*}
\expect{X_i^p\mid \H_i} &= \expect{ \left(\abs{x_i}+ Z_i  \right )^p \mid \H_i} \\ &= \abs{x_i}^p\expect{ \left(1 + \frac{Z_i}{\abs{x_i}}\right)^p} \\
& = \abs{x_i}^p \left(\sum_{v=0}^k \binom{p}{v} \frac{\expect{Z_i^v \mid \H_i}}{\abs{x_i}^{v}}
\pm \left\lvert \binom{p}{k+1} \right\rvert 8^{-(k+1)}\right) \\
& = \abs{x_i}^p  \left(1 \pm  \left(\frac{p}{k+1}\right)^{\lfloor p \rfloor} 8^{-(k+1)}\right) \\
& = \abs{x_i}^p \left(1  \pm  n^{-\Omega(1)}\right)
\end{align*}
The second to last step follows, since for $v =1,2, \ldots, k $, $\expect{ (Z_i)^v} = 0$, by Lemma~\ref{lem:expZipowv}. The error term is bounded by an application of Taylor's remainder term in the expansion of $(1 + z)^p$, for $\abs{z} \le 1/8$. We have also used the fact that for $k > p$, $\card{\binom{p}{k+1}} \le \left( \frac{p}{k+1} \right)^{\lceil p \rceil}$.
\end{proof}

\subsection{Analysis}
\label{sec:anal:ghss}

For $i \in [n]$, let $x_{li} $ be an indicator variable that is 1 iff $i \in \stream_l$.
For uniformity of notation, let  $X_i$ denote $\abs{\hat{x}_{i}}$ when $l_d(i) = L$ and otherwise,  let its meaning be unchanged.

\noindent Let $z_{il}$ be an indicator variable
that is 1 if $i \in \bar{G}_l $  and 0 otherwise.

We can now  write $\hat{F}_p$ as follows.
\begin{align} \label{eq:expFp0}
\hat{F}_p & =
\sum_{l=0}^L \sum_{\substack{i \in \bar{G}_l\\ l_d(i) < L}} 2^l X_i^p 
 = \sum_{i \in [n]} Y_i, \text{ where, } Y_i = \sum_{l'=0}^{L}2^{l'} z_{il'} X_i^p \enspace .
\end{align}


\subsection{Analysis of \ghss~sampled estimate}

 We now consider the analysis of the \ghss~sampled estimate.
Without loss of generality, for $l_d(i) = l$,  let $\abs{R_l(i)} = s$.


Let $\nocollision$ be defined as the event $\wedge_{i \in [n]} \nocollision(i)$. This means that, for each $i \in [n]$, (i) if $i \in G_l$ and $i$ maps to the random substream $\stream_l$, then, $\nocollision(i)$ holds at level $l_d(i)$, and (ii) if $i \in S_j$ and $\nocollision(i)$ holds at the shelf indexed $j_d(i)$, namely the shelf at which the item is discovered (the largest shelf index).
Let $\nocollision(J)$ denote $\wedge_{j \in J} \nocollision(j)$. Analogously, for a set $J$, define $\goodest(J)$.
For a set $J$, let  $\mathcal{H}_J$ be the event
$$\mathcal{H}_J = \G \wedge \nocollision(J) \wedge \goodest(J)\enspace . $$

\begin{lemma} \label{lem:expprod1}
 Let $r \in [0, \ldots, d], s \in \max(0,1-r),\ldots,d-r$ and $t \in \max(0,1-r), \ldots, d-r$. Let $\{i_1, \ldots, i_r\} \cup \{j_1, \ldots, j_s\} \cup \{k_1, \ldots, k_t\} \subset \lmargin(G_0) \cup_{l=1}^L G_l$ such that $\{i_1, \ldots, i_r\}$, $\{j_1, \ldots, \allowbreak j_s\}$ and $ \{k_1, \ldots, k_t\}$ are pair-wise disjoint. Let $(e_1, \ldots, e_{r+s}) $ and $(g_1, \ldots, g_{r+t})$ be positive integer vectors. Suppose $A = \{v \mid v\in [s], e_{r+v} = 1\}$ and $B = \{w \mid w \in [t], g_{r+w} = 1\}$.  Let the family $\{\omega _{lr}(i)\}_{i}$, be $k+1 = O(\log n)$-wise independent for each fixed $l \in [0,L]$ and $r \in [2s]$, and is independent across $l \in [0,L]$ and $r \in [2s]$.  Let $\gamma = \exp{-\Omega(s)} = n^{-\Omega(1)}$. 
 Then,
\begin{gather*}
 \E\left[ \prod_{u=1}^r (Y_{i_u} -\expect{Y_{i_u} \mid \H_J})^{e_u} (\conj{Y_{i_u}} - \expect{ \conj{Y_{i_u}}\mid \H_J})^{g_u} \prod_{v=1}^s  (Y_{j_{v}} - \expect{Y_{j_{v} }\mid \H_J})^{e_{r+v}} \right. \\
\left.\hspace*{1.0in} \prod_{w=1}^t (\conj{Y_{k_{w}}}- \expect{\conj{Y_{k_{w}}}\mid \H_J})^{g_{r+w}} \vert\H_J\right] \\
\le \gamma^{\abs{A} + \abs{B}} (1 \pm \gamma)^{(\abs{A}+ \abs{B})^2} \prod_{v \in A} \abs{x_{j_v}}^{p} \prod_{w \in B} \abs{x_{k_w}}^{p} \\
\E\left[  \prod_{u=1}^r (Y_{i_u} -\expect{Y_{i_u} \mid \H_J})^{e_u} (\conj{Y_{i_u}} - \expect{ \conj{Y_{i_u}}\mid \H_J})^{g_u} \prod_{v \in [s] \setminus A}(Y_{j_{v}} - \expect{Y_{j_{v}}\mid \H_J})^{e_{r+v}} \right.\\
\left.\hspace*{1.0in}\prod_{w\in [t] \setminus B}
(\conj{Y_{k_{w}}}- \expect{\conj{Y_{k_{w}}} \mid \H_J})^{g_{r+w}}
\mid \H_J\right]
\end{gather*}
\end{lemma}

\begin{proof}

We have, $\expect{Y_{s} \mid \H_J} = \abs{x_s}^p(1 \pm n^{-\Omega(1)}))$, for $s \in \{i_1, \ldots, i_r\} \cup \{j_1, \ldots, j_s\} \cup \{k_1, \ldots, k_t\}$.
Let $\gamma = n^{-\Omega(1)}$. Then,
\begin{align}\label{eq:dmoma1}
& \E\left[ \prod_{u=1}^r (Y_{i_u} -\expect{Y_{i_u} \mid \H_J})^{e_u} (\conj{Y_{i_u}} - \expect{\conj{Y_{i_u}} \mid \H_J})^{g_u}\right. \notag \\
& \hspace*{1.0in}  \left. \prod_{v=1}^s  (Y_{j_{v}} - \expect{Y_{j_{v}}\mid \H_J})^{e_{r+v}}
\prod_{w=1}^t (\conj{Y_{k_{w}}}- \expect{\conj{Y_{k_{w}}} \mid \H_J})^{g_{r+w}} \mid \H_J\right] \notag \\
& = \E\left[ \prod_{v \in A} (Y_{j_{v}} - \expect{Y_{j_{v}}\mid \H_J}) \prod_{w \in B} (\conj{Y_{k_{w}}}- \expect{\conj{Y_{k_{w}}} \mid \H_J}) \right.\\
&\hspace*{0.4in} \prod_{u=1}^r (Y_{i_u} -\expect{Y_{i_u} \mid \H_J})^{e_u} (\conj{Y_{i_u}} - \expect{\conj{Y_{i_u}} \mid \H_J})^{g_u} \notag \\
&  \hspace*{0.3in}\left.  \prod_{v \in [s] \setminus A}(Y_{j_{v}} - \expect{Y_{j_{v}}\mid \H_J})^{e_{r+v}} \prod_{w\in [t] \setminus B}
(\conj{Y_{k_{w}}}- \expect{\conj{Y_{k_{w}}} \mid \H_J})^{g_{r+w}} \mid \H_J \right]\notag \\
=
& \sum_{a=0}^{\abs{A}} \sum_{b=0}^{\abs{B}} \sum_{\substack{S_1 \subset A \\
\abs{S_1} = a}} \sum_{\substack{T_1 \subset B  \\ \abs{T_1} =b}} \prod_{v \in A\setminus S_1}(-\abs{x_{j_v}}^p (1 \pm \gamma))\prod_{ w\in B \setminus T_1} (-\abs{x_{k_w}}^p (1\pm \gamma))  \notag \\ &\hspace*{1.0in}\E\left[
\prod_{v \in S_1} Y_{j_v} \prod_{w \in T_1} \conj{Y_{k_w}} W_{A,B} \mid \H_J\right] \enspace .
\end{align}
where,
\begin{align*}
W_{A,B} & = \prod_{u=1}^r (Y_{i_u} -\expect{Y_{i_u} \mid \H_J})^{e_u} (\conj{Y_{i_u}} - \expect{\conj{Y_{i_u}} \mid \H_J})^{g_u} \prod_{v \in [s] \setminus A}(Y_{j_{v}} - \expect{Y_{j_{v}}\mid \H_J})^{e_{r+v}}\\
&\hspace*{1.5in} \prod_{w\in [t] \setminus B}
(\conj{Y_{k_{w}}}- \expect{\conj{Y_{k_{w}}} \mid \H_J})^{g_{r+w}} \enspace .
\end{align*}

\noindent
Consider the term inside the expectation in Eqn.~\eqref{eq:dmoma1}.   For $v\in S_1$, write\\ $Y_{j_v} = \sum_{l_v \text{ consist. w. } j_v} 2^{l_v} z_{j_v l_v} X_{j_v}^p$ and similarly for $Y_{k_w}$ for $k \in T_1$. Therefore,
\begin{align} \label{eq:dmoma21}
&\expect{ \prod_{v \in S_1} Y_{j_v} \prod_{w \in T_1} \conj{Y_{k_w}} W_{A,B} \mid \H_J} \notag
 = \notag\\
&\sum_{\substack{l_v \consistw~j_v \\ l_w \consistw~k_w}} 2^{\sum_{v \in S_1} l_v + \sum_{w \in T_1}} \expect{ \prod_{v \in S_1} X_{j_v}^p \prod_{w \in T_1} \conj{X_{k_w}}^p W_{A,B}  \prod_{v \in S_1} z_{j_vl_v} \prod_{w \in T_1} z_{k_w l_w} \mid \H_J} \enspace .  
\end{align}

Consider the term $\expect{\prod_{v \in S_1} X_{j_v}^p \prod_{w \in T_1} \conj{X_{k_w}}^p W_{A,B}  \prod_{v \in S_1} z_{j_vl_v} \prod_{w \in T_1} z_{k_w l_w} \mid \H_J}$. Let $\mathcal{R}$ denote the random bits used by the algorithm. Let $\omega$ denote the set of random bits used to form the complex roots of unity sketches. Then,
\begin{align}\label{eq:lem:expprod:1a}
&\expect{\prod_{v \in S_1} X_{j_v}^p \prod_{w \in T_1} \conj{X_{k_w}}^p W_{A,B}  \prod_{v \in S_1} z_{j_vl_v} \prod_{w \in T_1} z_{k_w l_w} \mid \H_J} \notag \\
&= \expectsub{\mathcal{R} \setminus \omega}{\expectsub{\omega}{
\prod_{v \in S_1} X_{j_v}^p \prod_{w \in T_1} \conj{X_{k_w}}^p W_{A,B}  \prod_{v \in S_1} z_{j_vl_v} \prod_{w \in T_1} z_{k_w l_w} \mid \H_J}} \enspace .
\end{align}
Now, $\H_J$ includes $\nocollision$ terms for each $i_u, j_v $ and $k_w$ as per the definition of the indices. Therefore, assuming $r+s+t$-wise independence of the $\omega_{lr}(i)$ family across the $i$'s and independence across $l \in [0, \ldots, L]$ and $r \in [2s]$, we have that
\begin{align*}
&\expectsub{\omega}{\prod_{v \in S_1} X_{j_v}^p \prod_{w \in T_1} \conj{X_{k_w}}^p W_{A,B}  \prod_{v \in S_1} z_{j_vl_v} \prod_{w \in T_1} z_{k_w l_w} \mid \H_J}\\
& = \prod_{v \in S_1}\expectsub{\omega}{X_{j_v}^p  \mid \H_J} \prod_{w \in T_1} \expectsub{\omega}{ \conj{X_{k_w}}^p \mid \H_J} \expectsub{\omega}{W_{A,B}\mid \H_J} \\
& = \prod_{v \in S_1} \left(\abs{x_{j_v}}^p(1\pm \gamma)\right) \prod_{w \in T_1} \left( \abs{x_{k_w}}^p (1 \pm \gamma) \right)  \expectsub{\omega}{W_{A,B}\mid \H_J} \enspace .
\end{align*}
Substituting in Eqn.~\eqref{eq:lem:expprod:1a}, we have,
\begin{align} \label{eq:lem:expprod:1b}
&\E\left[\prod_{v \in S_1} X_{j_v}^p \prod_{w \in T_1} \conj{X_{k_w}}^p W_{A,B}  \prod_{v \in S_1} z_{j_vl_v}\prod_{w \in T_1} z_{k_w l_w} \mid \H_J\right] \notag \\
&= \prod_{v \in S_1} \left(\abs{x_{j_v}}^p(1\pm \gamma)\right) \prod_{w \in T_1} \left( \abs{x_{k_w}}^p (1 \pm \gamma) \right) \expect{W_{A,B}  \prod_{v \in S_1} z_{j_vl_v} \prod_{w \in T_1} z_{k_w l_w}  \mid \H_J} \enspace .
\end{align}
Substituting in Eqn.~\eqref{eq:dmoma21}, we have,
\begin{align} \label{eq:lem:expprod:1c}
&\expect{ \prod_{v \in S_1} Y_{j_v} \prod_{w \in T_1} \conj{Y_{k_w}} W_{A,B} \mid H} \notag
 =\prod_{v \in S_1} \left(\abs{x_{j_v}}^p(1\pm \gamma)\right) \prod_{w \in T_1} \left( \abs{x_{k_w}}^p (1 \pm \gamma) \right) \notag \\
& \hspace*{1.0in} \sum_{\substack{l_v \consistw~j_v \\ l_w \consistw~k_w}} 2^{\sum_{v \in S_1} l_v + \sum_{w \in T_1}} \expect{W_{A,B}  \prod_{v \in S_1} z_{j_vl_v} \prod_{w \in T_1} z_{k_w l_w}  \mid \H_J}
\end{align}
Consider the term
\begin{align}\label{eq:lem:expprod:2a}
\sum_{\substack{l_v \consistw~j_v \\ l_w \consistw~k_w}} 2^{\sum_{v \in S_1} l_v + \sum_{w \in T_1}} \expect{W_{A,B}  \prod_{v \in S_1} z_{j_vl_v} \prod_{w \in T_1} z_{k_w l_w}  \mid \H_J}
\end{align}
Given the definition of the conditioning event $\H_J$, Lemma~\ref{lem:approxdwise:restate} essentially asserts that Eqn.~\eqref{eq:lem:expprod:2a} lies in
\begin{align*}
&\in \expect{W_{A,B} \mid \H_J} (1 \pm (\abs{S_1} + \abs{T_1})\gamma)^{\abs{S_1} + \abs{T_1}}\\ &  \in \expect{W_{A,B} \mid \H_J} (1 \pm  (\abs{A} + \abs{B})^2\gamma)
\end{align*}
where it is assumed that $\abs{A} + \abs{B} \le  n^{O(1)}$.

Substituting in Eqn.~\eqref{eq:lem:expprod:1c}, we obtain that
\begin{multline}  \label{eq:lem:expprod:2b}
\expect{ \prod_{v \in S_1} Y_{j_v} \prod_{w \in T_1} \conj{Y_{k_w}} W_{A,B} \mid \H_J}\\
 = (1 \pm  \gamma)^{(\abs{A} + \abs{B})^2}(1\pm \gamma)^{\abs{S_1} + \abs{T_1}} \prod_{v \in S_1} \abs{x_{j_v}}^p  \prod_{w \in T_1} \abs{x_{k_w}}^p \expect{W_{A,B} \mid \H_J}
\end{multline}

Substituting in Eqn.~\eqref{eq:dmoma1}, we have,
\begin{align} \label{eq:lem:expprod:4a}
& \mathbf{E} \left[ \prod_{u=1}^r (Y_{i_u} -\expect{Y_{i_u} \mid \H_J})^{e_u} (\conj{Y_{i_u}} - \expect{\conj{Y_{i_u}} \mid \H_J})^{g_u} \prod_{v=1}^s  (Y_{j_{v}} - \expect{Y_{j_{v}}\mid \H_J})^{e_{r+v}}\right. \notag \\
&\hspace*{0.5in}\left.\prod_{w=1}^t (\conj{Y_{k_{w}}}- \expect{\conj{Y_{k_{w}}} \mid \H_J})^{g_{r+w}} \mid \H_J \right] \notag \\
& =(1 \pm  \gamma)^{(\abs{A} + \abs{B})^2} \expect{W_{A,B} \mid \H_J} \notag \\
&~~\sum_{a=0}^{\abs{A}} \sum_{b=0}^{\abs{B}} \sum_{\substack{S_1 \subset A \\
\abs{S_1} = a}} \sum_{\substack{T_1 \subset B  \\ \abs{T_1} =b}} \prod_{v \in A\setminus S_1}(-\abs{x_{j_v}}^p (1 \pm \gamma))\prod_{ w\in B \setminus T_1} (-\abs{x_{k_w}}^p (1\pm \gamma)) \prod_{v \in S_1} \abs{x_{j_v}}^p  \prod_{w \in T_1} \abs{x_{k_w}}^p  \notag \\
& = (1 \pm  \gamma)^{(\abs{A} + \abs{B})^2} \expect{W_{A,B} \mid \H_J} \gamma^{\abs{A} + \abs{B}} \prod_{v \in A} \abs{x_{j_v}}^p  \prod_{w \in B} \abs{x_{k_w}}^p
\end{align}

\end{proof}

\begin{fact}\label{fact:1}
Suppose $i \in \lmargin(G_0) \cup_{l'=1}^L G_{l'}$ and suppose
$l$ is consistent with $i$ and $\beta \ge 1$ is an integer.  Then,
\begin{align*}
\abs{x_i}^{p\beta} 2^{l\beta} & \le 2^{l} \abs{x_i}^{p} \left( \frac{F_2}{B} \right)^{p(\beta-1)/2} 2^{(1+p/2)(\beta-1)} \\
 \end{align*}
 \end{fact}

 \begin{proof}[Proof of Fact~\ref{fact:1}.]
By the definition of $\level(i)$, $\abs{x_i} \le \left( \frac{ F_2}{(2\alpha)^{\level(i)-1} B}\right)^{1/2}$. Therefore,
\begin{gather}
 \abs{x_i}^{p\beta} 2^{l\beta}   = 2^l \cdot 2^{l(\beta-1)} \cdot \abs{x_i}^p \cdot \abs{x_i}^{p(\beta-1)}\le  2^{l}\abs{x_i}^{p} \left( \frac{ F_2}{(2\alpha)^{\level(i)-1} B} \right)^{p(\beta-1)/2} 2^{l (\beta-1)} \label{eq:fact:t1}
 \end{gather}
 Now $-1 \le l-\level(i)\le 1$ and
$ 2^{\ell}(2\alpha)^{-\ell p/2} = 2^{\ell (1 -(p/2) \log_2 (2\alpha))}$, for any $\ell$.

Since, $\alpha = 1- (1-2/p)\nu$, where, $\nu = 0.01$, $\ln (\alpha) \ge -2(1-2/p)\nu$. Hence,
 \begin{align*}
 1-(p/2) \log_2(2\alpha) & = 1 -  (p/2) (1 + \log_2(\alpha))\\
 & = 1-p/2 - \frac{p/2}{\ln 2} \ln (\alpha) \\
 & \le -(p/2-1) + \frac{(p/2)}{\ln 2}  (1-2/p)(2\nu))\\  &= -(p/2-1) + (p/2-1) \frac{(2 \nu)}{\ln 2}  \\
 &= -(p/2-1)(1- 2\nu/\ln (2)) \\ & < 0
 \end{align*}
 Hence, for any $\ell \ge 0$, $2^{\ell} (2\alpha)^{-\ell p/2} < 1$.  Therefore,  for any integer $\beta \ge 1$, we have,
 \begin{align*}
 &\abs{x_i}^{p\beta} 2^{l\beta} \\
 & \le  2^{l}\abs{x_i}^{p} \left( \frac{ F_2}{(2\alpha)^{\level(i)-1} B} \right)^{p(\beta-1)/2} 2^{l (\beta-1)} \\
 &  = 2^{l} \abs{x_i}^{p} \left( \frac{F_2}{B} \right)^{p(\beta-1)/2} \left( \frac{ 2^{\level(i)}}{(2\alpha)^{(p/2)\level(i)}} \right)^{\beta-1} (2\alpha)^{(p/2)(\beta-1)} 2^{(l -\level(i))(\beta-1)} \\
 & \le 2^{l} \abs{x_i}^{p} \left( \frac{F_2}{B} \right)^{p(\beta-1)/2} 2^{(1+p/2)(\beta-1)}
 \end{align*}
 \end{proof}

\begin{lemma} \label{lem:intermed1} Suppose $e_1, \ldots, e_r, g_1, \ldots, g_r \ge 1$ and $e_{r+1}, \ldots, e_{r+s'} \ge 2$ and $g_{r+1}, \ldots, g_{r+t} \ge 2$. Let $\{i_1, \ldots, i_r\} \cup \{ j_1, \ldots, j_{s'}\} \cup \{k_1, \ldots, k_t\} \subset \lmargin(G_0) \cup_{l=1}^L G_l$. Let  $s \ge 2\max_{u=1}^r(e_u,g_u)$ and $B/C \le 1/(72p^2)$. Let $J'$ denote the set $\{i_1, \ldots, i_r, j_1, \ldots, j_s', k_1, \ldots, k_r'\}$. Let $\H_{J'}$ denote $\G \wedge \nocollision(J') \wedge \accuest(J')$.
Then,
\begin{align*}
\mathbf{E} &\left[  \prod_{u=1}^r Y_{i_u}^{e_u} \conj{Y_{i_u}}^{g_u} \prod_{v \in [s']} Y_{j_v}^{e_{r+v}} \prod_{w \in [t]} \conj{Y_{k_w}}^{g_{r+v}} \mid \H_{J'}\right] \\
& \hspace*{1.0in} \le  2 e^{ \sum_{u \in [r]} (e_u+g_u)/8}\prod_{u=1}^r \abs{x_{i_u}}^p \prod_{v=1}^s \abs{x_{j_v}}^p \prod_{w=1}^t \abs{x_{k_w}}^p\\
& \hspace*{1.1in} \cdot \left( \frac{ 4 F_2}{B} \right)^{(p/2)\left(\sum_{u=1}^r (e_u+g_u-1)+ \sum_{v \in [s] } (e_{r+v}-1) + \sum_{w \in [t]} (e_{r+w}-1) \right)}\enspace .
 \end{align*}
\end{lemma}
\begin{proof} The expectations in this proof are conditioned on the conjunction of events $\G$, $\nocollision(J)$ and $\goodest(J)$. For brevity, let $\H_{J'} =  \G \wedge \nocollision(J')\wedge \goodest(J')$.
\begin{align} \label{eq:dmomb1}
&\expect{ \left.\prod_{u=1}^r Y_{i_u}^{e_u} \conj{Y_{i_u}}^{g_u} \prod_{v \in [{s'}]} Y_{j_v}^{e_{r+v}} \prod_{w \in [t]} \conj{Y_{k_w}}^{g_{r+w}}\right\vert \H_{J'}} \notag \notag \\
& = \E \left[\prod_{u=1}^r \left( \sum_{l_u \text{ consist. w. } i_u} 2^{l_u} z_{i_ul_u}X_{i_u}^p\right)^{e_u} \left( \sum_{l_u \text{ consist. w. } i_u} 2^{l_u} z_{i_ul_u}\conj{X_{i_u}}^p\right)^{g_u} \right. \notag \\
& \left.\prod_{v \in [{s'}]}\left( \sum_{l'_v \text{ consist. w. } j_v} 2^{l'_v} z_{j_v l'_v} X_{j_v}^p \right)^{e_{r+v}}  \left. \prod_{w\in [t]} \left( \sum_{\ell_w \text{ consist. w. } k_w} 2^{\ell_w} z_{k_w \ell_w} \conj{X_{k_w}}^p \right)^{ g_{r+w}}\right\vert \H_{J'} \right]
\end{align}
Consider one of the power terms in Eqn.~\eqref{eq:dmomb1}, say, $\left( \sum_{l_u \text{ consist. w. } i_u} 2^{l_u} z_{i_ul_u}X_{i_u}^p\right)^{e_u}$. Note that for each $i_u \in \lmargin(G_0) \cup_{l=1}^L G_l$, $i_u$ is sampled into at most one group $\bar{G}_l$, for $l \in [0, L]$. Since, $z_{i_u,l}$ is the indicator variable that is 1 iff $i_u$ is sampled into $\bar{G}_l$, therefore, for $z_{i_u,l} = 1$ for at most one $l \in [0,L]$. Hence, for $l \ne l'$, $z_{i_u,l} z_{i_u,l'} = 0$. Now,
\begin{align*}
&\left( \sum_{l_u \text{ consist. w. } i_u} 2^{l_u} z_{i_ul_u}X_{i_u}^p\right)^{e_u} = \sum_{\substack{
\sum_{l ~\consistw~i_u}  h_l = e_u \\ h_l\text{'s} \ge 0}}~~ \prod_{l ~\consistw~i_u} \left( 2^l z_{i_u,l} X_{i_u}^p \right)^{h_l}
\end{align*}
In the product term $\prod_{l ~\consistw~i_u} \left( 2^l z_{i_u,l} X_{i_u}^p \right)^{h_l}$, if $h_l $ and $h_{l'}$ are both non-zero for some $l\ne l'$, then, $z_{i_u,l}^{h_l} z_{i_u,l'}^{h_{l'}} = z_{i_u,l} z_{i_u,l'} = 0$, and therefore the product term is 0. Hence, the only contribution comes from the diagonal terms, that is
\begin{align*}
\left( \sum_{l_u \text{ consist. w. } i_u} 2^{l_u} z_{i_ul_u}X_{i_u}^p\right)^{e_u} & = \sum_{l_u ~ \consistw~i_u} 2^{l_u e_u} z_{i_u l_u} X_{i_u}^{pe_u} \\ & = X_{i_u}^{pe_u} \sum_{l_u \consistw~i_u} 2^{l_u e_u} z_{i_u l_u} \enspace .
\end{align*}

Substituting in Eqn.~\eqref{eq:dmomb1}, we have,
\begin{align}\label{eq:dmomb2}
\E &\left[\prod_{u=1}^r \left( \sum_{l_u \text{ consist. w. } i_u} 2^{l_u} z_{i_ul_u}X_{i_u}^p\right)^{e_u} \left( \sum_{l_u \text{ consist. w. } i_u} 2^{l_u} z_{i_ul_u}\conj{X_{i_u}}^p\right)^{g_u} \right. \notag \\
& \left.\prod_{v \in [{s'}]}\left( \sum_{l'_v \text{ consist. w. } j_v} 2^{l'_v} z_{j_v l'_v} X_{j_v}^p \right)^{e_{r+v}}  \left. \prod_{w\in [t]} \left( \sum_{\ell_w \text{ consist. w. } k_w} 2^{\ell_w} z_{k_w \ell_w} \conj{X_{k_w}}^p \right)^{ g_{r+w}}\right\vert \H_{J'} \right] \notag \\
& = \E \left[\prod_{u=1}^r X_{i_u}^{pe_u} \conj{X_{i_u}}^{pg_u} \prod_{v=1}^{s'} X_{j_v}^{pe_{r+v}} \prod_{w\in [t]} \conj{X_{k_w}}^{p  g_{r+w}} \prod_{u=1}^r \sum_{l_u \text{ consist. w. } i_u} 2^{l_u (e_u + g_u)} z_{i_u l_u} \right. \notag \\
& \hspace*{0.5in}\left. \left. \prod_{v=1}^{s'} \sum_{l'_v \text{ consist. w. } j_v} 2^{l'_v e_{r+v}} z_{j_v l'_v} \prod_{w=1}^t \sum_{\ell_w \text{ consist. w. } k_w} 2^{\ell_w g_{r+w}} z_{k_w \ell_w} \right\vert \H_{J'} \right] \enspace .
\end{align}
For any $i$, conditional on $\H_{J'}$, $X_i = \abs{x_i} + Z_i = \abs{x_i}\left(1 + \frac{Z_i}{\abs{x_i}}\right)$, where $Z_i$ is given by Eqns.~\eqref{eq:Xi} and ~\eqref{eq:Zi}. Then, Eqn.~\eqref{eq:dmomb2} equals
\begin{align} \label{eq:dmomb}
& = \prod_{u=1}^r \abs{x_{i_u}}^{p(e_u+g_u)} \prod_{v=1}^{s'} \abs{x_{j_v}}^{pe_{r+v}} \prod_{w=1}^t \abs{x_{k_w}}^{pg_{r+t}} \notag \\
& \hspace*{0.2in}\sum_{\substack{l_u \text{ consist. w. } i_u, u =1,\ldots, r \notag \\ l'_v \text{ consist. w. } j_v, v \in [{s'}] \notag \\ \ell_w \text{ con{s'}ist. w. } k_w, w \in [t]} } 2^{\sum_{u=1}^r {l_u(e_u+g_u)} + \sum_{v=1}^{s'} l'_v e_{r+v} + \sum_{w=1}^t \ell_w g_{r+w}}  \notag \\
&  \E\left[ \prod_{u=1}^r
\left(1 + \frac{Z_{i_u}}{\abs{x_{i_u}}} \right)^{pe_u} \left( 1 + \frac{\conj{Z_{i_u}}}{\abs{x_{i_u}}}\right)^{pg_u} \prod_{v=1}^{s'} \left( 1 + \frac{Z_{j_v}}{\abs{x_{j_v}}} \right)^{pe_{r+v}} \prod_{w=1}^t \left( 1 + \frac{ Z_{k_w}}{\abs{x_{k_w}}} \right)^{pg_{r+w}} \right. \notag \\
&  \left. \hspace{0.6in}\left\vert \bigwedge_{u=1}^r z_{i_u l_u}=1 \bigwedge_{v=1}^t z_{j_v l'_v} = 1 \bigwedge_{w=1}^t z_{k_w \ell_w}=1, \H_{J'} \right. \right] \notag \\ & \hspace*{0.6in} \cdot  \prob{ \bigwedge_{u=1}^r z_{i_u l_u}=1 \bigwedge_{v=1}^t z_{j_v l'_v} = 1 \bigwedge_{w=1}^t z_{k_w \ell_w}=1,\H_{J'}} \enspace .
\end{align}
We will denote the expression in the expectation as $P(I,J,K, e,g)$, where, $I = \{i_1, \ldots, i_r\}, J = \{j_1, \ldots, j_{s'}\}$ and $K = \{k_1, \ldots, k_t\}$. $e $ denotes the vector $(e_1, \ldots, e_{r+{s'}})$ and $g$ denotes the vector $(g_1, \ldots, g_{r+t})$, where, $e_1, \ldots, e_r \ge 1$, $g_1, \ldots, g_r \ge 1$, and $e_{r+1}, \ldots, e_{r+{s'}}, g_{r+1}, \ldots, g_{r+t} \ge 2$.  Then, Eqn.~\eqref{eq:dmomb} may be  written as
\begin{align*} 
&=\prod_{u=1}^r \abs{x_{i_u}}^{p(e_u+g_u)} \prod_{v=1}^{s'} \abs{x_{j_v}}^{pe_{r+v}} \prod_{w=1}^t \abs{x_{k_w}}^{pg_{r+t}} \notag \\ &\sum_{\substack{l_u \text{ consist. w. } i_u, u =1,\ldots, r \notag \\ l'_v \text{ consist. w. } j_v, v \in [{s'}] \notag \\ \ell_w \text{ con{s'}ist. w. } k_w, w \in [t]} } 2^{\sum_{u=1}^r {l_u(e_u+g_u)} + \sum_{v=1}^{s'} l'_v e_{r+v} + \sum_{w=1}^t \ell_w g_{r+w}}  \notag \\
&\hspace*{0.7in}
 \E\left[P(I,J,K,e,g)\left\vert \bigwedge_{u=1}^r z_{i_u l_u}=1 \bigwedge_{v=1}^t z_{j_v l'_v} = 1 \bigwedge_{w=1}^t z_{k_w \ell_w}=1, \H_{J'} \right. \right]\notag \\
 &\hspace*{1.0in}\prob{ \bigwedge_{u=1}^r z_{i_u l_u}=1 \bigwedge_{v=1}^t z_{j_v l'_v} = 1 \bigwedge_{w=1}^t z_{k_w \ell_w}=1,\H_{J'}} \enspace .
 \end{align*}

By $\nocollision (J)$, we have
\begin{align} \label{eq:dmt2}
&\expect{P(I,J,K,e,g) \mid \left\vert \bigwedge_{u=1}^r z_{i_u l_u}=1 \bigwedge_{v=1}^t z_{j_v l'_v} = 1 \bigwedge_{w=1}^t z_{k_w \ell_w}=1, \H_{J'} \right.}  \notag\\
& \E\left[ \left. \prod_{u=1}^r
\left(1 + \frac{Z_{i_u}}{\abs{x_{i_u}}} \right)^{pe_u} \left( 1 + \frac{\conj{Z_{i_u}}}{\abs{x_{i_u}}}\right)^{pg_u} \prod_{v=1}^{s'} \left( 1 + \frac{Z_{j_v}}{\abs{x_{j_v}}} \right)^{pe_{r+v}} \prod_{w=1}^t \left( 1 + \frac{ Z_{k_w}}{\abs{x_{k_w}}} \right)^{pg_{r+w}} \right\vert  \right. \notag \\
& \hspace*{1.0in}\left. \bigwedge_{u=1}^r z_{i_u l_u}=1 \bigwedge_{v=1}^t z_{j_v l'_v} = 1 \bigwedge_{w=1}^t z_{k_w \ell_w}=1, \H_{J'}  \right]  \notag \\
& = \prod_{u=1}^r \expect{\left. \left(1 + \frac{Z_{i_u}}{\abs{x_{i_u}}} \right)^{pe_u} \left( 1 + \frac{\conj{Z_{i_u}}}{\abs{x_{i_u}}}\right)^{pg_u} \right\rvert z_{i_u,l_u}=1, \H_{J'}}  \notag \\ &\hspace*{1.0in} \cdot \prod_{v=1}^{s'}\expect{\left. \left( 1 + \frac{Z_{j_v}}{\abs{x_{j_v}}} \right)^{pe_{r+v}}\right\vert z_{j_v,l'_v} =1,  \H_{J'}} \notag \\
& \hspace*{1.0in} \cdot \prod_{w=1}^t \expect{\left. \left( 1 + \frac{ Z_{k_w}}{\abs{x_{k_w}}} \right)^{pg_{r+w}} \right\vert z_{k_w, \ell_w} =1, \H_{J'}} \left(1  \pm n^{-\Omega(1)}\right) \enspace .
\end{align}

We now consider $\expect{\left. \left(1 + \frac{Z_{i_u}}{\abs{x_{i_u}}} \right)^{pe_u} \left( 1 + \frac{\conj{Z_{i_u}}}{\abs{x_{i_u}}}\right)^{pg_u} \right\rvert z_{i_u,l_u}=1, \H_{J'}}$. Then, since, $z_{i_u,l_u}=1$, we have  $\frac{\abs{Z_{i_u}}}{\abs{x_{i_u}}} \le (B/C)^{1/2} \le 1/(8p)$ with probability $1-n^{-\Omega(1)}$. Hence,
\begin{multline*}
\expect{\left. \left(1 + \frac{Z_{i_u}}{\abs{x_{i_u}}} \right)^{pe_u} \left( 1 + \frac{\conj{Z_{i_u}}}{\abs{x_{i_u}}}\right)^{pg_u} \right\rvert \H_{J'}}\\ \le \left( 1 + \frac{\abs{Z_{i_u}}}{\abs{x_{i_u}}} \right)^{pe_u} \left( 1+ \frac{ \abs{Z_{i_u}}}{\abs{x_{i_u}}}\right)^{pg_u}  \le  \left(1 + \frac{1}{8p} \right)^{p(e_u+g_u)} \enspace .
\end{multline*}
with probability $1-n^{-\Omega(1)}$.

Now consider $\expect{\left. \left( 1 + \frac{Z_{j_v}}{\abs{x_{j_v}}} \right)^{pe_{r+v}}\right\vert z_{j_v,l'_v}=1,\H_{J'}}$, for a fixed $v \in [s']$. Let $c_i = \binom{pe_{r+v}}{i}$, for $i=1,2, \ldots, k$. Since, $\frac{\abs{Z_{j_v}}}{\abs{x_{j_v}}} \le 1/8$ with probability $1-n^{-\Omega(1)}$, we have by the  Taylor's series  expansion  of $\left( 1 + \frac{Z_{j_v}}{\abs{x_{j_v}}} \right)^{pe_{r+v}}$ up to $k$ terms (and using $k+1$-wise independence of the family $\{\omega_{lr}(i)\}_{i \in [n]}$), that
\begin{align*}
&\left( 1 + \frac{Z_{j_v}}{\abs{x_{j_v}}} \right)^{pe_{r+v}}  = \sum_{i=0}^{k-1} c_i \frac{ Z_{j_v}^i}{\abs{x_{j_v}}^i} + c_k \frac{ Z_{j_v}^{'k}}{\abs{x_{j_v}}^k}
\end{align*}
where, $ \abs{Z'_{j_v}} \le \abs{Z_{j_v}}$.

Taking expectation, and using the fact that $ \expect{Z_{j_v}^i}=0$, for $r=1,2, \ldots, k-1$, and since, $ \abs{Z'_{j_v}} \le \abs{Z_{j_v}} \le (1/(8p))\abs{x_{j_v}}$, with probability $1-n^{-\Omega(1)}$, given that $z_{j_v, l'_v}=1$, we have,
\begin{align*}
\expect{\left( 1 + \frac{Z_{j_v}}{\abs{x_{j_v}}} \right)^{pe_{r+v}} \mid \H_{J'}}& = \sum_{r=0}^{k-1} c_r \frac{ \expect{Z_{j_v}^r \mid H}}{\abs{x_{j_v}}^r} + c_k \frac{\expect{ Z_{j_v}^{'k} \mid \H_{J'}}}{\abs{x_{j_v}}^k}
= 1 + c_k \frac{\expect{ Z_{j_v}^{'k} \mid \H_{J'}}}{\abs{x_{j_v}}^k} \\
& \le 1 + \abs{c_k} \frac{ \abs{Z_{j_v}^{'k}}}{\abs{x_{j_v}}^k}   \le  1 + \abs{c_k} \frac{ \abs{Z_{j_v}^k}}{\abs{x_{j_v}}^k} \le 1 + \abs{c_k} (8p)^{-k} \le 1 +  n^{-\Omega(1)}
\end{align*}
since, $c_k = \card{\binom{ pe_{r+v}}{k}} \le  \left( \frac{pe_{r+v}}{k} \right)^{\lceil p \rceil} < 1$, assuming $k \ge 2pe_{r+v}$.

Substituting in Eqn.~\eqref{eq:dmt2} we have,
\begin{align} \label{eq:dmt3}
&\expect{P(I,J,K,e,g) \mid \bigwedge_{u=1}^r z_{i_u l_u}=1 \bigwedge_{v=1}^t z_{j_v l'_v} = 1 \bigwedge_{w=1}^t z_{k_w \ell_w}=1, \H_{J'}}  \notag \\
& \le  \left(1 + \frac{1}{8p} \right)^{p \sum_{u \in [r]} (e_u + g_u)} \left(1 + n^{-\Omega(1)} \right)^{s'+t} \notag \\
& \le e^{ \sum_{u \in [r]} (e_u + g_u)/8} \left( 1 + n^{-\Omega(1)} \right) \enspace .
\end{align}

Substituting Eqn.~\eqref{eq:dmt3} in Eqn.~\eqref{eq:dmomb}, we have,
\begin{align} \label{eq:dmomt3}
&\expect{ \left.\prod_{u=1}^r Y_{i_u}^{e_u} \conj{Y_{i_u}}^{g_u} \prod_{v \in [s]} Y_{j_v}^{e_{r+v}} \prod_{w \in [t]} \conj{Y_{k_w}}^{g_{r+v}}\right\rvert \H_{J'}} \notag \notag \\
& = e^{ \sum_{u \in [r]} (e_u+g_u)/8} \left(1 + n^{-\Omega(1)} \right) \sum_{\substack{l_u \text{ consist. w. } i_u, u =1,\ldots, r \notag \\ l'_v \text{ consist. w. } j_v, v \in [s] \notag \\ \ell_w \text{ consist. w. } k_w, w \in [t]} } 2^{\sum_{u=1}^r {l_u(e_u+g_u)} + \sum_{v=1}^s l'_u + \sum_{w=1}^t \ell_w} \notag \\
& \prod_{u=1}^r \abs{x_{i_u}}^{p(e_u+g_u)} \prod_{v=1}^s \abs{x_{j_v}}^{pe_{r+v}} \prod_{w=1}^t \abs{x_{k_w}}^{pg_{r+t}} \notag \\
& \hspace*{1.0in}\prob{\left. \bigwedge_{u=1}^r z_{i_u l_u}=1 \bigwedge_{v=1}^t z_{j_v l'_v} = 1 \bigwedge_{w=1}^t z_{k_w \ell_w}=1\right\vert \H_{J'} } \enspace .
\end{align}
By Fact~\ref{fact:1}, we have, for any $i \in \lmargin(G_0) \cup_{l=1}^L G_l$, $l$ consistent with $i$ and  $\beta \ge 1$,   $$ \abs{x_i}^{p\beta} 2^{\beta l} \le  2^l \abs{x_i}^p\left( \frac{F_2}{B} \right)^{p(\beta-1) /2} 2^{(p/2+1)(\beta-1)} = 2^l \abs{x_i}^p \left( \frac{2^{1+2/p} F_2}{B} \right)^{(p/2)(\beta-1)} \enspace . $$ Thus, \\
$$ \abs{x_{i_u}}^{p(e_u+g_u)} 2^{(e_u+g_u) l_u} \le \abs{x_{i_u}}^p 2^{l_u} \left( \frac{2^{1+2/p}F_2}{B} \right)^{(p/2)(e_u+g_u-1)} $$ and $$ \abs{x_{j_v}}^{pe_{r+v}} 2^{e_{r+v} l'_v} \le \abs{x_{j_v}}^p 2^{l'_v} \left( \frac{2^{1+2/p}F_2}{B} \right)^{(p/2)(e_{r+v}-1)}  \enspace . $$
Using this,  Eqn.~\eqref{eq:dmomt3} is upper bounded as
\begin{align} \label{eq:dmomt4}
&\expect{\left. \prod_{u=1}^r Y_{i_u}^{e_u} \conj{Y_{i_u}}^{g_u} \prod_{v \in [s']} Y_{j_v}^{e_{r+v}} \prod_{w \in [t]} \conj{Y_{k_w}}^{g_{r+v}} \right\rvert \H_{J'}} \notag \notag \\
&\le e^{ \sum_{u \in [r]} (e_u+g_u)/8} \left(1 + n^{-\Omega(1)} \right) \prod_{u=1}^r \abs{x_{i_u}}^p \prod_{v=1}^s \abs{x_{j_v}}^p \prod_{w=1}^t \abs{x_{k_w}}^p
\notag \\
& \hspace*{0.5in}\left( \frac{2^{1+2/p}F_2}{B} \right)^{(p/2)\left(\sum_{u=1}^r (e_u+g_u-1)+ \sum_{v \in [s] } (e_{r+v}-1)  + \sum_{w \in [t]} (e_{r+w}-1) \right)}\notag \\
& \sum_{\substack{l_u \text{ consist. w. } i_u, u =1,\ldots, r  \\ l'_v \text{ consist. w. } j_v, v \in [s] \\ \ell_w \text{ consist. w. } k_w, w \in [t]} }
2^{\sum_{u \in [r]}  l_u + \sum_{v \in [s]} l_v + \sum_{w \in [t]} l_w} \notag\\
& \hspace*{1.5in}
\prob{ \bigwedge_{u=1}^r z_{i_u l_u}=1 \bigwedge_{v=1}^t z_{j_v l'_v} = 1 \bigwedge_{w=1}^t z_{k_w \ell_w}=1 \mid \H_{J'}}
\end{align}

Recall that for brevity purposes, in this proof, we have extended the event $\H_{J'}$ to also include $\nocollision(J)$ and $\goodest(J)$. Then,
by Lemma~\ref{lem:approxdwise},
\begin{gather*}
\sum_{\substack{l_u \text{ consist. w. } i_u, u =1,\ldots, r \notag \\ l'_v \text{ consist. w. } j_v, v \in [s] \notag \\ \ell_w \text{ consist. w. } k_w, w \in [t]} }
2^{\sum_{u \in [r]}  l_u + \sum_{v \in [s]} l_v + \sum_{w \in [t]} l_w}\\
\hspace*{1.0in} \cdot  \prob{ \bigwedge_{u=1}^r z_{i_u l_u}=1 \bigwedge_{v=1}^t z_{j_v l'_v} = 1 \bigwedge_{w=1}^t z_{k_w \ell_w}=1 \mid \H_{J'}}
 = 1 \pm n^{-\Omega(1)}
\end{gather*}

Substituting in Eqn.~\eqref{eq:dmomt4},  we obtain
\begin{align*}
&\expect{\left. \prod_{u=1}^r Y_{i_u}^{e_u} \conj{Y_{i_u}}^{g_u} \prod_{v \in [s']} Y_{j_v}^{e_{r+v}} \prod_{w \in [t]} \conj{Y_{k_w}}^{g_{r+v}} \right\rvert \H_{J'}} \notag \notag \\
&\le e^{ \sum_{u \in [r]} (e_u+g_u)/8} \left(1 + n^{-\Omega(1)} \right)^2 \prod_{u=1}^r \abs{x_{i_u}}^p \prod_{v=1}^s \abs{x_{j_v}}^p \prod_{w=1}^t \abs{x_{k_w}}^p\notag \\
& \hspace*{0.5in} \left( \frac{2^{1+2/p}F_2}{B} \right)^{(p/2)\left(\sum_{u=1}^r (e_u+g_u-1)+ \sum_{v \in [s] } (e_{r+v}-1)  + \sum_{w \in [t]} (e_{r+w}-1) \right)} \\
& \le 2 e^{ \sum_{u \in [r]} (e_u+g_u)/8}\prod_{u=1}^r \abs{x_{i_u}}^p \prod_{v=1}^s \abs{x_{j_v}}^p \prod_{w=1}^t \abs{x_{k_w}}^p \\
& \hspace*{0.5in}\left( \frac{ 4 F_2}{B} \right)^{(p/2)\left(\sum_{u=1}^r (e_u+g_u-1)+ \sum_{v \in [s] } (e_{r+v}-1) + \sum_{w \in [t]} (e_{r+w}-1) \right)}
\end{align*}
since, $p \ge 2$ and  $ 2^{1+2/p} \le  2^2 = 4$.

\end{proof}

\begin{lemma} \label{lem:centmompow}
Let $X$ be a non-negative real random variable with expectation $\expect{X} = \mu$. Then, for any integer $e \ge  2$, $\expect{ (X-\mu)^e } \le \expect{X^e} + \mu^e \enspace . $
\end{lemma}
\begin{proof} Let $p_x = \prob{X=x}$. Then,
\begin{align}
\expect{ (X-\mu)^e} & = \sum_{x} p_x (x-\mu)^e \notag \\
&= \sum_{x \le \mu} p_x (x-\mu)^e + \sum_{x>\mu} p_x (x-\mu)^e \label{eq:centmompow1} \\
& \le \sum_{x \le \mu} p_x (\mu-x)^e + \sum_{x > \mu} p_x x^e \notag \\
& \le \sum_{x \le \mu} p_x \mu^e + \expect{X^e} \notag \\
& \le \mu^e + \expect{X^e} \enspace .  \notag
\end{align}
Step 1 follows from the definition of expectation. Step 2 separates the summation into the ranges $x \le
\mu$ and $x > \mu$. Step 3 uses the fact  that  in the range $x \le \mu$, $\mu-x \le 0$ and therefore for any odd $e$, $(\mu-x)^e \le (x-\mu)^e$. For even $e$, the two powers $(\mu-x)^e$ and $(x-\mu)^e$ are the same. Hence for all $e \ge 2$,  $(\mu-x)^e  \le (x-\mu)^e$. Further,  for $x \ge \mu$,  $(x-\mu)^e \le x^e$, since, $x \ge 0$ by virtue of $X$ being a non-negative random variable.
Step 3 follows from the facts that (i) in the range $x \le \mu$, $0 \le (\mu-x) \le \mu$ and hence, $(\mu-x)^e \le \mu^e$, and,  (ii) by non-negativity of $X$,  $\sum_{x> \mu} p_x x^e \le \sum_{x \ge 0} p_x x^e  = \expect{X^e}$. The final step follows since $\sum_{x \le \mu} p_x \mu^e = \mu^e \prob{X \le \mu} \le \mu^e$.

\end{proof}

Recall that $G' = \lmargin(G_0) \cup_{l=1}^L G_l$.
\begin{lemma} \label{lem:dmf2:centtomom} 
Let $j \in G'$. Then, the following expectation is real and is bounded above as follows.
\begin{align*}
\expect{(Y_j-\expect{Y_j \mid \H})^e_j \mid \H} \le  (1+\delta'') (\expect{Y_j^{e_j} \mid \H} + \abs{x_j}^{pe_j}(1+\delta'))
\end{align*}
where, $\delta', \delta'' \in \min(O(\delta), n^{-\Omega(1)})$.
\end{lemma}
\begin{proof} By definition, $Y_j = \sum_{l \consistw j} 2^l z_{jl} X_j^p$, where, $z_{jl}=1$ iff $i \in \bar{G_l}$. Taking expectation of $Y_j$,
$
\expect{Y_j \mid \H} = \sum_{l \consistw j} 2^l \expect{z_{jl} X_j^p \mid \H}$. Note that $\expectsub{\omega}{X_j^p \mid \H} = \abs{x_j}^p$ as discussed earlier. This does not depend on the hash functions $g_1, \ldots, g_L$ which determine $z_{jl}$. Therefore,
$\expect{z_{jl} X_j^p \mid \H} = \expect{z_{jl} \mid \H} \abs{x_j}^p$. Thus, we have,
$$\expect{Y_j \mid \H} = \abs{x_j}^p \expect{ \sum_{l \consistw j} 2^l z_{jl} \mid \H} = \abs{x_j}^p (1 \pm 2^l \min (O(\delta), n^{-\Omega(1)}))  $$
where, the last step follows from Lemma~\ref{lem:margin}. In the following, we will refer to $(1 \pm 2^l \min (O(\delta), n^{-\Omega(1)}))$ as $1\pm \delta'$, where, $\delta' = \min (O(\delta), n^{-\Omega(1)}))$.

Therefore, letting $c_k = \binom{e_j}{k}$,
\begin{align} \label{eq:lem:dmf2:centtomom:t1}
&\expect{ (Y_j - \expect{Y_j \mid \H})^{e_j}\mid \H} = \sum_{j=0}^k c_k  \expect{ Y_j^k \mid \H} (-1)^{e_j-k} \abs{x_j}^{p(e_j-k)} (1\pm \delta')^{e_j-k}  \enspace .
\end{align}
Now, $\expect{ Y_j^k \mid \H} = 1$, if $k=0$. Otherwise, for $k \ge 1$, note that
$$ \left( \sum_{l \consistw j} 2^l z_{jl} X_j^p \right)^k = \sum_{l \consistw j} 2^{kl} z_{jl}X_j^{pk} \enspace . $$
Taking expectations of both sides, and noting that $\expectsub{\omega}{X_j^{pk}} = \abs{x_j}^{pk}$ and the expectation of the $z_{jl}$'s is independent of  $\omega$, we have,
\begin{align*}
\expect{ \left( \sum_{l \consistw j} 2^l z_{jl} X_j^p \right)^k} = \abs{x_j}^{pk} \expect{ \sum_{l \consistw j} 2^{kl} z_{jl}} \enspace .
\end{align*}
Substituting in Eqn.~\eqref{eq:lem:dmf2:centtomom:t1}, we have,
\begin{align*}
&\expect{ (Y_j - \expect{Y_j \mid \H})^{e_j}\mid \H}  = \sum_{j=0}^k c_k  \abs{x_j}^{pk} \expect{ \sum_{l \consistw j} 2^{kl} z_{jl}}(-1)^{e_j-k} \abs{x_j}^{p(e_j-k)} (1\pm \delta')^{e_j-k} \enspace .
\end{align*}
Define the random variable $W_j = \sum_{l \consistw j} 2^l z_{jl}$. Then, the above expectation can be written as  follows.
\begin{align}\label{eq:lem:dmf2:centtomom:t2}
 &\expect{ ((Y_j - \expect{Y_j \mid \H})^{e_j}\mid \H}\notag  \\
  &= \sum_{j=0}^k c_k  \abs{x_j}^{pe_j} \expect{W_j^k} (-1)^{e_j-k} \abs{x_j}^{p(e_j-k)} (1\pm \delta')^{e_j-k} \notag \\
 & = (1\pm \delta'')\abs{x_j}^{pe_j}\expect{ (W_j - \expect{W_j\mid \H})^{e_j} \mid \H}, \text{ where, $0 \le \delta'' \le \min(O(\delta), n^{-\Omega(1)})$} \notag \\
 & \le (1+\delta'') \abs{x_j}^{pe_j} (\expect{W_j^{e_j} \mid \H} + \expect{W_j \mid \H}^{e_j})
\end{align}
Now, $\abs{x_j}^{pe_j} \expect{W_j^{e_j} \mid \H} = \expect{Y_j^{e_j} \mid \H}$, by calculations as done above. Similarly,\\ $\abs{x_j}^{pe_j} (\expect{W_j \mid \H})^{e_j} = \abs{x_j}^{pe_j} (1\pm \delta')^{e_j}$. Substituting in Eqn.~\eqref{eq:lem:dmf2:centtomom:t1}, we have,
\begin{align*}
&\expect{ ((Y_j - \expect{Y_j \mid \H})^{e_j}\mid \H} \le (1+\delta'') (\expect{Y_j^{e_j} \mid \H} + \abs{x_j}^{pe_j}(1+\delta')) \enspace .
\end{align*}

\end{proof}

Lemma~\ref{lem:dmf2:centmom2} presents an approximation of the central moment $\expect{(Y_j-\expect{Y_j \mid \H})^{e_j} \mid \H}$ in terms of the (non-central) moment $\expect{Y_j^{e_j} \mid \H}$.

\begin{lemma} \label{lem:dmf2:centmom2}
Let $j \in G'$ and $e_j$ be an integer $\ge 0$.  Then,
\begin{align*}
\expect{(Y_j-\expect{Y_j \mid \H})^{e_j} \mid \H} \le (2 + \delta') \expect{Y_j^{e_j} \mid \H}
\end{align*}
where, $\delta'\le n^{-\Omega(1)}$.
\end{lemma}

\begin{proof}
We begin with the statement of Lemma~\ref{lem:dmf2:centtomom}.
\begin{align} \label{eq:lem:dmf2:centtomom:repeat}
\expect{(Y_j-\expect{Y_j \mid \H})^e_j \mid \H} \le  (1+\delta'') (\expect{Y_j^{e_j} \mid \H} + \abs{x_j}^{pe_j}(1+\delta'))
\end{align}
where, $\delta', \delta'' \le n^{-\Omega(1)}$.  We will now attempt to lower bound the term $\expect{Y_j^{e_j} \mid \H}$ in terms of $\abs{x_j}^{pe_j}$.
We will follow calculations similar to those in the proof of Lemma~~\ref{lem:dmf2:centtomom}. Firstly, if $e_j =0$, then, $\expect{Y_j^{e_j}} = 1 = \abs{x_j}^{pe_j}$. So now assume that $e_j \ge 1$ and integral.
\begin{align} \label{eq:lem:dmf2:centmom2:t1}
\expect{Y_j^{e_j} \mid \H} & = \abs{x_j}^{pe_j} \expect{ \biggl(\sum_{l \consistw j} 2^l z_{jl} \biggr)^{e_j}}  = \abs{x_j}^{pe_j} \expect{ \sum_{l \consistw j} 2^{le_j} z_{jl} } \enspace .
\end{align}
We write $2^{l e_j} z_{jl} = 2^{le_j-l} \cdot 2^l z_{jl}$, for $l \consistw j$.  Let $l  \consistw j$.

\noindent
\emph{Case 1:} If $j \in \midreg(G_r)$, then, the only value of  $l$ consistent with $j$ is $r$. In this case, $2^{l e_j} z_{jl} = 2^{r(e_j-1)} \cdot 2^l z_{jl}$.

\noindent \emph{Case 2:} If $j \in \rmargin(G_r)$, then, $2^{l e_j} z_{jl} \ge 2^{(r-1) (e_j-1)} \cdot 2^l z_{jl}$.

\noindent \emph{Case 3:} If $j \in \lmargin(G_r)$, then, $2^{l e_j} z_{jl} \ge 2^{r(e_j-1)} \cdot 2^l z_{jl}$.

Now  $j \in \lmargin(G_0) \cup_{l=1}^L G_l$. If $j \in G_r$ and $r \in [L]$, then, for any $l$ consistent with $j$,  $2^{l e_j} z_{jl} \ge 2^{(r-1) (e_j-1)} \cdot 2^l z_{jl}$. In this case,
$$\expect{ \sum_{l \consistw j} 2^{le_j} z_{jl}\mid \H } \ge 2^{(r-1)(e_j-1)} \expect{ \sum_{l \consistw j} 2^l z_{jl} \mid \H} = 2^{(r-1)(e_j-1)} (1\pm \delta') \ge 1- \delta' $$ by Lemma~\ref{lem:margin}, and since, $r \ge 1$ and $e_j \ge 1$.

 If $j \in \lmargin(G_0)$, then, for any $l$ consistent with $j$, $2^{l e_j} z_{jl} \ge 2^{0(e_j-1)} \cdot  2^l z_{jl} = 2^l z_{jl}$. In this case, $\expect{  \sum_{l \consistw j} 2^{le_j} z_{jl}\mid \H } = 1- \delta'$, by Lemma~\ref{lem:margin}. Therefore, in all cases, for $j \in \lmargin(G_0) \cup_{l=1}^L G_l$, we have,
 $$ \expect{ \sum_{l \consistw j} 2^{le_j} z_{jl}\mid \H } \ge 1- \delta' \enspace . $$
From Eqn.~\eqref{eq:lem:dmf2:centmom2:t1} for $e_j \ge 1$ we have,
$\displaystyle \expect{Y_j^{e_j} \mid \H} \ge (1-\delta') \abs{x_j}^{pe_j} \enspace . $

\noindent
Substituting in Eqn.~\eqref{eq:lem:dmf2:centtomom:repeat} we have,
\begin{align*}
\expect{(Y_j-\expect{Y_j \mid \H})^e_j \mid \H}  &
\le  (1+\delta'') (\expect{Y_j^{e_j} \mid \H} + \abs{x_j}^{pe_j}(1+\delta')) \\
& \le  (1+\delta'')  ( 1 + (1-\delta')^{-1} ) \expect{Y_j^{e_j} \mid \H}\\
& =  (2 + \delta''') \expect{Y_j^{e_j} \mid \H}
\end{align*}
where, $\delta''' \le  n^{-\Omega(1)})$.

Note that when $e_j=0$, $ \expect{Y_j^{e_j} \mid \H} = 1 = \abs{x_j}^{pe_j}$, and the \emph{RHS} above becomes $(1+\delta'')(2+\delta') = (2+ \delta''') \expect{Y_j^{e_j} \mid \H}$.

\end{proof}

\begin{lemma} \label{lem:dmf2:centmom3} Let $j \in \lmargin(G_0) \cup_{l=1}^L G_l$. Then, for any $e_j, g_j \ge 0$,
\begin{align*}
\expect{(Y_j - \expect{Y_j \mid \H})^{e_j} (\conj{Y_j} - \expect{\conj{Y_j} \mid \H})^{g_j} \mid \H}
\le  \left(\frac{2(1+\delta')}{1-\delta'} \right)^{e_j+g_j} \expect{Y_j^{e_j} \conj{Y_j}^{g_j}}
\end{align*}
where, $\delta' = n^{-\Omega(1)}$.
\end{lemma}

\begin{proof} We have $Y_j = \sum_{l_j \consistw j} 2^l z_{lj} X_j^p$. Further, $\expect{Y_j \mid \H} = \abs{x_j}^p (1 \pm \delta')$, where, $\delta' \le n^{-\Omega(1)}$. Let $c_k = \binom{e_j}{k}$ and $d_k =  \binom{g_j}{k}$. Therefore,
\begin{align}\label{eq:lem:dmf2:centmom3:t1}
&\expect{(Y_j - \expect{Y_j \mid \H})^{e_j} (\conj{Y_j} - \expect{\conj{Y_j} \mid \H})^{g_j} \mid \H} \notag  \notag \\
& = \expect{ \left( \sum_{l \consistw j} 2^l z_{lj} X_j^p - \abs{x_j}^p (1\pm \delta')\right)^{e_j}
\left(\sum_{l \consistw j} 2^{l} z_{lj} X_j^p - \abs{x_j}^p (1 \pm \delta')\right)^{g_j}} \notag  \notag \\
& = \textbf{E} \left[ \left( \sum_{k=0}^{e_j} c_k \left( \sum_{l \consistw j} 2^l z_{lj} X_j^p \right)^k (-1)^{e_j-k} \abs{x_j}^{p(e_j-k)}(1 \pm \delta')^{p(e_j-k)}\right) \right. \notag \\
&\hspace*{0.5in}\left.  \left( \sum_{k'=0}^{g_j} d_k \left( \sum_{l \consistw j} 2^l z_{jl} X_j^p\right)^{k'} (1-)^{g_j-k'}
\abs{x_j}^{p(g_j-k')} (1\pm \delta')^{p(e_j-k')} \right) \right]  \notag \\
& =  \sum_{k,k'=0,0}^{e_j, g_j} c_k d_{k'} \expect{ \left( \sum_{l \consistw j} 2^l z_{lj} X_j^p \right)^k \left( \sum_{l \consistw j} 2^l z_{jl} X_j^p\right)^{k'}} \notag \\
&\hspace*{1.0in} (-1)^{e_j + g_j - k-k'} \abs{x_j}^{p(e_j+g_j -k-k')} (1\pm \delta')^{e_j + g_j -k-k'} \enspace .
\end{align}
 Since, $z_{lj} z_{l'j} = 0$, for any distinct $l, l'$, we have for any fixed $0 \le k \le e_j$ and $0 \le k' \le g_j$ that,
\begin{align} \label{eq:lem:dmf2:cemtmom3:t3}
&\expect{\left( \sum_{l \consistw j} 2^l z_{lj} X_j^p \right)^k \left( \sum_{l \consistw j} 2^l z_{jl} X_j^p\right)^{k'}} \notag \\
& = \expect{\left( \sum_{l \consistw j} 2^{lk} z_{lj}^k X_j^{pk} \right)\left( \sum_{l \consistw j} 2^{lk'} z_{jl}^{k'} \conj{X_j}^{pk'} \right)} \notag \\
& = \expect{ \sum_{l \consistw j} 2^{l(k+k')} z_{lj}^{k+k'} X_j^{pk} \conj{X_j}^{pk'}} \enspace .
\end{align}

Let $a^{pk}_r = \binom{pk}{r}$ and $b_{r'}^{pk'}$ denote $\binom{pk'}{r'}$.
\begin{align} \label{eq:lem:dmf2:cemtmom3:t4}
\expectsub{\omega}{X_j^{pk} \conj{X_j}^{pk'} \mid \H} & = \abs{x_j}^{p(k+k')} \expectsub{\omega}{\left( 1 + \frac{Z_j}{\abs{x_j}} \right)^{pk} \left( 1 + \frac{\conj{Z_j}}{\abs{x_j}}\right)^{pk'} \mid \H} \notag\\
& = \abs{x_j}^{p(k+k')} \sum_{r,r'=0}^{pk, pk'} a^{pk}_r b^{pk'}_{r'}
\expectsub{\omega}{Z_j^r \conj{Z_j}^r} \abs{x_j}^{-(r+r')} \enspace .
\end{align}
Note that the expression for $Z_j^r \conj{Z_j}^{r'}$ has the multiplicative term $\omega_j^r \conj{\omega_j}^{r'}$. If $r \ne r'$, then, upon taking expectation with respect to all the various $\omega_i$'s, in particular, we have, $\expectsub{\omega_j}{ \omega_j^r \conj{\omega_j}^{r'}} = 0 $ and therefore that $\expectsub{|omega}{Z_j^r \conj{Z_j}^{r'}} =0$. Thus, Eqn.~\eqref{eq:lem:dmf2:cemtmom3:t4} is equivalent to the following expression where $r=r'$.
\begin{align} \label{eq:lem:dmf2:cemtmom3:t5}
\expectsub{\omega}{X_j^{pk} \conj{X_j}^{pk'} \mid \H}  = \abs{x_j}^{p(k+k')} \sum_{r} a^{pk}_r b^{pk'}_{r'} \expectsub{\omega}{ \frac{\abs{Z_j}^{2r}}{\abs{x_j}^{2r}} \mid \H} \enspace .
\end{align}
It follows that the expectation term in the above equation is real and non-negative.

Let $\expectsub{h}{\cdot}$ denote the expectation with respect to the hash functions and $\expectsub{\omega}{ \cdot}$  denote the expectation with respect to the complex roots of unity sketches (i.e., the remaining random bits).
Substituting Eqn~\eqref{eq:lem:dmf2:cemtmom3:t5} in Eqn.~\eqref{eq:lem:dmf2:cemtmom3:t4}, we have,
\begin{align*}
\expectsub{\omega}{X_j^{pk} \conj{X_j}^{pk'} \mid \H} & = \abs{x_j}^{p(k+k')} \sum_{r=0}^{ \min(pk, pk')} a_r^{pk} b_r ^{pk'} \expectsub{\omega}{\frac{\abs{Z_j}^{2r}}{\abs{x_j}^{2r}}} \enspace .
\end{align*}
Substituting the above  equation in Eqn.~\eqref{eq:lem:dmf2:cemtmom3:t3}, we have,
\begin{align*}
&\expect{\left( \sum_{l \consistw j} 2^l z_{lj} X_j^p \right)^k \left( \sum_{l \consistw j} 2^l z_{jl} X_j^p\right)^{k'}} \notag \\
& = \expect{ \sum_{l \consistw j} 2^{l(k+k')} z_{lj}^{k+k'} X_j^{pk} \conj{X_j}^{pk'}} \\
& = \abs{x_j}^{p(k+k')} \sum_{r=0}^{\min(pk, pk')} a_r^{pk} b_r^{pk'}  \expect{ \sum_{l \consistw j} 2^{l(k+k')}   z_{lj}^{k+k'} \frac{\abs{Z_j}^{2r}}{\abs{x_j}^{2r}}} \enspace .
\end{align*}
Substituting this equation in Eqn.~\eqref{eq:lem:dmf2:centmom3:t1}, we have,
\begin{align} \label{eq:lem:dmf2:centmom3:t6}
&\expect{(Y_j - \expect{Y_j \mid \H})^{e_j} (\conj{Y_j} - \expect{\conj{Y_j} \mid \H})^{g_j} \mid \H}\notag \\
& = \sum_{k,k'=0,0}^{e_j, g_j} c_k d_{k'} \expect{ \left( \sum_{l \consistw j} 2^l z_{lj} X_j^p \right)^k \left( \sum_{l \consistw j} 2^l z_{jl} X_j^p\right)^{k'} \mid \H} \notag \notag \\
&\hspace*{1.0in} (-1)^{e_j + g_j - k-k'} \abs{x_j}^{p(e_j+g_j -k-k')} (1\pm \delta')^{e_j + g_j -k-k'} \enspace . \notag \\
& =     \sum_{k,k'=0,0}^{e_j, g_j} c_k d_{k'} \abs{x_j}^{p(k+k')}  \sum_{r=0}^{\min(pk, pk')} a_r^{pk} b_r^{pk'} \expect{ \sum_{l \consistw j} 2^{l(k+k')}   z_{lj}^{k+k'} \frac{\abs{Z_j}^{2r}}{\abs{x_j}^{2r}}\mid \H}\notag \\
& \hspace*{1.0in} (-1)^{e_j + g_j - k-k'} \abs{x_j}^{p(e_j+g_j -k-k')} (1\pm \delta')^{e_j + g_j -k-k'} \enspace . \notag \\
& = \abs{x_j}^{p(e_j+g_j)} \sum_{k,k'=0,0}^{e_j, g_j} c_k d_{k'} \sum_{r=0}^{\min(pk, pk')} a_r^{pk} b_r^{pk'} \expect{ \sum_{l \consistw j} 2^{l(k+k')}   z_{lj}^{k+k'} \frac{\abs{Z_j}^{2r}}{\abs{x_j}^{2r}}\mid \H}\notag \\
&\hspace*{1.0in} (-1)^{e_j + g_j - k-k'}  (1\pm \delta')^{e_j + g_j -k-k'} \enspace .
\end{align}
Noting that the term in the expectation is always non-negative, replacing the powers of -1 by 1 cannot decrease the \emph{RHS}. Therefore, the \emph{RHS} of Eqn.~\eqref{eq:lem:dmf2:centmom3:t6} is  bounded above by
\begin{align} \label{eq:lem:dmf2:centmom3:t7}
\le & (1+\delta)^{e_j + g_j} \abs{x_j}^{p(e_j+g_j)} \sum_{k,k'=0,0}^{e_j, g_j} c_k d_{k'} \sum_{r=0}^{\min(pk, pk')} a_r^{pk} b_r^{pk'} \expect{ \sum_{l \consistw j} 2^{l(k+k')}   z_{lj}^{k+k'} \frac{\abs{Z_j}^{2r}}{\abs{x_j}^{2r}}\mid \H}  \notag \\
& = (1+ \delta) ^{e_j + g_j} \abs{x_j}^{p(e_j+g_j)} \sum_{k,k'=0,0}^{e_j, g_j} c_k d_{k'} \tau(k,k'),
\end{align}
where,
\begin{align} \label{eq:defn:tau}\tau(k,k') = \sum_{r=0}^{\min(pk, pk')} a_r^{pk} b_r^{pk'} \expect{ \sum_{l \consistw j} 2^{l(k+k')}   z_{lj}^{k+k'} \frac{\abs{Z_j}^{2r}}{\abs{x_j}^{2r}}\mid \H} \enspace .
\end{align}

Repeating similar calculations, we obtain,
\begin{align} \label{eq:lem:dmf2:centmom3:t8}
\expect{ Y_j^{e_j} \conj{Y_j}^{g_j} \mid \H} & \ge (1-\delta')^{e_j + g_j}\abs{x_j}^{p(e_j+g_j)} \sum_{r=0}^{\min(e_j, g_j)} a_r^{pe_j} b_r^{pg_j} \notag \\
& \hspace*{0.5in}\expect{ \sum_{l \consistw j} 2^{l(e_j + g_j)} z_{lj}^{e_j + g_j} \frac{\abs{Z_j}^{2r}}{\abs{x_j}^{2r}}\mid \H} \notag  \\
& = (1-\delta')^{e_j + g_j}\abs{x_j}^{p(e_j+g_j)} \tau(e_j, g_j) \enspace .
\end{align}

We now consider the $\tau$ function. Note that by definition, $\tau$ is a symmetric function, that is, $\tau(k,k') = \tau(k',k)$. We wish to show that $\tau$ is monotonic, that is, if $0 \le k \le e_j$ and $0 \le k' \le g_j$, then, $\tau(k,k') \le \tau(e_j, g_j)$. Assume that $0 \le k \le e_j$ and $0 \le k' \le g_j$.

\noindent
\emph{Case 1.} Suppose $k+k'=0$, or, equivalently, $k=0$ and $k'=0$. Then, $\tau(k,k')=1$. Further,
\begin{align*}
\tau(e_j, g_j ) & = \sum_{r=0}^{\min(pe_j, pg_j')} a_r^{pe_j} b_r^{pg_j} \expect{ \sum_{l \consistw j} 2^{l(e_j + g_j)}   z_{lj}^{e_j + g_j} \frac{\abs{Z_j}^{2r}}{\abs{x_j}^{2r}}\mid \H} \\
& = 1 + \sum_{r=1}^{\min(pe_j, pg_j')} a_r^{pe_j} b_r^{pg_j} \expect{ \sum_{l \consistw j} 2^{l(e_j + g_j)}   z_{lj}^{e_j + g_j} \frac{\abs{Z_j}^{2r}}{\abs{x_j}^{2r}}\mid \H} \\
& \ge 1 = \tau(k,k') \enspace .
\end{align*}
\emph{Case 2.} Now suppose $k+k' \ge 1$. We note that the function $a_r^n = \binom{n}{r}$ is monotonic in the first argument, that is, $\binom{n+1}{r} \ge \binom{n}{r}$. This  follows since, for $r=0$, these two terms are obviously equal to 1, and for $r > 0$,
\begin{align*}
\binom{n+1}{r} = \prod_{j=0}^{r-1} \frac{n+1-j}{j} > \prod_{j=0}^{r-1} \frac{n-j}{j} = \binom{n}{r} \enspace .
\end{align*}
Therefore,
\begin{align*}
\tau(k,k') &  = \sum_{r=0}^{\min(pk, pk')} \binom{pk}{r} \binom{pk'}{r} \expect{ \sum_{l \consistw j} 2^{l(k+k')}   z_{lj}^{k+k'} \frac{\abs{Z_j}^{2r}}{\abs{x_j}^{2r}}\mid \H} \\
& \le  \sum_{r=0}^{\min(pe_j, pg_j)} \binom{pe_j}{r} \binom{pe_j}{r} \expect{ \sum_{l \consistw j} 2^{l(e_j + g_j)}   z_{lj}^{e_j + g_j} \frac{\abs{Z_j}^{2r}}{\abs{x_j}^{2r}}\mid \H} \\
& = \tau(e_j, g_j) \enspace .
\end{align*}
The last step is obtained essentially by replacing each term in the expansion of $\tau(k,k')$ by a corresponding term that is a function of $e_j$ and $g_j$ and which is no smaller. First, the summation is extended from $\min(pk, pk')$ to $\min(pe_j, pg_j)$. Since the summation inside are all non-negative terms, the replacement cannot make it smaller. Secondly, the terms $\binom{pk}{r}$ and $ \binom{pk'}{r}$ are replaced by $\binom{pe_j}{r}$ and $\binom{pg_j}{r}$ respectively, which are each no smaller than its corresponding term. The term $2^{l(k + k')} \le 2^{l (e_j + g_j)}$. Now $z_{lj}$ is an indicator variable, and for $k+k' \ge 1$, $z_{lj} = z_{lj}^{k+k'} = z_{lj}^{e_j + g_j}$.

Thus, in all cases, $\tau$ is a monotonic function, that is, for $0 \le k \le e_j$ and  $\tau(k, k') \le \tau(e_j, g_j)$.

Continuing from Eqn.~\eqref{eq:lem:dmf2:centmom3:t7}, we have,
\begin{align*}
&\expect{(Y_j - \expect{Y_j \mid \H})^{e_j} (\conj{Y_j} - \expect{\conj{Y_j} \mid \H})^{g_j} \mid \H}\notag \\
&\le  (1+ \delta') ^{e_j + g_j} \abs{x_j}^{p(e_j+g_j)} \sum_{k,k'=0,0}^{e_j, g_j} c_k d_{k'} \tau(k,k') \\
& = (1+ \delta') ^{e_j + g_j} \abs{x_j}^{p(e_j+g_j)} \tau(e_j, g_j) \sum_{k,k'=0}^{e_j, g_j} c_k d_{k'} \\
& = (1+\delta')^{e_j + g_j} \abs{x_j}^{p(e_j+g_j)} \tau(e_j, g_j)  2^{e_j + g_j} \\
&  \left(\frac{2(1+\delta')}{1-\delta'} \right)^{e_j+g_j}  \expect{Y_j^{e_j} \conj{Y_j}^{g_j} \mid \H}
\end{align*}
\end{proof}

Let $G' =  \lmargin(G_0) \cup_{l=1}^L G_l$. Recall that $\H = \G \wedge \nocollision \wedge \accuest$. Although, for the next lemma, it would suffice to condition on the event $\H' = \G \wedge \nocollision(G') \wedge \accuest(G')$.
\begin{lemma} [Re-statement of Lemma~\ref{lem:dmf2}.]
Let $B \ge L  n^{1-2/p} \epsilon^{-4/p} \log^{2/p} (1/\delta))$ for a suitable constant $L$.
Then,
for integral $0 \le d_1, d_2 \le \lceil \log (1/\delta)\rceil $, we have,
$$\mathbf{E}\biggl[ \left( \sum_{i \in G'} (Y_i-\expect{Y_i \mid \H})\right)^{d_1} \left( \sum_{i \in G'} (\conj{Y_i} - \conj{\expect{Y_i \mid \H}})\right)^{d_2}\bigg\vert \H \biggr] \le \left( \frac{ \epsilon F_p}{20}\right)^{d_1+d_2} \enspace . $$
\end{lemma}

\begin{proof} 
For $d_1+d_2=0$, that is, $d_1=0$ and $d_2=0$, the statement of the lemma is vacuously true.
If $d_1 + d_2 =1$, then, say $d_1=1$, then,
\begin{align*}
\expect{ \sum_{i \in S} (Y_i-\expect{Y_i})} = \sum_{i \in S} \abs{x_i}^pn^{-\Omega(1)} = F_p n^{-\Omega(1)}
\end{align*}
implying the statement of the Lemma.

We therefore assume that $d_1+d_2 > 1$.
 Using the  notation that $S_1 = \{i_1, \ldots, i_r\}$ and $S_2 = \{j_1, \ldots, j_t\}$, we can  rewrite the summation by using three subsets: $T_1  = S_1 \cap S_2$, $T_2 = S_1 \setminus S_2$ and $T_3 = S_2 \setminus S_1$, where, $\abs{T_1} = r, \abs{T_2} = s$ and $\abs{T_3} = t$. Without loss of generality, let $d_1 \le d_2$---the other case is symmetric.

Then, the summation may be written as
\begin{align}
&\expect{ \left.\left( \sum_{i \in S} (Y_i-\expect{Y_i\mid \H})\right)^{d_1} \left( \sum_{i \in S} (\conj{Y_i} - \conj{\expect{Y_i\mid \H}})\right)^{d_2}\right\vert \H}\notag\\
& =\sum_{r=0}^{d_1} \sum_{s = \max(0, 1-r) }^{d_1-r} \sum_{t = \max(0,1-r)}^{d_2-r} \sum_{a=0}^{s} \sum_{b=0}^t
\sum_{\substack{e_1 + \ldots +e_{r+s} = d_1 \notag\\ e_1, \ldots, e_{r+s} \ge 1\\ \card{\{ j \in [s]: e_{r+j} = 1 \}} =a \\  A = \{j\in [s]: e_{r+j}=1\}}}
 \sum_{\substack{g_1 + \ldots + g_{r+t} =d_2 \\ \text{ $g_j's \ge 1$} \\ \card{\{j:g_{r+j} =1\}}=b \\ B = \{k\in [t]: g_{r+k}=1\}
 }}\binom{d_1}{e_1, \ldots, e_{r+s}} \notag\\
 & ~~~\binom{d_2}{g_1, \ldots, g_{r+t}}
\sum_{\substack{ \{ i_1, \ldots, i_r, j_1, \ldots, j_s, k_1, \ldots, k_t\}}}
\E\left[ \prod_{v \in A} (Y_{j_{v}} - \expect{Y_{j_{v}} \mid \H}) \prod_{w \in B} (\conj{Y_{k_{w}}}- \expect{\conj{Y_{k_{w}}} \mid \H}) \right.\notag\\
&  \hspace*{1.5in}\prod_{u=1}^r (Y_{i_u} -  \expect{Y_{i_u}\mid \H})^{ e_u} (\conj{Y_{i_u}} - \expect{ \conj{Y_{i_u}}})^{g_u}  \notag \\
& \hspace*{1.0in}\left. \left.\prod_{v \in [s] \setminus A}(Y_{j_{v}} - \expect{Y_{j_{v}} \mid \H})^{e_{r+v}}  \prod_{w\in [t] \setminus B}
(\conj{Y_{k_{w}}}- \expect{\conj{Y_{k_{w}}} \mid \H})^{g_{r+w}}\right\vert \H\right]\notag\\
 \label{eq:dmomc}
 \end{align}
 Let $\gamma = n^{-\Omega(1)}$.
 Using Lemma~\ref{lem:expprod1}, we have,
\begin{align} \label{eq:dmomd}
&\E\left[ \prod_{v \in A} (Y_{j_{v}} - \expect{Y_{j_{v}} \mid \H}) \prod_{w \in B} (\conj{Y_{k_{w}}}- \expect{\conj{Y_{k_{w}}} \mid \H}) \prod_{u=1}^r (Y_{i_u} -\expect{Y_{i_u}\mid \H})^{e_u} (\conj{Y_{i_u}} - \expect{ \conj{Y_{i_u}} \mid \H})^{g_u} \right.\notag\\
&  \left. \left.  \prod_{v \in [s] \setminus A}(Y_{j_{v}} - \expect{Y_{j_{v}} \mid \H})^{e_{r+v}} \prod_{w\in [t] \setminus B}
(\conj{Y_{k_{w}}}- \expect{\conj{Y_{k_{w}}} \mid \H})^{g_{r+w}}\right\vert \H\right] \notag \\
& \le \gamma^{a+b} \prod_{v \in A} \abs{x_{j_v}}^p \prod_{w \in B} \abs{x_{k_w}}^p  \mathbf{E} \left[ \prod_{u=1}^r (Y_{i_u} -\expect{Y_{i_u}\mid \H})^{e_u} (\conj{Y_{i_u}} - \expect{ \conj{Y_{i_u}} \mid \H})^{g_u}\right.  \notag \\
&\hspace*{0.5in} \left. \left.\prod_{v \in [s] \setminus A}(Y_{j_{v}} - \expect{Y_{j_{v}} \mid \H})^{e_{r+v}} \prod_{w\in [t] \setminus B}
(\conj{Y_{k_{w}}}- \expect{\conj{Y_{k_{w}}} \mid \H})^{g_{r+w} } \right\vert \H \right] \notag \\
& = \gamma^{a+b} \prod_{v \in A} \abs{x_{j_v}}^p \prod_{w \in B} \abs{x_{k_w}}^p
\prod_{u=1}^r  \expect{ (Y_{i_u} - \expect{Y_{i_u}})^{e_u}( \conj{Y_{i_u}} - \expect{ \conj{Y_{i_u}}})^{g_u} \mid \H} \notag \\
& \hspace*{0.5in} \prod_{v \in [s] \setminus A}\expect{(Y_{j_{v}} - \expect{Y_{j_{v}} \mid \H})^{e_{r+v}} \mid \H}
\prod_{w\in [t] \setminus B}\expect{
(\conj{Y_{k_{w}}}- \expect{\conj{Y_{k_{w}}} \mid \H})^{g_{r+w}} \mid \H } \enspace .
\end{align}

We will now use Lemma~\ref{lem:dmf2:centmom3}, which states that, for $e_j, g_j \ge 0$,
$$\expect{ (Y_j - \expect{Y_j \mid \H})^{e_j} ( \conj{Y_j} - \expect{ \conj{Y_j} \mid \H})^{g_j} \mid \H}
\le \left( 2\left(\frac{1+\delta'}{1-\delta'}\right)\right)^{e_j + g_j}\expect{ Y_j^{e_j} \conj{Y_j}^{g_j}}$$
where, $\delta' = n^{-\Omega(1)}$.

Note that for $\delta' = n^{-\Omega(1)}$, where, the constant term in $\Omega(1)$ can be made as large as needed by choosing the parameter of width of hash tables appropriately, and since, $e_j + g_j \le \log (1/\delta) \le n^{2/p}$, it follows that $\left( \frac{1+\delta'}{1-\delta'}\right)^{e_j + g_j} \le 1+\delta''$, where, $\delta'' = n^{-\Omega(1)}$ for a different constant in $\Omega(1)$. Also, $ (1+\delta'')^{\lceil \log (1/\delta) \rceil}  \le 2$, which we will use below.

Substituting in Eqn.~\eqref{eq:dmomd}, we have,
\begin{align*}
&\E\left[ \prod_{v \in A} (Y_{j_{v}} - \expect{Y_{j_{v}} \mid \H}) \prod_{w \in B} (\conj{Y_{k_{w}}}- \expect{\conj{Y_{k_{w}}} \mid \H}) \prod_{u=1}^r (Y_{i_u} -\expect{Y_{i_u}\mid \H})^{e_u} (\conj{Y_{i_u}} - \expect{ \conj{Y_{i_u}} \mid \H})^{g_u} \right.\notag\\
&  \left. \left.  \prod_{v \in [s] \setminus A}(Y_{j_{v}} - \expect{Y_{j_{v}} \mid \H})^{e_{r+v}} \prod_{w\in [t] \setminus B}
(\conj{Y_{k_{w}}}- \expect{\conj{Y_{k_{w}}} \mid \H})^{g_{r+w}}\right\vert \H\right] \notag \\
& \le  \gamma^{a+b} \prod_{v \in A} \abs{x_{j_v}}^p \prod_{w \in B} \abs{x_{k_w}}^p  \cdot (2) \cdot  2^{ \sum_{u=1}^r (e_u + g_u) + \sum_{v \in [s]} e_{r+v} + \sum_{w \in [t]} g_{r+w}} \\
& \hspace*{0.5in} \prod_{u=1}^r  \expect{ Y_{i_u}^{e_u}\conj{Y_{i_u}}^{g_u} \mid \H}  \prod_{v \in [s] \setminus A}\expect{Y_{j_{v}}^{e_{r+v}} \mid \H}
\prod_{w\in [t] \setminus B}\expect{
\conj{Y_{k_{w}}}^{g_{r+w}} \mid \H } \enspace .
\end{align*}
Now, the sum of the exponents in the power of 2 in the above expression is $d_1 + d_2 -a -b$. Thus, the above equation becomes,
\begin{align} \label{eq:dmf2:centtomom1}
&\E\left[ \prod_{v \in A} (Y_{j_{v}} - \expect{Y_{j_{v}} \mid \H}) \prod_{w \in B} (\conj{Y_{k_{w}}}- \expect{\conj{Y_{k_{w}}} \mid \H}) \prod_{u=1}^r (Y_{i_u} -\expect{Y_{i_u}\mid \H})^{e_u} (\conj{Y_{i_u}} - \expect{ \conj{Y_{i_u}} \mid \H})^{g_u} \right.\notag\notag \\
&  \left.   \prod_{v \in [s] \setminus A}(Y_{j_{v}} - \expect{Y_{j_{v}} \mid \H})^{e_{r+v}} \prod_{w\in [t] \setminus B}
(\conj{Y_{k_{w}}}- \expect{\conj{Y_{k_{w}}} \mid \H})^{g_{r+w}}\vert \H\right] \notag \notag \\
&\le  \gamma^{a+b} 2^{d_1 + d_2 -a-b+1} \prod_{v \in A} \abs{x_{j_v}}^p \prod_{w \in B} \abs{x_{k_w}}^p  \prod_{u=1}^r  \expect{ Y_{i_u}^{e_u}\conj{Y_{i_u}}^{g_u} \mid \H}  \prod_{v \in [s] \setminus A}\expect{Y_{j_{v}}^{e_{r+v}} \mid \H}
\notag \notag \\ &\hspace*{3.0in}\prod_{w\in [t] \setminus B}\expect{
\conj{Y_{k_{w}}}^{g_{r+w}} \mid \H }   \notag \\
& = \gamma^{a+b} 2^{d_1 + d_2 -a-b+1} \prod_{v \in A} \abs{x_{j_v}}^p \prod_{w \in B} \abs{x_{k_w}}^p \expect{ \prod_{u=1}^r  Y_{i_u}^{e_u}\conj{Y_{i_u}}^{g_u} \prod_{v \in [s] \setminus A}Y_{j_{v}}^{e_{r+v}}\prod_{w\in [t] \setminus B} \conj{Y_{k_{w}}}^{g_{r+w}} \mid \H},
\end{align}
where we have used the fact that $\sum_{u=1}^r (e_u+g_u) + \sum_{v=1}^s e_{r+v} + \sum_{w=1}^t g_{r+w} = d_1 + d_2$.

By Lemma~\ref{lem:intermed1}, we have, for  $e_1, \ldots, e_r, g_1, \ldots, g_r \ge 1$ and $e_{r+1}, \ldots, e_{r+t} \ge 2$ and $g_{r+1}, \ldots, g_{r+t} \ge 2$, that
\begin{align} \label{eq:dmf2:t2a}
&\expect{ \prod_{u=1}^r Y_{i_u}^{e_u} \conj{Y_{i_u}}^{g_u} \prod_{v \in [s]} Y_{j_v}^{e_{r+v}} \prod_{w \in [t]} \conj{Y_{k_w}}^{g_{r+v}}\vert \H}\notag\\
&\le 2 e^{ \sum_{u \in [r]} (e_u+g_u)/8}\prod_{u=1}^r \abs{x_{i_u}}^p \prod_{v=1}^s \abs{x_{j_v}}^p \prod_{w=1}^t \abs{x_{k_w}}^p \notag \\
&\hspace*{0.5in} \left( \frac{ 4 F_2}{B} \right)^{(p/2)\left(\sum_{u=1}^r (e_u+g_u-1)+ \sum_{v \in [s] } (e_{r+v}-1) + \sum_{w \in [t]} (e_{r+w}-1) \right)}  \notag\\
 & \le 2 e^{(d_1+d_2)/8} \prod_{u=1}^r \abs{x_{i_u}}^p \prod_{v\in [s]} \abs{x_{j_v}}^p \prod_{w \in [t]} \abs{x_{k_w}}^p\left( \frac{4F_2}{B} \right)^{(p/2)(d_1+d_2-r-s-t-a-b)}
 \end{align}
using the fact that $\sum_{u=1}^r (e_u+g_u) + \sum_{v=1}^s e_{r+v} + \sum_{w=1}^t g_{r+w} = d_1 + d_2$.

Substituting Eqn~\eqref{eq:dmf2:t2a} in  Eqn.~\eqref{eq:dmf2:centtomom1} we have,
\begin{align} \label{eq:dmf2:t1b}
&\E\left[ \prod_{v \in A} (Y_{j_{v}} - \expect{Y_{j_{v}} \mid \H}) \prod_{w \in B} (\conj{Y_{k_{w}}}- \expect{\conj{Y_{k_{w}}} \mid \H}) \prod_{u=1}^r (Y_{i_u} -\expect{Y_{i_u}\mid \H})^{e_u} (\conj{Y_{i_u}} - \expect{ \conj{Y_{i_u}} \mid \H})^{g_u} \right.\notag\\
&  \left.  \left. \prod_{v \in [s] \setminus A}(Y_{j_{v}} - \expect{Y_{j_{v}} \mid \H})^{e_{r+v}} \prod_{w\in [t] \setminus B}
(\conj{Y_{k_{w}}}- \expect{\conj{Y_{k_{w}}} \mid \H})^{g_{r+w}}\right\vert \H\right]  \notag \\
& \le 4\gamma^{a+b}2^{d_1+d_2} e^{(d_1+d_2)/8}\prod_{u=1}^r \abs{x_{i_u}}^p \prod_{v\in [s]} \abs{x_{j_v}}^p \prod_{w \in [t]} \abs{x_{k_w}}^p \left( \frac{4F_2}{B}\right)^{(p/2)(d_1+d_2-r-s-t-a-b)}
\end{align}

Substituting Eqn.~\eqref{eq:dmf2:t1b} in Eqn.~\eqref{eq:dmomc}, we have,
\begin{align} \label{eq:dmomsmp1}
&\expect{ \left.\left( \sum_{i \in S} (Y_i-\expect{Y_i\mid \H})\right)^{d_1} \left( \sum_{i \in S} (\conj{Y_i} - \conj{\expect{Y_i\mid \H}})\right)^{d_2}\right\vert \H}\notag\\
&\le 4\cdot 2^{d_1+d_2}e^{(d_1+d_2)/8} \sum_{r=0}^{d_1} \sum_{s = \max(0, 1-r) }^{d_1-r} \sum_{t = \max(0,1-r)}^{d_2-r}
\sum_{a=0}^{s} \sum_{b=0}^t
\sum_{\substack{e_1 + \ldots +e_{r+s} = d_1 \\ e_1, \ldots, e_{r+s} \ge 1\\ \card{\{ j \in [s]: e_{r+j} = 1 \}} =a}}
 \sum_{\substack{g_1 + \ldots + g_{r+t} =d_2 \\ \text{ $g_j's \ge 1$} \\ \card{\{j:g_{r+j} =1\}}=b
 }} \notag \\
 &~~\binom{d_1}{e_1, \ldots, e_{r+s}} \binom{d_2}{g_1, \ldots, g_{r+t}} \sum_{\substack{\{i_1, \ldots, i_r ,j_1 , \ldots j_s, k_1, \ldots, k_t\}}}
\gamma^{a+b}\prod_{u=1}^r \abs{x_{i_u}}^p \prod_{v\in [s]} \abs{x_{j_v}}^p \prod_{w \in [t]} \abs{x_{k_w}}^p  \notag \\
& \hspace*{1.0in} \left( \frac{4F_2}{B}\right)^{(p/2)(d_1+d_2-r-s-t-a-b)} \notag \\
& \le  2  \left( \frac{8e^{1/(4p)}F_2}{B} \right)^{(p/2)(d_1 + d_2)}\sum_{r=0}^{d_1} \sum_{s = \max(0, 1-r) }^{\lfloor (d_1-r)\rfloor } \sum_{t = \max(0,1-r)}^{\lfloor (d_2-r) \rfloor}\sum_{a=0}^{s} \sum_{b=0}^t
\sum_{\substack{e_1 + \ldots +e_{r+s} = d_1 \notag\\ e_1, \ldots, e_{r+s} \ge 1\\ \card{\{ j \in [s]: e_{r+j} = 1 \}} =a}} \notag \\
& \sum_{\substack{g_1 + \ldots + g_{r+t} =d_2 \notag\\ \text{ $g_j's \ge 1$} \\ \card{\{j:g_{r+j} =1\}}=b
 }} \binom{d_1}{e_1, \ldots, e_{r+s}} \binom{d_2}{g_1, \ldots, g_{r+t}} \frac{ \gamma^{a+b}F_p^{r+s+t} }{(4F_2/B)^{(p/2)(r+s+t+a+b)} (r+s+t)!} \notag \\
 & \le 2 \left( \frac{8e^{1/(4p)}F_2}{B} \right)^{(p/2)(d_1 + d_2)}\sum_{r=0}^{d_1} \sum_{s = \max(0, 1-r) }^{d_1-r } \sum_{t = \max(0,1-r)}^{d_2-r}\frac{ F_p^{r+s+t} }{(4F_2/B)^{(p/2)(r+s+t)} (r+s+t)!} \notag \\
 &\sum_{a=0}^{s} \sum_{b=0}^t  \frac{\gamma^{a+b} }{(4F_2/B)^{(p/2)(a+b)}}
\sum_{\substack{e_1 + \ldots +e_{r+s} = d_1 \\ e_1, \ldots, e_{r+s} \ge 1\\ \card{\{ j \in [s]: e_{r+j} = 1 \}} =a}}
 \sum_{\substack{g_1 + \ldots + g_{r+t} =d_2 \\ \text{ $g_j's \ge 1$} \\ \card{\{j:g_{r+j} =1\}}=b
 }} \binom{d_1}{e_1, \ldots, e_{r+s}} \binom{d_2}{g_1, \ldots, g_{r+t}}  1
  \end{align}

Fix $r,s,t,a$ and $b$ as per the constraints, and consider the inner summation. Given the vectors $e_1, \ldots, e_{r+s}$ and  $g_1, \ldots, g_{r+t}$, define the vectors $e'_1, \ldots, e'_{r+s}$ and $g'_1, \ldots, g'_{r+s}$ as follows. $e'_j = e_j-1$, for $j \in [r]$, $e'_{r+j} = e_{r+j}-2$, for $j \in [t]$ and $e_{r+j} \ne 1$, and $e'_{r+j} = 0$, if $e_{r+j}=1$. Similarly, the $r+t$-dimensional vector $g'$  is defined. Therefore,
\begin{align}\label{eq:dmf2:t2}
&\sum_{\substack{e_1 + \ldots +e_{r+s} = d_1 \\ e_1, \ldots, e_{r+s} \ge 1\\ \card{\{ j \in [s]: e_{r+j} = 1 \}} =a}}
 \sum_{\substack{g_1 + \ldots + g_{r+t} =d_2 \\ \text{ $g_j's \ge 1$} \\ \card{\{j:g_{r+j} =1\}}=b
 }} \binom{d_1}{e_1, \ldots, e_{r+s}} \binom{d_2}{g_1, \ldots, g_{r+t}}  1  \notag \\
& = \Biggl(\sum_{\substack{e_1 + \ldots +e_{r+s} = d_1 \notag\\ e_1, \ldots, e_{r+s} \ge 1\\ \card{\{ j \in [s]: e_{r+j} = 1 \}} =a}}\binom{d_1}{e_1, \ldots, e_{r+s}} 1\Biggr) \Biggl(
 \sum_{\substack{g_1 + \ldots + g_{r+t} =d_1 \notag\\ \text{ $g_j's \ge 1$} \\ \card{\{j:g_{r+j} =1\}}=b
 }} \binom{d_2}{g_1, \ldots, g_{r+t}}  1\Biggr) \notag \\
 & \le \Biggl(d_1^{r+a + 2(s-a)} \sum_{\substack{e'_1 + \ldots e'_{r+s} = d-(r+2s-a)\\
 \card{\{j: e'_{r+j} = 0\}} = a}} \binom{ d_1-(r+2s-a)}{e'_1, \ldots, e'_{r+s}} \Biggr) \notag\\
&\hspace*{0.5in} \Biggl( d_2^{r+b +2(t-b)}\sum_{\substack{g'_1 + \ldots g'_{r+t} = d_2-(r+2t-b)\\
 \card{\{j: g'_{r+j} = 0\}} = b}} \binom{ d_2-(r+2t-b)}{g'_1, \ldots, g'_{r+t}} \Biggr)
 \end{align}
 Now,
 \begin{align*}\sum_{\substack{e'_1 + \ldots e'_{r+s} = d_1-(r+2s-a)\\
 \card{\{j: e'_{r+j} = 0\}} = a}} & \binom{ d_1-(r+2s-a)}{e'_1, \ldots, e'_{r+s}}\\
  &= \binom{r+s}{a}\sum_{\substack{f_1 + \ldots, f_{r+s-a} = d_1-(r+2s-a)}} \binom{ d_1-(r+2s-a)}{f_1, \ldots, f_{r+s-a}} \\ & = \binom{r+s}{a} (r+s-a)^{d_1-(r+2s-a)}  \\
  & \le \frac{(r+s)^{d_1-(r+2s-a)+a}}{a!}
 \end{align*}
 Similarly,
 \begin{align*}\sum_{\substack{g'_1 + \ldots g'_{r+t} = d_2-(r+2t-b)\\
 \card{\{j: g'_{r+j} = 0\}} = b}} \binom{ d_2-(r+2t-b)}{g'_1, \ldots, g'_{r+t}} \le \frac{(r+t)^{d_2-(r+2t-b) + b}}{b!}  \enspace .
 \end{align*}
 Substituting in ~\eqref{eq:dmf2:t2}, we have
 \begin{align*}
&\sum_{\substack{e_1 + \ldots +e_{r+s} = d_1 \\ e_1, \ldots, e_{r+s} \ge 1\\ \card{\{ j \in [s]: e_{r+j} = 1 \}} =a}}
 \sum_{\substack{g_1 + \ldots + g_{r+t} =d_2 \\ \text{ $g_j's \ge 1$} \\ \card{\{j:g_{r+j} =1\}}=b
 }} \binom{d_1}{e_1, \ldots, e_{r+s}} \binom{d_2}{g_1, \ldots, g_{r+t}}  1 \\
 & \le \left(\frac{1}{a!b!}\right)d_1^{r+2s-a} d_2^{r+2t-b} (r+s)^{d_1-r-2s+2a} (r+t)^{d_2-r-2t+2b}
 \end{align*}
 Therefore,
 \begin{align}
 &\sum_{a=0}^{s} \sum_{b=0}^t \frac{\gamma^{a+b} }{(4F_2/B)^{(p/2)(a+b)}}\notag \\
&\hspace*{0.3in}\sum_{\substack{e_1 + \ldots +e_{r+s} = d \notag\\ e_1, \ldots, e_{r+s} \ge 1\\ \card{\{ j \in [s]: e_{r+j} = 1 \}} =a}}
 \sum_{\substack{g_1 + \ldots + g_{r+t} =d \notag\\ \text{ $g_j's \ge 1$} \\ \card{\{j:g_{r+j} =1\}}=b
 }} \binom{d}{e_1, \ldots, e_{r+s}} \binom{d}{g_1, \ldots, g_{r+t}}  1  \notag \\
 & \le d_1^{r+2s} d_2^{r+2t} (r+s)^{d_1-r-2s} (r+t)^{d_2-r-2t} \sum_{a=0}^{s} \sum_{b=0}^t   \left(\frac{\gamma}{(4F_2/B)^{p/2}}\right)^{a+b} \notag \\
 & \hspace*{1.5in}\left( \frac{ (r+s)^{2a} (r+t)^{2b}}{d_1^a d_2^b}  \right) \notag \\
 & = d_1^{r+2s} d_2^{r+2t} (r+s)^{d_1-r-2s} (r+t)^{d_2-r-2t} \left( \sum_{a=0}^s \left(\frac{\gamma(r+s)^{2}}{ d_1 (4F_2/B)^{(p/2)}}\right)^a \right) \notag \\
  & \hspace*{1.5in} \left( \sum_{b=0} ^t \left(\frac{ \gamma (r+s)^2}{d_2 (4F_2/B)^{p/2}} \right)^b\right) \notag  \\
 & \le  d_1^{r+2s} d_2^{r+2t} (r+s)^{d_1-r-2s} (r+t)^{d_2-r-2t} \left( 1  + \gamma'\right) \left( 1  + \gamma'\right) \label{eq:dmf2:t3a} \\
 & \le 2d_1^{r+2s} d_2^{r+2t} (r+s)^{d_1-r-2s} (r+t)^{d_2-r-2t} \enspace . \notag \\
 & \le 2 d_2^{2r+2s+2t}
 (r+s+t)^{d_1 + d_2 -2r -2s-2t} \label{eq:dmf2:t3}
 \end{align}
 Eqn.~\eqref{eq:dmf2:t3a} is obtained from the previous step as follows. Let $$ \gamma' = 2\max\left( \frac{ \gamma (r+s)^2}{d_1 (4F_2/B)^{p/2}}, \frac{\gamma (r+s)^2}{d_2 (4 F_2/B)^{p/2}}\right) \enspace . $$ Then,  $\gamma' = O(\gamma (d_1+d_2)^2) = n^{-\Omega(1)}$ and
 $$ \sum_{a=0}^s \left(\frac{\gamma(r+s)^{2}}{ d_1 (4F_2/B)^{(p/2)}}\right)^a \le 1 + \gamma', \text{ and } \sum_{b=0}^s \left(\frac{\gamma(r+s)^{2}}{ d_2 (4F_2/B)^{(p/2)}}\right)^b \le 1 + \gamma' \enspace. $$

 The last step, that is, Eqn.~\eqref{eq:dmf2:t3} is obtained from its previous equation  as follows. (1) $\displaystyle (r+s)^{d_1-r-2s} (r+t)^{d_2-r-2s} \le (r+s+t)^{d_1 - r-2s} (r+s+t)^{d_2-r-2t} = (r+s+t)^{d_1+d_2 - 2r -2s-2t}$, and (2) by assumption, $d_1 \le d_2$.

 Substituting Eqn.~\eqref{eq:dmf2:t3} in Eqn.~\eqref{eq:dmomsmp1}, we have,
 \begin{align}
 &\expect{ \left.\left( \sum_{i \in S} (Y_i-\expect{Y_i\mid \H})\right)^{d_1} \left( \sum_{i \in S} (\conj{Y_i} - \conj{\expect{Y_i\mid \H}})\right)^{d_2}\right\vert \H} \notag \\
 & \le  2^{d_1+d_2+2} e^{(d_1+d_2)/8} \left( \frac{4F_2}{B} \right)^{(p/2)(d_1 + d_2)}\sum_{r=0}^{d_1} \sum_{s = \max(0, 1-r) }^{d_1-r } \sum_{t = \max(0,1-r)}^{d_2-r}\frac{ F_p^{r+s+t} }{(4F_2/B)^{(p/2)(r+s+t)} (r+s+t)!} \notag \\
 &\hspace*{0.5cm}\sum_{a=0}^{s} \sum_{b=0}^t  \frac{\gamma^{a+b} }{(4F_2/B)^{(p/2)(a+b)}}
\sum_{\substack{e_1 + \ldots +e_{r+s} = d_1 \\ e_1, \ldots, e_{r+s} \ge 1\\ \card{\{ j \in [s]: e_{r+j} = 1 \}} =a}}
 \sum_{\substack{g_1 + \ldots + g_{r+t} =d_2 \\ \text{ $g_j's \ge 1$} \\ \card{\{j:g_{r+j} =1\}}=b
 }} \binom{d_1}{e_1, \ldots, e_{r+s}} \binom{d_2}{g_1, \ldots, g_{r+t}}  1  \notag \\
 &\le 2^2  \left( \frac{8e^{1/(4p)}F_2}{B} \right)^{(p/2)(d_1 + d_2)}\sum_{r=0}^{d_1} \sum_{s = \max(0, 1-r) }^{d_1-r } \sum_{t = \max(0,1-r)}^{d_2-r}\frac{ F_p^{r+s+t} }{(4F_2/B)^{(p/2)(r+s+t)} (r+s+t)!} \notag\\
 & \hspace*{0.5in} d_2^{2r+2s+2t} (r+s+t)^{d_1 + d_2 -2r -2s-2t} \enspace . \label{eq:dmf2:t4}
 \end{align}
 Letting $u = r+s+t$ in Eqn.~\eqref{eq:dmf2:t4}, we obtain,
 \begin{align}\label{eq:dmfa1}
 & \le 2^2 \left( \frac{8e^{1/(4p)}F_2}{B} \right)^{(p/2)(d_1+d_2)} \sum_{u=1}^{\lfloor (d_1+d_2)/2\rfloor}  \left(\frac{ F_p^u}{(4F_2/B)^{(p/2)u} u!}\right) d_2^{2u}u^{d_1+d_2-2u}
 \end{align}
 since, $2r + 2s + 2t \le d_1+d_2$ and therefore, $u=r+s+t \le \lfloor (d_1+d_2)/2\rfloor$.

 Taking ratio of the $u+1$st term to the $u$th term in the summation in Eqn.~\eqref{eq:dmfa1}, we have,
 \begin{align*}
 \frac{ F_p}{(4F_2/B)^{(p/2)}} \cdot \frac{ d_2^2}{(u+1)^3}\cdot  (1 +1/u)^{d_1 + d_2 -2u}
 \ge \frac{ F_p}{(4F_2/B)^{(p/2)}} \cdot \frac{ d_2^2}{(u+1)^3}
 \end{align*}
 Since, $F_2 \le n^{1-2/p} F_p^{2/p}$ and $B \ge Kn^{1-2/p} d_2^{2/p}$, it follows that
 $$ \frac{ F_p}{(4F_2/B)^{(p/2)}} \ge (K/4)^{p/2} d_2 \enspace . $$
 Therefore,  the ratio of $u+1$st term to the $u$th term is at least
 \begin{align*}
 \frac{ F_p}{(4F_2/B)^{(p/2)}} \cdot \frac{ d_2^2}{(u+1)^3} \ge  \frac{(K/4)^{p/2} d_2^3}{(u+1)^3} \ge 2
 \end{align*}
 for $K \ge 8$ since, $d_2 \ge \lfloor d_1 + d_2 \rfloor/2 \ge u+1$.

 Therefore, the  series $\displaystyle \sum_{u=1}^{\lfloor (d_1+d_2)/2\rfloor}\left(  \frac{ F_p^u}{(4F_2/B)^u u!}\right) d_2^{2u}u^{d_1+d_2-2u}$  is bounded above  by twice the last  term, that is,
 \begin{align} \label{eq:dmf2:t5a}
 &\sum_{u=1}^{\lfloor (d_1+d_2)/2\rfloor}\left(  \frac{ F_p^u}{(4F_2/B)^u u!}\right) d_2^{2u}u^{d_1+d_2-2u} \notag\\
 &\le (2) \left(\frac{ F_p^{\lfloor (d_1+d_2)/2\rfloor}}{ (4F_2/B)^{\lfloor (d_1+d_2)/2\rfloor} (\lfloor (d_1+d_2)/2\rfloor!)} \right) d_2^{2\lfloor (d_1+d_2)/2\rfloor}\notag \\
 &\hspace*{1.5in}\left( \lfloor (d_1+d_2)/2\rfloor \right)^{ d_1 + d_2 -2 \lfloor (d_1+d_2)/2\rfloor}
 \end{align}

  Substituting in Eqn.~\eqref{eq:dmfa1}, we obtain
 \begin{align} \label{eq:dmf2:t5}
 &\expect{ \left.\left( \sum_{i \in S} (Y_i-\expect{Y_i\mid \H})\right)^{d_1} \left( \sum_{i \in S} (\conj{Y_i} - \conj{\expect{Y_i\mid \H}})\right)^{d_2}\right\vert \H} \notag\\
 &\le 2^3 e^{(d_1+d_2)/8} \left( \frac{8F_2}{B} \right)^{(p/2)(d_1+d_2)}
 \left(\frac{ F_p^{\lfloor (d_1+d_2)/2\rfloor}}{ (4F_2/B)^{\lfloor (d_1+d_2)/2\rfloor} (\lfloor (d_1+d_2)/2\rfloor!)} \right) d_2^{2\lfloor (d_1+d_2)/2\rfloor}\notag \\
 & \hspace*{1.0in}\left( \lfloor (d_1+d_2)/2\rfloor \right)^{ d_1 + d_2 -2 \lfloor (d_1+d_2)/2\rfloor}
 \notag \\
 & \le 2^3 e^{(d_1+d_2)/8}  \left( \frac{8F_2}{B} \right)^{(p/2)\left((d_1+d_2)- \lfloor (d_1 + d_2)/2 \rfloor \right)} F_p^{\lfloor (d_1+d_2)/2 \rfloor} \notag\\
   &\hspace*{0.5in} \frac{d_2^{2\lfloor (d_1 + d_2)/2 \rfloor} }{ \lfloor (d_1 + d_2)/2 \rfloor !} \left(\lfloor (d_1 + d_2)/2 \rfloor \right)^{d_1+d_2 - 2\lfloor (d_1 + d_2)/2 \rfloor}
 \end{align}
 By Stirling's approximation, that is,  $\displaystyle n! > (2\pi n)^{1/2}\left(\frac{n}{e}\right)^n$, we have \\  $\displaystyle \lfloor (d_1 + d_2)/2 \rfloor ! > (2 \pi \lfloor (d_1 + d_2)/2 \rfloor))^{1/2} \left(\frac{ \lfloor (d_1 + d_2)/2 \rfloor}{e} \right)^{\lfloor (d_1 + d_2)/2 \rfloor}$.

 Since, $d_1 + d_2 > 1$, and we have assumed that $d_2 \ge d_1$,
 $d_2 \le 3\lfloor (d_1 + d_2)/2 \rfloor$. Therefore,
 \begin{align}
 \frac{d_2^{2\lfloor (d_1 + d_2)/2 \rfloor} }{ \lfloor (d_1 + d_2)/2 \rfloor !}
 &\le (3 \sqrt{e})^{d_1+d_2} \left(\frac{ \left(\lfloor (d_1 + d_2)/2 \rfloor\right)^{2\lfloor (d_1 + d_2)/2 \rfloor}}{\sqrt{2\pi}\left(\lfloor (d_1 + d_2)/2 \rfloor\right)^{\lfloor (d_1 + d_2)/2 \rfloor+1/2}}\right) \notag \\
 & = (2\pi)^{-1/2}(3 \sqrt{e})^{d_1+d_2} \eat{F_p^{\lfloor (d_1+d_2)/2 \rfloor}} \left(\lfloor (d_1 + d_2)/2 \rfloor\right)^{\lfloor (d_1 + d_2)/2 \rfloor -1/2} \label{eq:dmf2:t6}
 \end{align}
 The term $\left(\lfloor (d_1 + d_2)/2 \rfloor \right)^{d_1+d_2 - 2\lfloor (d_1 + d_2)/2 \rfloor} $ in Eqn.~\eqref{eq:dmf2:t5} is 1 if $d_1+d_2$ is odd and is  $\lfloor (d_1 + d_2)/2 \rfloor$ if $d_1 + d_2$ is even. We write this as the indicator variable $\displaystyle  \lfloor (d_1+d_2)/2 \rfloor^{ \1_{d_1 + d_2 \text{ odd } }}$.

 Using this and substituting Eqn.~\eqref{eq:dmf2:t6} into Eqn.~\eqref{eq:dmf2:t5}, we have,
 \begin{align}
 &\expect{ \left.\left( \sum_{i \in S} (Y_i-\expect{Y_i\mid \H})\right)^{d_1} \left( \sum_{i \in S} (\conj{Y_i} - \conj{\expect{Y_i\mid \H}})\right)^{d_2}\right\vert \H} \notag\\
 &\le 2^{3} e^{(d_1+d_2)/8} \left( \frac{8F_2}{B} \right)^{(p/2)\left((d_1+d_2)- \lfloor (d_1 + d_2)/2 \rfloor \right)} F_p^{\lfloor (d_1+d_2)/2 \rfloor} \notag \\
   &\hspace*{1.0in} \frac{d_2^{2\lfloor (d_1 + d_2)/2 \rfloor} }{ \lfloor (d_1 + d_2)/2 \rfloor !} \left(\lfloor (d_1 + d_2)/2 \rfloor \right)^{d_1+d_2 - 2\lfloor (d_1 + d_2)/2 \rfloor} \notag\\
 & \le 2^{3} e^{(d_1+d_2)/8} (2\pi)^{-1/2}(3 \sqrt{e})^{d_1+d_2}\left( \frac{8F_2}{B} \right)^{(p/2)\lceil (d_1+d_2)/2 \rceil} F_p^{\lfloor (d_1+d_2)/2 \rfloor}\notag \\
 & \hspace*{1.0in}\left(\lfloor (d_1 + d_2)/2 \rfloor\right)^{\lfloor (d_1 + d_2)/2 \rfloor -1/2 + \1_{d_1+d_2 \text{ odd}} } \notag \\
 & \le 2^3\left( 3  \cdot e \right)^{d_1 + d_2} \left( \frac{8F_2}{B} \right)^{(p/2)\lceil (d_1+d_2)/2 \rceil} F_p^{\lfloor (d_1+d_2)/2 \rfloor}\left(\lfloor (d_1 + d_2)/2 \rfloor\right)^{(d_1+d_2)/2} \label{eq:dmf2:t7}
 \end{align}

 Now, $B \ge L n^{1-2/p} \epsilon^{-4/p} \log^{2/p}(1/\delta)$, for an appropriate constant $L \ge (2^3)^{2/p} \cdot \left( (8 \cdot 20) (3e) \right)^{4/p}$. Using $F_2 \le n^{1-2/p} F_p^{2/p}$, we have,
\begin{align*}
\left( 2^{6/p} (3  e)^2 \cdot \frac{8F_2}{B} \right)^{p/2} \le  \frac{\epsilon^2 F_p}{400\lceil \log (1/\delta)\rceil}
\end{align*}
Therefore, assuming $d_1 + d_2 \ge 1$,  Eqn.~\eqref{eq:dmf2:t7} is bounded above by
\begin{align*}
& 2^3 \left( 3 e \right)^{d_1 + d_2} \left( \frac{8F_2}{B} \right)^{(p/2)\lceil (d_1+d_2)/2 \rceil} F_p^{\lfloor (d_1+d_2)/2 \rfloor}\left(\lfloor (d_1 + d_2)/2 \rfloor\right)^{(d_1+d_2)/2} \notag\\
& \le \left(\lfloor (d_1 + d_2)/2 \rfloor\right)^{(d_1+d_2)/2}  \left( \frac{\epsilon^2 F_p}{400\lceil \log (1/\delta)\rceil} \right)^{\lceil (d_1+d_2)/2 \rceil}F_p^{\lfloor (d_1+d_2)/2 \rfloor} \\
& \le \left(\frac{\epsilon F_p}{20} \right)^{d_1+ d_2} \enspace .
\end{align*}
Here, the last step uses that, (i) $\lfloor (d_1 + d_2)/2 \rfloor \le \lceil\log (1/\delta)\rceil$, (ii) $\epsilon^{2\lceil (d_1 + d_2)/2 \rceil} \le \epsilon^{d_1 + d_2}$, and \\ (iii) $F_p^{\lfloor (d_1 + d_2)/2 \rfloor + \lceil (d_1 + d_2)/2 \rceil} = F_p^{d_1+d_2}$.

\end{proof}

\subsection{ Analysis of contribution to $d$th moment from items in $\midreg(G_0)$}

Recall that  $\mathcal{H}$ denotes the event $\G \wedge \nocollision \wedge \goodest$.
\begin{lemma} [Re-stated (expanded) version of Lemma~\ref{lem:midcentral:conj}.]
Let $1 \le e,g \le \lceil\log (1/\delta)\rceil$ and $l \in \midreg(G_0)$. Assume $C \ge 72p^2 B$ and $\abs{x_l} \ge \left( \frac{\ftwores{C}}{B} \right)^{1/2}$.  Let the family $\{\omega_{lr}(i)\}_{i \in [n]}$ 
be $O(k)$-wise independent where, $k \ge O(\log (1/\delta) + \log n)$. Then,\\ $\displaystyle \expect{ \left( \left(1 + \frac{Z_l}{\abs{x_l}} \right)^p-1\right)^e \left( \left(1 + \frac{\conj{Z_l}}{\abs{x_l}} \right)^p-1\right)^g \mid \mathcal{H}}$ is real and
\begin{align*}
 0 \le \expect{ \left( \left(1 + \frac{Z_l}{\abs{x_l}} \right)^p-1\right)^e \left( \left(1 + \frac{\conj{Z_l}}{\abs{x_l}} \right)^p-1\right)^g \mid \mathcal{H}} \le \left( \frac{7p^2(e+g) \ftwores{C}}{\rho C \abs{x_l}^2} \right)^{(e+g)/2} \enspace .
\end{align*}
It follows that,
\begin{align*}
\expect{ \left( Y_l - \expect{Y_l \mid \H}\right)^e \left(\conj{Y_l} - \expect{ \conj{Y_l}\mid \H}\right)^g \mid \mathcal{H}}
\le \left( \frac{8p^2\abs{x_l}^{2p-2} \ftwores{C}}{\rho C} \right)^{(e+g)/2} \enspace .
\end{align*}
\end{lemma}

\begin{proof}
Let $c_r = \binom{p}{r}$, for $r\ge 0$.  Since, $ \frac{ \abs{Z_l}}{\abs{x_l}} \le 1/8$, the binomial series expansion of  $\left( 1 + \frac{Z_l}{\abs{x_l}} \right)^p $ holds as the power series $ \sum_{r \ge 0} c_r \left( \frac{Z_l}{\abs{x_l}}\right)^r$. For the first part of the proof, we will assume full independence of the family $\{\omega_{lr}(i)\}_{i \in [n]}$, for $l=0$ and $r \in [2s]$ so that we can use the above binomial expansion.
\begin{align}
&\expect{ \left( \left(1 + \frac{Z_l}{\abs{x_l}} \right)^p-1\right)^e \left( \left(1 + \frac{\conj{Z_l}}{\abs{x_l}} \right)^p-1\right)^g\mid \H } \notag \\
&= \expect{ \left(\sum_{r \ge 1} c_r \left(\frac{Z_l}{\abs{x_l}}\right)^r\right)^e
\left( \sum_{s \ge 1} c_s \left( \frac{\conj{Z_l}}{\abs{x_l}} \right)^s \right)^g\mid \H } \notag \\
& = \expect{ \left(\frac{c_1 Z_l}{\abs{x_l}}\right)^e \left( \frac{c_1 \conj{Z_l}}{\abs{x_l}} \right)^g \left(\sum_{r \ge 1} \frac{c_r}{c_1} \left( \frac{ Z_l}{\abs{x_l}}\right)^{r-1} \right)^e
\left(\sum_{s \ge 1} \frac{c_s}{c_1} \left( \frac{ \conj{Z_l}}{\abs{x_l}}\right)^{s-1} \right)^g \mid \H}\notag\\
& =  \sum_{a_1 + \ldots + a_k+ \ldots = e} \sum_{b_1 + \ldots + b_k + \ldots  = g}  c_1^{e+g} \prod_{r \ge 1} \left(\frac{c_r}{c_1}\right)^{a_r} \prod_{s \ge 1} \left(\frac{c_s}{c_1} \right)^{b_s}  \notag \\
&\hspace*{1.0in} \expect{\cfrac{Z_l^{e+\sum_{r \ge 1} (r-1)a_r} \conj{Z_l}^{g+\sum_{s \ge 1} (s-1)b_s}}{\abs{x_l}^{ \sum_{r \ge 1} (r-1)a_r + \sum_{s \ge 1} (s-1)b_r}} \mid \H} \label{eq:midreg:t1}
\end{align}

Consider the  term $\displaystyle \expect{\frac{Z_l^{e+\sum_{r \ge 1} (r-1)a_r} \conj{Z_l}^{g+\sum_{s \ge 1} (s-1)b_s}}{\abs{x_l}^{ \sum_{r \ge 1} (r-1)a_r + \sum_{s \ge 1} (s-1)b_r}} \mid \H}$. Subject to the constraint that $a_1 + \ldots + a_k+ \ldots = e$ and $b_1 + \ldots + b_k + \ldots  = g$, we have, $e+\sum_{r \ge 1} (r-1)a_r = \sum_{r \ge 1} ra_r$ and $g + \sum_{s \ge 1} (s-1) b_s = \sum_{s \ge 1} s b_s$. Hence,
\begin{align*}
\expect{Z_l^{e+\sum_{r \ge 1} (r-1)a_r} \conj{Z_l}^{g+\sum_{s \ge 1} (s-1)b_s} \mid \H}
= \expect{ Z_l^{\sum_{r \ge 1} ra_r} \conj{Z_l}^{\sum_{s \ge 1} s b_s} \mid \H} \end{align*}
This expectation is 0 if $ \sum_{r \ge 1} r a_r  \ne \sum_{s \ge 1} s b_s$.  Therefore,
\begin{align}\label{eq:midreg:t1a}
&\expect{ \left( \left(1 + \frac{Z_l}{\abs{x_l}} \right)^p-1\right)^e \left( \left(1 + \frac{\conj{Z_l}}{\abs{x_l}} \right)^p-1\right)^g \mid \H} \notag\\
&=\sum_{a_1 + \ldots + a_k+ \ldots = e} \sum_{\substack{b_1 + \ldots + b_k + \ldots  = g\\
\sum_{r \ge 1} r a_r = \sum_{s \ge 1} s b_s}}  c_1^{e+g} \prod_{r \ge 1} \left(\frac{c_r}{c_1}\right)^{a_r} \prod_{s \ge 1} \left(\frac{c_s}{c_1} \right)^{b_s}
 \frac{ \expect{(Z \conj{Z})^{\sum_{r \ge 1} r a_r} \mid \H}}{ \abs{x_l}^{2\sum_{r \ge 1} r a_r}}
\end{align}

From Eqn.~\eqref{eq:midreg:t1a}, it follows that the \emph{RHS} is a real number, and therefore the expectation in the \emph{LHS} is real.

For non-zero expectation, assuming $ \sum_{r \ge 1} r a_r  = \sum_{s \ge 1} s b_s$, we have,
\begin{align*}
0 \le \frac{\expect{ (Z_l\conj{Z_l})^{\sum_{r \ge 1} r a_r} \mid \H}}{\abs{x_l}^{2\sum_{r \ge 1} r a_r}}  &\le \frac{ \expect{ (Z_l\conj{Z_l})^{(e+g)/2} \mid\H}}{\abs{x_l}^{e+g}} \left( \frac{ \abs{Z_l}}{\abs{x_l}} \right)^{ 2(\sum_{r \ge 1} r a_r) - (e+g)}\\
&\le \abs{x_l}^{-(e+g)}\left( \frac{(e+g) \ftwores{C}}{2(\rho/9) C } \right)^{(e+g)/2} \varrho^{2(\sum_{r \ge 1} r a_r) - (e+g)}
\end{align*}
where, $\varrho = (B/C)^{1/2}$. Note that conditional on $\goodest$, $\abs{Z_l}/\abs{x_l} \le \rho$.

Substituting, and noting that $c_1 = \binom{p}{1} = p$,  the sum in Eqn.~\eqref{eq:midreg:t1} is bounded above as follows.
\begin{align}
&\sum_{a_1 + \ldots + a_k+ \ldots = e} \sum_{b_1 + \ldots + b_k + \ldots  = g}  c_1^{e+g} \prod_{r \ge 1} \left(\frac{c_r}{c_1}\right)^{a_r} \prod_{s \ge 1} \left(\frac{c_s}{c_1} \right)^{b_s}   \notag \\ & \hspace*{1.0in}\expect{\cfrac{Z_l^{e+\sum_{r \ge 1} (r-1)a_r} \conj{Z_l}^{g+\sum_{s \ge 1} (s-1)b_s}}{\abs{x_l}^{ \sum_{r \ge 1} (r-1)a_r + \sum_{s \ge 1} (s-1)b_r}} \mid \H} \notag \\
& \le \abs{x_l}^{-(e+g)} \left( \frac{c_1^2(e+g) \ftwores{C}}{2(\rho/9) C} \right)^{(e+g)/2}  \notag\\ &\hspace*{0.5cm}\sum_{a_1 + \ldots + a_k+ \ldots = e} \sum_{\substack{b_1 + \ldots + b_k + \ldots  = g \\
\sum_{r \ge 1} r a_r = \sum_{s \ge 1} s b_s}}  \prod_{r \ge 1} \left( \frac{\abs{c_r}}{c_1}\right)^{a_r} \prod_{s \ge 1} \left( \frac{\abs{c_s}}{c_1} \right)^{b_s}\varrho^{2(\sum_{r \ge 1} r a_r) - (e+g)} \notag \\
& \le \abs{x_l}^{-(e+g)}\left( \frac{c_1^2(e+g) \ftwores{C}}{2(\rho/9) C} \right)^{(e+g)/2}\notag\\ &\hspace*{0.5cm} \sum_{a_1 + \ldots + a_k+ \ldots = e} \sum_{\substack{b_1 + \ldots + b_k + \ldots  = g }}  \prod_{r \ge 1} \left(\frac{\abs{c_r}}{c_1}\right)^{a_r} \prod_{s \ge 1} \left( \frac{\abs{c_s}}{c_1} \right)^{b_s}\varrho^{\left(\sum_{r \ge 1} r a_r  +\sum_{s \ge 1} s b_s \right)- (e+g)} \label{eq:midrreg:t2}
\end{align}
Using, $\sum_{r \ge 1} a_r = e$ and $\sum_{s \ge 1} b_s = g$, we have,  $\left(\sum_{r \ge 1} r a_r  +\sum_{s \ge 1} s b_s \right)- (e+g) = \sum_{r \ge 1} (r-1)a_r + \sum_{s \ge 1} (s-1) b_s$.  Therefore, Eqn.~\eqref{eq:midrreg:t2} equals
\begin{align}
&
\abs{x_l}^{-(e+g)}\left( \frac{c_1^2(e+g) \ftwores{C}}{2(\rho/9) C} \right)^{(e+g)/2}
\left( \sum_{r \ge 1} \frac{\abs{c_r}}{c_1}\varrho^{r-1}\right)^e \left( \sum_{s \ge 1} \frac{\abs{c_s}}{c_1}\varrho^{s-1}\right)^g \label{eq:midrreg:t3}
\end{align}
The ratio of the $(r+1)$th term to the $r$th term in the summation $\sum_{r \ge 1} \frac{\abs{c_r}}{c_1}\varrho^{r-1}$ , for $r=1,2, \ldots, k-1$ is
\begin{align*}
\left\lvert \frac{ c_{r+1}}{c_r} \right\rvert \varrho= \frac{ \abs{p-r} \varrho}{r+1} \le \frac{1}{8}
\end{align*}
since, $\varrho = (B/C )^{1/2}\le 1/(8p)$. Therefore,
$\sum_{r \ge 1} \frac{\abs{c_r}}{c_1} \varrho^{r-1} \le \sum_{r \ge 1} (1/8)^{r-1} \le (9/8)$.

\noindent The summation in Eqn.~\eqref{eq:midrreg:t3} is therefore bounded above by  (since, $c_1 = p$)
\begin{align} \label{eq:midrreg:final1}
&\expect{ \left( \left(1 + \frac{Z_l}{\abs{x_l}} \right)^p-1\right)^e \left( \left(1 + \frac{\conj{Z_l}}{\abs{x_l}} \right)^p-1\right)^g \mid \H}\notag \\
& \le \abs{x_l}^{-(e+g)}\left( \frac{((9/8)p)^2(e+g) \ftwores{C}}{2(\rho/9) C} \right)^{(e+g)/2} \notag \\
&\le \left ( \frac{ 6p^2 (e+g) \ftwores{C}}{\abs{x_l}^2 \rho C} \right)^{(e+g)/2} \enspace .
\end{align}

\noindent
Proceeding  similarly, we  show that for any $k \ge 1$,  $\displaystyle \expect{ \left( \sum_{r=1}^k c_r  \frac{Z_l^r}{\abs{x_l}^r} \right)^e \left( \sum_{s=1}^k c_s \frac{ \conj{Z_l}^s}{\abs{x_l}^r} \right)^g \mid \H}$ is real and
\begin{align} \label{eq:midrreg:final2}
& 0 \le \expect{ \left( \sum_{r=1}^k c_r  \frac{Z_l^r}{\abs{x_l}^r} \right)^e \left( \sum_{s=1}^k c_s \frac{ \conj{Z_l}^s}{\abs{x_l}^r} \right)^g} \le \left ( \frac{ 6p^2 (e+g) \ftwores{C}}{\abs{x_l}^2 \rho C} \right)^{(e+g)/2} \enspace .
\end{align}

\emph{ Using $k$-wise independence.} Let $k = \Omega(\log n)$ be a parameter to be determined and assume that the family $\{\omega_{lr}(i)\}_{i \in [n]}$ is at least $k$-wise independent. With probability $1-n^{-\Omega(1)}$, $\abs{Z_l} \le \left( \frac{\ftwores{C}}{C} \right)^{1/2}$. Therefore,
$\displaystyle \frac{ \abs{Z_l}}{x_l} \le \left( \frac{ \ftwores{C}}{\abs{x_l}^2 C} \right)^{1/2}$.
By Taylor's series expansion up to $k$ terms, we have, $
\displaystyle \left(1 + \frac{Z_l}{\abs{x_l}} \right)^p = \sum_{r=0}^{k-1} c_r \frac{Z_l^r}{\abs{x_l}^r} + \gamma_k$, where, $\displaystyle \gamma_k = c_k \frac{Z_l^{'k}}{\abs{x_l}^k}$ and $\abs{Z'_l} \le \abs{Z_l}$. By the above discussion, we have, $\displaystyle   \abs{\gamma_k} = \frac{\abs{c_kZ_l^{'k}}}{\abs{x_l}^k}  \le  \left\lvert \binom{p}{k}\right\rvert \left( \frac{ \ftwores{C}}{\abs{x_l}^2 C} \right)^{k/2} = \zeta$ (say).

In the remainder of the proof, \emph{all expectations are conditional on} $\H$.

Let $\alpha_k $ denote the sum of the first $k-1$ terms in the Taylor series expansion of $\left(1+ \frac{Z_l}{\abs{x_l}} \right)^p$ except for the zeroth term, that is,
\begin{align*}
\alpha_k = \sum_{r=1}^{k-1} c_r \frac{ Z_l^r}{\abs{x_l}^r} \enspace .
\end{align*}
Hence, $\displaystyle \left(1+ \frac{Z_l}{\abs{x_l}}\right)^p-1 = \alpha_k + \gamma_k$. Let $\beta_k = \conj{\alpha_k}$  and so that $\displaystyle \left(1 + \frac{\conj{Z_l}}{\abs{x_l}}\right)^p-1 = \beta_k + \conj{\gamma_k}$.
Therefore,
\begin{align*}
&\expect{ \left( \left( 1 + \frac{Z_l}{\abs{x_l}}\right)^p - 1 \right)^e \left( \left(1 + \frac{ \conj{Z_l}}{\abs{x_l}} \right)^p -1 \right)^g} = \expect{ \left(\alpha_k + \gamma_k\right)^e \left(\beta_k + \conj{\gamma_k} \right)^g} \enspace .
\end{align*}
Let $d_r = \binom{e}{r}$, for $r=0,1, \ldots, e$ and $h_s = \binom{g}{s} $, for $s=0,1,\ldots, g$. Therefore,
\begin{align} \label{eq:midrreg:fi2:t1}
& \expect{ \left(\alpha_k + \gamma_k\right)^e \left(\beta_k + \conj{\gamma_k} \right)^g} \notag \\
& = \expect{ \left( \sum_{r=0}^e d_r \alpha_k^{e-r} \gamma_k^r\right)\left( \sum_{s=0}^g h_s \beta_k^{g-s} \conj{\gamma_k}^s \right)} \notag \\
& = \expect{\alpha_k^e \beta_k^g} + \sum_{\substack{ r=0 \ldots e, s=0 \ldots g\\ r+s \ge 1}}
d_r h_s \expect{ \alpha_k^{e-r} \beta_k^{g-s} \gamma_k^r\conj{\gamma_k}^{s}}  \enspace .
\end{align}

By Eqn.~\eqref{eq:midreg:t1a}, the \emph{LHS} of Eqn.~\eqref{eq:midrreg:fi2:t1} is non-negative and real.

By Eqn.~\eqref{eq:midrreg:final2}, $\displaystyle \expect{ \alpha_k^e \beta_k^g}$ is non-negative, real and $\displaystyle \expect{ \alpha_k^e \beta_k^g} \le \left ( \frac{ 6p^2 (e+g) \ftwores{C}}{\abs{x_l}^2 \rho C} \right)^{(e+g)/2}$.

We now consider the term $ \expect{ \alpha_k^{e-r} \beta_k^{g-s} \gamma_k^r\conj{\gamma_k}^{s}}$.
Note that
\begin{align*}
\abs{\alpha_k}  &\le \sum_{r=1}^{k-1} \abs{c_r} \frac{\abs{Z_l}^r}{\abs{x_l}^r}   \le  \sum_{r=1}^{k-1} \card{ \binom{p}{r}} \frac{1}{(8p)^r} \enspace .
\end{align*}
The ratio of the $r+1$st term to the $r$th term in the above summation is, $\displaystyle \frac{\abs{p-r}}{r+1} \cdot \frac{1}{8p} \le 1/8$. Thus, $\abs{\alpha_k} \le \displaystyle \sum_{r=1}^{k-1} \card{ \binom{p}{r}} \frac{1}{(8p)^r}\le \left( \frac{p}{8p} \right)\sum_{r=1}^{k-1} (1/8)^{r-1} = (1/7)$. Hence, $\abs{\beta_k} = \abs{\conj{\alpha_k}} = \abs{\alpha_k} \le 1/7$.

Substituting in Eqn.~\eqref{eq:midrreg:fi2:t1} and taking absolute values, we obtain,
\begin{align} \label{eq:midrreg:fi2:t2}
0 \le &\expect{ \left(\alpha_k + \gamma_k\right)^e \left(\beta_k + \conj{\gamma_k} \right)^g}  \notag\\
& \le  \expect{\alpha_k^e \beta_k^g} + \sum_{\substack{ r=0 \ldots e, s=0 \ldots g\\ r+s \ge 1}}d_r h_s  \abs{\alpha_k}^{e-r} \abs{\beta_k}^{g-s} \zeta_k^{r+s}  \notag \\
& \le \left ( \frac{ 6p^2 (e+g) \ftwores{C}}{\abs{x_l}^2 \rho C} \right)^{(e+g)/2} +
\sum_{\substack{ r=0 \ldots e, s=0 \ldots g\\ r+s \ge 1}} d_r h_s (1/7)^{e+g-r-s} \zeta_k^{r+s} \enspace .
\end{align}

We now consider the summation term in Eqn.~\eqref{eq:midrreg:fi2:t2}.
Then,
\begin{align} \label{eq:midrreg:fi2:t3}
&\sum_{\substack{ r=0 \ldots e, s=0 \ldots g\\ r+s \ge 1}} d_r h_s (1/7)^{e+g-r-s} \zeta_k^{r+s} \notag \\
& = (1/7)^{e+g}\sum_{s=1}^g h_s  (7\zeta_k)^s + (1/7)^{e+g}\sum_{r=1}^e d_r  (7\zeta_k)^r \sum_{s=0}^g h_s (7\zeta_k)^s \enspace .
\end{align}

Consider the summation  $\sum_{s=1}^g h_s (7\zeta_k)^s$. The ratio of the $s+1$th term to the $s$th term, for $s=1,2, \ldots, g-1$, is $\displaystyle \left(\frac{h_{s+1}}{h_s} \right)  (7\zeta_k) \le 7g \zeta_k/2  = n^{-\Omega(1)}$.  Thus,
$$\sum_{s=1}^g h_s  (7\zeta_k)^s \le    h_1 (7\zeta_k) \sum_{s=1}^g (7g\zeta_k/2)^{s-1} = 7g  \zeta_k (1 + O(g\zeta_k)) \le 8g \zeta_k \enspace . $$

\noindent
We now consider the second (double)-summation in Eqn.~\eqref{eq:midrreg:fi2:t3}, namely, $\displaystyle \sum_{r=1}^e d_r (7\zeta_k)^r \sum_{s=0}^g h_s (7\zeta_k)^s$. Proceeding as in the previous paragraph, this is at most $8e \zeta_k$. The second summation,  $\sum_{s=0}^g h_s (7\zeta_k)^s = (1+7\zeta_k)^g \le \exp{7\zeta_k g} \le 1 + 8g \zeta_k$, since, $g = O(\log (1/\zeta)) $ and $\zeta_k = n^{-\Omega(1)}$. Therefore, $\displaystyle \sum_{r=1}^e d_r (7\zeta_k)^r \sum_{s=0}^g h_s (7\zeta_k)^s \le 8e \zeta_k (1 + 8g \zeta_k)$.

Substituting in Eqn.~\eqref{eq:midrreg:fi2:t3}, we have,
\begin{align} \label{eq:midrreg:fi2:t4}
&(1/7)^{e+g}\sum_{s=1}^g h_s  (7\zeta_k)^s + (1/7)^{e+g}\sum_{r=1}^e d_r  (7\zeta_k)^r \sum_{s=0}^g h_s (7\zeta_k)^s \notag \\
& \le  (1/7)^{e+g}(9)\zeta_k (e+g)\notag \\
& \le  (9)(e+g) (1/7)^{e+g} \left( \frac{ \ftwores{C}}{\abs{x_l}^2 C} \right)^{k/2} \enspace .
\end{align}

Substituting in Eqn.~\eqref{eq:midrreg:fi2:t2}, we have,
\begin{align} \label{eq:midrreg:fi2:t5}
&\expect{ \left(\alpha_k + \gamma_k\right)^e \left(\beta_k + \conj{\gamma_k} \right)^g} \notag \\
&\le  \left ( \frac{ 6p^2 (e+g) \ftwores{C}}{\abs{x_l}^2 \rho C} \right)^{(e+g)/2} + (9)(e+g) (1/7)^{e+g} \left( \frac{ \ftwores{C}}{\abs{x_l}^2 C} \right)^{k/2} \notag \\
& = \left ( \frac{ 6p^2 (e+g) \ftwores{C}}{\abs{x_l}^2 \rho C} \right)^{(e+g)/2}\notag \\
& \hspace*{0.5in} \cdot  \left(1 + 9(e+g) \left( \frac{\rho}{6 (7p)^2(e+g)} \right)^{(e+g)/2}\left( \frac{ \ftwores{C}}{\abs{x_l}^2 C} \right)^{(k-(e+g))/2}\right) \enspace .
\end{align}
Let $d = 2\lceil \log (1/\delta) \rceil$, so that  $e +g \le d$.  Then,
\begin{align} \label{eq:midrreg:fi2:t5a}
& 9 (e+g) \left( \frac{\rho}{6 (7p)^2(e+g)} \right)^{(e+g)/2}  = \exp{ \ln (9(e+g)) + \left( \frac{e+g}{2} \right) \ln \left( \frac{\rho}{6 (7p)^2 (e+g)} \right)} \enspace .
\end{align}
The function $x \ln \frac{\rho}{ax}$ attains a maximum at $x = \frac{\rho}{e a}$ and the maximum value is $\frac{\rho}{\exp{1} a}$. Thus, the \emph{RHS} in Eqn.~\eqref{eq:midrreg:fi2:t5a} is bounded above  by $\displaystyle \exp{\ln (9d) + \frac{\rho}{\left(\exp{1}\right)6 (7p)^2}}$.
Further,
\begin{align*}
\left( \frac{ \ftwores{C}}{\abs{x_l}^2 C} \right)^{(k-(e+g))/2} \le \left(\frac{1}{8p} \right)^{ (k-(e+g))/2} &= \exp{-(\ln (8p)) (k- (e+g)/2)} \\ &\le  \exp{-2-\left( \ln (9d) + \frac{\rho}{\left(\exp{1}\right)6 (7p)^2}\right)},
\end{align*}
provided, $\displaystyle k \ge  d + \left(\frac{1}{\ln (8p)}\right) \left(2 +\left( \ln (9d) + \frac{\rho}{\left(\exp{1}\right)6 (7p)^2}\right)\right)  = O(\log (1/\delta) + \log n)$, since, $\rho = O(\log n)$.

Under this condition,
\begin{align*}
9(e+g) \left( \frac{\rho}{6 (7p)^2(e+g)} \right)^{(e+g)/2}\left( \frac{ \ftwores{C}}{\abs{x_l}^2 C} \right)^{(k-(e+g))/2} \le e^{-2} \enspace .
\end{align*}
Substituting in Eqn.~\eqref{eq:midrreg:fi2:t5}, we obtain
\begin{align} \label{eq:midrreg:fi2:t6}
&\expect{ \left(\alpha_k + \gamma_k\right)^e \left(\beta_k + \conj{\gamma_k} \right)^g} \notag \\
& \le \left ( \frac{ 6p^2 (e+g) \ftwores{C}}{\abs{x_l}^2 \rho C} \right)^{(e+g)/2} \left(1 + e^{-2} \right) \notag \\
& \le \left ( \frac{ 7p^2 (e+g) \ftwores{C}}{\abs{x_l}^2 \rho C} \right)^{(e+g)/2}
\end{align}

This proves the first statement of the lemma.

For the second statement of the lemma, $Y_l =  \abs{x_l}^p\left(1+ \frac{Z_l}{\abs{x_l}}\right)^p$.
Using $k$-wise independence of the $\omega_{lr}$'s family of random roots of unity, and since,  $\abs{Z_i}/\abs{x_i} \le \varrho = (B/C)^{1/2} \le 1/(8p)$, we have,
\begin{align*}
\expect{Y_l} &=  \abs{x_l}^p \left( \sum_{r=0}^{k-1} c_r \frac{\expect{Z_l^r}}{\abs{x_l}^r} + c_k \frac{ \expect{Z_l^{'k}}}{\abs{x_l}^k} \right) =\abs{x_l}^p\left(1 \pm  \gamma \right)
\end{align*}
since, $\expect{Z_l^r} = 0$, for $r \in [k-1]$ and where, $\abs{\gamma}  \le \abs{c_k} (8p)^{-k} = = \zeta$ (say), which is $n^{-\Omega(1)}$ since $k = \Omega(\log n)$. That is, we have shown that $ \displaystyle  \left\lvert \expect{Y_l} - \abs{x_l}^p \right\rvert \le \gamma \abs{x_l}^p$.

By a similar argument, $\expect{\conj{Y_l}} = \abs{x_l}^p(1 \pm \gamma)$.

Since, $ \displaystyle Y_l = (\abs{x_l} + Z_l)^p = \abs{x_l}^p \left( 1 + \frac{ Z_l}{\abs{x_i}} \right)^p$, we have,
\begin{align*}
&\expect{ \left( Y_l - \expect{Y_l}\right)^e \left( \conj{Y_l} - \expect{\conj{Y_l}} \right)^g}\\
& = \expect{ \left( \abs{x_l}^p \left( 1 + \frac{ Z_l}{\abs{x_l}} \right)^p -\abs{x_l}^p (1 + \gamma )\right)^e \left(\abs{x_l}^p \left( 1 + \frac{ \conj{Z_l}}{\abs{x_l}} \right)^p -\abs{x_l}^p (1 \pm \gamma)\right)^g} \\
& = \abs{x_l}^{p(e+g)}(1+\zeta)^{e+g} \notag \\
& \hspace*{0.5in}\expect{ \left( (1+\gamma)^{-1}\left( 1 + \frac{ Z_l}{\abs{x_l}} \right)^p -1 \right)^e
\left( (1+\conj{\gamma})^{-1}\left( 1 + \frac{ \conj{Z_l}}{\abs{x_l}} \right)^p -(1\pm \gamma )\right)^g} \\
& \le \abs{x_l}^{p(e+g)}(1+ \gamma)^{e+g}
\expect{ \left( \left( 1 + \frac{ Z_l}{\abs{x_l}} \right)^p -1\right)^e
\left( \left( 1 + \frac{ \conj{Z_l}}{\abs{x_l}} \right)^p -1 \right)^g} \\
& \le \abs{x_l}^{p(e+g)}(1+\gamma)^{e+g} \left( \frac{7p^2(e+g) \ftwores{C}}{\abs{x_l}^2 \rho C} \right)^{(e+g)/2} \\
& = \left( \frac{\abs{x_l}^{2p-2} a  (e+g) \ftwores{C}}{\rho C} \right)^{(e+g)/2}
\end{align*}
where, $a = 7p^2(1+ \gamma)^{e+g}\le  (7p^2) e^{(e+g)\gamma}  \le 9p^2$, since, $e+g \le O(\log (1/\delta)$ and $\gamma = \exp{-(\ln (8p)) k} \le O(1/(\log (1/\delta)))$, if $k  = O(\log\log (1/\delta))$.
The second to last step follows from the first statement of the lemma.

\end{proof}

\begin{lemma} \label{lem:midcentral:noconj}
Let $1 \le e \le \lceil\log (1/\delta)\rceil$ and $l \in \midreg(G_0)$. Suppose the random roots of unity family $ \{\omega_{lr}\}_{l,r}$ is $O(\log (n) \log (1/\delta))$-wise independent. Then,
$$ \expect{ \left(\left( 1 + \frac{Z_l}{\abs{x_l}} \right)^p - 1\right)^e \mid \H} \le n^{-\Omega(1)} \enspace . $$
Therefore,
$$\expect{ \left( Y_l - \expect{Y_l} \right)^e \mid \H} \le  \abs{x_l}^{pe} n^{-\Omega(1)} $$
and $$\expect{ \left( \conj{Y_l} - \expect{\conj{Y_l}} \right)^e \mid \H} \le  \abs{x_l}^{pe} n^{-\Omega(1)} \enspace .$$
\end{lemma}

\begin{proof} All expectations in this proof are conditional on $\H$.
Let $k $ be a parameter.
Let $c_r = \binom{p}{r}$, for $r=0,1, \ldots, k$. Then,
\begin{align} \label{eq:midreg:noconj:t1}
& \expect{\left( \left( 1 + \frac{Z_l}{\abs{x_l}} \right)^p -1 \right)^e \mid \H } = \expect{ \left( \left( \sum_{r=1}^{k-1} c_r \left( \frac{Z_l}{\abs{x_l}}\right)^r\right) + c_k  \left(\frac{Z'_i}{\abs{x_i}} \right)^k\right)^e \mid \H}
\end{align}
where, $\abs{Z'_i} \le \abs{Z_i}$.

Denote $ \sum_{r=1}^{k-1} c_r \left( \frac{Z_l}{\abs{x_l}}\right)^r$ by $\alpha$ and $c_k \frac{Z_i^{'k}}{\abs{x_i}^k}$ by $\beta$.  Let $d_r = \binom{e}{r}$, for $r=0,1, \ldots, e$. We have, $\abs{\alpha} \le  \sum_{r=1}^{k-1} \abs{c_r} (8p)^{-1}  \le (1/8) \sum_{r=1}^{k-1} (1/8)^{r-1} = 1/7$. Also, $\abs{\beta} \le (8p)^{-k}$.
Then, Eqn.~\eqref{eq:midreg:noconj:t1} can be written as
\begin{align} \label{eq:midreg:noconj:t1a}
\expect{\left( \left( 1 + \frac{Z_l}{\abs{x_l}} \right)^p -1 \right)^e \mid \H }& =\expect{ \left( \alpha + \beta\right)^e \mid \H} = \expect{\alpha^e \mid \H} + \sum_{r=1}^e d_r \expect{\alpha^{e-r} \beta^r \mid \H}  \enspace .
\end{align}
Now,
\begin{align*}
\expect{\alpha^e \mid \H} = \expect{\left(\sum_{r=1}^{k-1} c_r \left( \frac{Z_l}{\abs{x_l}}\right)^r\right)^e \mid \H} = \sum_{h_1 + \ldots + h_{k-1} = e}\prod_{r=1}^{k-1} c_r^{h_r} \frac{ \expect{Z_l^{ \sum_{r=1}^{k-1} r\cdot h_r}\mid \H}}{\abs{x_l}^{ \sum_{r=1}^{k-1} r\cdot h_r}} = 0 \enspace .
\end{align*}

Consider the sum $\displaystyle \sum_{r=1}^e d_r \abs{\alpha}^{e-r} \abs{\beta}^r $. The ratio of $r+1$th term to the $r$th term is $ \displaystyle \left( \frac{ e-r}{r+1}\right) \cdot 7\beta \le (7e\abs{\beta}/2)$. Assuming $k = O(\log n)$, $\abs{\beta} \le (8p)^{-k} = n^{-\Omega(1)}$. Therefore, $\displaystyle \sum_{r=1}^e d_r \abs{\alpha}^{e-r} \abs{\beta}^r \le e(1/7)^{e-1}\abs{\beta} (1 + O(en^{-\Omega(1)})) \le n^{-\Omega(1)}$.

Substituting these into Eqn.~\eqref{eq:midreg:noconj:t1a}, we have,
\begin{align*}
\expect{\left( \left( 1 + \frac{Z_l}{\abs{x_l}} \right)^p -1 \right)^e \mid \H }  = \expect{ \left( \alpha + \beta\right)^e \mid \H}  &=  0 + \sum_{r=1}^e d_r \expect{ \alpha^{e-r} \beta^r \mid \H} \enspace .
\end{align*}
Taking absolute values, we have,
\begin{align*}
\left\lvert \expect{\left( \left( 1 + \frac{Z_l}{\abs{x_l}} \right)^p -1 \right)^e \mid \H } \right\rvert
 \le \sum_{r=1}^e d_r \abs{\alpha}^{e-r} \abs{\beta}^r =  n^{-\Omega(1)}   \enspace .
\end{align*}


Now, $\displaystyle Y_l =  \abs{x_l}^p\left(1+ \frac{Z_l}{\abs{x_l}}\right)^p$ and therefore, $\displaystyle \left\lvert\frac{\expect{Y_l \mid \H}}{\abs{x_l}^p}-1 \right\rvert  \le  n^{-\Omega(1)}$, using $k= O(\log n)$-wise independence. Equivalently, $\displaystyle \card{\expect{Y_l \mid \H} -\abs{x_l}^p} \le  \abs{x_l}^pn^{-\Omega(1)}$. So, $\expect{Y_l \mid \H} = \abs{x_l}^p (1+ \gamma)$, where, $\abs{\gamma} \le n^{-\Omega(1)}$.  Therefore,
\begin{align*}\expect{ (Y_l - \expect{Y_l})^e \mid \H} &
= \abs{x_l}^{pe} \expect{ \left( \left( 1+ \frac{Z_l}{\abs{x_l}}\right)^p - (1+\gamma)\right)^e} \enspace .
\end{align*}
Taking absolute values,
\begin{align*}
\left \lvert  \expect{ (Y_l - \expect{Y_l})^e}  \right\rvert
& = \abs{x_l}^{pe} \left\lvert \expect{ \left( \left( 1+ \frac{Z_l}{\abs{x_l}}\right)^p - (1+\gamma)\right)^e \mid \H}\right\rvert \\
& \le \abs{x_l}^{pe}(1 + \abs{\gamma})^e  \left\lvert\expect{ \left( \left( 1+ \frac{Z_l}{\abs{x_l}}\right)^p -1 \right)^e \mid \H}\right\rvert \\
& \le \abs{x_l}^{pe}(1 + n^{-\Omega(1)})^e  n^{-\Omega(1)} \\
& \le \abs{x_l}^{pe} n^{-\Omega(1)} \enspace .
 \end{align*}
since, $e \le \lceil \log (1/\delta) \rceil \le O(\log n)$.

\end{proof}

\subsection{A combinatorial lemma}
In the calculation of the $d$th central moment for the  contribution from the items in $\midreg(G_0)$, we will need to estimate an upper bound on the following combinatorial sums defined in Eqns.~\eqref{eq:comb:basic} and ~\eqref{eq:comb:Rdef} respectively.
\begin{align*}
Q(S_1, S_2)& = \sum_{q=1}^{\min(S_1, S_2)} \sum_{\substack{e_1 + \ldots + e_q = S_1\\ e_j's \ge 1}}
\sum_{\substack{g_1 + \ldots + g_q = S_2 \\ g_j's \ge 1}} \binom{S_1}{e_1, \ldots, e_q} \binom{S_2}{g_1, \ldots, g_q} \notag \\
&\hspace*{0.2cm}\sum_{\{i_1, \ldots, i_q\}} \prod_{r=1}^q \abs{x_{i_r}}^{(p-1)(e_r+g_r)} \prod_{r=1}^q (e_r+g_r)^{(e_r+g_r)/2} \enspace . \\
R(S) & = \sum_{q=1}^{\lfloor  S/2 \rfloor} \sum_{h_1 + \ldots + h_q = S, h_j's \ge 2} \binom{S}{h_1, \ldots, h_q} \sum_{\{i_1, \ldots, i_q\}}\prod_{r\in [q]} \abs{x_{i_r}}^{(p-1)h_r} \prod_{r \in [q]} h_r^{h_r/2} \enspace . 
\end{align*}
Define the sum
\begin{align} \label{eq:comb:Pdef}
P(S) & = \sum_{q=1}^S \sum_{i_1, \ldots, i_q} \sum_{\substack{e_1 + \ldots + e_q = 2S\\ e_1, \ldots e_q \ge 2}} \binom{S}{e_1/2, \ldots, e_q/2} \prod_{r \in [q]} \abs{x_{i_r}}^{(2p-2)(g_r/2)} \enspace . 
\end{align}

The expression in Eqn.~\eqref{eq:comb:basic} is upper bounded by
\begin{lemma} \label{lem:comb:1}[Re-statement of first part of Lemma~\ref{lem:comb:final}.]
$Q(S_1, S_2) \le R(S_1+S_2)$.
\end{lemma}
\begin{proof}
Consider the multinomial  expansion of  $(a_1 + \ldots +a_q)^{S_1}$, for $a_1, \ldots, a_q \ge 0$ together with the constraint that in  each monomial of the form $\prod_{i\in [q]}{a_i}^{e_i}$, each $e_i \ge 1$. This equals $ \sum_{e_1 + \ldots + e_q = S_1, e_j's \ge 1} \binom{S_1}{e_1, \ldots, e_q} \prod_{i\in [q]} a_i^{e_i}$. The multinomial expansion of $(a_1 + \ldots + a_q)^{S_2}$ with the same constraint can be written as $\sum_{g_1 + \ldots + g_q = S_2, g_j's \ge 1}$  $ \binom{S_2}{g_1, \ldots, g_q} \prod_{i\in [q]} a_i^{g_i}$.  Now consider the multinomial expansion of $(a_1 + \ldots + a_q)^{S_1+S_2}$ subject to
the constraint that in  each monomial of the form $\prod_{i\in [q]}{a_i}^{h_i}$, each $h_i \ge 2$. This equals  $\sum_{h_1 + \ldots + h_q = S_1+S_2, h_j's \ge 2}$  $ \binom{S_1+S_2}{h_1, \ldots, h_q} \prod_{i\in [q]} a_i^{h_i}$. It therefore follows that
\begin{gather*}
\sum_{e_1 + \ldots + e_q = S_1, e_j's \ge 1} \binom{S_1}{e_1, \ldots, e_q} \prod_{i\in [q]} a_i^{e_i} \sum_{g_1 + \ldots + g_q = S_2, g_j's \ge 1} \binom{S_2}{g_1, \ldots, g_q} \prod_{i\in [q]} a_i^{g_i} \\ \hspace*{1.0in}\le \sum_{h_1 + \ldots + h_q = S_1+S_2, h_j's \ge 2} \binom{S_1+S_2}{h_1, \ldots, h_q} \prod_{i\in [q]} a_i^{h_i}
\end{gather*}
Here the variables  $h_i$ takes values (among other possibilities) $e_i+g_i$, for $i \in [q]$,  for each $q$-partition vector $e$ of $S_1$ and $q$-partition $g$ of $S_2$. Therefore,
\begin{gather}
\sum_{e_1 + \ldots + e_q = S_1, e_j's \ge 1} \binom{S_1}{e_1, \ldots, e_q} \prod_{i\in [q]} a_i^{e_i} \sum_{g_1 + \ldots + g_q = S_2, g_j's \ge 1} \binom{S_2}{g_1, \ldots, g_q} \prod_{i\in [q]} a_i^{g_i}  \notag \\ \hspace*{1.0in} \prod_{i \in [q]} (e_i + g_i)^{(e_i + g_i)/2}  \notag
\\ \le \sum_{h_1 + \ldots + h_q = S_1+S_2, h_j's \ge 2} \binom{S_1+S_2}{h_1, \ldots, h_q} \prod_{i\in [q]} a_i^{h_i} \prod_{i \in [q]} h_i^{h_i/2} \label{eq:comb:t2} \enspace .
\end{gather}
It therefore follows that by letting $a_i = \abs{x_i}^{p-1}$, for $i \in [n]$, that
\begin{align} \label{eq:comb:t3}
&Q(S_1, S_2) \\ &=
\sum_{q=1}^{\min(S_1, S_2)} \sum_{\{i_1, \ldots, i_q\} \subset[n]} \sum_{\substack{e_1 + \ldots + e_q = S_1\\ e_j's \ge 1}}
\sum_{\substack{g_1 + \ldots + g_q = S_2 \\ g_j's \ge 1}} \binom{S_1}{e_1, \ldots, e_q} \binom{S_2}{g_1, \ldots, g_q} \notag \\
& \hspace*{1.0in}\prod_{r=1}^q \abs{x_{i_r}}^{(p-1)(e_r+g_r)} \prod_{r=1}^q (e_r+ g_r)^{(e_r+g_r)/2} \notag \\
&\le \sum_{q=1}^{\lfloor ( S_1 + S_2)/2 \rfloor} \sum_{\{i_1, \ldots, i_q\}}\sum_{h_1 + \ldots + h_q = S_1+S_2, h_j's \ge 2} \binom{S_1+S_2}{h_1, \ldots, h_q} \prod_{r\in [q]} \abs{x_{i_r}}^{(p-1)h_r} \prod_{r \in [q]} h_r^{h_r/2} \notag \\
& = R(S_1+ S_2)
\end{align}
\end{proof}

Suppose we generalize the factorial notation $x!$ to mean $ \Gamma(x+1)$, when $x$ is a fraction  of the form $(n+1/2)$, for $n \ge 1$.  Using this, generalize the multinomial coefficient notation $\binom{S}{e_1, \ldots, e_q} $ to denote $\frac{S!}{e_1! \ldots e_q!}$, even when the $e_j$'s are fractional and of the form $(n + 1/2)$, for $n \ge 1$.

\begin{lemma}\label{lem:comb:RR}
$R^2(S) \le R(2S)$.
\end{lemma}

\begin{proof}

\begin{align*}
&R^2(S)\\
& = \sum_{q_1=1}^{\lfloor S/2 \rfloor} \sum_{q_2 = 1}^{\lfloor S/2 \rfloor} \sum_{\{i_1, \ldots, i_{q_1}\}} \sum_{j_1, \ldots, j_{q_2}} \sum_{\substack{e_1 + \ldots + e_{q_1} = S \\ e_j's \ge 2}}  \sum_{\substack{g_1 + \ldots + g_{q_2} = S\\ g_j's \ge 2}} \binom{S}{e_1, \ldots, e_{q_1}} \binom{S}{g_1, \ldots, g_{q_2}} \\
& \hspace*{1.0in}\prod_{r\in [q_1]} \abs{x_{i_r}}^{(p-1)e_r}  \prod_{s \in [q_2]} \abs{x_{j_s}}^{(p-1)g_s}\prod_{r \in [q_1]} e_r^{e_r/2} \prod_{r \in [q_2]} g_r^{g_r/2} \enspace .
\end{align*}

Define
\begin{align*}
a &= \card{\{i_1, \ldots i_{q_1}\} \cap \{j_1, \ldots, j_{q_2}\}} \\
b & = \{ \card{ \{i_1, \ldots, i_{q_1}\} \setminus \{j_1, \ldots, j_{q_2}\}}\\
c & = \{\card{ \{j_1, \ldots, j_{q_2}\} \setminus \{i_1, \ldots, i_{q_1}\} }
\end{align*}
Then, $R^2(S)$ can be written as
\begin{align}
& R^2(S) \notag \\
& = \sum_{a=0}^{S}\sum_{b=\max(0,1-a)}^{S-a} \sum_{c = \max(0,1-a)}^{S-a}
\sum_{\substack{e_1 + \ldots + e_{a+b} = S \\ e_j's \ge 2}} \sum_{\substack{ g_1 + \ldots + g_{a+c}= S \\ g_j's \ge 2}} \binom{S}{e_1, \ldots, e_{a+b}} \binom{S}{g_1, \ldots, g_{a+c}} \notag\\
& \sum_{\{i_1, \ldots, i_a, j_1, \ldots, j_b, k_1, \ldots, k_c\}} \prod_{l=1}^a \abs{x_{i_l}}^{(p-1)(e_l + g_l)} \prod_{m=1}^b \abs{x_{j_m}}^{(p-1)e_{l+m}} \prod_{n =1}^c \abs{x_{k_n}}^{(p-1)e_{l+n}} \notag \\
 & \hspace*{2.0in}\prod_{l=1}^a \left( e_l^{e_l/2} g_l^{g_l/2} \right) \prod_{m=1}^b e_{l+m}^{e_{l+m}/2} \prod_{n=1}^x g_{l+n}^{g_{l+n}/2} \enspace . \label{eq:comb:lem2:t1}
\end{align}
By a similar argument as in Lemma~\ref{lem:comb:1}, and using the fact that $\left( e_l^{e_l/2} g_l^{g_l/2} \right) \le (e_l+g_l)^{(e_l+g_l)/2}$, for any $e_l, g_l \ge 1$, we  have that the sum in Eqn.~\eqref{eq:comb:lem2:t1}, is bounded above by the following sum (using $q = a+b+c$).
\begin{align} \label{eq:comb:lem2:t2}
& \le  \sum_{q=1}^{S} \sum_{\substack{e_1 + \ldots + e_q = 2S \\ e_j's \ge 2}}
\binom{ 2S}{e_1, \ldots, e_q} \sum_{\{i_1, \ldots, i_q\}} \prod_{l=1}^q \abs{x_{i_l}}^{(p-1)e_j} \prod_{l=1}^q e_l^{e_l/2}\\
& = R(2S) \enspace . \notag
\end{align}

\end{proof}
\begin{lemma} \label{lem:comb:R2S}
\begin{align*}
R(2S) & \le (2eS)^S P(S) \\
&= (2eS)^S \sum_{q=1}^S \sum_{\substack{e_1 + \ldots + e_q = 2S\\ e_j's \ge 2} } \binom{S}{e_1/2, \ldots, e_q/2} \sum_{\{i_1, \ldots i_q\}} \prod_{r\in [q]} \abs{x_{i_r}}^{(2p-2)(e_r/2)} \enspace .
\end{align*}
\end{lemma}

\begin{proof}
Consider Eqn.~\eqref{eq:comb:lem2:t2}.

Fix $q$ and $\{i_1, \ldots, i_q\}$. Then, using Stirling's formula, we have,
\begin{align*}
&\binom{ 2S}{e_1, \ldots, e_q}  \prod_{l=1}^q e_l^{e_l/2} \\
& = \frac{ (2S)!}{e_1! \ldots e_q!} \prod_{l=1}^q e_l^{e_l/2} \\
& \le (2S)^S \left(\frac{ S! e^{e_1 + \ldots + e_q}}{ \left(\prod_{r\in [q]} \sqrt{2\pi e_r}\right) e_1^{e_1} \cdots e_q^{e_q}}\right) \prod_{l=1}^q e_l^{e_l/2} \\
& \le  (2eS)^{S} \left( \frac{ S!}{(e_1/2)! \ldots (e_q/2)!} \right) \\
& = (2eS)^S \binom{S}{e_1/2, \ldots, e_q/2}
\end{align*}
using the generalized factorial notation.
The statement of the lemma now follows.
\end{proof}
We would now like to establish a bound on $R(2S)$ in terms of $F_{2p-2}^S$. The two expressions are shown below.
\begin{align} \label{eq:comb:lem:R2S}
P(S) & \le  \sum_{q=1}^S  \sum_{\{i_1, \ldots i_q\}}\sum_{\substack{g_1 + \ldots + g_q = 2S \\ g_j's \ge 2} } \binom{S}{g_1/2, \ldots, g_q/2}  \prod_{r\in [q]} \abs{x_{i_r}}^{(2p-2)(g_r/2)} \enspace .\\
F_{2p-2}^S & = \sum_{q=1}^S \sum_{\{i_1, \ldots i_q\}} \sum_{\substack{e_1 + \ldots + e_q = S \\ e_j's \ge 1}} \binom{S}{e_1, \ldots, e_q}  \prod_{r\in [q]} \abs{x_{i_r}}^{(2p-2)e_r}  \enspace . \label{eq:comb:lem:F2p2}
\end{align}

\begin{lemma} \label{lem:comb:PSF2p2}
$P(S) \le 4^S F_{2p-2}^S$.
\end{lemma}

\begin{corollary} \label{lem:comb:R2SF2p2}
$R(2S) \le (8eS)^S F_{2p-2}^S$ \enspace .
\end{corollary}

\begin{proof} [Proof of Corollary~\ref{lem:comb:R2SF2p2}.]
We have from Lemma~\ref{lem:comb:R2S} that $R(2S) \le (2eS)^S P(S)$. By Lemma~\ref{lem:comb:PSF2p2}, we have $P(S) \le  4^S F_{2p-2}^S$. Combining, we obtain the corollary.
\end{proof}
\begin{proof} [Of Lemma ~\ref{lem:comb:PSF2p2}.]

Fix $q$ and fix a choice of the index set $\{i_1, \ldots, i_q\} \subset [n]$.

For $q \in [S]$,  define the sets $U_q = \{(g_1, \ldots, g_q): g_1 + \ldots + g_q = 2S, g_j's \ge 2\}$ and $V_q = \{(e_1, \ldots, e_q) : e_1 + \ldots + e_q = S, e_j's \ge 1\}$. The sets $U_q$ and $V_q$ may be viewed as two partitions of bi-partite graph $G_q = (U_q, V_q, E_q)$, where the edge-set $E_q$ is defined as follows.
Let $(g_1, \ldots, g_q) \in U_q$. Let $h$ be the number of indices $j$ such that $g_j$ is even and let these indices be  $k_1<  \ldots < k_h $ in  sequence. Thus,  $g_j $ is odd for any $j \not\in \{k_1, \ldots, k_h\}$.
Define,
\begin{align} \label{eq:comb:ekj}
e_{k_j} = g_{k_j}/2, ~~~j =1,2, \ldots, h \enspace .
\end{align}

Note that $q-h$ is even. This is because $0 = 2S \mod 2 = g_1 + \ldots + g_q \mod 2 = \left((g_{k_1} + \ldots + g_{k_h}\right) \mod 2 + \sum_{j \not\in \{k_1, \ldots, k_h\}}( g_j \mod 2) = q-h \mod 2 $, since each of $g_{k_r} $ is even, for $r \in [h]$ and each of $g_j$ is odd, for $j \in [q] \setminus \{k_1, \ldots, k_h\}$.

Let $l_1 < l_2 < \ldots < l_{q-h}$ be the sequence of all indices in $[q]$ such that $g_{l_j}$ is odd, for $l \in [q-h]$. Note that $q-h$ is even.  Let $t$ be the permutation of $\{l_1, \ldots, l_{q-h}\}$ such that
$\abs{x_{t_1}} \le \abs{x_{t_2}} \le \ldots \le \abs{x_{t_{q-h}}}$.  Define,
\begin{align} \label{eq:comb:etj}
e_{t_j} & = \begin{cases} \lfloor g_{t_j}/2 \rfloor &  ~~~j =1,2, \ldots, (q-h)/2 \\
\lceil g_{t_j}/2 \rceil & ~~~j = (q-h)/2 + 1, \ldots, q-h \enspace .
\end{cases}
\end{align}
Denote this mapping as $\phi(g_1, \ldots, g_q)$, where, $\phi: U_q \rightarrow V_q$.

We would now like to count $\abs{\phi^{-1}(e_1, \ldots, e_q)}$. Fix $(e_1, \ldots, e_q)$ and fix any  $H\subset [q]$.  We will let $H$ be the set of indices such that $g_i$ is even and $g_i = 2e_i$. Let $\abs{H} = h$. Let $l_1 < l_2 < \ldots < l_{q-h}$ be the sequence of indices in $[q] \setminus H$. Let $t$ be the permutation of $\{l_1, \ldots, l_{q-h}\}$ such that
$\abs{x_{t_1}} \le \abs{x_{t_2}} \le \ldots \le \abs{x_{t_{q-h}}}$. Define,
\begin{align} \label{eq:comb:gtj}
g_{t_j}& = \begin{cases}
2e_{t_j}-1 & \text{ for } j=(q-h)/2 + 1, \ldots, q-h \\
2e_{t_j}+1& \text{ for } j=1,2, \ldots, (q-h)/2 \enspace .
\end{cases}
\end{align}
For a fixed $H$, this function $\psi_H(e_1, \ldots, e_q)$ is a mapping from $V_q$ to $U_q$. Further,
if for a given $(g_1, \ldots, g_q)$, $H = \{ i: g_i \text{ is even }\}$, then,
$$ \psi_H(\phi(g_1,\ldots, g_q)) = (g_1, \ldots, g_q) \enspace . $$
Hence,
\begin{align} \label{eq:comb:phiinv}\abs{\phi^{-1}(e_1, \ldots, e_q)} \le \card{ \{H:H\subset [q]\}} = 2^q \enspace .
\end{align}

Another important fact  is that by construction, for any $(g_1, \ldots, g_q) \in U_q$,
if we let $(e_1, \ldots, e_q) = \phi(g_1, \ldots, g_q)$, then,
$$ \prod_{r=1}^q \abs{x_{i_r}}^{(p-1)g_r} \le \prod_{r=1}^q \abs{x_{i_r}}^{(2p-2)e_r} \enspace . $$
This can be seen as follows. Let $H = \{i: g_i \text{ is even }\}$ and let $h = \abs{H}$. Let $\bar{H} = [q] \setminus H$ and $\card{ \bar{H}} = q-h$. Then,
$\prod_{r \in [q]} \abs{x_r}^{(p-1) g_r} = \prod_{r \in H} \abs{x_r}^{(2p-2)e_r} \prod_{r \in \bar{H}} \abs{x_r}^{(p-1)g_r)}$. It suffices to show that $\prod_{r \in \bar{H}} \abs{x_r}^{g_r} \le \prod_{r \in \bar{H}} \abs{x_r}^{2e_r}$. This is equivalent to show that $ \sum_{r \in \bar{H}} g_r \ln \abs{x_r} \le \sum_{r \in \bar{H}} 2e_r \ln \abs{x_r}$. Let $\alpha_r = \ln \abs{x_r} \ge 0$. Let $t$ be a permutation of the indices of the elements  of $\bar{H}$ such that $\abs{x_{t_1}} \le \abs{x_{t_2}} \le \ldots \le \abs{x_{t_{q-h}}}$. So it suffices to show that $\sum_{r=1}^{q-h} g_{t_r} \alpha_{t_r} \le \sum_{r=1}^{q-h} 2e_{t_r} \alpha_{t_r}$. The permutation orders the indices so that $\alpha_{t_1} \le \alpha_{t_2} \le \ldots \le \alpha_{t_{q-h}}$. The sum $L(x) = \sum_{r=1}^{q-h} x_r \alpha_{t_r}$ is a linear function of $g_{t_r}$. Consider the  assignment for the $x_r$'s, where, $x_r = g_{t_r}$, where, $g_{t_r} \ge 2$ is odd for each $r$. For any $j \in 1, \ldots, (q-h)/2$, let $x'_r = g_{t_r}-1$  and $x'_{q-h-r} = g_{t_{q-h-r}} +1$ and let $x'_j$ equal $x_j$ for all other indices $j \in [q-h]$. Then, $L(x) \le L(x')$, since, $\alpha_{t_j}$'s are in non-descending order. Continuing this argument for each $r \in [(q-h)/2]$, we obtain that $L(x) \le L(x^*)$, where, $x^*_r = g_{t_r}-1 = 2 e_{t_r}$, for $r =1,2, \ldots, (q-h)/2$ and $x^*_r = g_{t_r}+1 = 2e_{t_r}$, for $r=1,2, \ldots, (q-h)/2$. This proves the assertion.

Let $(e_1, \ldots, e_q) = \phi(g_1, \ldots, g_q)$, for some $(g_1, \ldots, g_q) \in U_q$.
We are now interested in an upper bound for the following ratio
$\displaystyle
\cfrac{ \binom{S}{g_1/2, \ldots, g_q/2}}{\binom{S}{e_1, \ldots, e_q}}
$.
As before, let $H = \{i\in [q]: g_i \text{ is even }\}$ and let $E = \sum_{i \in H} (g_i/2)$. Let
the indices in $H$ be $k_1 < k_2< \ldots < k_h$, where, $h = \abs{H}$.  Let the remaining indices be $l_1 < l_2 < \ldots, l_{h'}$, where, $h' = (q-h)$. Then,
\begin{align*}
\binom{S}{g_1/2, \ldots, g_q/2} &= \frac{ S!}{(g_1/2)! \cdots (g_q/2)!} \\
&= \binom{S}{E} \binom{E}{(g_{k_1}/2), \ldots, (g_{k_h}/2)} \binom{S-E}{(g_{l_1}/2) \ldots, (g_{l_{h'}}/2)}\\
& = \binom{S}{E} \binom{E}{e_{k_1}, \ldots, e_{k_h}} \binom{S-E}{(g_{l_1}/2) \ldots, (g_{l_{h'}}/2)}
\end{align*}
since,
 $(e_1, \ldots, e_q) = \phi(g_1, \ldots, g_q)$. Similarly,
\begin{align*}
\binom{S}{e_1, \ldots, e_q} & = \binom{S}{E} \binom{E}{e_{k_1}, \ldots, e_{k_h}}
\binom{S-E}{e_{l_1} \ldots, e_{l_{h'}}}
\end{align*}
Taking ratios, let $R$ denote the ratio $\displaystyle \cfrac{\binom{S}{g_1/2, \ldots, g_q/2}}{\binom{S}{e_1, \ldots, e_q}}$.
\begin{align}\label{eq:comb:bincoeff:ratio1}
R & = \cfrac{\binom{S}{g_1/2, \ldots, g_q/2}}{\binom{S}{e_1, \ldots, e_q}}
= \cfrac{\binom{S-E}{(g_{l_1}/2) \ldots, (g_{l_{h'}}/2)}}{ \binom{S-E}{e_{l_1} \ldots, e_{l_{h'}}}}
 = \prod_{r=1}^{h'/2}\left( \cfrac{ e_{l_r}!}{g_{l_r}/2!}\right)\left( \cfrac{ e_{l_{h'-r+1}}!}{g_{l_{h'-r+1}}/2!}\right) \notag \\
 & = \prod_{r=1}^{h'/2}\cfrac{ e_{l_r}!}{ (e_{l_r}+1/2)!} \prod_{r=h'/2+1}^{h'} \cfrac{e_{l_r}!}{(e_{l_r}-1/2)!} \notag
  \le \prod_{r=h'/2+1}^{h'} \cfrac{e_{l_r}!}{(e_{l_r}-1/2)!} \\ &  \le \prod_{r=h'/2+1}^{h'} \cfrac{e_{l_r}!}{(e_{l_r}-1)!} \le \prod_{r=h'/2+1}^{h'} e_{l_r} \le \prod_{r=1}^{h'} e_{l_r} \enspace .
\end{align}
Further, $ \sum_{r=1}^{h'} e_{l_r} = S-E$ $= S'$ (say). Hence, $\prod_{l=1}^{h'} e_{l_r}$ is maximized when the $e_{l_r}$'s equals $S'/h$. The product is then at most $(S'/h')^{h'} = \exp{ h' \ln (S'/h')}$.
The function $g(x) = x \ln \frac{a}{x}$ attains its maximum in the range $1 \le x \le a$ at $x = a/e$. The corresponding maximum  value is $g^* = a/e$.
Therefore, $\displaystyle \exp{ h'\ln (S'/h')} \le \exp{S/e}$.
Substituting in Eqn.~\eqref{eq:comb:bincoeff:ratio1}, we have,
\begin{align}\label{eq:comb:bincoeff:ratio2}
R = \cfrac{\binom{S}{g_1/2, \ldots, g_q/2}}{\binom{S}{e_1, \ldots, e_q}} \le
e^{(S-E)/e}\enspace.
\end{align}
 Fix $q \in [S]$ and index set $\{i_1, \ldots, i_q\} \subset [n]$. We will now try to relate  the sub-summations from Eqn.~\eqref{eq:comb:lem:R2S} and Eqn.~\eqref{eq:comb:lem:F2p2} respectively.
\begin{align}
&\sum_{\substack{g_1 + \ldots + g_q = 2S \\ g_j's \ge 2} } \binom{S}{g_1/2, \ldots, g_q/2}  \prod_{r\in [q]} \abs{x_{i_r}}^{(2p-2)(g_r/2)}  \notag  \\
&= \sum_{\substack{e_1 + \ldots + e_q = S \\ e_j's \ge 1}} \sum_{\substack{g_1 + \ldots + g_q = 2S \\ g_j's \ge 2 \\ \phi(g_1, \ldots, g_q) = (e_1, \ldots, e_q)} } \binom{S}{g_1/2, \ldots, g_q/2}  \prod_{r\in [q]} \abs{x_{i_r}}^{(2p-2)(g_r/2)} \notag \\
& \le \sum_{\substack{e_1 + \ldots + e_q = S \\ e_j's \ge 1}} \sum_{\substack{g_1 + \ldots + g_q = 2S \\ g_j's \ge 2 \\ \phi(g_1, \ldots, g_q) = (e_1, \ldots, e_q)} }e^{S/e} \binom{S}{e_1, \ldots, e_q/2}
\prod_{r\in [q]} \abs{x_{i_r}}^{(2p-2)e_r} \notag \\
& = e^{S/e}\sum_{\substack{e_1 + \ldots + e_q = S \\ e_j's \ge 1}}\binom{S}{e_1, \ldots, e_q/2}
\prod_{r\in [q]} \abs{x_{i_r}}^{(2p-2)e_r} \left\lvert \phi^{-1}(e_1, \ldots, e_q)\right\rvert  \notag \\
& \le (2e^{1/e})^S \sum_{\substack{e_1 + \ldots + e_q = S \\ e_j's \ge 1}}\binom{S}{e_1, \ldots, e_q/2}
\prod_{r\in [q]} \abs{x_{i_r}}^{(2p-2)e_r} \label{eq:comb:lemR2Sb} \enspace .
\end{align}
Therefore, since, $q \le S$,
\begin{align*}
P(S) & \le  \sum_{q=1}^S  \sum_{\{i_1, \ldots i_q\}}\sum_{\substack{g_1 + \ldots + g_q = 2S \\ g_j's \ge 2} } \binom{S}{g_1/2, \ldots, g_q/2}  \prod_{r\in [q]} \abs{x_{i_r}}^{(2p-2)(g_r/2)} \\
& \le (2e^{1/e})^S\sum_{q=1}^S  \sum_{\{i_1, \ldots i_q\}}\sum_{\substack{e_1 + \ldots + e_q = S \\ e_j's \ge 1} } \binom{S}{e_1, \ldots, e_q/2}  \prod_{r\in [q]} \abs{x_{i_r}}^{(2p-2)e_r} \\
& \le 4^S F_{2p-2}^S \enspace .
\end{align*}
\end{proof}

\begin{corollary} \label{cor:comb:1}[Re-statement of second part of  Lemma~\ref{lem:comb:final}.]
$R(S) \le \left(4e^{1+1/e}S F_{2p-2} \right)^{S/2}$.
\end{corollary}
\begin{proof}
Follows from Lemma~\ref{lem:comb:R2SF2p2} and Lemma~\ref{lem:comb:RR}.
\end{proof}

\subsection{$d$th central moment calculations for $\midreg(G_0)$ resumed}

\begin{lemma} \label{lem:dcentmom:midG0:1} Let $a \ge 9p^2$ be a constant and let   $C \ge (400)ae(n^{1-2/p}/\log (n)) \epsilon^{-2}\log (1/\delta)$ . Then,
\begin{gather*}
 \sum_{q=1}^{\min(S_1, S_2)} \sum_{\substack{e_1 + \ldots + e_q = S_1\\ e_j \ge 1}} \sum_{\substack{g_1 + \ldots + g_q=S_2\\ g_j \ge 1}} \binom{S_1}{e_1, \ldots, e_q} \binom{S_2}{g_1, \ldots, g_q} \notag \\
 \hspace*{0.5in}\sum_{\{i_1, \ldots i_q\}} \prod_{r=1}^q \abs{x_{i_r}}^{ p(e_r + g_r) -2((e_r+g_r)/2)} \left( \frac{ aF_2}{\rho C} \right)^{\sum_{r=1}^q (e_r+g_r)/2} \prod_{r=1}^q (e_r+g_r)^{(e_r+g_r)/2} \\\le
 \left( \frac{\epsilon F_p}{20} \right)^{S_1+S_2} \enspace .
 \end{gather*}

\end{lemma}

\begin{proof}
\begin{align}\label{eq:dcentmom:midG0:t1}
& \sum_{q=1}^{\min(S_1, S_2)} \sum_{\substack{e_1 + \ldots + e_q = S_1\\ e_j \ge 1}} \sum_{\substack{g_1 + \ldots + g_q=S_2\\ g_j \ge 1}} \binom{S_1}{e_1, \ldots, e_q} \binom{S_2}{g_1, \ldots, g_q} \notag \\
& \hspace*{0.5in}\sum_{\{i_1, \ldots i_q\}} \prod_{r=1}^q \abs{x_{i_r}}^{ p(e_r + g_r) -2((e_r+g_r)/2)} \left( \frac{ aF_2}{\rho C} \right)^{\sum_{r=1}^q (e_r+g_r)/2} \prod_{r=1}^q (e_r+g_r)^{(e_r+g_r)/2} \notag \\
& = \left( \frac{ aF_2}{\rho C} \right)^{(S_1+S_2)/2} Q(S_1+S_2) \notag \\
& \le \left( \frac{ aF_2}{\rho C} \right)^{(S_1+S_2)/2}R(S_1+S_2),~~ \text{ by Lemma~\ref{lem:comb:1}} \notag \\
& \le \left( \frac{ aF_2}{\rho C} \right)^{(S_1+S_2)/2} \left( 16 e (S_1+S_2) F_{2p-2} \right)^{(S_1+S_2)/2}, ~~ \text{ by Corollary~\ref{lem:comb:R2SF2p2}}\notag \\
& = \left( \frac{ 16a e (S_1+S_2)F_2} {\rho C} \right)^{(S_1+S_2)/2} F_p^{(2-2/p)(S_1 + S_2)/2}
\end{align}
For $p \ge 2$, $2p-2 \ge p$ and so $F_{2p-2} \le F_p^{2-2/p}$. Since, $\rho \ge \log n$, $1 \le S_1, S_2 \le \lceil \log (1/\delta)\rceil$ and
$C =K(n^{1-2/p}/\log(n))\epsilon^{-2}\log (1/\delta))$, we have, for   $K=(400) (16 a e)$, that
\begin{align*}
\frac{ 16 a e(S_1+S_2)a F_2}{(\log n)K (n^{1-2/p}/\log (n)) \epsilon^{-2} \log(1/\delta)}
\le \frac{\epsilon^2 F_p^{2/p}}{400} \enspace .
\end{align*}
Substituting in Eqn.~\eqref{eq:dcentmom:midG0:t1}, we have,
\begin{align*}
\left( \frac{ 16a e (S_1+S_2)F_2} {\rho C} \right)^{(S_1+S_2)/2} F_p^{(2-2/p)(S_1 + S_2)/2}  & \le \left(\frac{\epsilon F_p}{20}\right)^{S_1+S_2} \enspace .
\end{align*}

\end{proof}

Let $\H$ denote the event $\G \wedge \nocollision \wedge \goodest$.

\begin{lemma}  [Re-statement of Lemma~\ref{lem:dcentmom:midG0:2}].
Let $a \ge 9p^2$ be a constant  and let  $C \ge K(16a e)(n^{1-2/p}/\log (n))$ $ \epsilon^{-2}\log (1/\delta)$, where, $K = 400 (16 a e)$. Then, for $0 \le d_1, d_2 \le \log(1/\delta)$ and integral, the following holds.
\begin{multline*}
\expect{ \left( \sum_{i \in \midreg(G_0)} (Y_i - \expect{Y_i \mid \H}) \right)^{d_1} \left( \sum_{i \in \midreg(G_0)} (\conj{Y_i} - \expect{ \conj{Y_i}\mid \H}) \right)^{d_2} \mid \H} \le \left( \frac{\epsilon F_p}{10} \right)^{d_1+d_2} \enspace .
\end{multline*}
\end{lemma}

\begin{proof}
Recall that  $\H$ denote the event $\G \wedge \nocollision(\midreg(G_0)) \wedge \goodest(\midreg(G_0))$.

Let $\alpha_0(q) = \max(0, 1-q)$. 

\begin{align} \label{eq:dcentmom:midG0:2:t1}
&\expect{ \left( \sum_{i \in \midreg(G_0)} (Y_i - \expect{Y_i \mid \G}) \right)^{d_1} \left( \sum_{i \in \midreg(G_0)} (\conj{Y_i} - \expect{ \conj{Y_i} \mid \H}) \right)^{d_2}} \notag  \\
& = \sum_{q=0}^{\min(d_1, d_2)} \sum_{s= \alpha_0(q)}^{d_1-q} \sum_{t = \alpha_0(q)}^{d_2-q} \sum_{e_1 + \ldots + e_{q+s} = d_1} \binom{ d_1}{e_1, \ldots, e_{q+s}} \sum_{g_1 + \ldots + g_{q+t}= d_2} \binom{d_2}{g_1, \ldots, g_{q+s}} \notag \\
& \sum_{\{i_1, \ldots, i_q, j_1, \ldots, j_s, k_1, \ldots, k_t\}}
\mathbf{E} \left[ \prod_{r=1}^q (Y_{i_r} - \expect{Y_{i_r} \mid \H})^{e_r} (\conj{Y_{i_r}} - \expect{\conj{Y_{i_r}}\mid \H})^{g_r}\right. \notag \\
&\hspace*{1.5in} \left.  \prod_{l=1}^s (Y_{j_l} - \expect{Y_{j_l} \mid \H})^{e_{q+s}}
 \prod_{m=1}^t ( Y_{k_m} - \expect{Y_{k_m} \mid \H})^{g_{q+t}}\right]
\end{align}
Note that for a given value of $q$, if $q=0$, then $s $ and $t$ are at least 1, and otherwise, for $q > 0$, $s$ and $t$ may be initialized from 0.

Consider the term in the expectation of Eqn.~\eqref{eq:dcentmom:midG0:2:t1}. By \nocollision~and its consequent property as explained earlier,
\begin{align} \label{eq:dcentmom:midG0:2:t2}
&\mathbf{E} \left[\prod_{r=1}^q (Y_{i_r} - \expect{Y_{i_r} \mid \H})^{e_r} (\conj{Y_{i_r}} - \expect{\conj{Y_{i_r}}\mid \H})^{g_r} \prod_{l=1}^s (Y_{j_l} - \expect{Y_{j_l} \mid \H})^{e_{q+s}}\right. \notag\\
&\left. \hspace*{1.0in} \prod_{m=1}^t ( Y_{k_m} - \expect{Y_{k_m} \mid \H})^{g_{q+t}} \right] \notag \\
& =\prod_{r=1}^q \mathbf{E}\left[(Y_{i_r} - \expect{Y_{i_r} \mid \H})^{e_r} (\conj{Y_{i_r}} - \expect{\conj{Y_{i_r}}\mid \H})^{g_r}
\prod_{l=1}^s \expect{ (Y_{j_l} - \expect{Y_{j_l} \mid \H})^{e_{q+s}}} \right. \notag \\
& \hspace*{0.75in}\left. \prod_{m=1}^t\expect{( Y_{k_m} - \expect{Y_{k_m} \mid \H})^{g_{q+t}}} \right]\notag \\
& \le \prod_{r=1}^q \left(\frac{ \abs{x_{i_r}}^{2p-2}a(e_r+g_r)F_2}{\rho C} \right)^{ (e_r+g_r)/2} \prod_{l=1}^s \abs{x_{j_l}}^{pe_{q+l}} n^{-\Omega(e_{q+l})} \prod_{m=1}^t \abs{x_{k_m}}^{pg_{q+m}}n^{-\Omega(g_{q+m})}
\end{align}
by Lemmas~\ref{lem:midcentral:conj} and ~\ref{lem:midcentral:noconj}.

Let $S_1 = e_1 + \ldots + e_q$ and $S_2 = g_1 + \ldots + g_q$. By the interpretation above, $S_2 =0$ iff $S_1 = 0$.
Rewrite Eqn.~\eqref{eq:dcentmom:midG0:2:t1} as follows. Since, $e_1 + \ldots + e_{q+s} = d_1$ and $g_1 + \ldots + g_{q+s} = d_2$,   we will use the identity that
$$ \binom{d_1}{e_1, \ldots, e_{q+s}} = \binom{d_1}{S_1} \binom{S_1}{e_1, \ldots, e_q} \binom{d_1-S_1}{ e_{q+1}, \ldots, e_{q+s}} \enspace . $$
Using Eqn.~\eqref{eq:dcentmom:midG0:2:t2}, Eqn.~\eqref{eq:dcentmom:midG0:2:t1} can be equivalently rewritten as
\begin{align}
&=\sum_{S_1=0}^{d_1} ~~\sum_{\substack{S_2=0 \\ S_1 =0 \Rightarrow S_2=0}}^{d_2} ~~\sum_{\substack{q=0\\ S_1=0 \Rightarrow q=0}}^{\min(S_1, S_2)} ~~\sum_{s = \alpha_0(q)}^{d_1 -S_1} \sum_{t=\alpha_0(q)}^{d_2-S_2}\notag \\
 &\sum_{e_1 + \ldots + e_q = S_1} \sum_{e_{q+1}, \ldots, e_{q+s} = d_1  - S_1} \binom{d_1}{S_1} \binom{S_1}{e_1, \ldots, e_q} \binom{d_1-S_1}{ e_{q+1}, \ldots, e_{q+s}} \binom{d_2}{S_2} \binom{S_2}{g_1, \ldots, g_q} \notag \\
 & \hspace*{2.0in}\binom{d_2-S_2} {g_{q+1}, \ldots, g_{q+s}} \notag \\
& \sum_{\{i_1, \ldots, i_q, j_1, \ldots, j_s, k_1, \ldots, k_t\}}
\E \left[ \prod_{r=1}^q (Y_{i_r} - \expect{Y_{i_r} \mid \H})^{e_r} (\conj{Y_{i_r}} - \expect{\conj{Y_{i_r}}\mid \H})^{g_r} \prod_{l=1}^s (Y_{j_l} - \expect{Y_{j_l} \mid \H})^{e_{q+s}} \right.\notag \\
& \hspace*{2.0in} \left. \prod_{m=1}^t ( Y_{k_m} - \expect{ Y_{k_m} \mid \H})^{g_{q+t}}\right] \notag \\
&\le \sum_{S_1=0}^{d_1} ~~\sum_{\substack{S_2=0 \\ S_1 =0 \Rightarrow S_2=0}}^{d_2} ~~\sum_{\substack{q=0\\ S_1=0 \Rightarrow q=0}}^{\min(d_1, d_2)} ~~\sum_{s = \alpha_0(q)}^{d_1 -S_1} \sum_{t=\alpha_0(q)}^{d_2-S_2}\sum_{e_1 + \ldots + e_q = S_1} \sum_{e_{q+1}, \ldots, e_{q+s} = d_1  - S_1} \notag \\
& \hspace*{0.8in} \binom{d_1}{S_1} \binom{S_1}{e_1, \ldots, e_q} \binom{d_1-S_1}{ e_{q+1}, \ldots, e_{q+s}} \binom{d_2}{S_2} \binom{S_2}{g_1, \ldots, g_q} \binom{d_2-S_2} {g_{q+1}, \ldots, g_{q+s}} \notag \\
& \sum_{\{i_1, \ldots, i_q, j_1, \ldots, j_s, k_1, \ldots, k_t\}} \prod_{r=1}^q \left(\frac{ \abs{x_{i_r}}^{2p-2}a(e_r+g_r)F_2}{\rho C} \right)^{ (e_r+g_r)/2}\prod_{l=1}^s \abs{x_{j_l}}^{pe_{q+l}} n^{-\Omega(e_{q+l})} \notag \\
& \hspace*{2.0in} \prod_{m=1}^t \abs{x_{k_m}}^{pg_{q+m}}n^{-\Omega(g_{q+m})}  \enspace . \label{eq:dcentmom:midG0:2:t3}
\end{align}

Define the following sums $P_1, P_2$ and $P_3$ as functions of $S_1$ and/or $S_2$. We will assume that the constraint: $S_2 = 0$ iff $S_1 = 0$, is satisfied by $S_1$ and $S_2$.
\begin{align} \label{eq:dcentmom:midG0:P1}
P_1(S_1, S_2) &= \sum_{\substack{q=1\\ q=0 \text{ iff } S_1=0}}^{\min(d_1, d_2)} \sum_{e_1 + \ldots + e_q = S_1} \sum_{g_1 + \ldots + g_q = S_2}\binom{ S_1}{e_1, \ldots, e_q} \binom{S_2}{g_1, \ldots, g_q}\notag \\
 & \hspace*{1.0in}\sum_{\{i_1, \ldots, i_q\}} \left(\frac{ \abs{x_{i_r}}^{2p-2}a(e_r+g_r)F_2}{\rho C} \right)^{ (e_r+g_r)/2} \\
P_2(d_1,S_1) & = \sum_{s = \alpha_0(S_1)}^{d_1 - S_1} \sum_{e_{q+1} + \ldots + e_{q+s} = d_1-S_1} \binom{ d_1-S_1}{e_{q+1}, \ldots, e_{q+s}} \sum_{\{j_1, \ldots, j_s\}} \prod_{l=1}^s \abs{x_{j_l}}^{pe_{q+l}} n^{-\Omega(e_{q+l})} \\
& = \left(n^{-\Omega(1)} F_p \right)^{d_1-S_1} \notag
\end{align}
By Lemma~\ref{lem:dcentmom:midG0:1}, and assuming $C \ge K (n^{1-2/p}/\log (n)) \epsilon^{-2} \log (1/\delta)$,   we have,
\begin{align*}
P_1(S_1,S_2) \le \left( \epsilon F_p/20\right)^{S_1+S_2} \enspace .
\end{align*}

 We note that the summation in Eqn.~\eqref{eq:dcentmom:midG0:2:t3} can be bounded above in terms of $P_1(S_1), P_2(d_1, S_1)$ and $P_2(d_2, S_2)$, as follows.
\begin{align}
&\sum_{S_1=0}^{d_1} ~~\sum_{\substack{S_2=0 \\ S_1 =0 \Rightarrow S_2=0}}^{d_2} ~~\sum_{\substack{q=0\\ S_1=0 \Rightarrow q=0}}^{\min(d_1, d_2)} ~~\sum_{s = \alpha_0(q)}^{d_1 -S_1} \sum_{t=\alpha_0(q)}^{d_2-S_2} \sum_{e_1 + \ldots + e_q = S_1} \sum_{e_{q+1}, \ldots, e_{q+s} = d_1  - S_1}
\notag \\
&
 \binom{d_1}{S_1} \binom{S_1}{e_1, \ldots, e_q} \binom{d_1-S_1}{ e_{q+1}, \ldots, e_{q+s}} \binom{d_2}{S_2} \binom{S_2}{g_1, \ldots, g_q} \binom{d_2-S_2} {g_{q+1}, \ldots, g_{q+s}} \notag \\
& \sum_{\{i_1, \ldots, i_q, j_1, \ldots, j_s, k_1, \ldots, k_t\}} \prod_{r=1}^q \abs{x_{i_r}}^{(p-1)(e_r+q_r)} \prod_{l=1}^s \abs{x_{j_l}}^{pe_{q+l}} n^{-\Omega(e_{q+l})} \prod_{m=1}^t \abs{x_{k_m}}^{pg_{q+m}}n^{-\Omega(g_{q+m})} \\
 & \le \sum_{S_1=0}^{d_1} \binom{d_1}{S_1} \sum_{S_2=0}^{d_2} \binom{d_2}{S_2}
~~\sum_{\substack{q=0\\ S_1=0 \Rightarrow q=0}}^{\min(d_1, d_2)} ~~\sum_{s = \alpha_0(q)}^{d_1 -S_1} \sum_{t=\alpha_0(q)}^{d_2-S_2}\sum_{e_1 + \ldots + e_q = S_1} \sum_{e_{q+1}, \ldots, e_{q+s} = d_1  - S_1} \notag \\
 &\hspace*{1.0in} \binom{S_1}{e_1, \ldots, e_q} \binom{d_1-S_1}{ e_{q+1}, \ldots, e_{q+s}} \binom{S_2}{g_1, \ldots, g_q} \binom{d_2-S_2} {g_{q+1}, \ldots, g_{q+s}} \notag \\
 &\sum_{\{i_1, \ldots, i_q\}} \sum_{\{j_1, \ldots, j_s\}} \sum_{\{k_1, \ldots, k_t\}} \prod_{r=1}^q \abs{x_{i_r}}^{(p-1)(e_r+q_r)} \prod_{l=1}^s \abs{x_{j_l}}^{pe_{q+l}} n^{-\Omega(e_{q+l})}  \prod_{m=1}^t \abs{x_{k_m}}^{pg_{q+m}}n^{-\Omega(g_{q+m})} \notag \\
& \le \sum_{S_1=0}^{d_1} \binom{d_1}{S_1} \sum_{S_2=0}^{d_2} \binom{d_2}{S_2} P_1(S_1, S_2) P_2(d_1, S_1) P_2(d_2, S_2)  \label{eq:dcentmom:midG0:2:t4}
\end{align}
Substituting the simplified expressions for $P_1(S_1, S_2)$ and $P_2(d_1, S_1)$ and $P_2(d_2, S_2)$, Eqn.~\eqref{eq:dcentmom:midG0:2:t4} is bounded above as follows.
\begin{align*}
& \sum_{S_1=0}^{d_1} \binom{d_1}{S_1} \sum_{S_2=0}^{d_2} \binom{d_2}{S_2} P_1(S_1, S_2) P_2(d_1, S_1) P_2(d_2, S_2) \\
& \le \sum_{S_1=0}^{d_1} \binom{d_1}{S_1} \sum_{S_2=0}^{d_2} \binom{d_2}{S_2} \left( \frac{\epsilon F_p}{20} \right)^{S_1+S_2} \left( n^{-\Omega(1)}F_p\right)^{d_1+d_2-S_1-S_2} \\
& \le \left( \frac{ \epsilon F_p}{20} + n^{-\Omega(1)} F_p \right)^{d_1+ d_2} \\
& \le \left( \frac{\epsilon F_p}{10} \right)^{d_1+d_2}
\end{align*}

\end{proof}

\subsection*{Putting things together}

The following lemma, which is a restatement of Lemma~\ref{lem:dcentmom} is applied  for the case when $\delta = O(n^{-c})$, for some constant $c$, but holds generally. Recall, that as stated, for the case $\delta = n^{-\Theta(1)}$, the shelf structure is not needed.
For the statement and proof of Lemma~\ref{lem:dcentmom}, let $\H = \G \wedge \nocollision\wedge \goodest$.
\begin{lemma} [Re-statement of Lemma~\ref{lem:dcentmom}]
Let $C \ge Kn^{1-2/p}\epsilon^{-2}\log(1/\delta)/\log (n) + L n^{1-2/p} \epsilon^{-4/p}\\ \log^{2/p}(1/\delta)$, where, $K$ is the constant from Lemma~\ref{lem:dcentmom:midG0:2} and $L$ is the constant from Lemma~\ref{lem:dmf2}.  Let $B$ be such that  $C/B \ge (5p)^2$.
Then,
\begin{align*}
\expect{ \left( \sum_{i \in [n]} (Y_i - \expect{Y_i \mid \H}) \right)^d \left( \sum_{i \in [n]} ( \conj{Y_i} - \expect{\conj{Y_i} \mid \H})\right)^d \mid \H} \le \left( \frac{ \epsilon F_p}{5} \right)^{2d}
\end{align*}
\end{lemma}

\begin{proof} Note that in this proof, the terms such as $\expect{Y_i}$ or $ \expect{ \conj{Y_i}}$ are written for brevity, they should be interpreted as $\expect{Y_i \mid \H}$ and $\expect{ \conj{Y_i} \mid \H}$ respectively.

Let $S_1 = \lmargin(G_0) \cup_{l=1}^L G_l$. Recall that for the case $ \delta = n^{-O(1)}$, $\midreg(G_0) = \{i : \abs{x_i} \ge T_0(1+\epsbar)\}$, corresponding to the frequency range $[T_0(1+\epsbar), \infty)$.
\begin{align*}
 &\left(\sum_{i \in S} (Y_i - \expect{Y_i \mid \H}) \right)^d\\
 &= \left ( \left( \sum_{i \in \midreg(G_0)} (Y_i - \expect{Y_i \mid \H}) \right)+ \left( \sum_{i \in S_1} (Y_i - \expect{Y_i \mid \H}) \right) \right)^d \\
& = \sum_{d_1=0}^d \binom{d}{d_1} \left( \sum_{i \in  \midreg(G_0)} (Y_i - \expect{Y_i \mid \H}) \right)^{d_1} \left( \sum_{i \in S_1} (Y_i - \expect{Y_i \mid \H}) \right) ^{d-d_1}
\end{align*}
Similarly,
\begin{multline*}
\left(\sum_{i \in [n]} (\conj{Y_i }- \expect{\conj{Y_i} \mid \H}) \right)^d =
\sum_{d_2=0}^d \binom{d}{d_2} \left( \sum_{i \in  \midreg(G_0)} (\conj{Y_i} - \expect{\conj{Y_i} \mid \H}) \right)^{d_2} \\ \left( \sum_{i \in S_1} (\conj{Y_i} - \expect{\conj{Y_i} \mid \H}) \right) ^{d-d_2}
\end{multline*}
Taking the product and then its expectation, we obtain,
\begin{align}
&\expect{ \left( \sum_{i \in [n]} (Y_i - \expect{Y_i \mid \H}) \right)^d \left( \sum_{i \in [n]} ( \conj{Y_i} - \expect{\conj{Y_i} \mid \H})\right)^d \mid \H}\notag \\
& = \sum_{d_1=0}^d  \sum_{d_2=0}^d\binom{d}{d_1} \binom{d}{d_2}\E\left[  \left( \sum_{i \in \midreg(G_0)} (Y_i - \expect{Y_i \mid \H}) \right)^{d_1}\left( \sum_{i \in \midreg(G_0)} (\conj{Y_i} - \expect{\conj{Y_i} \mid \H}) \right)^{d_2}\right. \notag \\
&\hspace*{1.0in}\left. \left( \sum_{i \in S_1} (Y_i - \expect{Y_i \mid \H}) \right) ^{d-d_1}\left( \sum_{i \in S_1} (\conj{Y_i} - \expect{\conj{Y_i} \mid \H}) \right) ^{d-d_2} \mid \H\right] \enspace .  \label{eq:dcentmom:t1}
\end{align}

 We now consider the expectation term in Eqn.~\eqref{eq:dcentmom:t1}. By properties of \nocollision~as discussed earlier, we have,
 \begin{align}
 &\E\left[  \left( \sum_{i \in \midreg(G_0)} (Y_i - \expect{Y_i \mid \H}) \right)^{d_1}\left( \sum_{i \in \midreg(G_0)} (\conj{Y_i} - \expect{\conj{Y_i} \mid \H}) \right)^{d_2}\right. \notag \\
&\left. \hspace*{0.5in}\left( \sum_{i \in S_1} (Y_i - \expect{Y_i \mid \H}) \right) ^{d-d_1}\left( \sum_{i \in S_1} (\conj{Y_i} - \expect{\conj{Y_i} \mid \H}) \right) ^{d-d_2} \mid \H\right] \notag \\
& = \E\left[  \left( \sum_{i \in \midreg(G_0)} (Y_i - \expect{Y_i \mid \H}) \right)^{d_1}\left( \sum_{i \in \midreg(G_0)} (\conj{Y_i} - \expect{\conj{Y_i} \mid \H}) \right)^{d_2} \mid \H\right] \notag \\
&\hspace*{0.5in} \expect{\left( \sum_{i \in S_1} (Y_i - \expect{Y_i \mid \H}) \right) ^{d-d_1}\left( \sum_{i \in S_1} (\conj{Y_i} - \expect{\conj{Y_i} \mid \H}) \right) ^{d-d_2} \mid \H} \notag \\
& \le \left( \frac{ \epsilon F_p}{10} \right)^{d_1 + d_2} \left( \frac{\epsilon   F_p}{10} \right)^{d-d_1+d-d_2}  \notag\\
&= \left( \frac{ \epsilon F_p}{10} \right)^{2d} \label{eq:dcentmom:t2}
\end{align}
where, the second to last equation follows from Lemmas~\ref{lem:dcentmom:midG0:2} and ~\ref{lem:dmf2} respectively.
Substituting in Eqn.~\eqref{eq:dcentmom:t1}, we obtain
\begin{align*}
&\expect{ \left( \sum_{i \in [n]} (Y_i - \expect{Y_i \mid \H}) \right)^d \left( \sum_{i \in [n]} ( \conj{Y_i} - \expect{\conj{Y_i} \mid \H})\right)^d \mid \H}\\
& \le \left( \frac{ \epsilon F_p}{10} \right)^{2d}\sum_{d_1=0}^d  \sum_{d_2=0}^d\binom{d}{d_1} \binom{d}{d_2}\\
& = \left( \frac{ \epsilon F_p}{10} \right)^{2d} 2^{2d} \\
& = \left( \frac{ \epsilon F_p}{5} \right)^{2d} \enspace .
\end{align*}

\end{proof}
\begin{lemma}[ Second part of Lemma~\ref{lem:dcentmom} restated.] \label{lem:dcentmom:part2}
Let $C \ge Kn^{1-2/p}\epsilon^{-2}\log(1/\delta)/\log (n) + L n^{1-2/p}  \allowbreak\epsilon^{-4/p} \log^{2/p}(1/\delta)$, where, $K,L$ are suitable constants.
Then, for $d = \lceil \log (1/\delta) \rceil$,
\begin{align*}
\mathbf{E}\biggl[ \biggl( \sum_{i \in [n]} (Y_i - \expect{Y_i}) \biggr)^d \biggl( \sum_{i \in [n]} ( \conj{Y_i} - \expect{\conj{Y_i}})\biggr)^d \mid \H\biggr] \le \biggl( \frac{ \epsilon F_p}{5} \biggr)^{2d}
\end{align*} It follows that
$\prob{ \left\lvert\sum_{i \in [n]}  (Y_i - \expect{Y_i}) \right\rvert \ge (\epsilon/2) F_p  \mid \H}
\le \delta^2 $. Hence, for $\delta = n^{-O(1)}$,
$$ \prob{ \left\lvert\sum_{i \in [n]}  (Y_i - \expect{Y_i}) \right\rvert \ge (\epsilon/2) F_p } \le \delta$$
\end{lemma}

\begin{proof}

\begin{align*}
&\prob{ \left\lvert\sum_{i \in [n]}  (Y_i - \expect{Y_i}) \right\rvert \ge (\epsilon/2) F_p \mid \H} \\
& = \prob{ \left\lvert\sum_{i \in [n]}  (Y_i - \expect{Y_i}) \right\rvert^{2d} \ge (\epsilon F_p/2)^{2d} \mid \H} \\
& \le \frac{\expect{ \left\lvert\sum_{i \in [n]}  (Y_i - \expect{Y_i}) \right\rvert^{2d} \mid \H}}{ (\epsilon F_p/2)^{2d}}\\
& = \frac{1}{ (\epsilon F_p/2)^{2d}} \expect{ \left(\sum_{i \in [n]}  (Y_i - \expect{Y_i}) \right)^{d}\left(\sum_{i \in [n]}  (\bar{Y}_i - \expect{\bar{Y}_i}) \right)^{d} \mid \H} \\
& \le \frac{\left( \epsilon F_p/5\right)^{2d}}{(\epsilon F_p/2)^{2d}}, \hspace*{1.0in}\text{by first part of  Lemma ~\ref{lem:dcentmom}.}\\
& \le (2/5)^{2d} \\
& \le (2/3)\delta^2
\end{align*}
since, $ d = \lceil \log (1/\delta) \rceil$ and $d \ge 1$.

Now, unconditioning with respect to $\H$, we have,
\begin{align} \label{eq:dcentmom:uncond}
&\prob{ \biggl\lvert\sum_{i \in [n]}  (Y_i - \expect{Y_i}) \biggr\rvert \ge (\epsilon/2) F_p } \notag \\
& = \prob{ \biggl\lvert\sum_{i \in [n]}  (Y_i - \expect{Y_i}) \biggr\rvert \ge (\epsilon/2) F_p \mid \H} \prob{\H} \notag \\
&~~~~~~ + \prob{ \biggl\lvert\sum_{i \in [n]}  (Y_i - \expect{Y_i}) \biggr\rvert \ge (\epsilon/2) F_p \mid \neg\H}  \prob{\neg H} \notag \\
& \le (2/3)\delta^2 \cdot  \prob{H} + (1-\prob{H})
\end{align}
We have $\prob{H} \ge 1- n^{-c}$, for any constant $c\ge 1$. Since, $\delta = n^{-O(1)}$, we can choose $c$ so that $ \prob{H} \ge 1- \frac{\delta}{10}$. Then, Eqn.~\eqref{eq:dcentmom:uncond} is bounded above by
\begin{align*}
& \le (2/3)\delta^2 ( 1- \delta/10) + \delta /10 \\
& \le \delta \enspace .
\end{align*}
\end{proof}

\begin{theorem} [Re-statement of Theorem~\ref{thm:ub}]For each $0 < \epsilon < 1$ and $7/8 \ge \delta \ge n^{-c}$, for any constant $c$, there is a sketching  algorithm that $(\epsilon, \delta)$-approximates $F_p$ with sketching dimension
$O\left(n^{1-2/p}\left(\epsilon^{-2} \log (1/\delta) + \epsilon^{-4/p} \log^{2/p}(1/\delta)\log n\right)\right)$ and update time (per stream update)  $O((\log^2 n)(\log (1/\delta)))$.
\end{theorem}
\begin{proof} 
The correctness of the algorithm follows from Lemma~\ref{lem:dcentmom:part2}.

The algorithm uses $O(C  \log (n)) = O(n^{1-2/p} \epsilon^{-2}\log (1/\delta)  + n^{1-2/p} \epsilon^{-4/p} \log^{2/p}(1/\delta) \log (n)) $ sketches  at the lowest level structure. The other  structures at levels $ 0,1,2, \ldots, L$ are geometrically decreasing in size with common ratio $\alpha$, and hence, the space is dominated by a constant times the space used at level $0$, that is,  $O(C \log n) = O(C \log^{2/p}(1/\delta) \epsilon^{-4/p} (\log n) )$ sketches.

The degree of independence used is $O((\log n) + (\log (1/\delta))$ for the roots of unity sketches and the hash functions for the \est~structure at each level. Consequently, the time taken to process a stream update is $O((\log n) + (\log (1/\delta)))$.
\end{proof}

\newpage

\section{Extending the algorithm for $n^{-\Omega(1)}$  failure probability to $2^{n^{-\Omega(1)}} $ failure probability  } \label{sec:shelf1}

We have so far shown an algorithm for computing an  $(\epsilon, \delta)$-approximation to  $F_p$ when  the failure probability $ \delta \ge n^{-c}$, for some constant $c$. We will now extend the analysis to the case when  $\delta$ is $n^{-\omega(1)}$ , that is,  $\delta = o(1/\text{poly}(n))$.

\subsection{The \ghss~structure and the event $\G$}

Consider the event $\G$ defined as a conjunction of events earlier in Section \ref{sec:algo}.
In \cite{g:arxiv15}, the set of good events  $\G$ as defined is a proper superset of the events constituting $\G$ in this work, and that set  was shown to hold with probability $1- n^{-\Omega(1)}$.
So far, we have looked at the case when  $\delta = n^{-O(1)}$. In this case, it suffices to show that $\G$ holds with probability $1- n^{-c'}$, for any constant $c'$. Since, this follows from the previous treatment in \cite{g:arxiv15}, no separate arguments were given.

Since now we consider the case when $\delta = n^{-\omega(1)}$  we have to at least show that $\G$ (as defined in Section~\ref{sec:algo}) holds with probability $1- \delta^{\Omega(1)}$. In order to do so, we first prove an extension of Lemma~34 from \cite{g:arxiv15} that shows  that not only is $C_L = n^{\Omega(1)}$ ( as shown in \cite{g:arxiv15}), but for a  suitable choice of the parameter $\nu = \Omega(1)$,  $C_L \ge (\epsilon^{-2}\log(1/\delta))^{1 + c}$, where, $c > 0$ is a constant. Using a theorem from \cite{sss:soda93}, this then implies that $\G$ holds with probability $1- \delta^{\Omega(1)}$. We first show that the  \smallres~event holds with probability $1-\delta/n^{\Omega(1)}$, and consequently, \lastlevel~also holds with probability $1-\delta/n^{\Omega(1)}$.

The events $\nocollision$, \accuest~and \smallhh~ are unchanged and hold with probability $1-n^{-\Omega(1)}$. These however are not part of the good event $\G$.

\begin{lemma} \label{lem:Clastlevel} [Extension  of Lemma 25 in \cite{g:arxiv15}.]
Let $\alpha = 1- (1-2/p)\nu$, for $\nu < (\ln 2)/8$. Then, for $p > 2$ and $\displaystyle \nu < \frac{ p/2 + 2/p-2}{(1-2/p)^2 (4/\ln 2)}$, we have, (1)  $C_L \ge n^{\Omega(1)}$ and (2) $C_L \ge (\epsilon^{-2}\log(1/\delta))^{1 + c}$, where, $c > 0$ is a constant depending on $p$.
\end{lemma}
\begin{proof}
Let $ \alpha = 1 - (1-2/p)\nu$ and let $\gamma = 1-\alpha$.
Following the notation of \cite{g:arxiv15}, we have,   $C = K n^{1-2/p}$, where, $K = \kappa \cdot \epsilon^{-4/p}\log^{2/p}(1/\delta)$. As shown in Lemma 25 in \cite{g:arxiv15},
\begin{align} \label{eq:lem:lastlevel:CL} C_L \ge  K^{1 + 4\gamma/(\ln 2)}  \cdot n^{1-2/p - (2/p)(4\gamma/(\ln 2))} \enspace .
\end{align}
From  the lower bound of $\Omega(n^{1-2/p} \epsilon^{-2} \log (1/\delta)$ for $(\epsilon, \delta)$-approximating $F_p$, we have that
$\log (1/\delta) = O(\epsilon^{2} n^{2/p})$  or that, $n \ge \Omega \left( \epsilon^{-p}\log ^{p/2} (1/\delta) \right)$.  Let $G  = \log (1/\delta)$. Substituting in Eqn.~\eqref{eq:lem:lastlevel:CL}, we have,
\begin{align}\label{eq:lem:lastlevel:t0}
C_L & \ge K^{1 + 4\gamma/(\ln 2)}  \cdot n^{1-2/p - (2/p)(4\gamma/(\ln 2))} \notag \\
& \ge (\epsilon^{-2} G)^{ (2/p)(1 + 4 \gamma/(\ln 2))} (\epsilon^{-2} G)^{ (p/2)\left( 1-2/p - (2/p)(4 \gamma/\ln 2) \right)} \notag \\
& = ( \epsilon^{-2}G)^{ (2/p) + (p/2)-1  - (1-2/p)(4 \gamma/(\ln 2))} \enspace .
\end{align}
The exponent of $\epsilon^{-2} G$ in ~\eqref{eq:lem:lastlevel:t0} equals
\begin{align}\label{eq:lem:lastlevel:t1}
&(p/2) + (2/p)-1 - (1-2/p)(4 \gamma/\ln 2) = (p/2) + (2/p) -1 - (1-2/p)^2 (4\nu/\ln 2)
\end{align}
using the fact that  $\gamma = 1- \alpha   = (1-2/p)\nu$.

In Eqn.~\eqref{eq:lem:lastlevel:t1}, note that $p/2 + 2/p = (\sqrt{(p/2)} - \sqrt{2/p})^2 +2 > 2$, since $p > 2$. By choosing $\displaystyle \nu < \frac{ p/2 + 2/p-2}{(1-2/p)^2 (4/\ln 2)}$, the exponent of $\epsilon^{-2} G$ in Eqn.~\eqref{eq:lem:lastlevel:t0} is greater than 1.
Therefore, $\nu$ can be chosen small enough so that
$$ C_L \ge (\epsilon^{-2} \log (1/\delta))^{1 +c } $$
for some constant $c > 0$.

\end{proof}

 Lemma~\ref{lem:Clastlevel} ensures that for levels $l \in [0, \ldots, L]$, $C_l > \omega(\log (1/\delta))$. This ensures that $ \log (1/\delta) \le  \lceil C_l e^{-1/3}\rceil $. This in turn implies that  Lemma 26 of \cite{g:arxiv15} holds with $d = O(\log (1/\delta))$ provided, the hash functions $g_1, \ldots, g_L$ are drawn from a $d$-wise independent family. It then follows that the events,  $\goodtopk(\{B_l\}_{l \in L}), \smallu( \{C_l\}_{l \in L})$, each holds with probability $1-(\delta/n)^{\Omega(1)}$. This implies that \smallres~ holds with probability $1 - (\delta/n^{\Omega(1)})$.

 The top-most level $L$ of  \ghss~uses the deterministic compressive sensing based algorithm  \cite{crt:ieeetit06a, donoho:tit06} for the recovery of $x_{iL}$ for those items $i$ that hash to level $L$. These techniques guarantee the deterministic  recovery of any $k$-sparse $n$-dimensional vector $y$ such that $ \norm{y - x}_2 \le C' \min_{k\text{-sparse } x'} \norm{x' - x_L}_1$, where, $C'>1$ is a constant,  using $m = O(k \log (n/k))$ measurements. Let $x_L$ denote the vector of frequencies of items that are sampled into level $L$, that is, $(x_L)_i = x_i$, if $i \in \stream_l$ and $(x_L)_i = 0$ otherwise. Following the arguments of Lemma 28 in \cite{g:arxiv15}, we have with probability $1 - \delta/n^{\Omega(1)}$ that, $\card{\{ i \in\stream_L\}} \le 2C_L$. Therefore, $x_L$ has at most $O(C_L)$ non-zero entries, and by compressive sensing,  using $m = O(C_L \log (n/C_L)) = O(C_L \log n)$ measurements, these entries are exactly recovered.  Hence, with this modification, \lastlevel~holds with probability $1- \delta/n^{\Omega(1)}$.

 We now prove Lemma~\ref{lem:Gfails}.

\eat{ \begin{lemma}[Restatement of Lemma \ref{lem:Gfails}]
 The number of items for which $\nocollision(i)$ does not hold is at most $O(\log(1/\delta)/\log (n))$ with probability $1 - \delta/10$. Analogously, the number of items for which $\accuest(i)$ and $\smallhh(i)$ does not hold is at most $O(\log (1/\delta)/\log (n))$ with probability $1- \delta/10$ each.
\end{lemma}

\begin{proof} First, we prove the statement of the lemma assuming full independence of the hash functions. Fix a level $l$ and let $i \in \stream_l$. As shown in \cite{g:arxiv15}, $\prob{\nocollision(i) \text{ fails} \mid  i \in \stream_l } \le n^{-13} = n^{-\Omega(1)} = q $\text{ say}.  Let $X = \{i \in \stream_l : \nocollision(i) ~\textrm{fails}\}$. Let $R = \card{\{ i: i \in \stream_l\}}$.
Hence, due to full independence,
\begin{align*}
\prob{X\ge t} = \sum_{k=t}^R \binom{n}{k} q^k (1-q)^{n-k} \le  \sum_{k = t}^n \binom{n}{k} q^k  \le  \sum_{k=t}^n (nq)^k = O(nq)^t = e^{t \ln (nq)} = e^{ -t \Omega(\ln n)}
\end{align*}
Therefore, for $t = O(\ln (1/\delta)/\ln n)$, $\prob{ X \ge t} \le \delta^{\Omega(1)}$.
\end{proof}
}

\subsubsection{Number of items not satisfying \textsc{goodest} or \textsc{nocollision}}
In this section, we try to calculate the number of items $i$ in the \ghss~levels  that do not satisfy \goodest$(i)$, with probability at most $1 - O(\delta)$. These are the items that can possibly be un-estimated, or underestimated, or overestimated, i.e., may cause $\hat{F}_p$ to be in error.

For the analysis, consider the \countsketch$(C,s)$ structure \cite{ccf:icalp02} with $s$ hash tables, denoted $T_1, \ldots, T_s$, each consisting of $C'=O(C)$ buckets. The hash functions for the respective tables are denoted by $h_1, \ldots, h_s$. Let $\{\xi_{ri}\}_{i\in [n]}$ denote the family of Rademacher variables used for the sketches in table $T_r$.  For the initial part of the analysis, we just consider one independent copy of the $s$ repetitions, and denote the hash table by $T$, the hash function by $h$ and the Rademacher family $\{\xi_j \mid j \in [n]\}$. For each $b \in [C]$, $$T[b] = \sum_{j: h(j)=b} x_j \xi_j \enspace .$$ Fix $i \in [n]$. Following the definition of \emph{estimate} for $x_i$ obtained from the single table copy under hash function $h$ as given in \cite{ccf:icalp02}, we have,
\begin{equation} \label{eq:cskbasic}\hat{x}_i = T[h(i)] \cdot \xi_i = x_i + \biggl(\sum_{j \ne i} x_j \xi_j  \chi_{ij}\biggr)\xi_i \end{equation}
where, $\chi_{ij}$ is an indicator variable such that $\chi_{ij} =1$ iff $h(i) = h(j)$.  We will say that the \emph{estimate for $i$ is good under $h$} if the following event holds:
$$ \text{Good Estimate}(i):~~ \card{\hat{x_i} - x_i} < \frac{\sqrt{24}}{\sqrt{C}} \norm{x}_2 $$
We consider the following question: Let $S_k = \{i_1, \ldots, i_k\} \subset [n]$ be any given fixed subset of $k$ distinct items, where, $k = C/20$. Obtain an upper bound for the probability $$\prob{ \ge \text{Good Estimate}(i_1) \wedge \cdots \wedge \neg  \text{Good Estimate}(i_k)}$$

Let $x^{h(i)}$ denote the $n$-dimensional vector $\begin{bmatrix} x_1 \cdot \chi_{i,1}, x_2 \cdot \chi_{i,2}, \ldots, x_n  \cdot \chi_{i,n}\end{bmatrix}$. Clearly,
$$ \norm{x^{h(i)}}_2^2 =  x_i^2 + \sum_{j\in [n], j \ne i} x_j^2 \chi_{i,j} $$
Taking expectations,
$$\expect{ \norm{x^{h(i)}}_2^2} = x_i^2 + \frac{\norm{x}_2^2 - x_i^2}{C} \le x_i^2 + \frac{ \norm{x}_2^2}{C} \enspace . $$
By Markov's inequality applied to the non-negative variable $\norm{x^{h(i)}}^2 - x_i^2$, we have,
$$\prob{ \norm{x^{h(i)}}^2 - x_i^2 \le \frac{ 8 \norm{x}_2^2}{C}} \ge \frac{7}{8} $$
Define the event \emph{GoodBucketNorm$(i)$} as follows:
$$ \emph{GoodBucketNorm}(i):  \norm{x^{h(i)}}^2 - x_i^2 \le \frac{ 8 \norm{x}_2^2}{C} $$
Fix the hash function $h$. This fixes the function $\chi_{ij}$. From Eqn.~\eqref{eq:cskbasic}, we have,
$$ \card{ \hat{x}_i - x_i} = \biggl\lvert \sum_{j \ne i} x_j \xi_j  \chi_{ij} \biggr\rvert $$
By Azuma-Hoeffding's bound
\begin{align*}
\prob{ \card{ \hat{x}_i - x_i} > t} \le 2\exp{ - \frac{2t^2}{\norm{x^{h(i)}}^2}}
\end{align*}
Choose $t = \frac{ \sqrt{24} \norm{x}_2}{\sqrt{C}}$.
Define the event
$$\text{\em SmallDeviation}(i): ~~\card{\hat{x}_i - x_i} \le \frac{ \sqrt{24} \norm{x}_2}{\sqrt{C}}$$
Conditioned on \emph{GoodBucketNorm}$(i)$, $$
\prob{\emph{SmallDeviation}(i) \mid \emph{GoodBucketNorm}(i)} \ge 1- 2 e^{-6} $$
Therefore,
\begin{align*}
\prob{ \emph{SmallDeviation}(i), \emph{GoodBucketNorm}(i)} \ge (1-2e^{-6}) \cdot \frac{7}{8} \ge \frac{3}{4}
\end{align*}
as deduced in \cite{ccf:icalp02}.

\emph{ Note 1.} that we use Azuma-Hoeffding's bound to obtain constant confidence of the form $1 - 2e^{-6}$, and hence the same can be obtained using a $d$th moment method for even and constant $d$. This in turn requires only $d = O(1)$-wise independence of the Rademacher variables and the hash function. \\
\noindent \emph{Note 2.} We intend to use the analysis for sketches involving complex roots of unity and not Rademacher variables. The Azuma-Hoeffding inequality can be applied for each of the real and complex part separately, since each of them are zero mean (since, $\expect{\omega} = 0$). The constant factor increases by a factor of 2.

We now return to the question posed earlier in the section. Let $S_k = \{i_1, \ldots, i_k\}$ be a fixed given set of items from $[n]$. Suppose the items in $[n]$ are populated as follows. First, say all items in $[n] \setminus S_k$ are inserted into the table. Next, the items $i_1, i_2, \ldots, i_k$ are inserted one by one incrementally. Say that the state of $i_1$ is good (i.e, 1) if \emph{GoodEstimate}$(i_1)$ holds and not good (i.e., 0) otherwise. Similarly, after the insertion of $i_1$, when $i_2$ is inserted, the state of $i_2$ can be either 1 or 0, and so on. Let $s^j \in \{0,1\}^j$ denote the state vector such that $s^j_t$ is 1 iff the state of $i_t$ is good and 0 otherwise, for $t \in [j]$. We wish to consider the probability
$$\prob{\emph{GoodEstimate}(i)\mid s^j} \enspace . $$
Define the event  \emph{Isolation}$(j+1)$ to mean that $i_{j+1}$ does not collide with $i_1, \ldots, i_j$. This happens with probability at least $1- \frac{t}{C} \ge 19/20$, since it is assumed that $j+1 \le k \le C/20$.
Therefore, by union bound,
\begin{align*}
\prob{ \emph{GoodBucketNorm}(i_{j+1}), \emph{Isolation}(i_{j+1}) \mid s^j} \ge 1- \frac{1}{20} - \frac{1}{8} = \frac{33}{40}
\end{align*}
Therefore,
\begin{align*}
&\prob{\emph{GoodEstimate}(i_{j+1}) \mid s^j}\\
& \ge \prob{\emph{SmallDeviation}(i_{j+1}), \emph{GoodBucketNorm}(i_{j+1}), \emph{Isolation}(i_{j+1}) \mid s^j} \\
& = \prob{ \emph{SmallDeviation}(i_{j+1}) \mid \emph{GoodBucketNorm}(i_{j+1}), \emph{Isolation}(i_{j+1}), s^j} \\
&\hspace*{1.0in} \cdot  \prob{\emph{GoodBucketNorm}(i_{j+1}), \emph{Isolation}(i_{j+1}) \mid s^j}\\
& \ge (1-e^{-6})(33/40)\\
& \ge \frac{4}{5} \enspace.
\end{align*}
We have thus shown that the probability that the item $i_{j+1} \in S_k$ satisfies \emph{GoodEstimate} is at least $4/5$, no matter what the state of the items $i_1, \ldots, i_j$ may be.

We can now introduce $s$ independent copies of the hash function $h$ as $h_1, \ldots, h_s$, and corresponding tables $T_1, \ldots, T_s$, where, the hash function $h_r$ is used for table $T_r$, $r \in [s]$. Let $\sigma^j = (\{0,1\}^{j})^s$ be the state vector for the state of items $i_1, \ldots, i_j$ in each of the tables $T_1, \ldots, T_s$. That is, the $j$-bit vector $\sigma^j_{1 \ldots j}$ represents the state of items $i_1, \ldots, i_j$ in $T_1$, and in general, the $j$-bit segment $\sigma^j_{ (r-1)j+1, \ldots, rj}$ represents the state of the items $i_1, \ldots, i_j$ in table $T_r$, for $r=1,2, \ldots, s$.
 Denote by \emph{GoodEstimate}$(r, i_j)$ the event that in table $T_r$, the item $i_j$ satisfies the event \emph{GoodEstimate}, that is, $\card{ \hat{x}_i - x_j} < \sqrt{\frac{24}{C}}\norm{x}_2$.
By independence  of the hash functions $h_1, \ldots, h_s$, for any $r \in [s]$,
$$ \prob{ \emph{GoodEstimate}(r, i_{j+1}) \mid \sigma^j} = \prob{ \emph{GoodEstimate}(r, i_{j+1}) \mid \sigma^j_{(r-1)j+1, \ldots, rj}} \ge 4/5 $$

For any $j \in [k]$, let $G_j$ denote the number of tables in which the bucket to which $i_j$ maps provides a good estimate, that is,
$$G_j = \card{ \{r \in [s] \mid \emph{GoodEstimate}(i_j) \text{ holds } \}}$$
For analysis purposes, let $\sigma^j$ denote the state of the buckets to which items $i_1, \ldots, i_j$ maps in each of the tables $T_1, \ldots, T_s$. Let
$$G_{j+1 \mid \sigma^j} = \card{ \{r \in [s] \mid \emph{GoodEstimate}(i_{j+1}) \text{ holds  conditional on state being } \sigma^j \}}$$
By the above calculation, the probability that \emph{GoodEstimate}$(i_{j+1})$ holds conditional on the state being any $\sigma^j$ is at least $4/5$.
What is the probability that the \countsketch~estimate, that is, \emph{median}$_{r=1}^s \hat{x}_{r, i_j}$ is not a good estimate? This is the probability $\prob{G_{j+1} \le s/2 \mid \sigma^j }$. Let $ g_{r,j+1}$ be the indicator variable that  is 1 iff \emph{GoodEstimate}$(i_{j+1})$ holds
 conditional on the state being $\sigma^j$. Then,
 $$G_{j+1} = \sum_{r=1}^s g_{r,j+1} $$ Now, $ \prob{ g_{r,j+1} \mid \sigma^j} \ge 4/5$. By Chernoff's bounds,
 \begin{align*}
 \prob{G_{j+1} < s/2 \mid \sigma^j} \le \exp{-\Theta(s)}, ~~~~ \text{ for any feasible $\sigma^j$}, j = 0, \ldots, k-1
 \end{align*}
Let $E_j$ be the union of any arbitrary states $\sigma^j$, that is, $E_j = \sigma^{j,1} \vee \ldots \vee \sigma^{j,N}$, for some $N$. For $N=2$,
\begin{align*}
&\prob{g_{r,j+1} \mid \sigma^{j,1} \vee \sigma^{j,2}} \\
& =
 \frac{\prob{g_{r,j+1} \wedge ( \sigma^{j,1} \vee \sigma^{j,2}) }}{\prob{\sigma^{j,1} \vee \sigma^{j,2}}}\\
& = \frac{ \prob{g_{r,j+1}, \sigma^{j,1}} + \prob{g_{r,j+1}, \sigma^{j,2}}}{\prob{\sigma^{j,1}} + \prob{\sigma^{j,2}}} \\
& = \frac{ \prob{g_{r,j+1} \mid \sigma^{j,1}} \prob{\sigma^{j,1}} + \prob{g_{r,j+1} \mid \sigma^{j,2}} \prob{ \sigma^{j,2}}}{\prob{\sigma^{j,1}} + \prob{\sigma^{j,2}}} \\
& \ge (4/5) \frac{ \prob{\sigma^{j,1}} + \prob{\sigma^{j,2}}}{\prob{\sigma^{j,1}} + \prob{\sigma^{j,2}}}\\
& = (4/5)
\end{align*}

Likewise, by induction, one can show that $ \prob{g_{r+1} \mid E_j} \ge 4/5$. Extending the argument, it can be shown that  $ \prob{G_{j+1} < s/2 \mid E_j} \le \exp{-\Theta(s)}$.
It follows that the probability that each of the $k$ median estimates is not a good estimate is,
\begin{align*}
& \prob{G_1 < s/2, G_2 < s/2, \ldots, G_k < s/2}\\
& = \prob{G_k < s/2 \mid G_{k-1} < s/2, \ldots, G_1 < s/2}  \cdot \prob{ G_{k-1} < s/2 \mid G_{k-2} < s/2 , \ldots, G_1 < s/2}\\
& \hspace*{0.5in}  \cdots \cdot \prob{G_1 < s/2} \\
& \le \exp{-\Theta(s)} \cdots \exp{-\Theta(s)} \\
& = \exp{-k \Theta(s)}
\end{align*}

For any set $S_k = \{i_1, \ldots, i_k\}$, let the $F_{S_k}$ be the event that for each $i_j \in S$, the estimate $\hat{x}_{i_j}$ is not a good estimate of $x_{i_j}$. Therefore, $\prob{F_{S_k}} \le \exp{-k \Theta(s)} $ For any given $k$, and noting that $s = \Theta(\log n)$, we have,
\begin{align*}
&\prob{\exists \text{ at least $k$ distinct items  whose estimates are not good estimates }} \\
& =\prob{ \exists S_k \subset [n], \abs{S_k} = k \text{ s.t. }  F_{S_k} \text{ holds }}\\
& \le \binom{n}{k} \exp{-k \Theta(s)} = \exp{-k(\Theta(s) - \log n)} = \exp{-k\Theta(\log n)}\\
& \le \delta^{O(1)}
\end{align*}
provided, $k \ge O(\log (1/\delta))/\log (n)$.

In an analogous way, it can be shown that the probability that at most $O(\log (1/\delta))/\log (n)$ items may fail to satisfy \nocollision~with probability $1-\delta$.
We have proved the following lemma.
\begin{lemma} [Restatement of Lemma~\ref{lem:Gfails}.]
With probability $1- O(\delta)$, the number of elements for which $\goodest$ fails is at most $O(\log (1/\delta))/(\log n)$.  With probability  $1- O(\delta)$, the number of elements for which $\nocollision$ fails is $O(\log (1/\delta))/(\log n)$.
\end{lemma}

\subsection{Extending the analysis}

 Lemma 29 in \cite{g:arxiv15} can be directly extended to obtain  $\card{\discover_l } \le 2 B_l$, with probability $1 - \delta^{\Omega(1)}$. Further,  $B_l = \Omega(C_L) = \omega(\log (1/\delta))$, with probability $1-\delta n^{-\Omega(1)}$. Using the analysis of Lemma~\ref{lem:Gfails}, the number of elements which are discovered  (i.e., $\abs{\hat{x}_i} \ge Q_l$) and for which $\goodest$ fails is $O(\log (1/\delta)/\log n)$ with probability $1 - \delta/n^{\Omega(1)}$. Hence, $\card{\discover_l} \le 2B_l + O(\log (1/\delta)/\log n \le 2B_l + o(B_L/\log n) \le 3 B_l$, with probability $1 - \delta/n^{\Omega(1)}$.

We now consider the effect of items in $\left(\lmargin(G_0) \cup_{l=1}^L G_l \right) \setminus \gooditems_1$. This set contains  items that for  reasons such as, non-discovery, or discovery followed by collisions, get dropped and their scaled contribution to $\hat{F}_p$ are not added. These also include items that were mistakenly discovered and added to samples. This is because $\accuest$ holds only with $1- n^{-\Omega(1)}$ probability.  Let $\text{Error}^{\ghss}$ denote the total contribution to $\hat{F}_p$ of such items that are  either dropped or erroneously estimated and  misclassified within  in the \ghss~structure. Let $\text{Error}^{\shelf}$ denote the total contribution to $\hat{F}_p$ of items that are either dropped or erroneously misclassified within the shelf structure.

\emph{Note} that it is possible for an item to belong to some group $G_l$ as per the \ghss~grouping, but the item's frequency could be significantly over-estimated so that it is classified into one of shelves. The error due to over-estimation of an item that belongs to a \ghss~group but gets classified into a shelf sample is calculated when we analyze the error in the estimation from the shelf structure.

Let $\gooditems_1$ be the set of items that satisfy $\nocollision \cap \accuest$.

\begin{lemma}[ Restatement of Lemma ~\ref{lem:error1}.]
$\text{Error}^{\ghss} \le O\left(\displaystyle \frac{\epsilon^2 F_p}{\log n} \right) $ with probability $1 - \delta/n^{\Omega(1)}$.
\end{lemma}

\begin{proof}

Consider the set $D_1 \subset \left(\lmargin(G_0) \cup_{l=1}^L G_l \right) \setminus \gooditems_1$ consisting of legitimate  items that are dropped due to error in estimation or due to collisions. 
Let $D_2 \subset \left(\lmargin(G_0) \cup_{l=1}^L G_l \right) \setminus \gooditems_1$ consisting of items that are incorrectly discovered or misclassified due to error in estimation.

If $i \in D_1 \cap G_l$ then,
$\abs{x_i} \le T_{l-1}$, if $l \ge 1$, or otherwise, $\abs{x_i} \le T_0 (1 + \epsbar)$, where, $\epsbar = 1/(54p)$. Suppose $i \in D_2$ and $i \in G_{l'}$. However, due to error in estimation, suppose that  $i \in \bar{G_l}$, for some $l < l'$, causing an over-estimate. Such items are in $D_2$. For $l \ge 1$, the estimate for $\abs{x_i}$, namely, $\abs{X_i}$ is bounded by $T_{l-1}(1+\bar{\epsilon})$, otherwise, the estimate is dropped from $\bar{G_l}$.
If $i \in \bar{G}_0$, the upper bound is $U_1(1+\epsbar)$, as is defined in the shelf structure. Here $U_1 \le 2^{1/p} T_0$.

We also have, $$T_0 = O\left( \frac{F_2}{n^{1-2/p} \epsilon^{-4/p} \log^{2/p}(1/\delta)}\right)^{1/2} \le O\left( \frac{\epsilon^4/p F_p^{2/p}}{\log^{2/p}(1/\delta)}\right)^{1/2}
\le O\left( \frac{\epsilon^2  F_p}{\log (1/\delta)} \right)^{1/p}
$$
Therefore, $T_0^p \le O(\epsilon^2 F_p/\log (1/\delta))$. Further,
$T_l^p \le O\left(\displaystyle (2\alpha)^{-lp/2} \epsilon^2 F_p/\log (1/\delta) \right)$.
 Then, the error contribution due to under-estimation or dropping of items is bounded above as follows.
\begin{align} \label{eq:error1}
\text{Error}_1 & \le \sum_{i \in D_1, i \in \lmargin(G_0)}  T_0^p(1+\epsbar)^p + \sum_{l=1}^L \sum_{i \in D_1, i \in G_l} 2^l T_{l-1}^p \notag \\
& \le O\left(\frac{\epsilon^2 F_p}{\log (1/\delta)}\right) \left(\card{D_1 \cap   \lmargin(G_0)} + \sum_{l=1}^L 2^l (2\alpha)^{-lp/2}\card{D_1 \cap G_l} \right) \notag \\
& \le O\left(\frac{\epsilon^2 F_p}{\log (1/\delta)}\right) \card{ D_1 \cap \left( \lmargin(G_0) \cup_{l=1}^{L-1} G_l\right)} \left(1 + \sum_{l=1}^L 2^l (2\alpha)^{-lp/2} \right) \notag \\
& \le O\left( \frac{\epsilon^2 F_p}{\log n} \right) \left(1 + \sum_{l=1}^L 2^l (2\alpha)^{-lp/2} \right)
\end{align}
since, as argued above, the total number of items dropped or mis-estimated is at most $O(\log (1/\delta)/\log n)$.

Recall that $\alpha = 1 - (1-2/p)\nu$, for a small constant $\nu$. Let $\gamma = (1-2/p)\nu$ so that $\alpha = 1- \gamma$. Therefore, $\ln \alpha = \ln (1-\gamma) \ge -2\gamma $.
Now, $$(2\alpha)^{-p/2} = 2^{-(p/2) \log_2 (2\alpha)} = 2^{-(p/2)(1 +  \ln \alpha/\ln 2)} \le 2^{-(p/2)(1 -\gamma (2/\ln 2))} \enspace . $$ Since, $(p/2)\gamma = (p/2-1)\nu$, therefore,
$$2 (2\alpha)^{-p/2}  = 2^{ 1 -(p/2)(1-\gamma(2/\ln 2))} = 2^{(1-p/2) + (p/2-1) \nu(2/\ln 2)}
= 2^{(1-p/2)(1- 2\nu/\ln 2)}
$$
The value $2^{(1-p/2)(1- 2\nu/\ln 2)} $ is a constant $< 1$ since, $p/2 > 1$ and $\nu = 0.01 < \ln 2/2$.
Thus, we have,
\begin{align*}
\sum_{l=1}^L 2^l (2\alpha)^{-lp/2} = \sum_{l=1}^L \left( \frac{2}{(2\alpha)^{p/2}} \right)^l
= \sum_{l=1}^L 2^{l(1-p/2)(1-2\nu/\ln 2)}  = \Theta(1)
\end{align*}
Substituting this in Eqn.~\eqref{eq:error1}, we have,
$$\text{Error}_1 \le O\left( \frac{\epsilon^2 F_p}{\log n} \right)  \enspace . $$

The error contribution of the items in $D_2$ is bounded as follows.
\begin{align*}
\text{Error}_2 = \sum_{i \in D_2 \cap \bar{G_0}}  (U_1(1+\epsbar))^p + \sum_{l=1}^L \sum_{i \in D_2, i \in G_l} 2^l T_{l-1}^p
\end{align*}
This sum is bounded similarly, since $U_1 \le 2^{1/p} T_0$, and therefore, $\text{Error}_2 \le O\left(\epsilon^2 F_p/\log n\right)$.

Therefore, combining, $\text{Error}^{\ghss} = \text{Error}_1 + \text{Error}_2 \le O\left(\epsilon^2 F_p/\log n\right)$.
\end{proof}

\newpage

\section{Analysis of Shelf Structure} \label{sec:shelf}

As discussed in Section~\ref{sec:algo}, the shelf structure is needed when $\log (1/\delta) = \omega( \log (n))$.
In this section,  we will assume that $\log (1/\delta) = \omega( \log (n))$. Further, as discussed,
we will  emphasize  the interesting case when $H_J = o(H_0)$.

Define the following two constraints on   the parameters $a$ and  $b$.
\begin{align}\label{eq:shelf:constr:ab}
(1) ~~ \abs{\ln (ab)} = \Omega(1) ~~~~~\text{ and }~~~~(2)~~ b = \Omega(1) \enspace .
\end{align}

Constraint 1 is derived from the following consideration. The number of measurements required by the $j$th shelf is $O(H_jw_j)$.

\noindent
\emph{Case 1.} Suppose $ab < 1$. Then, by constraint (1), $1-ab = 1-\exp{-\abs{\ln ab}} = 1- O(1) = \Omega(1)$.  The sum of the number of measurements required by all the shelves is $$ \sum_{j=0}^J H_j w_j= H_0 w_0 \sum_{j=0}^J (ab)^j \le \frac{H_0 w_0}{(1-ab)}  = O(H_0 w_0) \enspace .$$
\emph{Case 2.}  Suppose  $ab > 1$. Then, by constraint (1), $1-(ab)^{-1} = 1-\exp{-\abs{\ln ab}} = 1 - O(1) = \Omega (1)$. The sum of the number of measurements across all shelves becomes $$\sum_{j=0}^J H_j w_j = H_J w_J \sum_{k=0}^J (ab)^{-k} \le  \frac{H_J w_J}{1- (ab)^{-1}} = O(H_Jw_J) \enspace . $$ So, by constraint 1, in either case, the total space used by the shelf structure is $O(H_0 w_0 + H_Jw_J)$.
The motivation for Constraint 2 arises from the need to bound the \text{Error}$^{\shelf}$ term, as in the proof of Lemma~\ref{lem:error1}. We now outline how $a$ and $b$ can be chosen to satisfy constraints 1 and 2.

We have, $a^J = w_J/w_0 = \Theta(\ln (1/\delta)/(\ln n)) = \omega(1)$ and $b^J = H_J/H_0 = o(1) $. Therefore,
\begin{align*}
(ab)^J &=  \frac{w_J}{w_0} \cdot   \frac{H_J}{H_0} = \Theta(1) \cdot \frac{\log (1/\delta)}{\log n} \cdot \frac{ \epsilon^{-2}}{\epsilon^{-4/p} \log^{2/p} (1/\delta)} = \Theta(1) \cdot \frac{\left(  \epsilon^{-2} \log(1/\delta) \right)^{1-2/p}}{\log n}
\end{align*}
Let
\begin{align*}
L_0 & = \ln \left(H_Jw_J/(H_0 w_0)\right)= J \ln (ab) \enspace .  
\end{align*}
Let $J = L_0$ so that $\ln (ab) =1$ or $ ab = e$. Set $b = 1/2$ and therefore $a = 2e$.
\eat{In order for $0 < \ln (ab)  = \Omega(1)$, set  $J = L_0$. Therefore, $e^{L_0} = (ab)^J $, or, $ab = e$. We also note that
$$a^J = \frac{w_J}{w_0} = \frac{\log(1/\delta)}{\log n} \le ((ab)^J)^{q_1}$$
for a constant $q_1 \le  (1-2/p)^{-1}$. Similarly,
$$b^{-J} = \frac{H_0}{H_J} =  \Theta(\epsilon^{2-4/p} \log^{2/p}(1/\delta)) \le  ((ab)^J)^{q_2}$$
for another constant $q_2 < (2/p)(1-2/p)^{-1}$.
Define $b = (H_J/H_0)^{1/J}$ and $a = (w_J/w_0)^{1/J}$.
Then,
\begin{align*}
b^{-1} = \left(\frac{H_0}{H_J}\right)^{1/J} \le  (ab)^{Jq_2/J} = e^{q_2} = O(1) \enspace.
\end{align*}}
This shows that $a$ and $b$ can be chosen to satisfy constraints 1 and 2.

\subsection{Error Analysis for Shelf Structure}
We first  extend the definition of events  $\accuest_1$, $\smallhh_1$ and $\nocollision_1$ to  the shelf structure. Let $\gooditems_1 = \nocoll_1 \cap \smallhh_1 \cap \nocollision_1$. We first estimate the error  arising in the  estimate $\hat{F}_p^{\shelf}$  due to items with frequency in the range $[T_0, \infty)$ $\setminus \gooditems_1$, with probability $1- \delta/n^{\Omega(1)}$. That is, these items were either dropped on account of collision, or due to inaccurate estimation, or they were misclassified into a larger shelf, also due to inaccurate estimation.

Suppose $i$ belongs to the frequency range of shelf $j$. Then, the probability of $\nocollision(i)$ failing, or $\accuest(i)$ failing or $\smallhh_j$ failing is at most $e^{-\Omega(w_j)}$. In particular, since, $w_J = \Theta (\log (1/\delta))$, it follows that items belonging to the range of shelf numbered $J$ satisfy \accuest~ and \smallhh~ and \nocollision, with probability $1- \delta^{\Omega(1)}$, and therefore, there is no contribution to error from the last shelf up to probability $1-\delta^{\Omega(1)}$.

\begin{lemma}[Restatement of  first part of  Lemma \ref{lem:error1}.] Assuming constraints 1 and 2,
$$\displaystyle \text{Error}^{\shelf} \le \max\left(O\left( \frac{\epsilon^2 F_p}{\log n} \right), O(\epsilon^p F_p) \right) $$ with probability $1- \delta/n^{\Omega(1)}$.
\end{lemma}

\begin{proof}

For shelf index $j \in \{0,1,\ldots, J-1\}$, let  $D^{(j)}$ denote the set of items that belong to the frequency range of shelf indexed $j$ but do not belong to $ \gooditems_1$. Let $\abs{D^{(j)}} = d_j$.  The contribution to the error term from these items is at most  $U_{j+1}^p d_j$. Further, by the calculation in Lemma~\ref{lem:Gfails}, $d_j \le O\left(\log (1/\delta)/w_j\right)$. More generally, from Lemma~\ref{lem:Gfails}, we have that $ \sum_{j=0}^J w_j d_j = O(\log (1/\delta))$.
Therefore,
\begin{align}\label{eq:shelf:lp}
\text{Error}^{\shelf} \le & \text{ Maximize}  \sum_{j=0}^{J-1} U_{j+1}^p d_j  \text{ subject to } \sum_{j=0}^{ J-1} w_j d_j = O(\log (1/\delta))
\end{align}
This is a linear program with feasible region $ \sum_{j=0}^{J-1} w_j d_j = O(\log (1/\delta))$ and $d_j \ge 0$, for $j =0, \ldots, J-1$. The optimal value of this linear program  lies on a vertex of the corresponding  polygonal face in $J$-dimensional space $\R^J$. A vertex of this face is of the form $\hat{d}_{j-1} e_j = (0, \ldots, 0, \hat{d}_{j-1}, 0, \ldots, 0)^T$, where, $\hat{d}_{j-1} = O(\log (1/\delta))/w_j$, for each $j=0, 1, \ldots, J-1$. Here, $e_j$ is the $j$th column of the $J \times J$ identity matrix.

\noindent
We have,  $$U_{j+1}^p = \left( \frac{F_2}{H_0 b^{j+1}} \right)^{p/2} \le \left( \frac{\epsilon^2 F_p}{\log (1/\delta)} \right) b^{-(j+1)(p/2)} $$
The objective value at the vertex $ \hat{d}_j e_j = (0, \ldots, 0, \hat{d}_j, 0, \ldots, 0)$ is
\begin{align} \label{eq:shelf:objval} U_{j+1}^p \hat{d}_j \le \left( \frac{\epsilon^2 F_p}{\log (1/\delta)} \right)b^{-(j+1)(p/2)} \left( \frac{ O(\log (1/\delta))}{w_j } \right) = \left( \frac{O(\epsilon^2 F_p)}{b^{p/2}\log (n)}\right) \left( \frac{1}{(ab^{p/2})^j}\right)
\end{align}
using $w_j = \Theta(\log n) a^j$.

\noindent
As has been discussed, we can choose the parameters $a$ and $b$, so that $\abs{\ln (ab)} = \Omega(1)$ and  $b = \Omega(1)$. Therefore, $b^{p/2} = \Omega(1)$.\\

\noindent
\emph{Case 1: $ab^{p/2} < 1$}. By Eqn.~\eqref{eq:shelf:objval},  the vertex maximizing the objective function occurs at $\hat{d}_{J-1} e_{J}$, that is, the vertex corresponding to $j = J-1$. Then,
\begin{align*}(ab^{p/2})^J  &= (w_J/w_0) (H_j/H_0)^{p/2}  = \left( \frac{ \log (1/\delta)}{\log (n)} \right) \left( \frac{ \epsilon^{-2+4/p}}{\log^{2/p}(1/\delta)} \right)^{p/2}\\& = \left( \frac{ \log (1/\delta)}{\log (n)} \right) \left( \frac{ \epsilon^{-(p-2)}}{\log (1/\delta)} \right)
= \frac{ \epsilon^{-(p-2)}}{\log n} \enspace .
\end{align*}
Therefore, the maximum objective value is
\begin{align*}
U_J^p \hat{d}_J &  \le  \left( \frac{ O(\epsilon^2 F_p)}{b^{p/2}\log (n)}\right)\left( \frac{1}{(ab^{p/2})^J}\right)\\
&  = \left( \frac{O(\epsilon^2) F_p}{b^{p/2}\log (n)}\right) \left(\epsilon^{p-2} \log n\right) \\
& = O\left(\epsilon^p b^{-p/2} F_p\right)\\
& = O\left(\epsilon^p F_p\right)
\end{align*}
assuming $b = \Omega (1)$.

\noindent
\emph{Case 2: $ab^{p/2} > 1$.} Then the vertex maximizing the objective value occurs at $\hat{d}_0 e_1$, that is corresponding to $j=0$. This value is
\begin{align*}
 U_{1}^p \hat{d}_0  &\le \left( \frac{ O(\epsilon^2 F_p)}{b^{p/2}\log (n)}\right) = \frac{ O(\epsilon^2 F_p)}{\log n}
\end{align*}
assuming $b = \Omega(1)$.

\noindent\emph{Case 3: $ab^{p/2}=1$.} In this case, all vertices have the same objective value, which is $\frac{ O(\epsilon^2 F_p)}{\log n}$. This proves the Lemma.
\end{proof}

\begin{lemma} [Expanded Restatement of Lemma~\ref{lem:refine:shelf:conj}.]
Let $1 \le e,g \le \lceil\log (1/\delta)\rceil$ and $l \in S_j$. Let $\log (1/\delta) = \omega (\log n)$. Assume that $\accuest(l)$ holds and  $H_j \ge \Omega(p^2 E_J)$ and $\abs{x_l} \ge \left( \frac{F_2}{E_j} \right)^{1/2}$.  Let the family $\{\omega_{lr}(i)\}_{i \in [n]}$ 
be $O(log(1/\delta))$-wise independent . Then, $\displaystyle \expect{ \left( \left(1 + \frac{Z_l}{\abs{x_l}} \right)^p-1\right)^e \left( \left(1 + \frac{\conj{Z_l}}{\abs{x_l}} \right)^p-1\right)^g\mid \H}$ is real and non-negative and is bounded above by $\displaystyle c^h \abs{x_l}^{-h} \left( \frac{F_2}{H_j} \right)^{h/2} \left( \min\left( \frac{h}{w_j}, 1 \right) \right)^{h/2}$, 
where, $h = e+g$ and $c $ is an absolute constant.
Further,
\begin{align*}
\expect{ \left( Y_l - \expect{Y_l}\right)^e \left(\conj{Y_l} - \expect{ \conj{Y_l}}\right)^g\mid \H}
\le c^h \abs{x_l}^{(p-1)h} \left( \frac{F_2}{H_j} \right)^{h/2} \left( \min\left( \frac{h}{w_j}, 1 \right) \right)^{h/2} \enspace .
\end{align*}
\end{lemma}

\emph{Note.} It suffices to condition on the conjunction of events $\G \wedge \nocollision(l) \wedge \goodest(l)$ instead of $\G \wedge \nocollision \wedge \goodest$.

\begin{proof} Following the initial part of the proof of Lemma~\ref{lem:midcentral:conj}, we have,
\begin{align}\label{eq:ref:midreg:t1a}
&\expect{ \left( \left(1 + \frac{Z_l}{\abs{x_l}} \right)^p-1\right)^e \left( \left(1 + \frac{\conj{Z_l}}{\abs{x_l}} \right)^p-1\right)^g \mid \H} \notag\\
&=\sum_{a_1 + \ldots + a_k+ \ldots = e} \sum_{\substack{b_1 + \ldots + b_k + \ldots  = g\\
\sum_{r \ge 1} r a_r = \sum_{s \ge 1} s b_s}}  c_1^{e+g} \prod_{r \ge 1} \left(\frac{c_r}{c_1}\right)^{a_r} \prod_{s \ge 1} \left(\frac{c_s}{c_1} \right)^{b_s}
 \frac{ \expect{(Z \conj{Z} )^{\sum_{r \ge 1} r a_r}\mid \H}}{ \abs{x_l}^{2\sum_{r \ge 1} r a_r}}
\end{align}

\emph{Case 1}: $h/2 \le \Theta(w_j)$. \\
Following earlier calculation,
for non-zero expectation, we have to assume that  $ \sum_{r \ge 1} r a_r  = \sum_{s \ge 1} s b_s$. Thus,  we have,
\begin{align*}
0 \le \frac{\expect{ (Z_l\conj{Z_l})^{\sum_{r \ge 1} r a_r} \mid \H}}{\abs{x_l}^{2\sum_{r \ge 1} r a_r}}  &\le \frac{ \expect{ (Z_l\conj{Z_l})^{h/2} \mid \H}}{\abs{x_l}^{h}} \left( \frac{ \abs{Z_l}}{\abs{x_l}} \right)^{ 2(\sum_{r \ge 1} r a_r) - h}\\
&\le \abs{x_l}^{-h}\left( \frac{h F_2}{\Theta(w_j) H_j } \right)^{h/2} \varrho^{2(\sum_{r \ge 1} r a_r) - h}
\end{align*}
Again, following analogous calculation in Lemma~\ref{lem:midcentral:conj},
\begin{align}
\expect{ \left( \left(1 + \frac{Z_l}{\abs{x_l}} \right)^p-1\right)^e \left( \left(1 + \frac{\conj{Z_l}}{\abs{x_l}} \right)^p-1\right)^g\mid \H}  \le \abs{x_l}^{-h} c'^h \left( \frac{h F_2}{O(w_j) H_j} \right)^{h/2}
\end{align}
for some constant $c'$.

\emph{Case 2}: $h/2 > \Theta(w_j)$. Let $\Theta(w_j) = w'_j$.
\begin{align} \label{eq:ref:midcentral:a}
\frac{\expect{ (Z_l\conj{Z_l})^{\sum_{r \ge 1} r a_r} \mid \H}}{\abs{x_l}^{2\sum_{r \ge 1} r a_r}}  &
\le \frac{ \expect{ (Z_l \conj{Z_l})^{w'_j} \mid H}}{ \abs{x_l}^{2w'_j}} \varrho^{2(\sum_{r} ra_r) -h} \left( \frac{\abs{Z_l}}{\abs{x_l}}\right)^{h - 2w'_j} \varrho^{2(\sum_{r} ra_r) -h}
\end{align}
Now $ \expect{(Z_l \conj{Z_l})^{w'_j} \mid \H}\le  \left( \frac{w'_jF_2}{w'_j H_j} \right)^{w'_j}$, and by \accuest$(l)$, $\abs{Z_l} \le \left( \frac{ F_2}{H_j} \right)^{1/2}$. Substituting in Eqn.~\eqref{eq:ref:midcentral:a}, we have,
\begin{align*}
\frac{\expect{ (Z_l\conj{Z_l})^{\sum_{r \ge 1} r a_r}\mid \H}}{\abs{x_l}^{2\sum_{r \ge 1} r a_r}}  &
\le \abs{x_l}^{-h_j} \left( \frac{F_2}{H_j} \right)^{h_j/2} \varrho^{2\sum_{r} ra_r - h}
\end{align*}
From here, the proof may proceed along the lines of Lemma~\ref{lem:midcentral:conj}. This yields,
\begin{align*}
\expect{ \left( \left(1 + \frac{Z_l}{\abs{x_l}} \right)^p-1\right)^e \left( \left(1 + \frac{\conj{Z_l}}{\abs{x_l}} \right)^p-1\right)^g \mid \H}
\le \abs{x_l}^{-h} c'^h \left( \frac{F_2}{ H_j} \right)^{h/2}
\end{align*}
for some constant $c'$.

Combining cases 1 and 2, we have in general that,
\begin{align*}
\expect{ \left( \left(1 + \frac{Z_l}{\abs{x_l}} \right)^p-1\right)^e \left( \left(1 + \frac{\conj{Z_l}}{\abs{x_l}} \right)^p-1\right)^g \mid \H}
\le \abs{x_l}^{-h} c'^h \left( \frac{F_2}{ H_j} \right)^{h/2} \min\left( \frac{h}{w_j},1 \right)^{h/2} \enspace .
\end{align*}
The remainder of the proof proceeds identically along the lines of the proof of Lemma~\ref{lem:midcentral:conj}.

\end{proof}

\eat{\begin{lemma} [Restatement of second part of \label{lem:dmom:shelf}].
Let $0 \le d_1, d_2 \le \lceil \log (1/\delta) \rceil$ and integral. Then, for some constants $c_1,c_2 > 0$
\begin{align*}
\mathbf{E}\biggl[ & \biggl( \sum_{i \in \cup_{j=0}^{J-1} (S_j \cap \nocollision)} (Y_i - \expect{Y_i}) \biggr)^{d_1} \biggl( \sum_{i \in \cup_{j=0}^{J-1} (S_j \cap \nocollision)} (\conj{Y_i} - \conj{\expect{Y_i}}) \biggr)^{d_2}\biggr]  \\
& \hspace*{3.0in}\le \left( c_1\epsilon F_p \right)^{d_1 + d_2} \enspace . \\
\mathbf{E}\biggl[ & \biggl( \sum_{i \in S_J} (Y_i - \expect{Y_i}) \biggr)^{d_1} \biggl( \sum_{i \in S_J} (\conj{Y_i} - \conj{\expect{Y_i}}) \biggr)^{d_2}\biggr] \le \left( c_2\epsilon F_p \right)^{d_1 + d_2}
\end{align*}
\end{lemma}
}

We will decompose the two statements of Lemma~\ref{lem:dmom:shelf} into two lemmas, one corresponding to the contribution to the $2d$th central moment from all the shelves except the outermost shelf, and the second corresponding to the contribution to the same from the outermost shelf. We prove them separately and then combine the results.  Let $S' = S_1 \cup \ldots \cup S_{J-1}$.

\begin{lemma}\label{lem:dmom:shelf1}
Let $0 \le d_1, d_2 \le \log (1/\delta)$ and integral. Then, there exists an absolute constant $c$ such that
\begin{equation*}
\expect{ \left( \sum_{i \in S'} (Y_i - \expect{Y_i}) \right)^{d_1} \left( \sum_{i \in  S'} (\conj{Y_i} - \conj{\expect{Y_i}}) \right)^{d_2} \mid \H} \le \left( c\epsilon F_p \right)^{d_1 + d_2} \enspace .
\end{equation*}
\end{lemma}

\begin{lemma} \label{lem:dcentmom:shelf:largest}
Let $0 \le d_1, d_2 \le \log (1/\delta)$ and integral. Then, there exists an absolute constant $c$ such that
\begin{align*}
\expect{ \left( \sum_{i \in S_J} (Y_i - \expect{Y_i}) \right)^{d_1} \left( \sum_{i \in S_J} (\conj{Y_i} - \conj{\expect{Y_i}}) \right)^{d_2} \mid \H} \le \left( c\epsilon F_p \right)^{d_1 + d_2} \enspace .
\end{align*}
\end{lemma}

\eat{We can now apply the above lemmas.
\begin{lemma} \label{lem:dcentmomsmalldelta}
Let $\log (1/\delta) = \omega(\log n)$, $C = \Theta( n^{1-2/p} \epsilon^{-4/p} \log^{2/p}(1/\delta))$ and $H_J = \Theta(n^{1-2/p} \epsilon^{-2}$.
Then, for $d = \lceil \log (1/\delta)$, there exists a constant $c>0$ such that
\begin{align*}
\mathbf{E} \biggl[ \biggl(  \sum_{i \in  \G_1} (Y_i - \expect{Y_i}) \biggr)^d \biggl( \sum_{i \in  \G_1} ( \conj{Y_i} - \expect{\conj{Y_i}})\biggr)^d \biggr]  \le \left( c \epsilon F_p \right)^{2d}
\end{align*}
Setting $d = \lceil \log (1/\delta) \rceil$, we obtain $\prob{ \card{ \hat{F}_p - F_p} \le (\epsilon F_p)/2} \le  \delta/2 \enspace.$
\end{lemma}

\begin{lemma}[Restatement of  Lemma~\ref{lem:dmom:shelf1}.]
Let $0 \le d_1, d_2 \le \log (1/\delta)$.
\begin{align*}
\expect{ \left( \sum_{i \in S'} (Y_i - \expect{Y_i}) \right)^{d_1} \left( \sum_{i \in S'} (\conj{Y_i} - \conj{\expect{Y_i}}) \right)^{d_2} \mid \G } \le \left( c\epsilon F_p \right)^{d_1 + d_2}
\end{align*}
for some absolute constant $c$.
\end{lemma}
}

\begin{proof} [Proof of Lemma~\ref{lem:dmom:shelf1}.] Consider the case when say $d_1=0$ and $d_2 > 0$. Then, the expression in the expectation is
\begin{align*} &\expect{ \left( \sum_{i \in S'} (\conj{Y_i} - \conj{\expect{Y_i}}) \right)^{d_2} \mid \H} \\
& = \sum_{q=1}^{d_2} \sum_{e_1 + \ldots + e_q = d_2} \sum_{\{i_1, \ldots, i_q\} \subset S'} \prod_{u=1}^r \expect{ (\conj{Y_i} - \conj{\expect{Y_i}})^{e_u} \mid \H} \\
& = \sum_{q=1}^{d_2} \sum_{e_1 + \ldots + e_q = d_2} \sum_{\{i_1, \ldots, i_q\} \subset S'} \prod_{u=1}^r \abs{x_{i_u}}^{e_u} n^{-\Omega(e_u)} \\
& \le F_p^{d_2} n^{-\Omega(d_2)}
\end{align*}
which proves the claim. The proof for $d_2=0$ and $d_1 > 0$  is analogous.

The \emph{LHS} of the expression in the lemma may be written as follows. This is the same expression as in Eqn.~\eqref{eq:dcentmom:midG0:2:t1} in the proof of Lemma~\ref{lem:dcentmom:midG0:2}. Let $\alpha_0(q) = \max(0, 1-q)$. Using \nocollision~and the argument in Lemma~\ref{lem:dcentmom:midG0:2}
\begin{align} \label{eq:dcentmom:shelf:t1}
&\expect{ \left( \sum_{i \in S'} (Y_i - \expect{Y_i}) \right)^{d_1}
\left( \sum_{i \in S'} (\conj{Y_i} - \conj{\expect{Y_i}}) \right)^{d_2} \mid \H} \notag \\
& = \sum_{q=1}^{\min(d_1, d_2)} \sum_{s=\alpha_0(q)}^{d_1-q} \sum_{t=\alpha_0(q)}^{d_2-q}   \sum_{\substack{e_1 + \ldots + e_{q+s} =d_1 \\ e_1, \ldots, e_{q+s}\ge 1}}
\sum_{\substack{g_1 + \ldots + g_{q+s}=d_2 \\ g_1, \ldots, g_{q+t} \ge 1}}
\sum_{\{i_1, \ldots, i_q, j_1, \ldots, j_{s}, k_1, \ldots, k_t\} \subset [n]} \notag\\
&\hspace*{0.5in}\prod_{r=1}^q
\expect{  (Y_{i_r} - \expect{Y_{i_r}})^{e_r} \left(\conj{Y_{i_r}} - \expect{ \conj{Y_{i_r}}}\right)^{g_r}\mid \H}
\prod_{l=1}^s
\expect{ (Y_{j_l} - \expect{Y_{j_l}})^{e_{q+l}}\mid \H}  \notag \\
&\hspace*{2.0in}\prod_{m=1}^t \expect{ (Y_{k_m} - \expect{Y_{k_m}})^{e_{q+m}} \mid \H} \enspace .
\end{align}

Analogous to the definitions of $P_1$ and $P_2$ in Eqn.~\eqref{eq:dcentmom:midG0:P1}, define
\begin{align} \label{eq:dcentmom:shelf:alpha}
\alpha(s_1, s_2) &= \sum_{\substack{q=1\\ q=0 \text{iff} s_1=0}}^{\min(s_1, s_2)} \sum_{e_1 + \ldots + e_q = s_1} \sum_{g_1 + \ldots + g_q = s_2}\binom{ s_1}{e_1, \ldots, e_q} \binom{s_2}{g_1, \ldots, g_q}\notag \\
 & \hspace*{1.0in}\sum_{\{i_1, \ldots, i_q\}} \prod_{r=1}^q
\expect{  (Y_{i_r} - \expect{Y_{i_r} \mid \H})^{e_r} \left(\conj{Y_{i_r}} - \expect{ \conj{Y_{i_r}} \mid \H}\right)^{g_r}\mid \H} \enspace . \\
\beta(d_1,s_1) & = \sum_{s = \alpha_0(s_1)}^{d_1 - s_1} \sum_{e_{q+1} + \ldots + e_{q+s} = d_1-s_1} \binom{ d_1-s_1}{e_{q+1}, \ldots, e_{q+s}} \sum_{\{j_1, \ldots, j_s\}} \prod_{l=1}^s \abs{x_{j_l}}^{pe_{q+l}} n^{-\Omega(e_{q+l})} \notag\\
& = \left(n^{-\Omega(1)} F_p \right)^{d_1-s_1} \enspace . \label{eq:dcentmom:shelf:beta}
\end{align}
Following the calculations of Lemma~\ref{lem:dcentmom:midG0:2}~, the expression in Eqn.~\eqref{eq:dcentmom:shelf:t1} is bounded above as
\begin{multline}\label{eq:dcentmom:shelf:t2}
\sum_{s_1=0}^{d_1} \sum_{s_2=0}^{d_2} \binom{d_1}{s_1} \binom{d_2}{s_2} \alpha(s_1, s_2) \beta (d_1, s_1) \beta(d_2, s_2) \\
  \le  \sum_{s_1=0}^{d_1} \sum_{s_2=0}^{d_2} \binom{d_1}{s_1} \binom{d_2}{s_2} \alpha(s_1, s_2) \left(n^{-\Omega(1)} F_p \right)^{d_1+d_2-s_1+s_2} \enspace .
\end{multline}
We now make the following claim.
\begin{claim} \label{lem:dcentmom:shelf:claim}
$\alpha(s_1, s_2) \le (c' \epsilon F_p)^{s_1 + s_2}$ for some constant $c'$.
\end{claim}
Using this claim, we have from Eqn.~\eqref{eq:dcentmom:shelf:t2} that this expression is bounded above by
\begin{align*}
\sum_{s_1=0}^{d_1} \sum_{s_2=0}^{d_2} \binom{d_1}{s_1} \binom{d_2}{s_2} (c'\epsilon F_p)^{s_1 + s_2} \left(n^{-\Omega(1)} F_p \right)^{d_1+d_2-s_1+s_2} = \left((c'\epsilon + n^{-\Omega(1)}) F_p \right)^{d_1 + d_2} \enspace .
\end{align*}
thereby proving the statement of the Lemma.

For the remainder of the proof, it will suffice to prove
 Claim \ref{lem:dcentmom:shelf:claim}. Let $s_1$ and $s_2$ be each non-zero, otherwise, the sum is vacuous.

For $i_u \in S_1 \cup S_2 \cup \ldots \cup S_{J-1}$, let $j(i_u)$ denote the index $j$ of the shelf $S_j$ such that $i \in S_j$. From Lemma~\ref{lem:refine:shelf:conj}, we have
\begin{align} \label{lem:dcentmom:shelf:t3}
\alpha(s_1,s_2) &= \sum_{q=1}^{\min(s_1, s_2)} \sum_{\substack{e_1 + \ldots + e_q = s_1 \\ e_j's \ge 1}} \sum_{\substack{g_1 + \ldots + g_q = s_2\\ g_j's \ge 1}}\binom{ s_1}{e_1, \ldots, e_q} \binom{s_2}{g_1, \ldots, g_q}\notag \\
 & \sum_{\{i_1, \ldots, i_q\}} \prod_{r=1}^q
\expect{  (Y_{i_r} - \expect{Y_{i_r} \mid \H})^{e_r} \left(\conj{Y_{i_r}} - \expect{ \conj{Y_{i_r}} \mid \H}\right)^{g_r}\mid \H} \notag  \\
& \le c_0^{s_1+s_2} \sum_{q=1}^{\min(s_1, s_2)} \sum_{\substack{e_1 + \ldots + e_q = s_1 \\ e_j's \ge 1}} \sum_{\substack{g_1 + \ldots + g_q = s_2\\ g_j's \ge 1}}\binom{ s_1}{e_1, \ldots, e_q} \binom{s_2}{g_1, \ldots, g_q}\notag \\
 & \hspace*{0.5in}\sum_{\{i_1, \ldots, i_q\}} \prod_{u=1}^q \abs{x_{i_u}}^{(p-1)h_u} \left( \frac{F_2}{H_{j(i_u)}} \right)^{h_u/2} \left( \min\left(\frac{h_u}{w_{j(i_u)}}, 1\right)\right)^{h_u/2},
\end{align}
where, $h_u = e_u + g_u$. 

\noindent By definition of multinomial coefficients, we have,
\begin{align} \label{lem:dcentmom:shelf:t4}
\alpha(s_1, s_2) & \le \sum_{q=1}^{\min(s_1, s_2)} \sum_{\substack{h_1 + \ldots + h_q = s_1+s_2 \\ h_1, \ldots, h_q \ge 2}} \binom{ s_1 + s_2} { h_1, \ldots, h_q} \notag  \\
& \hspace*{0.3in}\sum_{\{i_1, \ldots, i_q\}} \prod_{u=1}^q \abs{x_{i_u}}^{(p-1)h_u} \left( \frac{F_2}{H_{j(i_u)}} \right)^{h_u/2} \left( \min\left(\frac{h_u}{w_{j(i_u)}}, 1\right)\right)^{h_uc_0^{s_1+s_2}/2} \enspace .
\end{align}
For $i_u \in s_{j(i_u)}$, where, $j(i_u) \in \{0, 1, \ldots, J-1\}$, it follows that $\abs{x_{i_u}} = \Theta ((F_2/H_{j(i_u)})^{1/2})$. Therefore, for $\beta = p/2-1$, we have,
\begin{align*}
\abs{x_{i_u}}^{(p-1) h_u} \left( \frac{ F_2}{H_{j(i_u)}} \right)^{h_u/2} \le c_1^{h_u} \abs{x_{i_u}}^{(p/2) h_u} \left( \frac{ F_2}{H_{j(i_u)}} \right)^{p h_u/4}
\end{align*}
for some constant $c_1$.
Further, by the definition of $H_{j} = H_0 b^j$, we have,
\begin{align} \label{eq:Hju}
\left( \frac{ F_2}{H_{j(i_u)}} \right)^{p h_u/4} =
\left( \frac{F_2}{H_0} \right)^{ph_u/4} b^{-j(i_u) p h_u/4} \enspace .
\end{align}
Now, for any $j \in \{0,1, \ldots, J\}$, $$b^{-j} = b^{-J (j/J)} = \left( \frac{H_0}{H_J} \right)^{j/J} = \left( \epsilon^{2-4/p}  \log^{2/p}(1/\delta)\right)^{j/J} \le \left( \log (1/\delta)\right)^{(2/p)(j/J)} \enspace . $$
Substituting this in Eqn.~\eqref{eq:Hju} and then in Eqn.~\eqref{lem:dcentmom:shelf:t4}, we obtain,
\begin{align*}
\left( \frac{ F_2}{H_{j(i_u)}} \right)^{p h_u/4} \le
\left( \frac{F_2}{H_0} \right)^{ph_u/4}  \left( \log (1/\delta)\right)^{jh_u/(2J)} \enspace .
\end{align*}
Further, 
\begin{align}\label{eq:lem:dcentmom:shelf:t5}
\alpha(s_1, s_2) &
  \le c_0^{s_1+s_2}\sum_{q=1}^{\min(s_1, s_2)} \sum_{\substack{h_1 + \ldots + h_q = s_1+s_2 \\ h_1, \ldots, h_q \ge 2}} \binom{ s_1 + s_2} { h_1, \ldots, h_q}  \sum_{\{i_1, \ldots, i_q\}} \prod_{u=1}^q c_1^{h_u} \notag \\
   & \hspace*{0.5in}\abs{x_{i_u}}^{(p/2) h_u} \left( \frac{ F_2}{H_0} \right)^{p h_u/4} \left( \log (1/\delta)\right)^{(j(i_u)h_u/(2J)}  \left( \min\left(\frac{h_u}{w_{j(i_u)}}, 1\right)\right)^{h_u/2}  \notag\\
& = (c_0c_1)^{s_1+s_2}  \left( \frac{ F_2}{H_0} \right)^{p (s_1 + s_2)/4} \sum_{q=1}^{\min(s_1, s_2)} \sum_{\substack{h_1 + \ldots + h_q = s_1+s_2 \\ h_1, \ldots, h_q \ge 2}} \binom{ s_1 + s_2} { h_1, \ldots, h_q} \notag  \\
& \hspace*{0.3in}\sum_{\{i_1, \ldots, i_q\}} \prod_{u=1}^q \abs{x_{i_u}}^{(p/2) h_u} \left( \log (1/\delta)\right)^{(j(i_u)h_u/(2J)} \left( \min\left(\frac{h_u}{w_{j(i_u)}}, 1\right)\right)^{h_u/2} \enspace .
\end{align}

Consider the product $(\log (1/\delta))^{j(i_u)/J} \min((h_u/w_{j(i_u)}),1)$ and suppose that we wish to maximize it as a function of $j(i_u)$. First suppose that $h_u \ge w_{j(i_u)}$. Then, the product is simply $(\log (1/\delta))^{j(i_u)/J} $ and increases with increasing $j(i_u)$. Thus, in this case, the maximum of the product is $\log (1/\delta)$. Now suppose $h_u < w_{j(i_u)}$. For simplicity, let $j_u$ denote $j(i_u)$.
Incrementing $j_u$ by 1, and noting that $w_{j+1}/w_j = a$  and $a^J O(\log (n)) = \log (1/\delta)$, consider  the ratio
\begin{align*}
\frac{(\log (1/\delta))^{(j_u+1)/J} (h_u/ w_{j_u+1})}{(\log (1/\delta))^{j_u/J} (h_u/ w_{j_u})} = (\log(1/\delta))^{1/J} a^{-1} = (O(\log n))^{1/J}  \enspace .
\end{align*}
Hence, the ratio increases by incrementing $j(i_u)$. Hence, in both cases, the maximum is obtained when $j{i_u} = J-1$, that is,  $w_{j(i_u)} = w_{J-1}  = O(\log (1/\delta)$. Substituting in Eqn.~\eqref{eq:lem:dcentmom:shelf:t5}, we have,
\begin{align}\label{eq:lem:dcentmom:shelf:t6}
&\alpha(s_1, s_2)\\
&  \le
(c_0c_1)^{s_1+s_2}  \left( \frac{ F_2}{H_0} \right)^{p (s_1 + s_2)/4} \sum_{q=1}^{\min(s_1, s_2)} \sum_{\substack{h_1 + \ldots + h_q = s_1+s_2 \\ h_1, \ldots, h_q \ge 2}} \binom{ s_1 + s_2} { h_1, \ldots, h_q} \notag  \\
& \hspace*{1.0in}\sum_{\{i_1, \ldots, i_q\}} \prod_{u=1}^q \abs{x_{i_u}}^{(p/2) h_u} \left( \log (1/\delta)\right)^{(J-1/J)h_u/2} \left(\frac{h_u}{O(\log (1/\delta))} \right)^{h_u/2} \notag \\
& \le (c_0c_1c_2)^{s_1+s_2}  \left( \frac{ F_2}{H_0} \right)^{p (s_1 + s_2)/4} \sum_{q=1}^{\min(s_1, s_2)} \sum_{\substack{h_1 + \ldots + h_q = s_1+s_2 \\ h_1, \ldots, h_q \ge 2}} \binom{ s_1 + s_2} { h_1, \ldots, h_q} \prod_{u=1}^q h_u^{h_u/2} \notag \\
& \hspace*{2.0in}\sum_{\{i_1, \ldots, i_q\}} \prod_{u=1}^q \abs{x_{i_u}}^{(p/2) h_u} \enspace .
\end{align}
Now
$$
 \left( \frac{F_2}{H_0} \right)^{p/2} = \left( \frac{\epsilon^{4/p} F_p^{2/p}}{\log^{2/p}(1/\delta)} \right)^{p/2} = \frac{ \epsilon^2 F_p}{\log (1/\delta)} \enspace. $$
 Therefore, $\displaystyle \left( \frac{ F_2}{H_0} \right)^{p (s_1 + s_2)/4}  \le \left( \frac{ \epsilon^2 F_p}{\log (1/\delta)} \right)^{(s_1+s_2)/2}$.\\
Further, since, each $h_u \ge 2$, we have,
\begin{align*}
\binom{s_1 + s_2}{h_1, \ldots, h_q} \prod_{u=1}^q h_u^{h_u/2}
&= \frac{(s_1+s_2)!}{\prod_{u=1}^q h_u!} \prod_{u=1}^q h_u^{h_u/2} \\& \le c_3^{s_1+s_2} (s_1+s_2)^{(s_1+s_2)/2} \binom{s_1 + s_2}{h_1/2, \ldots, h_q/2} \enspace .
\end{align*}
for some constant $c_3$.
Therefore, Eqn.~\eqref{eq:lem:dcentmom:shelf:t6} can be written as
\begin{align}
\alpha(s_1, s_2) &\le
(c_0 c_1 c_2 c_3)^{s_1+s_2} \left( \frac{ \epsilon^2 F_p}{\log (1/\delta)} \right)^{(s_1+s_2)/2} (s_1+s_2)^{(s_1+s_2)/2} \notag \\
\sum_{q=1}^{\min(s_1, s_2)} \sum_{\substack{h_1 + \ldots + h_q = s_1+s_2 \\ h_1, \ldots, h_q \ge 2}} \binom{(s_1 + s_2)/2}{h_1/2, \ldots, h_q/2}\sum_{\{i_1, \ldots, i_q\}} \\ & \hspace*{0.5in} \prod_{u=1}^q \abs{x_{i_u}}^{(p/2) h_u} \label{eq:lem:dcentmom:shelf:t7}.
\end{align}
Following the arguments in the proof of Lemma~\ref{lem:comb:PSF2p2}, we can show that
\begin{gather*}
\sum_{q=1}^{\min(s_1, s_2)} \sum_{\substack{h_1 + \ldots + h_q = s_1+s_2 \\ h_1, \ldots, h_q \ge 2}} \binom{(s_1 + s_2)/2}{h_1/2, \ldots, h_q/2}\sum_{\{i_1, \ldots, i_q\}} \prod_{u=1}^q \abs{x_{i_u}}^{(p/2) h_u} \le  \left(c_4 F_p\right)^{(s_1+s_2)/2}  \enspace .
\end{gather*}
for some absolute constant $c_4$.\\
Substituting in Eqn.~\eqref{eq:lem:dcentmom:shelf:t6}, we have,
\begin{align*}
\alpha(s_1, s_2) &\le C'^{s_1+s_2} \left( \frac{ \epsilon^2 F_p}{\log (1/\delta)} \right)^{(s_1+s_2)/2} (s_1+s_2)^{(s_1+s_2)/2} \left(c_4 F_p\right)^{(s_1+s_2)/2} \notag \\
& \le \left( C' \epsilon F_p \right)^{s_1+s_2}
\end{align*}
assuming $ s_1, s_2$ are each $O(\log (1/\delta))$, for some absolute constant $C'$.

\end{proof}

\begin{lemma} [Restatement of Lemma~\ref{lem:dcentmom:shelf:largest}.]
Let $0 \le d_1, d_2 \le \log (1/\delta)$ and integral. Then,
\begin{align*}
\expect{ \left( \sum_{i \in S_J} (Y_i - \expect{Y_i}) \right)^{d_1} \left( \sum_{i \in S_J} (\conj{Y_i} - \conj{\expect{Y_i}}) \right)^{d_2} \mid \H} \le \left( c\epsilon F_p \right)^{d_1 + d_2}  \enspace .
\end{align*}
for some absolute constant $c$.
\end{lemma}

\begin{proof} [Proof of Lemma~\ref{lem:dcentmom:shelf:largest}.]
If $d_1 =0$ and $d_2 > 0$ or vice-versa, then, as proved in the initial part of the proof of Lemma~\ref{lem:refine:shelf:conj}, the expectation in the statement of the lemma is bounded above as $(F_p n^{-\Omega(1)})^{d_1 + d_2}$, thereby satisfying the statement of the lemma.

So now we assume that both $d_1, d_2 > 0$. Continuing as in the proof of Lemma~\ref{lem:refine:shelf:conj}, let us  define $ \alpha(s_1, s_2)$, for $1 \le s_1 \le d_1$ and $1 \le s_2 \le d_2$ as in Eqn.~\eqref{eq:dcentmom:shelf:alpha}, which is  reproduced below.
\begin{align*}
\alpha(s_1, s_2) &= \sum_{\substack{q=1\\ q=0 \text{iff} s_1=0}}^{\min(s_1, s_2)} \sum_{e_1 + \ldots + e_q = s_1} \sum_{g_1 + \ldots + g_q = s_2}\binom{ s_1}{e_1, \ldots, e_q} \binom{s_2}{g_1, \ldots, g_q}\notag \\
 & \hspace*{1.0in}\sum_{\{i_1, \ldots, i_q\} \subset S_J} \prod_{r=1}^q
\expect{  (Y_{i_r} - \expect{Y_{i_r}\mid \H})^{e_r} \left(\conj{Y_{i_r}} - \expect{ \conj{Y_{i_r}}\mid \H}\right)^{g_r} \mid \H}
\end{align*}
Following  the arguments in the proof of Lemma~\ref{lem:refine:shelf:conj}, to prove the given lemma it suffices to show that $\alpha(s_1,s_2)  \le (c'\epsilon F_p)^{s_1 + s_2}$, for some constant $c'$.

\begin{align*}&\alpha(s_1,s_2) \\
&= \sum_{q=1}^{\min(s_1, s_2)} \sum_{\substack{e_1 + \ldots + e_q = s_1 \\ e_j's \ge 1}} \sum_{\substack{g_1 + \ldots + g_q = s_2\\ g_j's \ge 1}}\binom{ s_1}{e_1, \ldots, e_q} \binom{s_2}{g_1, \ldots, g_q}\notag \\
 & \hspace*{1.0in}\sum_{\{i_1, \ldots, i_q\}\subset S_J} \prod_{r=1}^q
\expect{  (Y_{i_r} - \expect{Y_{i_r} \mid \H})^{e_r} \left(\conj{Y_{i_r}} - \expect{ \conj{Y_{i_r}} \mid \H}\right)^{g_r} \mid \H} \notag  \\
& \le c_0^{s_1+s_2} \sum_{q=1}^{\min(s_1, s_2)} \sum_{\substack{e_1 + \ldots + e_q = s_1 \\ e_j's \ge 1}} \sum_{\substack{g_1 + \ldots + g_q = s_2\\ g_j's \ge 1}}\binom{ s_1}{e_1, \ldots, e_q} \binom{s_2}{g_1, \ldots, g_q}\notag \\
 & \hspace*{1.0in}\sum_{\{i_1, \ldots, i_q\} \subset S_J} \prod_{u=1}^q \abs{x_{i_u}}^{(p-1)h_u} \left( \frac{F_2}{H_{J}} \right)^{h_u/2} \left(\frac{h_u}{w_J}\right)^{h_u/2} , \\ &\hspace*{3.0in} \{ h_u = e_u + g_u, u \in [q]\}\\
& \le c_0^{s_1+s_2} \left( \frac{F_2}{H_J} \right)^{(s_1+s_2)/2} (c_2 \log (1/\delta))^{-(s_1+s_2)/2} \sum_{q=1}^{\min(s_1, s_2)} \\
&\sum_{\substack{ h_1 + \ldots + h_q = s_1+s_2\\ h_j's \ge 2}} \binom{ s_1 + s_2}{h_1, \ldots, h_q} \prod_{u=1}^r h_u^{h_u/2} \sum_{\{i_1, \ldots, i_q\} \subset S_J} \prod_{u=1}^r \abs{x_{i_u}}^{(p-1)h_u} \enspace .
\end{align*}
Following the arguments in the proof of Lemma~\ref{lem:comb:PSF2p2}, we have
\begin{align*}
& \sum_{q=1}^{\min(s_1, s_2)} \sum_{\substack{ h_1 + \ldots + h_q = s_1+s_2\\ h_j's \ge 2}} \binom{ s_1 + s_2}{h_1, \ldots, h_q} \prod_{u=1}^r h_u^{h_u/2} \sum_{\{i_1, \ldots, i_q\} \subset S_J} \prod_{u=1}^r \abs{x_{i_u}}^{(p-1)h_u} \\
& \le (s_1 +s_2)^{(s_1 + s_2)/2} \sum_{q=1}^{\min(s_1, s_2)} \sum_{\substack{ h_1 + \ldots + h_q = s_1+s_2\\ h_j's \ge 2}} \binom{ (s_1+s_2)/2} { h_1/2, \ldots, h_q/2} \\
& \hspace*{1.5in}\sum_{\{i_1, \ldots, i_q\} \subset S_J} \prod_{u=1}^r \abs{x_{i_u}}^{(2p-2)h_u/2} \\
& \le (s_1 +s_2)^{(s_1 + s_2)/2}c_3^{s_1 + s_2} F_{2p-2}^{(s_1+s_2)/2}  \enspace .
\end{align*}
Substituting in the expression for $\alpha(s_1, s_2)$, we obtain,
\begin{align*}
\alpha(s_1, s_2) & \le (c_0 c_2 c_3)^{s_1+s_2} \left( \frac{F_2}{H_J} \right)^{(s_1+s_2)/2} ((s_1 + s_2)/\log (1/\delta))^{(s_1+s_2)/2} F_{2p-2}^{(s_1 + s_2)/2}  \enspace .
\end{align*}
Now $(s_1 + s_2)/\log (1/\delta)) \le 2$. Further, $$\displaystyle \left( \frac{F_2}{H_J} \right)^{(s_1+s_2)/2} \le \left( \epsilon^{2} F_p^{2/p} \right)^{(s_1+s_2)/2} = \epsilon^{(s_1+s_2)} F_p^{(s_1+s_2)/p} \enspace. $$
Also, $$F_{2p-2}^{(s_1+s_2)/2} \le F_p^{(2-2/p)(s_1+s_2)/2} \le F_p^{(1-1/p)(s_1 + s_2)}\enspace. $$ Substituting these simplifications, we have,
\begin{align*}
\alpha(s_1, s_2) & \le (2c_0 c_2 c_3)^{s_1+s_2}\epsilon^{(s_1+s_2)} F_p^{(s_1+s_2)/p}F_p^{(1-1/p)(s_1 + s_2)} = \left( c' \epsilon F_p \right)^{s_1 + s_2}
\end{align*}
This implies the statement of the lemma, as argued earlier.

\end{proof}

\begin{lemma} [Restatement of Lemma~\ref{lem:dcentmomsmalldelta}.]
Let $\log (1/\delta) = \omega(\log n)$.
Let $C \ge L n^{1-2/p} \epsilon^{-4/p} \allowbreak \log^{2/p}(1/\delta)$ and $H_J = L' n^{1-2/p} \epsilon^{-2}$, for suitable constants $L,L'$.
Then, for $d = \lceil \log (1/\delta) \rceil$,
\begin{align*}
\expect{ \left( \sum_{i \in [n]} (Y_i - \expect{Y_i}) \right)^d \left( \sum_{i \in [n]} ( \conj{Y_i} - \expect{\conj{Y_i}})\right)^d \mid \H } \le \left( c \epsilon F_p \right)^{2d}
\end{align*}
for some constant $c$.
\end{lemma}
\begin{proof} The proof follows the lines of the proof of Lemma~\ref{lem:dcentmom} with the modification that we use Lemmas~\ref{lem:dcentmom:shelf:largest}, ~\ref{lem:dmom:shelf} and ~\ref{lem:dmf2} instead. We sketch an outline below.
Let $G' = \lmargin(G_0) \cup_{l=1}^L G_l$. Let $S'= S_1 \cup \ldots \cup S_{J-1}$.
\begin{align*}
\left( \sum_{i \in [n] } (Y_i - \expect{Y_i \mid \H}  \right)^d &= \sum_{d_1+d_2 + d_3 = d} \binom{d}{d_1 d_2 d_3} \left( \sum_{i \in G' } ( Y_i - \expect{Y_i \mid \H}\right)^{d_1} \\ & \hspace*{0.2in}\left( \sum_{i \in S'} (Y_i - \expect{Y_i\mid \H}) \right)^{d_2}  \left( \sum_{i \in S_J} (Y_i - \expect{Y_i \mid \H} \right)^{d_3}  \enspace .
\end{align*}
Similarly, we expand the conjugate expression $\left( \sum_{i \in [n] \cap \G_1} (\conj{Y_i} - \expect{\conj{Y_i} \mid \H} ) \right)^d$.
\begin{align*}
\left( \sum_{i \in [n] } (\conj{Y_i} - \expect{\conj{Y_i}\mid \H})  \right)^d
&= \sum_{d'_1+d'_2 + d'_3 = d} \binom{d}{d'_1  d'_2 d'_3} \left( \sum_{i \in G'} ( \conj{Y_i} - \expect{\conj{Y_i} \mid \H})\right)^{d'_1}\\ & \hspace*{0.2in}  \left( \sum_{i \in S'} (\conj{Y_i} - \expect{\conj{Y_i} \mid \H}) \right)^{d'_2}  \left( \sum_{i \in S_J} (\conj{Y_i} - \expect{\conj{Y_i} \mid \H}) \right)^{d'_3}  \enspace .
\end{align*}
Taking the product  and then its expectation, we obtain,
\begin{align*}
&\expect{\left( \sum_{i \in [n] \cap \G_1} (Y_i - \expect{Y_i \mid \H})  \right)^d
\left( \sum_{i \in [n] \cap \H_1} (\conj{Y_i} - \expect{\conj{Y_i} \mid \H})  \right)^d \mid \H } \\
& = \sum_{d_1+d_2 + d_3 = d}  \sum_{d'_1+d'_2 + d'_3 = d} \binom{d}{d_1 d_2 d_3}  \binom{d}{d'_1  d'_2 d'_3}\\
& \hspace*{0.1in} \cdot \expect{ \left( \sum_{i \in G'} ( Y_i - \expect{Y_i \mid \H})\right)^{d_1}\left( \sum_{i \in G'} ( \conj{Y_i} - \expect{\conj{Y_i} \mid \H})\right)^{d'_1} \mid \H} \\
&\hspace*{0.1in} \cdot \expect{ \left( \sum_{i \in S'} (Y_i - \expect{Y_i \mid \H}) \right)^{d_2} \left( \sum_{i \in S'} (\conj{Y_i} - \expect{\conj{Y_i} \mid \H}) \right)^{d'_2} \mid \H} \\
&\hspace*{0.1in} \cdot  \expect{ \left( \sum_{i \in S_J} (Y_i - \expect{Y_i \mid \H}) \right)^{d_3} \left( \sum_{i \in S_J} (\conj{Y_i} - \expect{\conj{Y_i} \mid \H}) \right)^{d'_3} \mid \H}
\end{align*}
By Lemmas~\ref{lem:dcentmom:shelf:largest}, ~\ref{lem:dmom:shelf} and ~\ref{lem:dmf2}, there is some constant $c'$ such that  the above expression is  upper-bounded by
\begin{align*}
& \le \sum_{d_1+d_2 + d_3 = d}  \sum_{d'_1+d'_2 + d'_3 = d} \binom{d}{d_1 d_2 d_3}  \binom{d}{d'_1  d'_2 d'_3}(c' \epsilon F_p)^{d_1 + d'_1 + d_2 + d'_2 + d_3 + d'_3}\\
& = (c' \epsilon F_p)^{2d} \sum_{d_1+d_2 + d_3 = d}  \sum_{d'_1+d'_2 + d'_3 = d} \binom{d}{d_1 d_2 d_3}  \binom{d}{d'_1  d'_2 d'_3} \\
& = (c' \epsilon F_p)^{2d} (1+1+1)^d (1+1+1)^d \\
& = (c \epsilon F_p)^{2d}
\end{align*}
for some (other) constant $c$.
\end{proof}

\begin{lemma} \label{lem:final}
For $\log (1/\delta)= \omega(\log n)$, $0 < \epsilon < 1/2$,  $C  \ge L n^{1-2/p} \epsilon^{-4/p} \log^{2/p}(1/\delta)$ and $H_J = L' n^{1-2/p} \epsilon^{-2}$, for suitable constants $L,L'$, then,
\begin{align*}
\prob{ \card{ \hat{F}_p - F_p} \le (\epsilon F_p)/2} \le  \delta/n^{\Theta(1)} \enspace.
\end{align*}
\end{lemma}

\begin{proof}
In the statement of Lemma~\ref{lem:dcentmomsmalldelta}, let $\epsilon' = \epsilon/(8c c')$, where, $c' \ge 1$ is a constant to be chosen below, and choose $L$ and $L'$ appropriately, by increasing them by a proportionate  constant factor. Then,  by Lemma~\ref{lem:dcentmomsmalldelta},  for
$d = \lceil \log (1/\delta)$,
\begin{align*}
\expect{ \left( \sum_{i \in [n]} (Y_i - \expect{Y_i \mid \H}) \right)^d \left( \sum_{i \in [n]} ( \conj{Y_i} - \expect{\conj{Y_i} \mid \H})\right)^d \mid \H} \le \left( (\epsilon/(8c')) F_p \right)^{2d}
\end{align*}
Therefore,
\begin{align*}
 \prob{ \left\lvert \sum_{i \in [n]} (Y_i - \expect{Y_i \mid \H}) \right\rvert > (\epsilon/2) F_p \mid \H}
&\le \frac{ \expect{ \left\vert \sum_{i \in [n]} (Y_i - \expect{Y_i \mid \H}) \right\rvert^{2d}\mid \H}}{((\epsilon/2) F_p)^{2d}}\\ & \le (4c')^{-2\log (1/\delta)}\\  & = \delta^{4+\log c'}
\end{align*}
From definition,  $ \expect{Y_i \mid H} = \abs{f_i}^p$, and therefore, $\expect{ \sum_{i \in [n]} Y_i \mid \H} = F_p$. Therefore,
\begin{align*}
\prob{ \abs{\hat{F}_p -F_p} > (\epsilon/2) F_p \mid \H} \le \delta^{4 + \log c'} \enspace .
\end{align*}
Now,
\begin{align} \label{eq:ubfinal:1}
&\prob{ \abs{\hat{F}_p - F_p} >  \epsilon F_p} \notag \\
&  \le  \prob{\abs{\hat{F}_p - F_p} > (\epsilon/2) F_p\mid \H} \prob{\H} +   \prob{\abs{\hat{F}_p - F_p} > \epsilon F_p\mid \neg \H} \prob{\neg \H}  \notag \\
& \le \delta^4 (1 - n^{-\Theta(1)}) + \prob{\abs{\hat{F}_p - F_p} > \epsilon F_p\mid \neg \H} \prob{\neg \H}
\end{align}
As shown in Lemma~\ref{lem:error1}, the unaccounted error term, in addition to the error term conditional under $\H$ is  $ \text{Error}^{\ghss} + \text{Error}^{\shelf} =  O(\epsilon^{\prime 2} F_p/\log (n) + O(\epsilon^{\prime p} F_p) \le 2\epsilon^{'2} $, with probability $1-\delta$. Let $c' = 4$.
Thus, conditional on $\H$ failing, the error increases
to  $$ (\epsilon/2) F_p + 2\epsilon^{'2} F_p \le  \epsilon F_p$$
with probability $1- \delta^{O(1)}$ by Lemma~\ref{lem:error1}. Hence, $\prob{\abs{\hat{F}_p - F_p} >  \epsilon F_p\mid \neg \H} \le \delta/n^{\Theta(1)}$. Combining with Eqn.~\eqref{eq:ubfinal:1}, we have,
\begin{align*}
&\prob{ \abs{\hat{F}_p - F_p} >  \epsilon F_p} \\
& \le \delta^4 (1 - n^{-\Theta(1)}) + \prob{\abs{\hat{F}_p - F_p} > \epsilon F_p\mid \neg \H} \prob{\neg \H}\\
& \le \delta^4 (1 - n^{-\Theta(1)}) + \frac{\delta}{n^{\Theta(1)}}  n^{-O(1)} \\
& \le \frac{\delta}{n^{\Theta(1)}} , ~~~~ \text{ since, $\delta = n^{-\omega(1)}$}\enspace .
\end{align*}

\end{proof}

\section{Full Version of Our Lower Bounds}\label{sec:lbLong}
\subsection{The $\Omega(n^{1-2/p} \epsilon^{-2} \log(1/\delta))$ Measurement Lower Bound}\label{sec:firstlb}
We consider the problem of designing a distribution $\mu$ over $r \times n$ matrices $S$ so that for any
fixed vector $x \in \mathbb{R}^n$, from $S \cdot x$, one can estimate $\|x\|_p^p = \sum_{i=1}^n |x_i|^p$
up to a $(1 \pm \epsilon)$ factor with probability
at least $1-\delta$. We say the algorithm $(\epsilon, \delta)$-approximates $\|x\|_p^p$.
We start by showing a lower bound on $r = \Omega(n^{1-2/p}\epsilon^{-2} \log(1/\delta))$, assuming that
$\log(1/\delta) \leq C \eps^2 n^{2/p}$, for a sufficiently small constant $C > 0$. This assumption is necessary up
to constant factors, as otherwise one can always just take $S$ to be the $n \times n$ identity matrix.

By Yao's minimax principle, it
suffices to design two distributions $\alpha$ and $\beta$ on $\mathbb{R}^n$ so that
\begin{enumerate}
\item For all $y$ in the support of $\alpha$, and all $z$ in the support of $\beta$, $\|y\|_p^p \leq (1-\epsilon) \|z\|_p^p$.
\item There is a constant $\kappa > 0$ so that
  for any fixed matrix $S \in \mathbb{R}^{r \times n}$ with $r = \kappa n^{1-2/p} \epsilon^{-2} \log(1/\delta)$,
  $D_{TV}(\bar{\alpha}, \bar{\beta}) < 1-\delta$, where $\bar{\alpha}$ is the distribution of $S x$ for $x \sim \alpha$ and
  $\bar{\beta}$ is the distribution of $S x$ for $x \sim \beta$.
\end{enumerate}

\subsubsection{Preliminaries}
We use  the following lemma (re-stated from Section~\ref{sec:lb:intro}) concerning distances between multivariate Gaussian distributions.
\begin{lemma}[Re-statement of Lemma~\ref{lem:gtvd}.]
  Let $P_1$ denote the $N(0, I_r)$ Gaussian distribution, and $P_2$ the
  $N(\tau, I_r)$ Gaussian distribution. Then
  $$D_{TV}(P_1, P_2) = \Pr[|N(0,1)| \leq \|\tau\|_2/2],$$
  where $N(0,1)$ denotes a standard one-dimensional normal random variable.
\end{lemma}
\begin{proof}
  Let $U$ be an arbitrary $r \times r$ orthonormal matrix.
  Note that $D_{TV}(P_1, P_2) = D_{TV}(P_1', P_2')$, where $P_i'$ is the distribution
  of $Ux$, where $x \sim P_i$. Let $U$ be any such matrix which rotates $\tau$ to $\|\tau\|_2 \cdot e_1$,
  where $e_1$ is the first standard unit vector. Then by rotational invariance,
  $D_{TV}(P_1, P_2) = D_{TV}(N(0,1), N(\|\tau\|_2, 1))$. Then by Section 3 of \cite{tvd},
  $$D_{TV}(N(0,1), N(\|\tau\|_2, 1) = \Pr[|N(0,1)| \leq \|\tau\|_2/2].$$
\end{proof}
We also need a lemma concerning concentration of $\|x\|_p$ for $x \in N(0, I_n)$.
\begin{lemma}\label{lem:concentrate}
  For $x \in N(0, I_n)$, for all $t \geq 0$ we have
  $$\Pr[|\|x\|_p - {\bf E}[\|x\|_p]| \geq t] \leq 2e^{-\frac{t^2}{2}}.$$
\end{lemma}
\begin{proof}
  Say that a function $f:\mathbb{R}^n \rightarrow \mathbb{R}$ is $L$-Lipshitz with respect to the Euclidean norm
  if $|f(x)-f(y)| \leq L\|x-y\|_2$ for all $x, y \in \mathbb{R}^n$.
  We invoke the following standard theorem on $L$-Lipshitz functions with respect to the Euclidean norm
  \begin{theorem}(see, e.g., Theorem 2.4 of \cite{w16})\label{thm:concentration}
    Let $x \sim N(0, I_n)$ and let $f$ be $L$-Lipshitz with respect to the Euclidean norm. Then
    $$\Pr[|f(x) - {\bf E}[f(x)]| \geq t] \leq 2 e^{-\frac{t^2}{2L^2}},$$
    for all $t \geq 0$.
  \end{theorem}
  Let $f(x) = \|x\|_p$. Then for $x, y \in \mathbb{R}^n$,
  \begin{eqnarray*}
    |\|x\|_p -\|y\|_p| & \leq & \|x-y\|_p \leq \|x-y\|_2,
  \end{eqnarray*}
  where the first inequality is the triangle inequality, and the second uses that $\|z\|_p \leq \|z\|_2$ for $p \geq 2$ and
  any vector $z \in \mathbb{R}^n$. Hence, $f$ is $1$-Lipshitz with respect to the Euclidean norm, and applying Theorem \ref{thm:concentration},
  $$\Pr[|\|x\|_p - {\bf E}[\|x\|_p]| \geq t] \leq 2^{-\frac{t^2}{2}}.$$
\end{proof}
Since ${\bf E}_{X \sim N(0,1)} [|X|^p]$ is a positive constant for any constant $p$, we have that ${\bf E}[\|x\|_p^p] = \Theta(n)$
and ${\bf Var}[\|x\|_p^p] = O(n)$. Consequently, by Chebyshev's inequality there is an absolute constant $\kappa > 0$
so that with probability at least $1/2$, it holds that $\frac{n}{\kappa} \leq \|x\|_p^p \leq \kappa n$.
Therefore because of the tail bounds given in
Lemma \ref{lem:concentrate}, it follows that
${\bf E}_{x \sim N(0, I_n)} [\|x\|_p] = \Theta(n^{1/p})$. We let $E_n = {\bf E}_{x \sim N(0, I_n)} [\|x\|_p]$.

\subsubsection{The Hard Distribution}\label{sec:dist}
{\bf Case 1:}
Suppose $y = n^{1/p}e_I + x$, where $x \sim N(0, I_n)$ and $I$ is independent and uniformly random in $[n] = \{1, 2, \ldots, n\}$.
Let $\alpha'$ be the distribution of $y$.
Then $\|y\|_p^p = (n^{1/p} e_I - x_I)^p + \|\bar{x}\|_p^p$, where $\bar{x}$ denotes the vector $x$ with the $I$-th coordinate
removed. Let $\mathcal{E}$
be the event that $|x_I| \leq \eps n^{1/p}/p$. By independence of $I$ and
$x$, we have
\begin{eqnarray*}
  \Pr_{I, x}[\mathcal{E}] & = & (1/n) \sum_i \Pr_x[\mathcal{E} \mid I = i]\\
  & = & (1/n) \sum_i \Pr_x[|x_i| \leq \epsilon n^{1/p}/p].
  \end{eqnarray*}
By tail bounds for a standard normal random variable,
$$\Pr[\mathcal{E}] \geq 1 - 2 e^{-\frac{\eps^2 n^{2/p}}{2 \cdot p^2}} \geq 1 - \frac{\delta}{10},$$
provided $2e^{-\frac{\eps^2 n^{2/p}}{2 \cdot p^2}} \leq \delta/10$, which holds if $\log(1/\delta) \leq C \eps^2 n^{2/p}$, for a
sufficiently small constant $C > 0$,
as we have assumed. Conditioned on $\mathcal{E}$,
\begin{eqnarray*}
  \|y\|_p^p & \leq & (n^{1/p} (1+\epsilon/p))^p + \|\bar{x}\|_p^p \leq n (1+2\epsilon) + \|\bar{x}\|_p^p,
\end{eqnarray*}
using that $(1+x)^p \leq 1+2px$ for $px \leq 1/2$. Let $\mathcal{F}$ be the event that
$\|\bar{x}\|_p \leq E_{n-1} + 10\sqrt{\log(1/\delta)}$. By Lemma \ref{lem:concentrate}, $\Pr[\mathcal{F}] \geq 1-\delta/10$. Note also
that $10 \sqrt{\log(1/\delta)} \leq 10 \sqrt{C} \eps n^{1/p}$, under our assumption on $\log(1/\delta)$, and for a sufficiently small
constant $C > 0$,
this is at most $(\eps/p)E_{n-1}$, using that $E_{n-1} = \Theta(n^{1/p})$. Hence, conditioned on $\mathcal{E}$ and $\mathcal{F}$,
\begin{eqnarray}\label{eqn:lower}
\|y\|_p^p \leq n(1+2\epsilon) + E_{n-1}^p(1+\epsilon/p)^p = (1+2\epsilon) \cdot (n + E_{n-1}^p).
\end{eqnarray}

{\bf Case 2:}
Now let $z = (1+C'\epsilon)n^{1/p} e_I + x$, where $x \sim N(0,I_n)$ and $I$ is independent and uniformly random in $[n]$.
Here $C' > 0$ be a sufficiently large constant. Let $\beta'$ be the distribution of $z$.
Then $\|z\|_p^p = ((1+C'\epsilon)n^{1/p} e_I + x_I)^p + \|\bar{x}\|_p^p$.
Let $\mathcal{E}'$ be the event that $|x_I| \leq \eps n^{1/p}/p$. As for the event $\mathcal{E}$,
we have $\Pr[\mathcal{E}'] \geq 1-\delta/10$. Also, let
$\mathcal{F}'$ be the event hat $\|\bar{x}\|_p \geq E_{n-1} - 10\sqrt{\log(1/\delta)}$. As in the previous paragraph, we have $\Pr[\mathcal{F}'] \geq 1-\delta/10$
and conditioned on $\mathcal{E}'$ and $\mathcal{F}'$, that
$$\|z\|_p^p \geq (1+C' \epsilon) n(1-\epsilon/p)^p + E_{n-1}^p (1-\epsilon/p)^p \geq (1+C'\epsilon/2) n + (1-\epsilon)E_{n-1}^p,$$
where here we use Bernoulli's inequality that $(1+x)^p \geq (1+xp)$ for $p \geq -1$, and that $C' > 0$ is sufficiently large.
As argued above, $E_{n-1}^p = O(n)$, and consequently for large enough $C'$, we have that
\begin{eqnarray}\label{eqn:upper}
  \|z\|_p^p \geq (1+4 \epsilon)(n + E_{n-1}^p).
\end{eqnarray}

\subsubsection{Conditioning the Cases:}\label{sec:firstlbReuse}
Let $\alpha$ be the distribution of $\alpha'$ conditioned on $\mathcal{E}$ and $\mathcal{F}$. Similarly, let $\beta$ be the distribution of
$\beta'$ conditioned on $\mathcal{E}'$ and $\mathcal{F}'$. Recall that for a distribution $\gamma$ on $x$, the distribution $\bar{\gamma}$ is the
distribution of $Sx$.

Combining (\ref{eqn:lower}) and (\ref{eqn:upper}) it follows that any algorithm which can $(\epsilon, \delta)$-approximate $\|x\|_p^p$
of an arbitrary vector $x$ can also be used to decide, with probability at least $1-\delta$, if $x$ is drawn from $\alpha$ or if
$x$ is drawn from $\beta$. Consequently, $D_{TV}(\bar{\alpha}, \bar{\beta}) \geq 1-\delta$.

On the other hand, we have
\begin{eqnarray*}
  D_{TV}(\bar{\alpha}, \bar{\beta}) & \leq & D_{TV}(\bar{\alpha'}, \bar{\beta'}) + D_{TV}(\bar{\alpha'}, \bar{\alpha}) + D_{TV}(\bar{\beta}, \bar{\beta'})\\
  & \leq & D_{TV}(\bar{\alpha'}, \bar{\beta'}) + D_{TV}(\alpha, \alpha') + D_{TV}(\beta, \beta')\\
  & \leq & D_{TV}(\bar{\alpha'}, \bar{\beta'}) + \delta/10 + \delta/10\\
  & \leq & D_{TV}(\bar{\alpha'}, \bar{\beta'}) + \delta/5,
\end{eqnarray*}
where the first inequality is the triangle inequality, the second uses the data processing inequality, and the third uses
that $E \cap F$ and $E' \cap F'$ are the events that realize the variation distance in the two cases.
Therefore, to obtain our lower bound, it suffices to show for small sketching dimension $r$
that $D_{TV}(\bar{\alpha'}, \bar{\beta'}) < 1-2\delta$.

\subsubsection{Completing the Argument}
Let us fix an $r \times n$ matrix $S$. Without loss of generality we can assume the rows of $S$ are orthonormal, since we can always
perform a change of basis to its rowspace to make the rows orthonormal, and any change of basis preserves $D_{TV}(Sx, Sy)$.

We let $\bar{\alpha'}_i$ denote the distribution of $\bar{\alpha'}$ conditioned on $I = i$, and similarly define
$\bar{\beta'}_i$. Then
\begin{eqnarray}\label{eqn:contradict}
  D_{TV}(\bar{\alpha'}, \bar{\beta'}) & = & D_{TV}(\frac{1}{n} \sum_i \bar{\alpha'}_i, \frac{1}{n} \sum_i \bar{\beta'}_i)
  \leq \frac{1}{n} \sum_i D_{TV}(\bar{\alpha'}_i, \bar{\beta'}_i).
\end{eqnarray}
To complete the argument, as argued in Section \ref{sec:dist}, we just need to show that $ \frac{1}{n} \sum_i D_{TV}(\bar{\alpha'}_i, \bar{\beta'}_i)$ is at most $1-2\delta$.
Suppose, towards a contradiction, that it were larger than $1-2\delta$. Since it is an average of $n$ summands, each bounded by $1$,
it implies that for at least $1-3\delta$ fraction of summands, the summand value is at least $1-3\delta$. Indeed, otherwise the summation
would be at most
$(1-3\delta)^2 + 3\delta \leq 1-3\delta + 9\delta^2 < 1-2\delta$, where the final inequality follows for $\delta$ smaller than
a small enough constant. This is a contradiction.

On the other hand, since $\|S\|_F^2 = r$, by averaging, for at least $n/2$ columns $S_i$ of $S$, we have $\|S_i\|_2^2 \leq 2r/n$. Since $\delta$
is smaller than a small enough constant, we can assume $1-3\delta \geq 2/3$, and therefore by a union bound there exists an $i^* \in [n]$
for which both (1) $\|S_{i^*}\|_2^2 \leq 2r/n$ and (2) $D_{TV}(\bar{\alpha'}_{i^*}, \bar{\beta'}_{i^*}) \geq 1-3\delta$.

Since $S$ has orthonormal rows, $\bar{\alpha'}_{i^*} \sim N(n^{1/p}S_{i^*}, I_r)$ while $\bar{\beta'}_{i^*} \sim N((1+C'\epsilon)n^{1/p}S_{i^*}, I_r)$.
Shifting both distributions by $n^{1/p} S_{i^*}$, it follows that
$$D_{TV}(\bar{\alpha'}_{i^*}, \bar{\beta'}_{i^*}) = D_{TV}(N(0, I_r), N(C'\epsilon n^{1/p} S_{i^*}, I_R)).$$
Applying Lemma \ref{lem:gtvd},
\begin{eqnarray*}
  D_{TV}(N(0, I_r), N(C'\epsilon n^{1/p} S_{i^*}, I_R) & = & \Pr[|N(0,1)| \leq C' \epsilon n^{1/p} \|S_{i^*}\|_2/2]\\
  & \leq & \Pr[|N(0,1)| \leq C' \epsilon n^{1/p} \sqrt{2r/n}/2].
\end{eqnarray*}
We can assume $C' \epsilon n^{1/p} \sqrt{2r/n}/2 \geq 2$, as otherwise this probability is $1- \Omega(1)$, which is a contradiction
to it being at least $1-3\delta$ for small enough constant $\delta > 0$. A standard bound \cite{d10} is then
that for $t \geq 2$,
$$\Pr[|N(0,1)| > t] \geq \frac{e^{-t^2/2}}{\sqrt{2 \pi}} \cdot \frac{1}{2t}.$$
Consequently, we have
\begin{eqnarray*}
  \Pr[|N(0,1)| \leq C' \epsilon n^{1/p} \sqrt{2r/n}/2]
  & \leq & 1 - 2 \frac{e^{-(C')^2 \epsilon^2 n^{2/p} r/(2n)}}{\sqrt{2\pi}} \cdot \frac{1}{C' \epsilon n^{1/p} \sqrt{2r/n}}.
\end{eqnarray*}
It follows that if $r = o(n^{1-2/p} \epsilon^{-2} \log(1/\delta)$, this probability is strictly less than $1-3\delta$.

Let us recap the argument. We found an $i^*$ with two properties:
(1) $\|S_{i_*}|_2^2 \leq 2r/n$, and
(2) $D_{TV}(\bar{\alpha'}_{i_*}, \bar{\beta'}_{i_*}) \geq 1-3\delta$.
Using (1), we were able to apply Lemma \label{lem:gtvd} to upper bound $D_{TV}(\bar{\alpha'}_{i^*}, \bar{\beta'}_{i^*})$
by a quantity that was strictly less than $1-3\delta$, thereby contradicting (2). It follows that there cannot
exist an $i^*$, which means that our hypothesis in (\ref{eqn:contradict}) did not hold. Consequently,
$D_{TV}(\bar{\alpha'}, \bar{\beta'}) = D_{TV}(\frac{1}{n} \sum_i \bar{\alpha'}_i, \frac{1}{n} \sum_i \bar{\beta'}_i) \leq \frac{1}{n} \sum_i D_{TV}(\bar{\alpha'}_i, \bar{\beta'}_i) \leq 1-2\delta$, and so by the argument in
Section \ref{sec:firstlbReuse} the proof is complete.

\subsection{The $\Omega(n^{1-2/p} \eps^{-2/p} (\log^{2/p} 1/\delta) \log n)$ Measurement Lower Bound}
We can assume $\log(1/\delta) \leq C \eps^2 n^{2/p}$ for a small enough constant $C > 0$, as otherwise
the lower bound we proved in Section \ref{sec:firstlb} is $\Omega(n)$, and one can always let $S$ be the $n \times n$ identity matrix
which would be optimal up to a constant factor in this regime. We refer to this as $\delta$-{\bf Bound1}. Note also that since $p > 2$,
this implies $n = \omega(\log(1/\delta))$, which implies $E_{n-t} = \Theta(n^{1/p})$ whenever $t = O(\log(1/\delta))$. We use this
fact later.

Furthermore, we can give two other bounds on $\delta$, similar to $\delta$-{\bf Bound1}.

First, we can assume that $n^{1-2/p} \eps^{-2/p} (\log^{2/p} 1/\delta) \log n = \Omega(n^{1-2/p} \epsilon^{-2} \log 1/\delta)$,
as otherwise the lower bound in Section \ref{sec:firstlb} is stronger.
This is equivalent to assuming $\log(1/\delta) \leq C \epsilon^{2-2/p} (\log^{2/p} 1/\delta) \log n$ for a sufficiently small constant $C > 0$.
 We refer to this as $\delta$-{\bf Bound2}.

Second, we can assume $n^{1-2/p} \eps^{-2/p} (\log^{2/p} 1/\delta) \log n \leq Cn$ for a sufficiently small constant $C > 0$, as otherwise the lower bound
we are proving is $\Omega(n)$ and one can always just let $S$ be the $n \times n$ identity matrix which would be optimal up to a constant
factor in this regime. This assumption is equivalent to $\log^{2/p} 1/\delta \leq C \eps^{2/p} n^{2/p}/\log n$. We refer to this as $\delta$-{\bf Bound3}.

In fact, we will need to assume $\delta$-{\bf Bound4}, which is that $\log(1/\delta) \leq (n^{1-2/p} \eps^{-2/p} (\log^{2/p} 1/\delta) \log n)^{1/4} n^{-c'}$,
for a sufficiently small constant $c' > 0$. Since $p > 2$ is an absolute constant, independent of $n$, this just states that $\delta \geq 2^{-n^{c''}}$ for a sufficiently small constant $c'' > 0$.
We note that $\delta$-{\bf Bound4} implies some of the bounds above, but we state it separately since unlike the previous three bounds on $\delta$, which
are optimal, it may be possible to relax this one for a larger constant $c'' > 0$.


\subsubsection{Preliminaries}
Let $p$ and $q$ be probability density functions of continuous distributions. The $\chi^2$-divergence from $p$ to $q$ is
$$\chi^2(p, q) = \int_x \left (\frac{p(x)}{q(x)} - 1 \right )^2 q(x) dx.$$

\begin{fact}(\cite{Tsybakov}, p.90) \label{chitvd}
For any two distributions $p$ and $q$, we have $D_{TV}(p, q) \leq \sqrt{\chi^2(p,q)}$.
  \end{fact}
We need a fact about the distance between a Gaussian location mixture to a Gaussian distribution.
\begin{fact}(p.97 of \cite{is03})\label{mixture}
  Let $p$ be a distribution on $\mathbb{R}^n$. Then
  $$\chi^2(N(0,I_n) \ast p, N(0,I_n)) = {\bf E}[e^{\langle X, X' \rangle}] - 1,$$
  where $X$ and $X'$ are independently drawn from $p$.
\end{fact}

\subsubsection{The Hard Distribution}
Let $T$ be a sample of $t \eqdef \log_3(1/\sqrt{\delta})$ coordinates $i \in [n]$ without replacement.

{\bf Case 1:}
Suppose $y \sim N(0,I_n)$, and let $\alpha'$ be the distribution of $y$.
Let $\bar{y}$ denote the vector $y$ with the coordinates in the set $T$ removed, and let $x$ be $y$ restricted to the coordinates in $T$.
By the triangle inequality, $\|y\|_p \leq \|\bar{y}\|_p + \|x\|_p$.
Let $\mathcal{E}$ be the event that $\|\bar{y}\|_p \leq E_{n-t} + 10 \sqrt{\log(1/\delta)}$.
By Lemma \ref{lem:concentrate}, $\Pr[\mathcal{E}] \geq 1-\delta/10$.
Let $\mathcal{F}$ be the event that $\|x\|_p \leq E_{t} + 10 \sqrt{\log(1/\delta)}$.
Again by Lemma \ref{lem:concentrate}, $\Pr[\mathcal{F}] \geq 1-\delta/10$.
Note also that $E_t = \Theta(t^{1/p}) \leq \sqrt{\log(1/\delta)}$ using the definition of $t$ and that $p > 2$.
Consequently, if $\mathcal{E}$ and $\mathcal{F}$ occur, then
$$\|y\|_p \leq E_{n-t} + 21 \sqrt{\log(1/\delta)}.$$
By $\delta$-{\bf Bound1}, we have
$\sqrt{\log(1/\delta)} \leq C^{1/2} \eps n^{1/p}$ for a sufficiently small constant $C > 0$. This implies that
$21 \sqrt{\log(1/\delta)} \leq (\epsilon/p) E_{n-t}$, using that $E_{n-t} = \Theta(n^{1/p})$.
Consequently, conditioned on $\mathcal{E}$ and $\mathcal{F}$, we have that
\begin{eqnarray*}
  \|y\|_p^p & \leq & (1+\epsilon/p)^p E_{n-t}^p\\
  & \leq & (1+2\epsilon) E_{n-t}^p,
\end{eqnarray*}
where the second inequality follows for $\epsilon$ less than a sufficiently small positive constant.

{\bf Case 2:}
Let $z = x + \sum_{i \in T} \frac{C' \epsilon^{1/p} E_{n-t}}{t^{1/p}} e_i,$ where $x \sim N(0, I_n)$.
Note that $x$ and $T$ are independent. Also, $C' > 0$ is a sufficiently large constant.
Let $\beta'$ be the distribution of $z$. Let $\mathcal{E}'$ be the event
that for all $i \in T$, $|x_i| \leq \epsilon^{1/p} n^{1/p}/(p t^{1/p})$. Then
$$\Pr[\mathcal{E}'] \geq 1 - t \cdot 2 e^{-\frac{\epsilon^{2/p} n^{2/p}}{2 p^2 t^{2/p}}} \geq 1 - \frac{\delta}{10},$$
which holds provided $\log(1/\delta) \leq C \frac{\epsilon^{2/p} n^{2/p}}{\log^{2/p} 1/\delta}$ for a sufficiently small constant $C > 0$. To see that
the latter holds, note that it is equivalent to the constraint that
$\log^{1-2/p}(1/\delta) \leq C \epsilon^{2/p} n^{2/p}$. To see that this latter constraint holds, observe
\begin{eqnarray*}
  \log^{1-2/p}(1/\delta) & \leq & \log(1/\delta)\\
  & \leq & C \epsilon^{2-2/p} \log^{2/p} (1/\delta) \log n\\
  & \leq & C^2 \epsilon^{2-2/p} n^{2/p} \epsilon^{2/p}\\
  & = & C^2 \epsilon^2 n^{2/p}\\
  & \leq & C^2 \epsilon^{2/p} n^{2/p},
  \end{eqnarray*}
where the first inequality follows since $p > 0$, the second inequality follows from $\delta$-{\bf Bound2}, the third inequality
follows from $\delta$-{\bf Bound3}, and the final inequality holds for any $p \geq 1$. Thus, the above constraint holds, since $C^2 > 0$ can be
made sufficiently small.

Conditioned on $\mathcal{E}'$, and using that $E_{n-t} = \Theta(n^{1/p})$ and $C' > 0$ is a sufficiently large constant, we have
$$\|z\|_p^p \geq t \cdot \frac{\epsilon C'' E_{n-t}^p }{t} \left (1-1/p \right )^p + \|\bar{x}\|_p^p,$$
where $\bar{x}$ denotes $x$ with coordinates $i \in T$ removed, and where $C'' > 0$ can be made an arbitrarily large constant, provided
$C' > 0$ is a sufficiently large constant. Let $\mathcal{F}'$ be the event that
$\|\bar{x}\|_p \geq E_{n-t} - 10\sqrt{\log(1/\delta)}$. By Lemma \ref{lem:concentrate}, $\Pr[\mathcal{F}'] \geq 1-\delta/10$.
By $\delta$-{\bf Bound1}, we have
$\sqrt{\log(1/\delta)} \leq C^{1/2} \eps n^{1/p}$ for a sufficiently small constant $C > 0$,
which implies $10 \sqrt{\log(1/\delta)} \leq (\epsilon/p) E_{n-t}$, using
that $E_{n-t} = \Theta(n^{1/p})$. Consequently,
conditioned on $\mathcal{E}'$ and $\mathcal{F}'$, we have
\begin{eqnarray*}
  \|z\|_p^p & \geq & \epsilon C'' E_{n-t}^p (1-1/p)^p + (1-\epsilon/p)^p E_{n-t}^p\\
            & \geq & (1+4\epsilon) E_{n-t}^p,
\end{eqnarray*}
where the second inequality follows for $C''$ a sufficiently large constant.

\subsubsection{Conditioning the Cases:}
The argument here is the same as in Section \ref{sec:firstlbReuse}.

Let $\alpha$ be the distribution of $\alpha'$ conditioned on $\mathcal{E}$ and $\mathcal{F}$. Similarly, let $\beta$ be the distribution of
$\beta'$ conditioned on $\mathcal{E}'$ and $\mathcal{F}'$. Recall that for a distribution $\gamma$ on $x$, the distribution $\bar{\gamma}$ is the
distribution of $Sx$.

Combining (\ref{eqn:lower}) and (\ref{eqn:upper}) it follows that any algorithm which can $(\epsilon, \delta)$-approximate $\|x\|_p^p$
of an arbitrary vector $x$ can also be used to decide, with probability at least $1-\delta$, if $x$ is drawn from $\alpha$ or if
$x$ is drawn from $\beta$. Consequently, $D_{TV}(\bar{\alpha}, \bar{\beta}) \geq 1-\delta$.

On the other hand, we have
\begin{eqnarray*}
  D_{TV}(\bar{\alpha}, \bar{\beta}) & \leq & D_{TV}(\bar{\alpha'}, \bar{\beta'}) + D_{TV}(\bar{\alpha'}, \bar{\alpha}) + D_{TV}(\bar{\beta}, \bar{\beta'})\\
  & \leq & D_{TV}(\bar{\alpha'}, \bar{\beta'}) + D_{TV}(\alpha, \alpha') + D_{TV}(\beta, \beta')\\
  & \leq & D_{TV}(\bar{\alpha'}, \bar{\beta'}) + \delta/10 + \delta/10\\
  & \leq & D_{TV}(\bar{\alpha'}, \bar{\beta'}) + \delta/5,
\end{eqnarray*}
where the first inequality is the triangle inequality, the second uses the data processing inequality, and the third uses
that $E \cap F$ and $E' \cap F'$ are the events that realize the variation distance in the two cases.

It follows that $D_{TV}(\bar{\alpha'}, \bar{\beta'}) \geq 1-\delta -\delta/5 > 1-2\delta$.

\subsubsection{A Further Useful Conditioning}
Fix an $r \times n$ matrix $S$ with orthonormal rows. Important to our proof will be the existence of a
subset $W$ of $n/2$ of the columns for which $\|S_i\|^2 \leq 2r/n$ for all $i \in W$.
To see that $W$ exists,
consider a uniformly random column $S_i$ for $i \in [n]$. Then ${\bf E}[\|S_i\|^2] = r/n$ and so by Markov's inequality, at least a $1/2$-fraction
of columns $S_i$ satisfy $\|S_i\|^2 \leq 2r/n$. We fix $W$ to be an arbitrary subset of $n/2$ of these columns.

Suppose we sample $t$ columns of $S$ without replacement, indexed by $T \subset [n]$.
Let $\mathcal{G}$ be the event that the set $T$ of sampled columns belongs to the set $W$.

\begin{lemma}\label{lemma:condition}
  $\Pr[\mathcal{G}] \geq \sqrt{\delta}$.
\end{lemma}
\begin{proof}
The probability that $T \subset W$ is equal to
$$\frac{|W|}{n} \cdot \frac{|W|-1}{n-1} \cdots \frac{|W|-\log_3(1/\sqrt{\delta})-1}{n-\log_3(1/\sqrt{\delta})-1}
\geq \left (\frac{|W| - \log_3(1/\sqrt{\delta})}{n-\log_3(1/\sqrt{\delta})} \right )^{\log_3(1/\sqrt{\delta})}
\geq \left (\frac{1}{3} \right )^{\log_3(1/\sqrt{\delta})} \geq \sqrt{\delta},$$
where we have used $t = \log_3(1/\sqrt{\delta}) \leq \frac{n}{4}$ and $|W| =n/2$. Hence, $\Pr[\mathcal{G}] \geq \sqrt{\delta}$.
\end{proof}

Let $\alpha_G = \bar{\alpha'} \mid \mathcal{G}$ and
$\beta_G = \bar{\beta'} \mid \mathcal{G}$. By the triangle inequality,
\begin{eqnarray*}
  1- 2\delta & \leq & D_{TV}(\bar{\alpha'}, \bar{\beta'}) \leq \Pr[\mathcal{G}] D_{TV}(\alpha_g, \beta_G) + 1-\Pr[\mathcal{G}]
  \leq \frac{\sqrt{\delta}}{2} D_{TV}(\alpha_G, \beta_G) + 1-\frac{\sqrt{\delta}}{2},
\end{eqnarray*}
which implies that $1 - 4\sqrt{\delta} \leq D_{TV}(\alpha_G, \beta_G)$. We can assume $\delta$ is less than a sufficiently small positive
constant, and so it suffices to show for sketching dimension $r = o(n^{1-2/p} \eps^{-2/p} (\log^{2/p} 1/\delta) \log n)$, that
$D_{TV}(\alpha_G, \beta_G) \leq 1/2$. By Fact \ref{chitvd}, it suffices to show $\chi^2(\alpha_G, \beta_G) \leq 1/4$.

Since $S$ has orthonormal
rows, $\bar{\alpha'}$ is distributed as $N(0, I_r)$. Note that, by definition of $\alpha$, we in fact have $\bar{\alpha'} = \alpha_G$ since
conditioning on $\mathcal{G}$ does not affect this distribution. On the other hand, $\beta_G$ is a Gaussian location mixture, that is, it has
the form $N(0, I_r) \ast p$, where $p$ is the distribution of a random variable
chosen by sampling a set $T$ subject to event $\mathcal{G}$ occurring
and outputting $\sum_{i \in T} \frac{C' \epsilon^{1/p} E_{n-t} S_i}{t^{1/p}}$.
We can thus apply Fact \ref{mixture} and
it suffices to show for $r = o(n^{1-2/p} \eps^{-2/p} (\log^{2/p} 1/\delta) \log n)$ that
$${\bf E}[e^{\frac{(C')^2 \epsilon^{2/p} E_{n-t}^2}{t^{2/p}} \langle \sum_{i \in T} S_i, \sum_{j \in U} S_j \rangle}] - 1 \leq \frac{1}{4},$$
where the expectation is over independent samples $T$ and $U$ conditioned on $\mathcal{G}$. Note that
under this conditioning $T$ and $U$ are uniformly random subsets of $W$.

\subsubsection{Analyzing the $\chi^2$-Divergence}
To bound the $\chi^2$-divergence, we define variables $x_{T,U}$, where
\begin{eqnarray*}
  x_{T, U} & = & \frac{(C')^2 \epsilon^{2/p} E_{n-t}^2}{t^{2/p}} \langle \sum_{i \in T} S_i, \sum_{j \in U} S_j \rangle.
\end{eqnarray*}
Consider the following, where the expectation is over
independent samples $T$ and $U$ conditioned on $\mathcal{G}$:
\begin{eqnarray*}\label{eqn:exp}
  {\bf E}[e^{\frac{(C')^2 \epsilon^{2/p} E_{n-t}^2}{t^{2/p}} \langle \sum_{i \in T} S_i, \sum_{j \in U} S_j \rangle}] & = & {\bf E} [e^{x_{T,U}}]\\
  &= & \sum_{0 \leq j < \infty} {\bf E} \left [\frac{x_{T,U}^j}{j!} \right ]\\
  &= & 1 + \sum_{j \geq 1} \frac{(C')^{2j} \eps^{2j/p} E_{n-t}^{2j}}{t^{2j/p} j!} {\bf E} \left [\langle \sum_{i \in T} S_i, \sum_{j \in U} S_j \rangle^j \right ]\\
  & = & 1 + \sum_{j \geq 1} \frac{O(1)^{2j} \eps^{2j/p} n^{2j/p}}{t^{2j/p} j!} {\bf E} \left [\langle \sum_{i \in T} S_i, \sum_{j \in U} S_j \rangle^j \right ],
   \end{eqnarray*}
The final equality uses that $E_{n-t} = \Theta(n^{1/p})$ and here
$O(1)^{2j}$ denotes an absolute constant raised to the $2j$-th power.

We can think of $T$ as indexing a subset of rows of $S^T S$ and $U$ indexing a subset of columns.
Let $M$ denote the resulting $t \times t$ submatrix of $S^T S$.
Then $\langle \sum_{i \in T} S_i, \sum_{j \in U} S_j \rangle = \sum_{i, j \in [t]} M_{i,j} \leq \sum_{i, j \in [t]} |M_{i,j}| \eqdef P$, and we seek
to understand the value of ${\bf E}[P^j]$ for integers $j \geq 1$.

\begin{lemma}\label{lem:key}
  For integers $j \geq 1$,
  $${\bf E}[P^j] \leq \left (\frac{t^2}{r^{1/2}} \right ) \cdot \left (\frac{16r}{n} \right )^j.$$
  \end{lemma}
\begin{proof}
 We have,
  \begin{eqnarray}\label{eqn:split}
    {\bf E}[P^j] & = & {\bf E} \left [\sum_{a_1,\ldots, a_j \in T, b_1, \ldots, b_j \in U} \prod_{w = 1}^j |\langle S_{a_w}, S_{b_w} \rangle | \right ]\\ \notag
    & = & \sum_{a_1, \ldots, a_j, b_1, \ldots, b_j \in W} \Pr[a_1, \ldots, a_j \in T] \cdot \Pr[b_1, \ldots, b_j \in U] \cdot \prod_{w = 1}^j |\langle S_{a_w}, S_{b_w} \rangle |,
  \end{eqnarray}
  where recall $W$ is our subset of $n/2$ columns of $S$ which all have squared norm at most $2r/n$.

  To analyze (\ref{eqn:split}), we define a $y$-{\it pattern} $P$ to be a
  partition of $\{1, 2, \ldots, j\}$ into $y$ non-empty parts, where $1 \leq y \leq j$. The
  number $n_y$ of $y$-patterns, that is, the number of partitions of a set of size $j$ into $y$ non-empty parts is
  known to equal ${j-1 \choose y-1}$. Note that $\sum_y n_y = 2^{j-1}$ is the total number of {\it patterns}.

  We partition $W^j$ into patterns, where a particular $j$-tuple
  $(a_1, \ldots, a_j)$ is in some $y$-pattern $P$, for some $1 \leq y \leq j$, if for each non-empty piece
  $\{d_1, \ldots, d_{\ell}\} \in P$, we have $a_{d_1} = a_{d_2} = \cdots = a_{d_{\ell}}$. Moreover, if $d, d' \in \{1, 2, \ldots, j\}$
  are in different pieces of $P$, then $a_d \neq a_{d'}$. If $(a_1, \ldots, a_j)$ is in some $y$-pattern $P$ we say it is {\it valid}
  for $P$. We similarly say a $j$-tuple $(b_1, \ldots, b_j)$ is in some $z$-pattern $Q$, for some $1 \leq z \leq j$, if for each
  non-empty piece $\{e_1, \ldots, e_m\} \in Q$, we have $b_{e_1} = b_{e_2} = \cdots = b_{e_m}$. Moreover, if $e, e' \in \{1, 2, \ldots, j\}$
  are in different pieces of $Q$, then $b_e \neq b_{e'}$. We also say $(b_1, \ldots, b_j)$ is valid for the $z$-pattern $Q$ in this case.

  Thus, each pair of $j$-tuples is valid for exactly one pair $P, Q$ of patterns.

  We show for each
  pair $P, Q$ of patterns, where $P$ is a $y$-pattern for some $y$ and $Q$ is a $z$-pattern for some $z$,
  \begin{eqnarray}\label{eqn:toshow}
    \sum_{a_1, \ldots, a_j \in P \textrm{ and } b_1, \ldots, b_j \in Q} \Pr[a_1, \ldots, a_j \in T] \Pr[b_1, \ldots, b_j \in U] \prod_{w = 1}^j |\langle S_{a_w}, S_{b_w} \rangle |
    \leq \left (\frac{t^2}{r^{1/2}} \right ) \cdot \left (\frac{4r}{n} \right )^j.
  \end{eqnarray}
  Notice the sum is only over pairs of $j$-tuples valid for $P$ and $Q$.
  As the number of pairs $P, Q$ of patterns is at most $4^{j-2}$, the lemma will follow given (\ref{eqn:toshow}).

  We have $\Pr[a_1, \ldots, a_j \in T] = \left (\frac{t}{|W|} \right )^y$ and $\Pr[b_1, \ldots, b_j \in U] = \left (\frac{t}{|W|} \right )^z$. To analyze
  $$\sum_{a_1, \ldots, a_j \in P \textrm{ and } b_1, \ldots, b_j \in Q} \prod_{w = 1}^j |\langle S_{a_w}, S_{b_w} \rangle |,$$ by Cauchy-Schwarz this is at most
  \begin{eqnarray}\label{eqn:cauchySchwarz}
    |W|^{(y+z)/2} \cdot \left (\sum_{a_1, \ldots, a_j \in P \textrm{ and } b_1, \ldots, b_j \in Q} \prod_{w = 1}^j \langle S_{a_w}, S_{b_w} \rangle^2 \right )^{1/2}.
  \end{eqnarray}
  The valid pairs of $j$-tuples for $P$ and $Q$
  define a bipartite multi-graph as follows.
  In the left partition we create a node for each non-empty piece of $P$, and in the right
  partition we create a node for each non-empty piece of $Q$. We include an edge between a node
  $a$ in the left and a node $b$ in the right if $i \in a$ and $i \in b$ for some $i \in \{1, 2, \ldots, j\}$.
  If there is more than one such $i$, we include the edge with multiplicity corresponding to the number of such $i$.
  This bipartite graph only depends on $P$ and $Q$.

  Fix a largest matching $M$ in this multi-graph, meaning that if an edge occurs with multiplicity more than one, it can occur at most once
  in a matching. After
  choosing $M$, greedily construct a set $G$ of edges for which one endpoint is not incident to an edge which is already chosen.
  Finally, let $H$ be the set of remaining edges, each of which has both endpoints incident to an edge already chosen. Write
  $$\prod_{w = 1}^j \langle S_{a_w}, S_{b_w} \rangle^2 = \prod_{\{a,b\} \in M} \langle S_a, S_b \rangle^2
  \cdot \prod_{\{a,b\} \in G} \langle S_a, S_b \rangle^2 \cdot \prod_{\{a,b\} \in H} \langle S_a, S_b \rangle^2.$$
  We bound
  \begin{eqnarray*}
    && \left (\sum_{a_1, \ldots, a_j \in P \textrm{ and } b_1, \ldots, b_j \in Q} \prod_{w = 1}^j \langle S_{a_w}, S_{b_w} \rangle^2 \right )^{1/2}\\
    && = \left (\sum_{a_1, \ldots, a_j \in P \textrm{ and } b_1, \ldots, b_j \in Q} \prod_{\{a,b\} \in M} \langle S_a, S_b \rangle^2
    \cdot \prod_{\{a,b\} \in G} \langle S_a, S_b \rangle^2 \cdot \prod_{\{a,b\} \in H} \langle S_a, S_b \rangle^2 \right )^{1/2}
    \end{eqnarray*}
  We peel off the edges of the bipartite multi-graph constructed above one at a time.

  First, we have $\prod_{\{a,b\} \in H} \langle S_a, S_b \rangle^2 \leq \left (\frac{4r^2}{n^2} \right )^{|H|}$, conditioned on
  $\mathcal{G}$, so that $S_a$ and $S_b$ belong to $W$.

  Next, each $\{a,b\} \in G$ is incident to a vertex which is not incident to any edges chosen before $\{a,b\}$. Suppose w.l.o.g. this vertex is $b$.
  Consider any assignment to the vertices incident to all edges chosen before $\{a,b\}$. For this fixed assignment, we will sum over at most
  $|W|$ possible assignments to the vertex $b$ (note, typically it may be much fewer than $n$ assignments since $b$ can only be assigned
  to an $S_i$ for $i \in W$ if no vertex incident to all edges chosen before $\{a,b\}$ is assigned to $i$). Since the rows of $S$
  are orthonormal, it follows that this sum over assignments to $b$ is at most $\|S_a\|_2^2$,
  which is at most $\frac{2r}{n}$ since $a \in W$.

  Finally, each $\{a,b\} \in M$ is incident to two vertices not incident to any edges chosen before $\{a,b\}$. Consider any assignment to all
  vertices incident to all edges chosen before $\{a,b\}$. For this fixed assignment, we will sum over at most $|W|^2$ possible assignments
  to the vertices $a$ and $b$. This is at most $\sum_{i,j \in [n]} \langle S_i, S_j \rangle^2 = r$, using the fact that the rows of $S$
  are orthonormal.

We thus have:
\begin{eqnarray*}
  && \sum_{a_1, \ldots, a_j \in P \textrm{ and } b_1, \ldots, b_j \in Q} \prod_{(a,b) \in M} \langle S_a, S_b \rangle^2 \cdot \prod_{(a,b) \in G} \langle S_a, S_b \rangle^2 \cdot \prod_{(a,b) \in H} \langle S_a, S_b \rangle^2\\
  & \leq & r^{|M|} \cdot \left (\frac{2r}{n} \right )^{|G|} \cdot \left (\frac{4r^2}{n^2} \right )^{|H|}\\
  & \leq & r^{\min(y,z)} \cdot \left (\frac{2r}{n} \right )^{\max(y,z) - \min(y,z)} \cdot \left (\frac{4r^2}{n^2} \right )^{j-\max(y,z)},
\end{eqnarray*}
where the final inequality uses that $r \geq \frac{2r}{n} \geq \frac{4r^2}{n^2}$ and bounds on the maximum size of $M$ and $G$ given the number of vertices in
the two parts of the bipartite graph. Rearranging, we have
\begin{eqnarray}\label{eqn:blah}
 \sum_{a_1, \ldots, a_j \in P \textrm{ and } b_1, \ldots, b_j \in Q}\prod_{(a,b) \in M} \langle S_a, S_b \rangle^2
  \cdot \prod_{(a,b) \in G} \langle S_a, S_b \rangle^2 \cdot \prod_{(a,b) \in H} \langle S_a, S_b \rangle^2
  & \leq & \frac{4^j r^{2j - \max(y,z)}}{n^{2j-\max(y,z) - \min(y,z)}}
  \end{eqnarray}
Combining (\ref{eqn:blah}) with our earlier (\ref{eqn:cauchySchwarz}) we have
\begin{eqnarray*}
  && \sum_{a_1, \ldots, a_j \in P \textrm{ and } b_1, \ldots, b_j \in Q} \Pr[a_1, \ldots, a_j \in T] \Pr[b_1, \ldots, b_j \in U] \prod_{w = 1}^j |\langle S_{a_w}, S_{b_w} \rangle |\\
  & \leq &  \left (\frac{t}{|W|} \right )^{y+z}  \cdot |W|^{(y+z)/2} \cdot \left (\frac{4^j r^{2j - \max(y,z)}}{n^{2j-\max(y,z) - \min(y,z)}} \right )^{1/2}\\
  & \leq & \frac{2^{j} 2^j t^{y+z} r^{j-\max(y,z)/2}}{n^{j-\max(y,z)/2-\min(y,z)/2 + (y+z)/2}}\\
  & \leq & \left (\frac{t^4}{r} \right )^{(y+z)/4} \cdot \left (\frac{4r}{n} \right )^j\\
  & \leq & \left (\frac{t^2}{r^{1/2}} \right )  \cdot \left (\frac{4r}{n} \right )^j,
\end{eqnarray*}
which establishes (\ref{eqn:toshow}), and completes the proof.
\end{proof}
By $\delta$-{\bf Bound4}, we have $\frac{t^2}{r^{1/2}} = \frac{1}{n^{\Omega(1)}}$, and therefore Lemma \ref{lem:key} establishes that
 $${\bf E}[P^j] \leq \frac{1}{n^{\Omega(1)}} \cdot \left (\frac{16r}{n} \right )^j.$$

We thus have:
\begin{eqnarray*}
  {\bf E}[e^{\frac{(C')^2 \epsilon^{2/p} E_{n-t}^2}{t^{2/p}} \langle \sum_{i \in T} S_i, \sum_{j \in U} S_j \rangle}] & = & {\bf E} [e^{x_{T,U}}]\\
  & = & 1 + \frac{1}{n^{\Omega(1)}} \cdot \sum_{j \geq 1} \frac{O(1)^{2j} \epsilon^{2j/p} n^{2j/p}}{j! t^{2j/p}} \cdot \left (\frac{r}{n} \right )^j\\
  & = & 1 + \frac{1}{n^{\Omega(1)}} \cdot \sum_{j \geq 1} \frac{(c \log n)^j}{j!}\\
  & \leq & 1 + \frac{1}{n^{\Omega(1)}} \cdot e^{c (\log n)}\\
  & \leq & 1 + \frac{1}{4},
\end{eqnarray*}
since $c > 0$ is an arbitrarily small constant independent of the constant in the $n^{\Omega(1)}$. The proof is complete.


\subsection{Lower Bound for $1 \leq p < 2$}\label{sec:smallp}
We prove the following theorem, that for $1 \leq p < 2$, the sketching
dimension is $\Omega(\epsilon^{-2} \log(1/\delta))$, which as discussed
in Section \ref{sec:intro}, matches known upper bounds up to a constant
factor.

\begin{theorem}[Re-statement of Theorem ~\ref{thm:smallp}.]
  The sketching dimension for $(\epsilon,\delta)$-approximating $F_p$
  for $1 \leq p < 2$ is $\Omega(\epsilon^{-2} \log(1/\delta))$.
\end{theorem}
\begin{proof}[Proof of Theorem~\ref{thm:smallp}.]
  By Yao's minimax principle we can fix the sketching matrix
  $S \in \R^{r \times n}$, which
  as discussed we can w.l.o.g. assume has orthonormal rows. We will
  show the lower bound for $n = \Omega(\epsilon^{-2} \log(1/\delta))$.
%
  We will again use Theorem \ref{thm:concentration} and a similar argument
  to that of Lemma \ref{lem:concentrate} and the discussion following it.
  This time we apply Theorem \ref{thm:concentration} to the function
  $f(x) = \|x\|_p$ for $1 \leq p < 2$. By the triangle inequality
  and norm inequalities,
  \begin{eqnarray*}
    |\|x\|_p - \|y\|_p| & \leq & \|x-y\|_p \leq n^{1/p-1/2}\|x-y\|_2,
  \end{eqnarray*}
  and so $f$ is $n^{1/p-1/2}$-Lipshitz with respect to the Euclidean norm.

  Applying Theorem \ref{thm:concentration} for $x \sim N(0, I_n)$,
  \begin{eqnarray}\label{eqn:smallp}
    \Pr[|\|x\|_p - {\bf E}[\|x\|_p]| \geq t] \leq 2^{-\frac{t^2}{2n^{2/p-1}}}.
    \end{eqnarray}
  Since ${\bf E}_{X \sim N(0,1)} [|X|^p]$ is a positive constant for any
  constant $p$, we have ${\bf E}[\|x\|_p^p] = \Theta(n)$ and
  ${\bf Var}[\|x\|_p^p] = O(n)$. Consequently, because of the tail bound
  in (\ref{eqn:smallp}) and Chebyshev's inequality, necessarily we have
  ${\bf E}[\|x\|_p] = \Theta(n^{1/p})$. We let $E_n=  {\bf E}_{x \sim N(0, I_n)} [\|x\|_p]$.
  Consequently, by
  (\ref{eqn:smallp}),
\begin{eqnarray*}
  \Pr_{x \sim N(0, I_n)}[|\|x\|_p - {\bf E}[\|x\|_p]| \geq \epsilon E_n] \leq 2^{-\frac{\Theta(\epsilon^2 n^{2/p})}{n^{2/p-1}}}
     & = & 2^{-\Theta(\epsilon^2 n)} \leq \delta/10,
    \end{eqnarray*}
    where the final inequality uses that $n = \Theta(\epsilon^{-2} \log(1/\delta))$.

     By linearity,
    for $x \sim (1+3\epsilon) \cdot N(0, I_n)$, we have ${\bf E}_{x \sim (1+3\epsilon) \cdot N(0, I_n)} [\|x\|_p] = (1+3\epsilon)E_n$.
    We also have
\begin{eqnarray*}
  \Pr_{x \sim (1+3\epsilon) N(0, I_n)}[|\|x\|_p - {\bf E}[\|x\|_p]| \geq \epsilon (1+3\epsilon)E_n] \leq 2^{-\frac{\Theta(\epsilon^2 n^{2/p})}{n^{2/p-1}}} & = & 2^{-\Theta(\epsilon^2 n)} \leq \delta/10.
    \end{eqnarray*}
        It follows that for both distributions $N(0, I_n)$ and $(1+3\epsilon) \cdot N(0, I_n)$, that a random sample from the distribution
        is within a $(1+\epsilon)$ factor of its expectation with probability $1 - \delta/10$, and consequently, by the same argument as in
        Section \ref{sec:dist}, the distributions of $S \cdot x$ and $S \cdot y$ must have variation distance $1-\Theta(\delta)$, where
        $x \sim N(0,I_n)$ and $y \sim (1+3\epsilon)N(0, I_n)$.

        The key point now though is that the distribution of $S \cdot x$ is equal
        to $N(0, I_r)$, while the distribution of $S \cdot y$ is equal to
        $(1+3\epsilon) N(0,I_r)$, by using the fact that the rows of $S$ are orthonormal
        and the rotational invariance of the Gaussian distribution.

        We use the following variation distance bound (it is standard, but see, e.g.,
        Lemma 22 of \cite{kmv}).
        \begin{fact}\label{fact:kmv}
          For $r$-dimensional distributions $N(0,\Sigma_1)$ and $N(0, \Sigma_2)$, let
          $\lambda_1, \ldots, \lambda_n > 0$
          be the eigenvalues of $\Sigma_1^{-1} \Sigma_2$. Then
          $$D_{TV}(N(0,\Sigma_1), N(0, \Sigma_2))^2 \leq \sum_{i=1}^r \left (\lambda_i + \frac{1}{\lambda_i} - 2 \right )$$
        \end{fact}
        Let us split $r$ into $\Theta(\log(1/\delta) \cdot s$, where $s = C\epsilon^{-2}$ for a sufficiently
        small constant $C > 0$.
        Applying Fact \ref{fact:kmv} to $N(0,I_s)$ and $(1+3\epsilon) N(0, I_s) = N(0, (1+3\epsilon)^2 I_r)$, we have that all eigenvalues
        $\lambda_i$ in Fact \ref{fact:kmv} are $(1 \pm O(\epsilon))$, and $\lambda_i + \frac{1}{\lambda_i} = \Theta(\epsilon^{-2})$, and consequently for $C > 0$ sufficiently small,
        $D_{T,V}(N(0, I_S))$ and $(1+3\epsilon) N(0,I_s)) \leq 1/10$.

        Recall that the squared Hellinger distance $h^2(p,q)$ between distributions $p$ and $q$
        satisfies $h^2(p,q) \leq D_{TV}(p,q)$ (see, e.g., Lemma 2.3 of \cite{ww15}) and so
        $h^2(N(0, I_s), (1+3\epsilon)N(0,I_s)) \leq 1/10$. We also have
        $h^2(p^m, q^m) = 1 - \prod_{i=1}^m (1-h^2(p,q))$, where $p^m$ and $q^m$ denote the distributions
        of $m$ independent samples from $p$ and $q$ respectively (see, e.g., Fact 2.2 of \cite{ww15}).
        Setting $m = r/s = \Theta(\log(1/\delta))$, we have for appropriate choice of constant in the definition of $m$,
        $h^2(N(0,I_r), (1+3\epsilon)N(0,I_r)) \leq 1-(9/10)^m \leq 1-B\sqrt{\delta}$, for a sufficiently
        large constant $B > 0$. But we also have
        $D_{TV}(N(0, I_r), (1+3\epsilon)N(0, I_r)) \leq h \sqrt{2-h^2}$, where
        $h = h(N(0, I_r), (1+3\epsilon)N(0,I_r))$ (see, e.g., Lemma 2.3 of \cite{ww15}). Consequently,
        $D_{TV}(N(0, I_r), (1+3\epsilon)N(0, I_r)) \leq (1-B\sqrt{\delta})^{1/2} \sqrt{2-(1-B\sqrt{\delta})} < 1-\delta$, for $B$ sufficiently large. This is a contradiction and so necessarily $r = \Omega(\epsilon^{-2} \log(1/\delta))$.
\end{proof}

\end{document}